\newif\iffull
\newif\ifnotes
\newif\iflater

\fulltrue
\notesfalse

\latertrue
\iffull
 \documentclass[10pt,letter]{article}
 \usepackage[\ifnotes
 			 draft,
             notes=true,
             \iflater later=true \else later=false \fi
             \else
             notes=false
             \fi]{dtrt}
 \usepackage{fullpage}
 \else
\documentclass[runningheads]{llncs}
 \usepackage[notes=xxx]{dtrt}
\fi
\iffull
	\usepackage{amsthm}
\fi
\usepackage[utf8]{inputenc}
\usepackage{times}
\usepackage{amsmath}
\usepackage{comment}
\usepackage{xspace}
\usepackage{paralist}
\usepackage{enumitem}
  \setlist[itemize]{leftmargin=*}
  \setlist[enumerate]{leftmargin=*}
\usepackage{makecell}
\usepackage{tikz}
	\usetikzlibrary{matrix,shapes,arrows,positioning,chains,calc}
\usepackage{floatrow}
\usepackage{hypcap}
\usepackage{mleftright}
\iffull
	\usepackage[margin=1cm,footnotesize]{caption}
\else
	\usepackage[footnotesize]{caption}
\fi
\usepackage[binary-units]{siunitx}
\usepackage{bbm}
\usepackage{colortbl}
\usepackage{tabu}
\usepackage{multirow}
\usepackage{mathrsfs}
\usepackage{adjustbox}
\usepackage{microtype}
\usepackage{etoolbox}
\usepackage{graphicx}
\usepackage{breakcites}
\usepackage{booktabs}
\usepackage{mdframed}
\usepackage{soul}
\usepackage{amssymb}
\usepackage{changepage}

\floatstyle{boxed}
\restylefloat{figure}
\newcommand{\ale}[1]{\dtcolornote[Ale]{red}{#1}}

\newcommand{\eli}[1]{\dtcolornote[Eli]{orange}{#1}}

\newcommand{\doclearpage}{%
\iffull
\clearpage
\else
\fi
}

\iffull
  \newcommand{\keywords}[1]{\bigskip\par\noindent{\small\textbf{Keywords\/}: #1}}
\fi

\newcommand{\sunderline}[1]{\underline{\smash{#1}}}

\newcommand{\defemph}[1]{\textbf{\emph{#1}}}

\newcommand{\FormatAuthor}[3]{
\begin{tabular}{c}
#1 \\ {\small\texttt{#2}} \\ {\small #3}
\end{tabular}
}

\iffull

\theoremstyle{plain} % italics
\newtheorem{theorem}{Theorem}[section]

\newtheorem{lemma}[theorem]{Lemma}

\newtheorem{claim}[theorem]{Claim}
\newtheorem{corollary}[theorem]{Corollary}
\newtheorem{definition}[theorem]{Definition}

\newtheorem{construction}[theorem]{Construction}

\newtheorem*{uclaim}{Claim}

\theoremstyle{definition} % not italics
\newtheorem{remark}[theorem]{Remark}
\newtheorem{example}[theorem]{Example}

\theoremstyle{remark} % ??

\newtheoremstyle{soundnessstyle} % name
    {\topsep}                    % Space above
    {\topsep}                    % Space below
    {}                           % Body font
    {0em}                        % Indent amount
    {\scshape}                   % Theorem head font
    {.}                          % Punctuation after theorem head
    {.5em}                       % Space after theorem head
    {}  % Theorem head spec (can be left empty, meaning ?normal?)
\theoremstyle{soundnessstyle}

\else

\spnewtheorem{soundnessdefinition}[theorem]{Definition}{\scshape}{\normalfont}
\spnewtheorem{soundnessclaim}[theorem]{Claim}{\scshape}{\normalfont}
\spnewtheorem{adaptivedefinition}[theorem]{Definition}{\scshape}{\normalfont}
\spnewtheorem{adaptiveclaim}[theorem]{Claim}{\scshape}{\normalfont}

\fi%iffull
\newcommand{\secref}[1]{Section~\protect\ref{#1}}
\newcommand{\appref}[1]{Appendix~\protect\ref{#1}}
\newcommand{\defref}[1]{Definition~\protect\ref{#1}}

\newcommand{\conref}[1]{Construction~\protect\ref{#1}}

\newcommand{\thmref}[1]{Theorem~\protect\ref{#1}}

\newcommand{\lemref}[1]{Lemma~\protect\ref{#1}}
\newcommand{\clmref}[1]{Claim~\protect\ref{#1}}
\newcommand{\corref}[1]{Corollary~\protect\ref{#1}}

\newcommand{\eqnref}[1]{(\protect\ref{#1})}

\newcommand{\exmpref}[1]{Example~\protect\ref{#1}}
\newcommand{\remref}[1]{Remark~\protect\ref{#1}}
\newcommand{\stepref}[1]{Step~\protect\ref{#1}}
\newcommand{\condref}[1]{Condition~\protect\ref{#1}}

\newcommand{\CB}{\allowbreak}

\newcommand{\pair}[2]{(#1 ,\CB #2)}

\DeclareMathOperator{\poly}{poly}
\DeclareMathOperator{\polylog}{polylog}

\newcommand{\Ot}[1]{\tilde{O}(#1)}

\newcommand{\FConc}{\|}

\newcommand{\yes}{{\scriptscriptstyle\mathsf{YES}}}
\newcommand{\no}{{\scriptscriptstyle\mathsf{NO}}}

\newcommand{\Bits}{\{0,1\}}
\newcommand{\Strings}{\Bits^{*}}
\newcommand{\Naturals}{\mathbb{N}}

\newcommand{\GF}[1]{\mathbb{F}_{#1}}

\newcommand{\IdxSet}{I}

\newcommand{\Complement}[1]{\overline{#1}}

\newcommand{\DefineEqual}{:=}

\newcommand{\Set}[1]{\{#1\}}
\newcommand{\SetCardinality}[1]{|#1|}
\newcommand{\ConditionalSet}[2]{\Set{#1 \mid #2}}

\newcommand{\BitSize}[1]{|#1|}

\newcommand{\Proof}{\pi\xspace}

\newcommand{\Randomness}{r}

\newcommand{\FormatComplexityClass}[1]{\mathbf{#1}}
\newcommand{\NTIME}{\FormatComplexityClass{NTIME}}
\newcommand{\BPP}{\FormatComplexityClass{BPP}}
\newcommand{\NP}{\FormatComplexityClass{NP}}

\newcommand{\sharpP}{\FormatComplexityClass{\#P}}
\newcommand{\PSPACE}{\FormatComplexityClass{PSPACE}}
\newcommand{\NEXP}{\FormatComplexityClass{NEXP}}

\newcommand{\PCPP}{\FormatComplexityClass{PCPP}}
\newcommand{\IPCP}{\FormatComplexityClass{IPCP}}
\newcommand{\PZKIPCP}{\FormatComplexityClass{PZK\mbox{-}IPCP}}
\newcommand{\IOP}{\FormatComplexityClass{IOP}}
\newcommand{\IOPP}{\FormatComplexityClass{IOPP}}
\newcommand{\PZKIOP}{\FormatComplexityClass{PZK\mbox{-}IOP}}
\newcommand{\PZKIOPP}{\FormatComplexityClass{PZK\mbox{-}IOPP}}
\newcommand{\AM}{\FormatComplexityClass{AM}}
\newcommand{\coAM}{\FormatComplexityClass{coAM}}

\newcommand{\TextNumRounds}{\mathrm{rounds}}
\newcommand{\TextAnswerAlphabet}{\mathrm{answer\;alphabet}}
\newcommand{\TextProofLength}{\mathrm{proof\;length}}

\newcommand{\TextQueryComplexity}{\mathrm{query\;complexity}}
\newcommand{\TextQueryBound}{\mathrm{query\;bound}}

\newcommand{\TextSoundnessError}{\mathrm{soundness\;error}}
\newcommand{\TextProximityParameter}{\mathrm{proximity\;parameter}}

\newcommand{\TextProverTime}{\mathrm{prover\;time}}
\newcommand{\TextProverSpace}{\mathrm{prover\;space}}

\newcommand{\TextVerifierTime}{\mathrm{verifier\;time}}

\newcommand{\TextVerifierRandomness}{\mathrm{verifier\;randomness}}
\newcommand{\Domain}{D}
\newcommand{\Range}{R}
\newcommand{\SubDomain}{\tilde{\Domain}}
\newcommand{\Alphabet}{\Sigma}
\newcommand{\Restrict}[2]{#1|_{#2}}
\newcommand{\EvaluationDomain}{\Domain}

\newcommand{\IndexSet}{I}

\newcommand{\CDomain}{\tilde{\Domain}}
\newcommand{\CCode}{\tilde{\Code}}

\newcommand{\Puncture}[2]{#1_{\subseteq #2}}

\newcommand{\Field}{\mathbb{F}}
\newcommand{\SubField}{\mathbb{K}}
\newcommand{\FieldSize}{q}

\newcommand{\VariableX}{X}

\newcommand{\Poly}{Q}
\newcommand{\OtherPoly}{T}

\NewDocumentCommand{\IndividualDegree}{m o}{\IfValueTF{#2}{\mathrm{deg}_{#2}(#1)}{\mathrm{deg}(#1)}}

\newcommand{\PolynomialRing}[3]{#1[#3_{1},\dots,#3_{#2}]}
\newcommand{\PolynomialRingIndOne}[4]{#1^{< #4}[#3_{1},\dots,#3_{#2}]}
\newcommand{\PolynomialRingIndTwo}[4]{#1^{< #4}[#3_{2},\dots,#3_{#2}]}

\newcommand{\pST}{\; \middle\vert \;}
\newcommand{\MakeDistribution}[1]{\mathcal{#1}}
\newcommand{\Distribution}{\MakeDistribution{D}}

\newcommand{\Relation}{\mathscr{R}}
\newcommand{\Language}{\mathscr{L}}
\newcommand{\Witnesses}[2]{#1\vert_{#2}}
\newcommand{\GetRelation}[1]{\mathrm{Rel}(#1)}
\newcommand{\GetLanguage}[1]{\mathrm{Lan}(#1)}

\newcommand{\Instance}{\mathbbmss{x}}
\newcommand{\Witness}{\mathbbmss{w}}
\newcommand{\InstanceSize}{n}

\newcommand{\DeciderMachine}{M}
\newcommand{\DeciderTime}{T}

\newcommand{\Prover}{P\xspace}
\newcommand{\Verifier}{V\xspace}
\newcommand{\Simulator}{S\xspace}

\newcommand{\IOPize}[1]{#1}

\newcommand{\IOPProver}{\IOPize{\Prover}}
\newcommand{\IOPVerifier}{\IOPize{\Verifier}}

\newcommand{\IOPSimulator}{\IOPize{\Simulator}}
\newcommand{\Message}{m}
\newcommand{\OMessage}{\Proof}

\newcommand{\IOPPProver}{\IOPize{\Prover}}
\newcommand{\IOPPVerifier}{\IOPize{\Verifier}}
\newcommand{\IOPPSimulator}{\IOPize{\Simulator}}

\newcommand{\CircuitSize}{s}

\newcommand{\SCSymbol}{\mathrm{SC}}
\newcommand{\SCSubset}{H}
\newcommand{\SCVars}{m}
\newcommand{\SCDegree}{d}
\newcommand{\SCPoly}{F}
\newcommand{\SCConstant}{v}

\newcommand{\SCRelation}{\Relation_{\SCSymbol}}
\newcommand{\RandPoly}{R}
\newcommand{\MaskedPoly}{Q}
\newcommand{\AnsTable}[1]{\mathsf{ans}_{#1}}

\newcommand{\ZKSCProver}{\IOPProver_{\SCSymbol}}
\newcommand{\ZKSCVerifier}{\IOPVerifier_{\SCSymbol}}
\newcommand{\ZKSCSimulator}{\IOPSimulator_{\SCSymbol}}
\newcommand{\IPSCProver}{\Prover_{\mathrm{IP}}}
\newcommand{\IPSCVerifier}{\Verifier_{\mathrm{IP}}}
\newcommand{\IPSCSimulator}{\Simulator_{\mathrm{IP}}}
\newcommand{\CodeSimAlgorithm}{\mathcal{A}}
\newcommand{\EmptyVector}{\bot}
\newcommand{\ListSize}{\ell}
\newcommand{\SlowSimulator}{\Simulator_{\mathrm{slow}}}
\newcommand{\MakeF}[1]{#1_{\mathrm{bi}}}
\newcommand{\MakeP}[1]{#1_{\mathrm{pf}}}
\newcommand{\MakeBox}[1]{#1_{{\scriptscriptstyle\Box}}}
\newcommand{\MakeProx}[1]{#1_{\mathrm{px}}}

\newcommand{\BSCode}{\mathrm{BS\text{-}RS}}
\newcommand{\RSBox}{\MakeBox{\RS}}
\newcommand{\BSBaseDim}{k}
\newcommand{\BSSpaceDimSuperLocal}{\ell}
\newcommand{\BSBalance}{\mu}
\newcommand{\SymbolRS}{\mathrm{rs}}
\newcommand{\SymbolProx}{\mathrm{px}}
\newcommand{\SymbolCol}{\mathrm{col}}
\newcommand{\SymbolRow}{\mathrm{row}}
\newcommand{\GetBSProof}[2]{\pi_{#1}(#2)}
\newcommand{\GetBSCode}[2]{#1^{#2}}

\newcommand{\DetAlgorithm}{\mathcal{D}}

\newcommand{\Interact}[2]{\langle #1,#2 \rangle}

\newcommand{\LessThan}[1]{[<#1]}

\newcommand{\InnerProduct}[2]{\langle #1,#2 \rangle}
\DeclareMathOperator{\Span}{span}
\DeclareMathOperator{\rank}{rank}
\newcommand{\Cover}{S}
\newcommand{\CoverAtLayer}[2]{#1|_{#2}}
\newcommand{\CoverAtVertex}[2]{#1_{#2}}
\newcommand{\TreeCover}{T}
\newcommand{\TreeRoot}{r}
\newcommand{\TreeLayer}[1]{\mathrm{layer}(#1)}
\newcommand{\TreeDepth}[1]{\mathrm{depth}(#1)}
\newcommand{\IdpParam}{\kappa}

\newcommand{\CDepth}{d}
\newcommand{\InterDomain}[1]{\mathrm{di}(#1)}
\newcommand{\Undefined}{\emptyset}
\newcommand{\defDomain}[1]{\mathrm{def}(#1)}
\newcommand{\SetA}{V}
\newcommand{\SetB}{W}
\newcommand{\IntParam}{c}
\newcommand{\Children}[1]{\mathrm{successors}(#1)}
\newcommand{\LocalSet}{W}

\newcommand{\Malicious}[1]{\tilde{#1}}
\newcommand{\Simulated}[1]{#1_{\mathrm{sim}}}

\newcommand{\View}{\mathrm{View}}
\newcommand{\FView}[1]{v({#1})}
\newcommand{\IOPView}[2]{\View\;\Interact{#1}{#2}}
\newcommand{\IOPPView}[2]{\View\;\Interact{#1}{#2}}
\newcommand{\IPCPView}[2]{\View\;\Interact{#1}{#2}}

\newcommand{\IndepTag}{\text{`independent'}}

\newcommand{\Distance}{\Delta}
\newcommand{\AbsoluteHammingDistance}{\Distance} % maybe later change to specific symbol

\newcommand{\Randomizability}{t}
\newcommand{\RandomizeTime}{r}

\newcommand{\Dimension}[1]{\mathrm{dim}(#1)}

\newcommand{\Support}[1]{\mathrm{supp}(#1)}

\newcommand{\Class}[1]{\mathscr{#1}}

\newcommand{\LinearSpace}{S}

\newcommand{\Code}{C}
\newcommand{\Subcode}{C'}
\newcommand{\Dual}[1]{#1^{\perp}}
\newcommand{\CodeClass}{\Class{\Code}}
\newcommand{\OfCode}[1]{#1_{\CodeClass}}
\newcommand{\CodeRate}{\rho}

\newcommand{\CodeBlockLength}{\ell}
\newcommand{\CodeMessageLength}{k}
\newcommand{\CodeAbsoluteDistance}{d}
\newcommand{\CodeRelativeDistance}{\tau}

\newcommand{\CodeTime}{T}
\newcommand{\Codeword}{w}
\newcommand{\OtherCodeword}{z}
\newcommand{\SubcodeCodeword}{w'}

\newcommand{\CodeIdx}{\mathbbmss{n}}
\newcommand{\CodeSamplingTime}{S}
\newcommand{\RM}{\mathrm{RM}}
\newcommand{\RMCode}[4]{\RM[#1,#2,#3,#4]}
\newcommand{\RMDomain}{S}

\newcommand{\RMVars}{m}
\newcommand{\RMDegree}{d}
\newcommand{\PSRM}{\mathrm{\Sigma RM}}
\newcommand{\PSRMCode}[4]{\PSRM[#1,#2,#3,#4]}

\newcommand{\RS}{\mathrm{RS}}
\newcommand{\RSCode}[3]{\RS[#1,#2,#3]}
\newcommand{\RSDomain}{\RMDomain}

\newcommand{\RSDegree}{\RMDegree}

\newcommand{\ARS}{\mathrm{RS}^{{\scriptscriptstyle+}}}

\newcommand{\BSSpace}{L}
\newcommand{\VanishingPoly}[1]{Z_{#1}}

\newcommand{\VariantA}[1]{#1^{\prime}}
\newcommand{\VariantB}[1]{#1^{\prime\prime}}

\newcommand{\sharpSATLanguage}{\Language_{\mathrm{\#3SAT}}}

\newcommand{\Formula}{\phi}

\newcommand{\NumSats}{N}
\newcommand{\NumVars}{n}
\newcommand{\NumClauses}{c}
\newcommand{\Arithmetize}[1]{p_{#1}}

\newcommand{\Assignment}{w}

\newcommand{\OtherCode}{L}
\newcommand{\FormatComplexityFunction}[1]{\mathsf{#1}}

\newcommand{\ProverTime}{\FormatComplexityFunction{tp}}
\newcommand{\ProverSpace}{\FormatComplexityFunction{sp}}

\newcommand{\VerifierTime}{\FormatComplexityFunction{tv}}

\newcommand{\VerifierRandomness}{\FormatComplexityFunction{rv}}
\newcommand{\SoundnessError}{\FormatComplexityFunction{\varepsilon}}
\newcommand{\ProofLength}{\FormatComplexityFunction{l}}
\newcommand{\QueryComplexity}{\FormatComplexityFunction{q}}

\newcommand{\ProximityParameter}{\FormatComplexityFunction{\delta}}

\newcommand{\DistanceMeasure}{\Distance}
\newcommand{\NumRounds}{\FormatComplexityFunction{k}}
\newcommand{\QueryBound}{\FormatComplexityFunction{b}}
\newcommand{\AnyBound}{\FormatComplexityFunction{*}}

\newcommand{\PCPPProver}{\Prover}
\newcommand{\PCPPVerifier}{\Verifier}

\newcommand{\RandCodeword}{z}
\newcommand{\MaskedCodeword}{w'}

\newcommand{\RSVRS}{\mathrm{ERS}^{{\scriptscriptstyle+}}}
\newcommand{\RSVRSCode}[5]{\RSVRS[#1,#2,#3,#4,#5]}
\newcommand{\RSVRSSubdomain}{L}
\newcommand{\RSVRSVanishing}{H}

\newcommand{\LAize}[1]{#1}
\newcommand{\RLAize}[1]{#1}

\newcommand{\IOPPize}[1]{\hat{#1}}

\newcommand{\LACSPSymbol}{{\scriptscriptstyle\mathsf{LACSP}}}
\newcommand{\LACSPRelation}{\Relation_{\LACSPSymbol}}
\newcommand{\LACSPLanguage}{\Language_{\LACSPSymbol}}
\newcommand{\LACSPRelationP}{
(\LACSPRelation^{\yes}, \LACSPLanguage^{\no})[
\Field(\InstanceSize),
\MakeZero{\Code}(\InstanceSize),
\MakeOne{\Code}(\InstanceSize),
\CodeBlockLength(\InstanceSize),
\CodeRelativeDistance(\InstanceSize),
\MapLocality(\InstanceSize),
\MapEfficiency(\InstanceSize)]
}
\newcommand{\LACSPRelationQ}{
(\LACSPRelation^{\yes}, \LACSPLanguage^{\no})[
\Field,
\MakeZero{\Code},
\MakeOne{\Code},
\LAize{\CodeBlockLength},
\LAize{\CodeRelativeDistance},
\LAize{\MapLocality},
\LAize{\MapEfficiency}]
}

\newcommand{\RLACSPSymbol}{{\scriptscriptstyle\mathsf{RLACSP}}}
\newcommand{\RLACSPRelation}{\Relation_{\RLACSPSymbol}}
\newcommand{\RLACSPLanguage}{\Language_{\RLACSPSymbol}}
\newcommand{\RLACSPRelationP}{
(\RLACSPRelation^{\yes},\RLACSPLanguage^{\no})[
\Field(\InstanceSize),
\MakeZero{\Code}(\InstanceSize),
\MakeOne{\Code}(\InstanceSize),
\CodeBlockLength(\InstanceSize),
\CodeRelativeDistance(\InstanceSize),
\MapLocality(\InstanceSize),
\MapEfficiency(\InstanceSize),
\Randomizability(\InstanceSize),
\RandomizeTime(\InstanceSize)]
}
\newcommand{\RLACSPRelationQ}{
(\RLACSPRelation^{\yes},\RLACSPLanguage^{\no})[
\Field,
\MakeZero{\Code},
\MakeOne{\Code},
\RLAize{\CodeBlockLength},
\RLAize{\CodeRelativeDistance},
\RLAize{\MapLocality},
\RLAize{\MapEfficiency},
\RLAize{\Randomizability},
\RLAize{\RandomizeTime}
]
}

\newcommand{\MapEfficiency}{c}
\newcommand{\MapLocality}{q}

\newcommand{\LocalMapAlphabet}{\Sigma}
\newcommand{\LocalMap}{g}
\newcommand{\LocalMapDomainSize}{n}
\newcommand{\LocalMapRangeSize}{m}
\newcommand{\LACodeRelation}{\GetRelation{\MakeZero{\Code} \times \MakeOne{\Code}}}

\newcommand{\InstanceMap}{\mathsf{inst}}
\newcommand{\WitnessMap}{\mathsf{wit}_{1}}
\newcommand{\ExtractorMap}{\mathsf{wit}_{2}}
\renewcommand{\TextField}{\mathrm{field}}
\newcommand{\TextCodeZero}{\mathrm{first\;code}}
\newcommand{\TextCodeOne}{\mathrm{second\;code}}
\newcommand{\TextCodeBlockLength}{\mathrm{block\;length}}

\newcommand{\TextCodeRelativeDistance}{\mathrm{relative\;distance}}
\newcommand{\TextMapLocality}{\mathrm{map\;locality}}
\newcommand{\TextMapEfficiency}{\mathrm{map\;efficiency}}

\newcommand{\TextRandomizability}{\mathrm{randomizability}}
\newcommand{\TextRandomizeTime}{\mathrm{randomize\;time}}
\newcommand{\MakeZero}[1]{#1_{0}}
\newcommand{\MakeOne}[1]{#1_{1}}
\newcommand{\IOPProverAssignment}{w}
\newcommand{\IOPProverNonce}{u'}
\newcommand{\IOPProverRandAssignment}{w'}

\newcommand{\CodesProver}{\IOPPize{\Prover}}
\newcommand{\CodesVerifier}{\IOPPize{\Verifier}}
\newcommand{\CodesSimulator}{\IOPPize{\Simulator}}

\newcommand{\LAProver}{\Prover_{\LACSPSymbol}}
\newcommand{\LAVerifier}{\Verifier_{\LACSPSymbol}}
\newcommand{\LASimulator}{\Simulator_{\LACSPSymbol}}

\newcommand{\RLAProver}{\Prover_{\RLACSPSymbol}}
\newcommand{\RLAVerifier}{\Verifier_{\RLACSPSymbol}}
\newcommand{\RLASimulator}{\Simulator_{\RLACSPSymbol}}
\newcommand{\RLAWitness}{\Witness'}
\newcommand{\RLAEntries}{E}
\newcommand{\SimRLAWitness}{\hat{\IOPProverAssignment}}

\newcommand{\BsDepth}[2]{\left\lfloor \log \Dimension{#1} - \log(#2) \right \rfloor}
\newcommand{\BsOtherDepth}[2]{\log_{2} \Dimension{#1} - \log_{2}( \log_{\SetCardinality{\SubField}} #2 + \BSBalance + 2) - 1}
\newcommand{\OtherBSSpace}{\tilde{\BSSpace}}

\newcommand{\dCols}{d_{\mathsf{cols}}}
\newcommand{\dRows}{d_{\mathsf{rows}}}
\newcommand{\SCols}{S_{\mathsf{cols}}}
\newcommand{\SRows}{S_{\mathsf{rows}}}
\newcommand{\SPoints}{S_{\mathsf{pnts}}}
\newcommand{\gExists}{g}
\newcommand{\gCol}[1]{\gExists_{\mathsf{col},{#1}}}
\newcommand{\gRow}[1]{\gExists_{\mathsf{row},{#1}}}

\newcommand{\eps}{\ensuremath{\epsilon}\xspace}
\newcommand{\F}{\ensuremath{\mathbb{F}}\xspace}

\newcommand{\x}{\ensuremath{\mathbf{x}}}
\newcommand{\sumrange}{\ensuremath{{H^m}}}
\newcommand{\sumOnRange}[1]{\ensuremath{\sum_{x\in \sumrange}#1(x)}}
\newcommand{\polys}[1]{\ensuremath{\PolynomialRingIndOne{\Field}{\SCVars}{\VariableX}{#1}}\xspace}
\newcommand{\Z}{\ensuremath{Z_{H^m}}\xspace}
\newcommand{\shiftedRange}{\ensuremath{H'^m}}
\renewcommand{\deg}[1]{\ensuremath{\mathrm{deg}(#1)}}
\newcommand{\eval}[2]{\ensuremath{#1|_{#2}}}
\newcommand{\distOn}[3]{\ensuremath{\Delta_{#1}(#2,#3)}\xspace}

\begin{document}
%%%%%%%%%%%%%%%%%%%%%%%%%%%%%%%%%%%%%%%%%%%%%%%%%%%%%%%%%%%%%%%%%%%%%%%%%%%%%%%%
%%%%%%%%%%%%%%%%%%%%%%%%%%%%%%%%%%%%%%%%%%%%%%%%%%%%%%%%%%%%%%%%%%%%%%%%%%%%%%%%
%%%%%%%%%%%%%%%%%%%%%%%%%%%%%%%%%%%%%%%%%%%%%%%%%%%%%%%%%%%%%%%%%%%%%%%%%%%%%%%%

%%%%%%%%%%%%%%%%%%%%%%%%%%%%%%%%%%%%%%%%%%%%%%%%%%%%%%%%%%%%%%%%%%%%%%%%%%%%%%%%
\title{%
On Probabilistic Checking in Perfect Zero Knowledge
}
%%%%%%%%%%%%%%%%%%%%%%%%%%%%%%%%%%%%%%%%%%%%%%%%%%%%%%%%%%%%%%%%%%%%%%%%%%%%%%%%

%%%%%%%%%%%%%%%%%%%%%%%%%%%%%%%%%%%%%%%%%%%%%%%%%%%%%%%%%%%%%%%%%%%%%%%%%%%%%%%%
\author{
\begin{tabular}[h!]{ccc}
\FormatAuthor{Eli Ben-Sasson}{eli@cs.technion.ac.il}{Technion}
 & \FormatAuthor{Alessandro Chiesa}{alexch@berkeley.edu}{UC Berkeley}
 & \FormatAuthor{Michael A. Forbes}{miforbes@csail.mit.edu}{Stanford University} \\
 &&\\
\FormatAuthor{Ariel Gabizon}{arielga@cs.technion.ac.il}{Technion}
 &  \FormatAuthor{Michael Riabzev}{mriabzev@cs.technion.ac.il}{Technion}
 &  \FormatAuthor{Nicholas Spooner}{spooner@cs.toronto.edu}{University of Toronto}
\end{tabular}
}
%%%%%%%%%%%%%%%%%%%%%%%%%%%%%%%%%%%%%%%%%%%%%%%%%%%%%%%%%%%%%%%%%%%%%%%%%%%%%%%%

%%%%%%%%%%%%%%%%%%%%%%%%%%%%%%%%%%%%%%%%%%%%%%%%%%%%%%%%%%%%%%%%%%%%%%%%%%%%%%%%
\iffull
  \date{\today}
\fi
%%%%%%%%%%%%%%%%%%%%%%%%%%%%%%%%%%%%%%%%%%%%%%%%%%%%%%%%%%%%%%%%%%%%%%%%%%%%%%%%

%%%%%%%%%%%%%%%%%%%%%%%%%%%%%%%%%%%%%%%%%%%%%%%%%%%%%%%%%%%%%%%%%%%%%%%%%%%%%%%%
\maketitle
%%%%%%%%%%%%%%%%%%%%%%%%%%%%%%%%%%%%%%%%%%%%%%%%%%%%%%%%%%%%%%%%%%%%%%%%%%%%%%%%

%%%%%%%%%%%%%%%%%%%%%%%%%%%%%%%%%%%%%%%%%%%%%%%%%%%%%%%%%%%%%%%%%%%%%%%%%%%%%%%%
%%%%%%%%%%%%%%%%%%%%%%%%%%%%%%%%%%%%%%%%%%%%%%%%%%%%%%%%%%%%%%%%%%%%%%%%%%%%%%%%
%%%%%%%%%%%%%%%%%%%%%%%%%%%%%%%%%%%%%%%%%%%%%%%%%%%%%%%%%%%%%%%%%%%%%%%%%%%%%%%%
\begin{abstract}
We present the first constructions of \emph{single}-prover proof systems that achieve \emph{perfect} zero knowledge (PZK) for languages beyond $\NP$, under no intractability assumptions:
\begin{enumerate}

  \item The complexity class $\sharpP$ has PZK proofs in the model of Interactive PCPs (IPCPs) \cite{KalaiR08}, where the verifier first receives from the prover a PCP and then engages with the prover in an Interactive Proof (IP).

  \item The complexity class $\NEXP$ has PZK proofs in the model of Interactive Oracle Proofs (IOPs) \cite{BenSassonCS16,ReingoldRR16}, where the verifier, in every round of interaction, receives a PCP from the prover.

\end{enumerate}
Unlike PZK multi-prover proof systems \cite{BenOrGKW88}, PZK single-prover proof systems are elusive: PZK IPs are limited to $\AM \cap \coAM$ \cite{Fortnow87,AielloH91}, while known PCPs and IPCPs achieve only \emph{statistical} simulation \cite{KilianPT97,GoyalIMS10}. Recent work \cite{BenSassonCGV16} has achieved PZK for IOPs but only for languages in $\NP$, while our results go beyond it.

Our constructions rely on \emph{succinct} simulators that enable us to ``simulate beyond $\NP$'', achieving exponential savings in efficiency over \cite{BenSassonCGV16}. These simulators crucially rely on solving a problem that lies at the intersection of coding theory, linear algebra, and computational complexity, which we call the \emph{succinct constraint detection} problem, and consists of detecting dual constraints with polynomial support size for codes of exponential block length. Our two results rely on solutions to this problem for fundamental classes of linear codes:
\begin{itemize}

  \item An algorithm to detect constraints for Reed--Muller codes of exponential length.

  \item An algorithm to detect constraints for PCPs of Proximity of Reed--Solomon codes \cite{BS08} of exponential degree.

\end{itemize}
The first algorithm exploits the Raz--Shpilka \cite{RazS05} deterministic polynomial identity testing algorithm, and shows, to our knowledge, a first connection of algebraic complexity theory with zero knowledge.
Along the way, we give a perfect zero knowledge analogue of the celebrated sumcheck protocol \cite{LundFKN92}, by leveraging both succinct constraint detection and low-degree testing.
The second algorithm exploits the recursive structure of the PCPs of Proximity to show that small-support constraints are ``locally'' spanned by a small number of small-support constraints.

\bigskip
\keywords{\mbox{probabilistically checkable proofs, interactive proofs, sumcheck, zero knowledge, polynomial identity testing}}
\end{abstract}

%%%%%%%%%%%%%%%%%%%%%%%%%%%%%%%%%%%%%%%%%%%%%%%%%%%%%%%%%%%%%%%%%%%%%%%%%%%%%%%%
%%%%%%%%%%%%%%%%%%%%%%%%%%%%%%%%%%%%%%%%%%%%%%%%%%%%%%%%%%%%%%%%%%%%%%%%%%%%%%%%
%%%%%%%%%%%%%%%%%%%%%%%%%%%%%%%%%%%%%%%%%%%%%%%%%%%%%%%%%%%%%%%%%%%%%%%%%%%%%%%%

\iffull
\clearpage
\setcounter{tocdepth}{2}
{\footnotesize \tableofcontents}
\clearpage
\fi

\doclearpage
%%%%%%%%%%%%%%%%%%%%%%%%%%%%%%%%%%%%%%%%%%%%%%%%%%%%%%%%%%%%%%%%%%%%%%%%%%%%%%%%
%%%%%%%%%%%%%%%%%%%%%%%%%%%%%%%%%%%%%%%%%%%%%%%%%%%%%%%%%%%%%%%%%%%%%%%%%%%%%%%%
%%%%%%%%%%%%%%%%%%%%%%%%%%%%%%%%%%%%%%%%%%%%%%%%%%%%%%%%%%%%%%%%%%%%%%%%%%%%%%%%
\section{Introduction}
\label{sec:introduction}

We study proof systems that unconditionally achieve soundness and zero knowledge, that is, without relying on intractability assumptions. We present two new constructions in single-prover models that combine aspects of probabilistic checking and interaction:
\begin{inparaenum}[(i)]
  \item \emph{Interactive Probabilistically-Checkable Proofs} \cite{KalaiR08}, and
  \item \emph{Interactive Oracle Proofs} (also called \emph{Probabilistically Checkable Interactive Proofs}) \cite{BenSassonCS16,ReingoldRR16}.
\end{inparaenum}

%%%%%%%%%%%%%%%%%%%%%%%%%%%%%%%%%%%%%%%%%%%%%%%%%%%%%%%%%%%%%%%%%%%%%%%%%%%%%%%%
%%%%%%%%%%%%%%%%%%%%%%%%%%%%%%%%%%%%%%%%%%%%%%%%%%%%%%%%%%%%%%%%%%%%%%%%%%%%%%%%
\subsection{Proof systems with unconditional zero knowledge}
\label{sec:try}

Zero knowledge is a central notion in cryptography. A fundamental question, first posed in \cite{BenOrGKW88}, is the following: \emph{what types of proof systems achieve zero knowledge unconditionally?} We study this question, and focus on the setting of \emph{single}-prover proof systems, as we now explain.

%%%%%%%%%%%%%%%%%%%%%%%%%%%%%%%%%%%%%%%%
\parhead{IPs and MIPs}
An \sunderline{interactive proof} (IP) \cite{BabaiM88,GoldwasserMR89} for a language $\Language$ is a pair of interactive algorithms $(\Prover,\Verifier)$, where $\Prover$ is known as the prover and $\Verifier$ as the verifier, such that:
\begin{inparaenum}[(i)]
  \item (completeness)
  for every instance $\Instance \in \Language$, $\Prover(\Instance)$ makes $\Verifier(\Instance)$ accept with probability $1$;
  \item (soundness)
  for every instance $\Instance \not\in \Language$, every prover $\Malicious{\Prover}$ makes $\Verifier(\Instance)$ accept with at most a small probability $\SoundnessError$.
\end{inparaenum}
A \sunderline{multi-prover interactive proof} (MIP) \cite{BenOrGKW88} extends an IP to the case where the verifier interacts with \emph{multiple} non-communicating provers.

%%%%%%%%%%%%%%%%%%%%%%%%%%%%%%%%%%%%%%%%
\parhead{Zero knowledge}
An IP is \emph{zero knowledge} \cite{GoldwasserMR89} if malicious verifiers cannot learn any information about an instance $\Instance$ in $\Language$, by interacting with the prover, besides the fact $\Instance$ is in $\Language$: for any efficient verifier $\Malicious{\Verifier}$ there is an efficient simulator $\Simulator$ such that $\Simulator(\Instance)$ is ``indistinguishable'' from $\Malicious{\Verifier}$'s view when interacting with $\Prover(\Instance)$. Depending on the notion of indistinguishability, one gets different types of zero knowledge: computational, statistical, or perfect zero knowledge (respectively CZK, SZK, or PZK), which require that the simulator's output and the verifier's view are computationally indistinguishable, statistically close, or equal. Zero knowledge for MIPs is similarly defined \cite{BenOrGKW88}.

%%%%%%%%%%%%%%%%%%%%%%%%%%%%%%%%%%%%%%%%
\parhead{IPs vs.\ MIPs}
The complexity classes captured by the above models are (conjecturally) different: all and only languages in $\PSPACE$ have IPs \cite{Shamir92}; while all and only languages in $\NEXP$ have MIPs \cite{BabaiFL91}. Yet, when requiring zero knowledge, the difference between the two models is even greater. Namely, the complexity class of SZK (and PZK) IPs is contained in $\AM \cap \coAM$ \cite{Fortnow87,AielloH91} and, thus, does not contain $\NP$ unless the polynomial hierarchy collapses \cite{BoppanaHZ87}.\footnote{A seminal result in cryptography says that if one-way functions exist then every language having an IP also has a CZK IP \cite{GoldwasserMR89,ImpagliazzoY87,BenOrGGHKMR88}; also, if one-way functions do not exist then CZK IPs capture only ``average-case'' $\BPP$ \cite{Ostrovsky91,OstrovskyW93}. Thus, prior work gives minimal assumptions that make IPs a powerful tool for zero knowledge. However, we focus on zero knowledge achievable \emph{without} intractability assumptions.} In contrast, the complexity class of PZK MIPs equals $\NEXP$ \cite{BenOrGKW88,LapidotS95,DworkFKNS92}.

%%%%%%%%%%%%%%%%%%%%%%%%%%%%%%%%%%%%%%%%
\parhead{Single-prover proof systems beyond IPs}
The limitations of SZK IPs preclude many information-theoretic and cryptographic applications; at the same time, while capturing all of $\NEXP$, PZK MIPs are difficult to leverage in applications because it is hard to exclude all communication between provers.  Research has thus focused on obtaining zero knowledge by considering alternatives to the IP model for which there continues to be a \emph{single} prover. E.g., an extension that does not involve multiple provers is to allow the verifier to probabilistically check the prover's messages, instead of reading them in full, as in the following models.
\begin{inparaenum}[(1)]
  \item
  In a \sunderline{probabilistically checkable proof} (PCP) \cite{FortnowRS88,BFLS91,FGLSS96,AroraS98,AroraLMSS98}, the verifier has oracle access to a single message sent by the prover; zero knowledge for this case was first studied by \cite{KilianPT97}.
  \item
  In an \sunderline{interactive probabilistically checkable proof} (IPCP) \cite{KalaiR08}, the verifier has oracle access to the first message sent by the prover, but must read in full all subsequent messages (i.e., it is a PCP followed by an IP); zero knowledge for this case was first studied by \cite{GoyalIMS10}.
  \item
  In an \sunderline{interactive oracle proof} (IOP) \cite{BenSassonCS16,ReingoldRR16}, the verifier has oracle access to all of the messages sent by the prover (i.e., it is a ``multi-round PCP''); zero knowledge for this case was first studied by \cite{BenSassonCGV16}.
\end{inparaenum}
The above models are in order of increasing generality; their applications include unconditional cryptography \cite{GoyalIMS10}, zero knowledge with black-box cryptography \cite{IshaiMS12,IshaiMSX15}, and zero knowledge with random oracles \cite{BenSassonCS16}.

%%%%%%%%%%%%%%%%%%%%%%%%%%%%%%%%%%%%%%%%%%%%%%%%%%%%%%%%%%%%%%%%%%%%%%%%%%%%%%%%
%%%%%%%%%%%%%%%%%%%%%%%%%%%%%%%%%%%%%%%%%%%%%%%%%%%%%%%%%%%%%%%%%%%%%%%%%%%%%%%%
\subsection{What about perfect zero knowledge for single-prover systems?}
\label{sec:limitations}

In the setting of multiple provers, one can achieve the best possible notion of zero knowledge: perfect zero knowledge for all $\NEXP$ languages, even with only two provers and one round of interaction \cite{BenOrGKW88,LapidotS95,DworkFKNS92}. In contrast, in the setting of a single prover, \emph{all} prior PCP and IPCP constructions achieve only \emph{statistical} zero knowledge (see \appref{sec:prior-work-summary} for a summary of prior work). This intriguing state of affairs motivates the following question:
\begin{center}
\emph{What are the powers and limitations of perfect zero knowledge for single-prover proof systems?}
\end{center}

The above question is especially interesting when considering that prior PCP and IPCP constructions all follow the \emph{same} paradigm: they achieve zero knowledge by using a ``locking scheme'' on an intermediate construction that is zero knowledge against (several independent copies of) the honest verifier. Perfect zero knowledge in this paradigm is not possible because a malicious verifier can always guess a lock's key with positive probability. Also, besides precluding perfect zero knowledge, this paradigm introduces other limitations:
\begin{inparaenum}[(a)]
  \item the prover incurs polynomial blowups to avoid collisions due to birthday paradox; and
  \item the honest verifier is adaptive (first retrieve keys for the lock, then unlock).
\end{inparaenum}

Recent work explores alternative approaches but these constructions suffer from other limitations:
\begin{inparaenum}[(i)]
  \item \cite{IshaiWY16} apply PCPs to leakage-resilient circuits, and obtain PCPs for $\NP$ with a non-adaptive verifier but they are only witness (statistically) indistinguishable;
  \item \cite{BenSassonCGV16} exploit algebraic properties of PCPs and obtain $2$-round IOPs that are perfect zero knowledge but only for $\NP$.
\end{inparaenum}

In sum, prior work is limited to statistical zero knowledge (or weaker) or to languages in $\NP$. Whether perfect zero knowledge is feasible for more languages has remained open. Beyond the fundamental nature of the question of understanding perfect zero knowledge for single-prover proof systems, recall that perfect zero knowledge, in the case of single-prover interactive proofs, is desirable because it implies general concurrent composability while statistical zero knowledge does not always imply it \cite{KushilevitzLR10}, and one may hope to prove analogous results for IPCPs and IOPs.

%%%%%%%%%%%%%%%%%%%%%%%%%%%%%%%%%%%%%%%%%%%%%%%%%%%%%%%%%%%%%%%%%%%%%%%%%%%%%%%%
%%%%%%%%%%%%%%%%%%%%%%%%%%%%%%%%%%%%%%%%%%%%%%%%%%%%%%%%%%%%%%%%%%%%%%%%%%%%%%%%
\subsection{Results}
\label{sec:results}

We present the first constructions of \emph{single}-prover proof systems that achieve perfect zero knowledge for languages beyond $\NP$; prior work is limited to languages in $\NP$ or statistical zero knowledge. We develop and apply linear-algebraic algorithms that enable us to ``simulate beyond $\NP$'', by exploiting new connections to algebraic complexity theory and local views of linear codes. In all of our constructions, we do not use locking schemes, the honest verifier is non-adaptive, and the simulator is straightline (i.e., does not rewind the verifier). We now summarize our results.

%%%%%%%%%%%%%%%%%%%%%%%%%%%%%%%%%%%%%%%%%%%%%%%%%%%%%%%%%%%%%%%%%%%%%%%%%%%%%%%%
\subsubsection{Perfect zero knowledge beyond $\NP$}
\label{sec:pzk-beyond-np}

%%%%%%%%%%%%%%%%%%%%%%%%%%%%%%%%%%%%%%%%
\parhead{(1) PZK IPCPs for $\sharpP$}
We prove that the complexity class $\sharpP$ has IPCPs that are perfect zero knowledge. We do so by constructing a suitable IPCP system for the counting problem associated to $\mathrm{3SAT}$, which is $\sharpP$-complete.

\begin{theorem}[informal statement of Thm.\ \ref{thm:zk-sharpp}]
\label{thm:intro-sharpp}
The complexity class $\sharpP$ has Interactive PCPs that are perfect zero knowledge against unbounded queries (that is, a malicious verifier may ask an arbitrary polynomial number of queries to the PCP sent by the prover).
\end{theorem}

\noindent
The above theorem gives new tradeoffs relative to prior work. First, \cite{KilianPT97} construct PCPs for $\NEXP$ that are statistical zero knowledge against unbounded queries. Instead, we construct IPCPs (which generalize PCPs) for $\sharpP$, but we gain perfect zero knowledge. Second, \cite{GoyalIMS10} construct IPCPs for $\NP$ that
\begin{inparaenum}[(i)]
  \item are statistical zero knowledge against unbounded queries, and
  \item have a polynomial-time computable oracle.
\end{inparaenum}
Instead, we construct IPCPs that have an exponential-time computable oracle, but we gain perfect zero knowledge and all languages in $\sharpP$.

%%%%%%%%%%%%%%%%%%%%%%%%%%%%%%%%%%%%%%%%
\parhead{(2) PZK IOPs for $\NEXP$}
We prove that $\NEXP$ has $2$-round IOPs that are perfect zero knowledge against unbounded queries. We do so by constructing a suitable IOP system for $\NTIME(\DeciderTime)$ against query bound $\QueryBound$, for each time function $\DeciderTime$ and query bound function $\QueryBound$; the result for $\NEXP$ follows by setting $\QueryBound$ to be super-polynomial.

\begin{theorem}[informal statement of Thm.\ \ref{thm:ntime}]
\label{thm:intro-ntime}
For every time bound $\DeciderTime$ and query bound $\QueryBound$, the complexity class $\NTIME(\DeciderTime)$ has $2$-round Interactive Oracle Proofs that are perfect zero knowledge against $\QueryBound$ queries, and where the proof length is $\Ot{\DeciderTime+\QueryBound}$ and the (honest verifier's) query complexity is $\polylog(\DeciderTime + \QueryBound)$.
\end{theorem}

\noindent
The above theorem gives new tradeoffs relative to \cite{KilianPT97}'s result, which gives PCPs for $\NEXP$ that have statistical zero knowledge, polynomial proof length, and an adaptive verifier. Instead, we construct IOPs (which generalize PCPs) for $\NEXP$, but we gain perfect zero knowledge, quasilinear proof length, and a non-adaptive verifier. Moreover, our theorem extends \cite{BenSassonCGV16}'s result from $\NP$ to $\NEXP$, positively answering an open question raised in that work. Namely, \cite{BenSassonCGV16}'s result requires $\DeciderTime,\QueryBound$ to be polynomially-bounded because their simulator runs in $\poly(\DeciderTime+\QueryBound)$ time; in contrast, our result relies on a simulator that runs in time $\poly(\Malicious{q} + \log \DeciderTime+ \log \QueryBound)$, where $\Malicious{q}$ is the actual number of queries made by the malicious verifier; this is the exponential improvement that enables us to ``go up to $\NEXP$''. Overall, we learn that ``$2$ rounds of PCPs'' are enough to obtain perfect zero knowledge (even with quasilinear proof length) for any language in $\NEXP$.

%\ariel{May want to mention, as I've understood now about KPT, that KPT will always produce a PCP of at least quadratic  length in $\InstanceSize$, unless there exists a PCP construction with randomness complexity $2\cdot \log(\ProofLength )-\omega(1)$, which seems unlikely. In particular, assuming all known PCP constructions have randomness complexity at least $2\cdot \log(\ProofLength) -O(1)$ (am pretty sure of this), KPT produces quadratic ZK PCPs at best case}

%%%%%%%%%%%%%%%%%%%%%%%%%%%%%%%%%%%%%%%%%%%%%%%%%%%%%%%%%%%%%%%%%%%%%%%%%%%%%%%%
\subsubsection{Succinct constraint detection for Reed--Muller and Reed--Solomon codes}
\label{intro:succinct-constraint-detection}

Our theorems in the previous section crucially rely on solving a problem that lies at the intersection of coding theory, linear algebra, and computational complexity, which we call the \emph{constraint detection problem}. We outline the role played by this problem in \secref{techniques:masking}, while now we informally introduce it and state our results about it.

%%%%%%%%%%%%%%%%%%%%%%%%%%%%%%%%%%%%%%%%
\parhead{Detecting constraints in codes}
Constraint detection is the problem of determining which linear relations hold across all codewords of a linear code $\Code \subseteq \Field^{\Domain}$, when considering only a given subdomain $\IdxSet \subseteq\Domain$ of the code rather than all of the domain $\Domain$. This problem can always be solved in time that is polynomial in $\SetCardinality{\Domain}$ (via Gaussian elimination); however, if there is an algorithm that solves this problem in time that is \emph{polynomial in the subdomain's  size} $\SetCardinality{\IdxSet}$, rather than the domain's size $\SetCardinality{\Domain}$, then we say that the code has \emph{succinct} constraint detection; in particular, the domain could have  \emph{exponential} size and the algorithm would still run in polynomial time.

\begin{definition}[informal]
\label{def:informal-constraint-detection}
We say that a linear code $\Code \subseteq \Field^{\Domain}$ has \defemph{succinct constraint detection} if there exists an algorithm that, given a subset $\IdxSet \subseteq \Domain$, runs in time $\poly(\log \SetCardinality{\Field}  +\log \SetCardinality{\Domain} + \SetCardinality{\IdxSet})$ and outputs $\OtherCodeword \in \Field^{\IdxSet}$ such that $\sum_{i \in \IdxSet} \OtherCodeword(i) \Codeword(i) = 0$ for all $\Codeword \in \Code$, or ``no'' if no such $\OtherCodeword$ exists. (In particular, $\SetCardinality{\Domain}$ may be exponential.)
\end{definition}

\noindent
We further discuss the problem of constraint detection in \secref{techniques:detecting-constraints}, and provide a formal treatment of it in \secref{sec:definition-of-constraint-detection}. Beyond this introduction, we shall use (and achieve) a stronger definition of constraint detection: the algorithm is required to output a basis for the space of dual codewords in $\Dual{\Code}$ whose support lies in the subdomain $\IdxSet$, i.e., a basis for the space $\{ \OtherCodeword \in \Domain^{\IdxSet} : \; \forall\, \Codeword \in \Code\,,\, {\sum_{i \in \IdxSet} \OtherCodeword(i) \Codeword(i) = 0}  \}$. Note that in our discussion of succinct constraint detection we do not leverage the distance property of the code $\Code$, but we do leverage it in our eventual applications.

Our zero knowledge simulators' strategy includes sampling a ``random PCP'': a random codeword $\Codeword$ in a linear code $\Code$ with exponentially large domain size $\SetCardinality{\Domain}$ (see \secref{techniques:masking} for more on this). Explicitly sampling $\Codeword$ requires time $\Omega(\SetCardinality{\Domain})$, and so is inefficient. But a verifier makes only polynomially-many queries to $\Codeword$, so the simulator has to only simulate $\Codeword$ when restricted to polynomial-size sets $\IdxSet \subseteq \Domain$, allowing the possibility of doing so in time $\poly(\SetCardinality{\IdxSet})$. Achieving such a simulation time is an instance of (efficiently and perfectly) ``implementing a huge random object'' \cite{GoldreichGN10} via a \emph{stateful} algorithm \cite{BogdanovW04}. We observe that if $\Code$ has succinct constraint detection then this sampling problem for $\Code$ has a solution: the simulator maintains the set $\{(i,a_{i})\}_{i \in \IdxSet}$ of past query-answer pairs; then, on a new verifier query $j \in \Domain$, the simulator uses constraint detection to determine if $\Codeword_{j}$ is linearly dependent on $\Codeword_{\IdxSet}$, and answers accordingly (such linear dependencies characterize the required probability distribution, see \lemref{lem:efficient-codeword-simulator}).

Overall, our paper thus provides an application (namely, obtaining PZK simulators for classes beyond $\NP$) where the problem of efficient implementation of huge random objects arises naturally, and also where \emph{perfect} implementation (no probability of error in the simulated distribution) is a requirement of the application.

We now state our results about succinct constraint detection.

%%%%%%%%%%%%%%%%%%%%%%%%%%%%%%%%%%%%%%%%
\parhead{(1) Reed--Muller codes, and their partial sums}
We prove that the family of linear codes comprised of evaluations of low-degree multivariate polynomials, along with their partial sums, has succinct constraint detection. This family is closely related to the \emph{sumcheck protocol} \cite{LundFKN92}, and indeed we use this result to obtain a PZK analogue of the sumcheck protocol (see \secref{techniques:zk-sumcheck} and \secref{sec:zk-sumcheck}), which yields \thmref{thm:intro-sharpp} (see \secref{techniques:zk-sharpp} and \secref{sec:zk-sharpp}).

Recall that the family of Reed--Muller codes, denoted $\RM$, is indexed by tuples $\CodeIdx = (\Field,\SCVars,\SCDegree)$, where $\Field$ is a finite field and $\SCVars,\SCDegree$ are positive integers, and the $\CodeIdx$-th code consists of codewords $\Codeword \colon \Field^{\SCVars} \to \Field$ that are the evaluation of an $\SCVars$-variate polynomial $Q$ of individual degree less than $\SCDegree$ over $\Field$. We denote by $\PSRM$ the family that extends $\RM$ with evaluations of all partial sums over certain subcubes of a hypercube:

\begin{definition}[informal]
We denote by $\PSRM$ the linear code family that is indexed by tuples $\CodeIdx = (\Field,\SCVars,\SCDegree,\SCSubset)$, where $\SCSubset$ is a subset of $\Field$, and where the $\CodeIdx$-th code consists of codewords $(\Codeword_{0},\dots,\Codeword_{\SCVars})$ such that there exists an $\SCVars$-variate polynomial $Q$ of individual degree less than $\SCDegree$ over $\Field$ for which $\Codeword_{i} \colon \Field^{\SCVars-i} \to \Field$ is the evaluation of the $i$-th partial sum of $Q$ over $\SCSubset$, i.e, $\Codeword_{i}(\vec{\alpha}) = \sum_{\vec{\gamma} \in \SCSubset^{i}} Q(\vec{\alpha}, \vec{\gamma})$ for every $\vec{\alpha} \in \Field^{\SCVars-i}$.
\end{definition}

The domain size for codes in $\PSRM$ is $\Omega(\SetCardinality{\Field}^{\SCVars})$, but our detector's running time is exponentially smaller.

\begin{theorem}[informal statement of Thm.\ \ref{thm:psrm-succinct-constraint-detection}]
\label{thm:intro-psrm}
The family $\PSRM$ has succinct constraint detection:
\smallskip\\
\centerline{there is a detector algorithm for $\PSRM$ that runs in time $\poly(\log \SetCardinality{\Field} + \SCVars + \SCDegree + \SetCardinality{\SCSubset} + \SetCardinality{\IdxSet})$.}
\end{theorem}

\noindent
We provide intuition for the theorem's proof in \secref{techniques:detect-psrm} and provide the proof's details in \secref{sec:partial-sums}; the proof leverages tools from algebraic complexity theory. (Our proof also shows that the family $\RM$, which is a restriction of $\PSRM$, has succinct constraint detection.) We note that our theorem implies perfect and stateful implementation of a random low-degree multivariate polynomial and its partial sums over any hypercube; this extends \cite{BogdanovW04}'s result, which handles the special case of parity queries to boolean functions on subcubes of the boolean hypercube.

%%%%%%%%%%%%%%%%%%%%%%%%%%%%%%%%%%%%%%%%
\parhead{(2) Reed--Solomon codes, and their PCPPs}
Second, we prove that the family of linear codes comprised of evaluations of low-degree univariate polynomials concatenated with corresponding BS proximity proofs \cite{BS08} has succinct constraint detection. This family is closely related to quasilinear-size PCPs for $\NEXP$ \cite{BS08}, and indeed we use this result to obtain PZK proximity proofs for this family (see \secref{techniques:pzk-iopp-bsrs} and \secref{sec:pzk-codes}), from which we derive \thmref{thm:intro-ntime} (see \secref{techniques:zk-nexp} and \secref{sec:pzk-ntime}).

\begin{definition}[informal]
\label{def:intro-bsrs}
We denote by $\BSCode$ the linear code family indexed by tuples $\CodeIdx = (\Field,\BSSpace,\RSDegree)$, where $\Field$ is an extension field of $\Field_{2}$, $\BSSpace$ is a linear subspace in $\Field$, and $\RSDegree$ is a positive integer; the $\CodeIdx$-th code consists of evaluations on $\BSSpace$ of univariate polynomials $Q$ of degree less than $\RSDegree$, concatenated with corresponding \cite{BS08} proximity proofs.
\end{definition}

The domain size for codes in $\BSCode$ is $\Omega(\SetCardinality{\BSSpace})$, but our detector's running time is exponentially smaller.

\begin{theorem}[informal statement of Thm.\ \ref{thm:bsrs-succinct-constraint-detection}]
\label{thm:intro-bsrs}
The family $\BSCode$ has succinct constraint detection:
\smallskip\\
\centerline{there is a detector algorithm for $\BSCode$ that runs in time $\poly(\log \SetCardinality{\Field} + \Dimension{\BSSpace} + \SetCardinality{\IdxSet})$.}
\end{theorem}

\noindent
We provide intuition for the theorem's proof in \secref{techniques:detect-bsrs} and provide the proof's details in \secref{sec:bsrs-succinct-constraint-detection}; the proof leverages combinatorial properties of the recursive construction of BS proximity proofs.

%\doclearpage
%%%%%%%%%%%%%%%%%%%%%%%%%%%%%%%%%%%%%%%%%%%%%%%%%%%%%%%%%%%%%%%%%%%%%%%%%%%%%%%%
%%%%%%%%%%%%%%%%%%%%%%%%%%%%%%%%%%%%%%%%%%%%%%%%%%%%%%%%%%%%%%%%%%%%%%%%%%%%%%%%
%%%%%%%%%%%%%%%%%%%%%%%%%%%%%%%%%%%%%%%%%%%%%%%%%%%%%%%%%%%%%%%%%%%%%%%%%%%%%%%%
\section{Techniques}
\label{sec:techniques}

We informally discuss intuition behind our algorithms for detecting constraints (\secref{techniques:detecting-constraints}), their connection to perfect zero knowledge (\secref{techniques:masking}), and how we ultimately derive our results about $\sharpP$ and $\NEXP$ (\secref{techniques:deriving-our-results}). Throughout, we provide pointers to the technical sections that contain further details.

%%%%%%%%%%%%%%%%%%%%%%%%%%%%%%%%%%%%%%%%%%%%%%%%%%%%%%%%%%%%%%%%%%%%%%%%%%%%%%%%
%%%%%%%%%%%%%%%%%%%%%%%%%%%%%%%%%%%%%%%%%%%%%%%%%%%%%%%%%%%%%%%%%%%%%%%%%%%%%%%%
\subsection{Detecting constraints for exponentially-large codes}
\label{techniques:detecting-constraints}

As informally introduced in \secref{intro:succinct-constraint-detection}, the \emph{constraint detection problem} corresponding to a linear code family $\CodeClass=\Set{\Code_{\CodeIdx}}_{\CodeIdx}$ with domain $\Domain(\cdot)$ and alphabet $\Field(\cdot)$ is the following: given an index $\CodeIdx \in \Bits^{*}$ and subset $\IdxSet \subseteq \Domain(\CodeIdx)$, output a basis for the space $\{ z \in \Domain(\CodeIdx)^{I} : \; \forall\, \Codeword \in \Code_{\CodeIdx}\,,\, {\sum_{i \in I} z(i) \Codeword(i) = 0}  \}$. In other words, for a given subdomain $\IdxSet$, we wish to determine all linear relations which hold for codewords in $\Code_{\CodeIdx}$ restricted to the subdomain $\IdxSet$.

If a generating matrix for $\Code_{\CodeIdx}$ can be found in polynomial time, this problem can be solved in $\poly(\BitSize{\CodeIdx} + \SetCardinality{\Domain(\CodeIdx)})$ time via Gaussian elimination (such an approach was implicitly taken by \cite{BenSassonCGV16} to construct a perfect zero knowledge simulator for an IOP for $\NP$). However, in our setting $\SetCardinality{\Domain(\CodeIdx)}$ is \emph{exponential} in $\BitSize{\CodeIdx}$, so the straightforward solution is inefficient. With this in mind, we say that $\CodeClass$ has \emph{succinct constraint detection} if there exists an algorithm that solves its constraint detection problem in $\poly(\BitSize{\CodeIdx} + \SetCardinality{\IdxSet})$ time, even if $\SetCardinality{\Domain(\CodeIdx)}$ is exponential in $\BitSize{\CodeIdx}$.

The formal definition of succinct constraint detection is in \secref{sec:definition-of-constraint-detection}. In the rest of this section we provide intuition for two of our theorems: succinct constraint detection for the family $\PSRM$ and for the family $\BSCode$. As will become evident, the techniques that we use to prove the two theorems differ significantly. Perhaps this is because the two codes are quite different: $\PSRM$ has a simple and well-understood algebraic structure, whereas $\BSCode$ is constructed recursively using proof composition.

%%%%%%%%%%%%%%%%%%%%%%%%%%%%%%%%%%%%%%%%%%%%%%%%%%%%%%%%%%%%%%%%%%%%%%%%%%%%%%%%
\subsubsection{From algebraic complexity to detecting constraints for Reed--Muller codes and their partial sums}
\label{techniques:detect-psrm}

The purpose of this section is to provide intuition about the proof of \thmref{thm:intro-psrm}, which states that the family $\PSRM$ has succinct constraint detection. (Formal definitions, statements, and proofs are in \secref{sec:partial-sums}.) We thus outline how to construct an algorithm that detects constraints for the family of linear codes comprised of evaluations of low-degree multivariate polynomials, along with their partial sums. Our construction generalizes the proof of \cite{BogdanovW04}, which solves the special case of parity queries to boolean functions on subcubes of the boolean hypercube by reducing this problem to a probabilistic identity testing problem that is solvable via an algorithm of \cite{RazS05}.

Below, we temporarily ignore the partial sums, and focus on constructing an algorithm that detects constraints for the family of Reed--Muller codes $\RM$, and at the end of the section we indicate how we can also handle partial sums.

%%%%%%%%%%%%%%%%%%%%%%%%%%%%%%%%%%%%%%%%
\parhead{Step 1: phrase as linear algebra problem}
Consider a codeword $\Codeword \colon \Field^{\SCVars} \to \Field$ that is the evaluation of an $\SCVars$-variate polynomial $Q$ of individual degree less than $\SCDegree$ over $\Field$. Note that, for every $\vec{\alpha} \in \Field^{\SCVars}$, $\Codeword(\vec{\alpha})$ equals the inner product of $Q$'s coefficients with the vector $\phi_{\vec{\alpha}}$ that consists of the evaluation of all $\SCDegree^{\SCVars}$ monomials at $\vec{\alpha}$. One can then argue that constraint detection for $\RM$ is equivalent to finding the nullspace of $\{ \phi_{\vec{\alpha}} \}_{\vec{\alpha} \in \IdxSet}$. However, ``writing out'' this $\SetCardinality{\IdxSet} \times \SCDegree^{\SCVars}$ matrix and performing Gaussian elimination is too expensive, so we must solve this linear algebra problem \emph{succinctly}.

%%%%%%%%%%%%%%%%%%%%%%%%%%%%%%%%%%%%%%%%
\parhead{Step 2: encode vectors as coefficients of polynomials}
While each vector $\phi_{\vec{\alpha}}$ is long, it has a succinct description; in fact, we can construct an $\SCVars$-variate polynomial $\Phi_{\vec{\alpha}}$ whose coefficients (after expansion) are the entries of $\phi_{\vec{\alpha}}$, but has an arithmetic circuit of only size $O(\SCVars\SCDegree)$: namely, $\Phi_{\vec{\alpha}}(\vec{\VariableX})
\DefineEqual
\prod_{i=1}^{\SCVars}
(1 + \alpha_{i} \VariableX_{i} + \alpha_{i}^{2} \VariableX_{i}^{2} + \cdots + \alpha_{i}^{\SCDegree-1} \VariableX_{i}^{\SCDegree-1})$.
Computing the nullspace of $\{ \Phi_{\vec{\alpha}} \}_{\vec{\alpha} \in \IdxSet}$ is thus equivalent to computing the nullspace of $\{ \phi_{\vec{\alpha}} \}_{\vec{\alpha} \in \IdxSet}$.

%%%%%%%%%%%%%%%%%%%%%%%%%%%%%%%%%%%%%%%%
\parhead{Step 3: computing the nullspace}
Computing the nullspace of a set of polynomials is a problem in algebraic complexity theory, and is essentially equivalent to the Polynomial Identity Testing (PIT) problem, and so we leverage tools from that area.\footnote{PIT is the following problem: given a polynomial $f$ expressed as an algebraic circuit, is $f$ identically zero? This problem has well-known randomized algorithms \cite{Zippel79,Schwartz80}, but deterministic algorithms for all circuits seem to be beyond current techniques \cite{KabanetsI04}. PIT is a central problem in algebraic complexity theory, and suffices for solving a number of other algebraic problems. We refer the reader to \cite{ShpilkaY10} for a survey.} While there are simple randomized algorithms to solve this problem (see for example \cite[Lemma 8]{Kayal10}), these algorithms, due to a nonzero probability of error, suffice to achieve statistical zero knowledge \emph{but do not suffice to achieve perfect zero knowledge}. To obtain perfect zero knowledge, we need a solution that has \emph{no probability of error}. Derandomizing PIT for arbitrary algebraic circuits seems to be beyond current techniques (as it implies circuit lower bounds \cite{KabanetsI04}), but derandomizations are currently known for some restricted circuit classes. The polynomials that we consider are special: they fall in the well-studied class of ``sum of products of univariates'', and for this case we can invoke the deterministic algorithm of \cite{RazS05} (see also \cite{Kayal10}). (It is interesting that derandomization techniques are ultimately used to obtain a qualitative improvement for an inherently probabilistic task, i.e., perfect sampling of verifier views.)

\medskip
\noindent
The above provides an outline for how to detect constraints for $\RM$. The extension to $\PSRM$, which also includes partial sums, is achieved by considering a more general form of vectors $\phi_{\vec{\alpha}}$ as well as corresponding polynomials $\Phi_{\vec{\alpha}}$. These polynomials also have the special form required for our derandomization. See \secref{sec:partial-sums} for details.

%%%%%%%%%%%%%%%%%%%%%%%%%%%%%%%%%%%%%%%%%%%%%%%%%%%%%%%%%%%%%%%%%%%%%%%%%%%%%%%%
\subsubsection{From recursive code covers to detecting constraints for Reed--Solomon codes and their PCPPs}
\label{techniques:detect-bsrs}

The purpose of this section is to provide intuition about the proof of \thmref{thm:intro-bsrs}, which states that the family $\BSCode$ has succinct constraint detection. (Formal definitions, statements, and proofs are in \secref{sec:bsrs-succinct-constraint-detection}.) We thus outline how to construct an algorithm that detects constraints for the family of linear codes comprised of evaluations of low-degree univariate polynomials concatenated with corresponding BS proximity proofs \cite{BS08}.

Our construction leverages the recursive structure of BS proximity proofs: we identify key combinatorial properties of the recursion that enable ``local'' constraint detection. To define and argue these properties, we introduce two notions that play a central role throughout the proof:
\begin{center}
\begin{minipage}{0.9\textwidth}
A \emph{(local) view} of a linear code $\Code \subseteq \Field^{\Domain}$ is a pair $(\CDomain,\CCode)$ such that $\CDomain \subseteq \Domain$ and $\CCode=\Restrict{\Code}{\CDomain}\subseteq\Field^{\CDomain}$. \\ A \emph{cover} of $\Code$ is a set of local views $\Cover =\Set{(\CDomain_{j},\CCode_{j})}_{j}$ of $\Code$ such that $\Domain = \cup_{j} \CDomain_{j}$.
\end{minipage}
\end{center}

%%%%%%%%%%%%%%%%%%%%%%%%%%%%%%%%%%%%%%%%
\parhead{Combinatorial properties of the recursive step}
Given a finite field $\Field$, domain $\Domain \subseteq \Field$, and degree $d$, let $\Code \DefineEqual \RS[\Field,\Domain,d]$ be the Reed--Solomon code consisting of evaluations on $\Domain$ of univariate polynomials of degree less than $d$ over $\Field$; for concreteness, say that the domain size is $\SetCardinality{\Domain}=2^{n}$ and the degree is $d=\SetCardinality{\Domain}/2=2^{n-1}$.

The first level of \cite{BS08}'s recursion appends to each codeword $f \in \Code$ an auxiliary function $\GetBSProof{1}{f} \colon \Domain'\to\Field$ with domain $\Domain'$ disjoint from $\Domain$. Moreover, the mapping from $f$ to $\GetBSProof{1}{f}$ is linear over $\Field$, so the set $\GetBSCode{\Code}{1} \DefineEqual \{f\FConc\GetBSProof{1}{f}\}_{f\in \Code}$, where $f\FConc\GetBSProof{1}{f} \colon \Domain \cup \Domain' \to\Field$ is the function that agrees with $f$ on $\Domain$ and with $\GetBSProof{1}{f}$ on $\Domain'$, is a linear code over $\Field$. The code $\GetBSCode{\Code}{1}$ is the ``first-level'' code of a BS proximity proof for $f$.

The code $\GetBSCode{\Code}{1}$ has a naturally defined cover $\GetBSCode{\Cover}{1} = \Set{(\CDomain_{j},\CCode_{j})}_{j}$ such that each $\CCode_{j}$ is a Reed--Solomon code $\RS[\Field,\CDomain_{j},d_{j}]$ with $2d_{j} \leq \SetCardinality{\CDomain_{j}} = O(\sqrt{d})$, that is, with rate $1/2$ and block length $O(\sqrt{d})$. We prove several combinatorial properties of this cover:
\begin{itemize}

\item \emph{$\GetBSCode{\Cover}{1}$ is $1$-intersecting.}
For all distinct $j,j'$ in $J$, $\SetCardinality{\CDomain_{j} \cap \CDomain_{j'}} \leq 1$ (namely, the subdomains are almost disjoint).

\item \emph{$\GetBSCode{\Cover}{1}$ is $O(\sqrt{d})$-local.}
Every partial assignment to $O(\sqrt{d})$ domains $\CDomain_{j}$ in the cover that is \emph{locally consistent} with the cover can be extended to a \emph{globally consistent} assignment, i.e., to a codeword of $\GetBSCode{\Code}{1}$. That is, there exists $\IdpParam = O(\sqrt{d})$ such that every partial assignment $h \colon \cup_{\ell=1}^{\IdpParam}\CDomain_{j_{\ell}} \to \Field$ with $\Restrict{h}{\CDomain_{j_\ell}} \in \CCode_{j_\ell}$ (for each $\ell$) equals the restriction to the subdomain $\cup_{\ell=1}^{\IdpParam}\CDomain_{j_{\ell}}$ of some codeword $f \FConc \GetBSProof{1}{f}$ in $\GetBSCode{\Code}{1}$.

\item{\emph{$\GetBSCode{\Cover}{1}$ is $O(\sqrt{d})$-independent}.}
The ability to extend locally-consistent assignments to ``globally-consistent'' codewords of $\GetBSCode{\Code}{1}$ holds in a stronger sense: even when the aforementioned partial assignment $h$ is extended \emph{arbitrarily} to $\IdpParam$ additional point-value pairs, this new partial assignment still equals the restriction of some codeword $f \FConc \GetBSProof{1}{f}$ in $\GetBSCode{\Code}{1}$.

\end{itemize}
The locality property alone already suffices to imply that given a subdomain $\IdxSet\subseteq \Domain\cup\Domain'$ of size $\SetCardinality{\IdxSet}<\sqrt{d}$, we can solve the constraint detection problem on $\IdxSet$ by considering only those constraints that appear in views that intersect $\IdxSet$ (see \lemref{lemma:bsrs-cover-based-implies-succinct}). But $\Code$ has exponential block length so a ``quadratic speedup'' does not yet imply succinct constraint detection. To obtain it, we also leverage the intersection and independence properties to reduce ``locality'' as follows.

%%%%%%%%%%%%%%%%%%%%%%%%%%%%%%%%%%%%%%%%
\parhead{Further recursive steps}
So far we have only considered the first recursive step of a BS proximity proof; we show how to obtain covers with smaller locality (and thereby detect constraints with more efficiency) by considering additional recursive steps. Each code $\CCode_{j}$ in the cover $\GetBSCode{\Cover}{1}$ of $\GetBSCode{\Code}{1}$ is a Reed--Solomon code $\RS[\Field,\CDomain_{j},d_{j}]$ with $\SetCardinality{\CDomain_{j}},d_{j} = O(\sqrt{d})$, and the next recursive step appends to each codeword in $\CCode_{j}$ a corresponding auxiliary function, yielding a new code $\GetBSCode{\Code}{2}$. In turn, $\GetBSCode{\Code}{2}$ has a cover $\GetBSCode{\Cover}{2}$, and another recursive step yields a new code $\GetBSCode{\Code}{3}$, which has its own cover $\GetBSCode{\Cover}{3}$, and so on. The crucial technical observation (\lemref{lemma:bsrs-recursive-and-independent-implies-local}) is that the intersection and independence properties, which hold recursively, enable us to deduce that $\GetBSCode{\Code}{i}$ is $1$-intersecting, $O(\sqrt[2^{i}]{d})$-local, and $O(\sqrt[2^{i}]{d})$-independent; in particular, for $r=\log \log d+O(1)$, $\GetBSCode{\Cover}{r}$ is $1$-intersecting, $O(1)$-local, $O(1)$-independent.

Then, recalling that detecting constraints for local codes requires only the views in the cover that intersect $\IdxSet$ (\lemref{lemma:bsrs-cover-based-implies-succinct}), our constraint detector works by choosing $i\in\{1,\dots, r\}$ such that the cover $\GetBSCode{\Cover}{i}$ is $\poly(\SetCardinality{\IdxSet})$-local, finding in this cover a $\poly(\SetCardinality{\IdxSet})$-size set of $\poly(\SetCardinality{\IdxSet})$-size views that intersect $\IdxSet$, and computing in $\poly(\SetCardinality{\IdxSet})$ time a basis for the dual of each of these views --- thereby proving \thmref{thm:intro-bsrs}.

\begin{remark}
For the sake of those familiar with $\BSCode$ we remark that the domain $\Domain'$ is the carefully chosen subset of $\Field \times \Field$ designated by that construction, the code $\GetBSCode{\Code}{1}$ is the code that evaluates bivariate polynomials of degree $O(\sqrt{d})$ on $\Domain\cup\Domain'$ (along the way mapping $\Domain\subseteq\Field$ to a subset of $\Field\times \Field$), the subdomains $\CDomain_j$ are the axis-parallel ``rows'' and ``columns'' used in that recursive construction, and the codes $\CCode_{j}$ are Reed--Solomon codes of block length $O(\sqrt{d})$. The $O(\sqrt{d})$-locality and independence follow from basic properties of bivariate Reed--Muller codes; see \exmpref{example:rm-line-cover}.
\end{remark}

\begin{remark}
It is interesting to compare the above result with \emph{linear lower bounds on query complexity} for testing proximity to random low density parity check (LDPC) codes \cite{BenSassonHR05,BenSassonGKSV10}. Those results are proved by obtaining a basis for the dual code such that every small-support constraint is spanned by a small subset of that basis. The same can be observed to hold for $\BSCode$, even though this latter code is locally testable with \emph{polylogarithmic query complexity} \cite[Thm.~2.13]{BS08}. The difference between the two cases is due to the fact that, for a random LDPC code, an assignment that satisfies all but a single basis-constraint is (with high probability) far from the code, whereas the recursive and $1$-intersecting structure of $\BSCode$ implies the existence of words that satisfy all but a single basis constraint, yet are negligibly close to being a codeword.
\end{remark}

%%%%%%%%%%%%%%%%%%%%%%%%%%%%%%%%%%%%%%%%%%%%%%%%%%%%%%%%%%%%%%%%%%%%%%%%%%%%%%%%
%%%%%%%%%%%%%%%%%%%%%%%%%%%%%%%%%%%%%%%%%%%%%%%%%%%%%%%%%%%%%%%%%%%%%%%%%%%%%%%%
\subsection{From constraint detection to perfect zero knowledge via masking}
\label{techniques:masking}

We provide intuition about the connection between constraint detection and perfect zero knowledge (\secref{techniques:masking-explanation}), and how we leverage this connection to achieve two intermediate results:
\begin{inparaenum}[(i)]
  \item a sumcheck protocol that is perfect zero knowledge in the Interactive PCP model (\secref{techniques:zk-sumcheck}); and
  \item proximity proofs for Reed--Solomon codes that are perfect zero knowledge in the Interactive Oracle Proof model (\secref{techniques:pzk-iopp-bsrs}).
\end{inparaenum}

%%%%%%%%%%%%%%%%%%%%%%%%%%%%%%%%%%%%%%%%%%%%%%%%%%%%%%%%%%%%%%%%%%%%%%%%%%%%%%%%
\subsubsection{Local simulation of random codewords}
\label{techniques:masking-explanation}

Suppose that the prover and verifier both have oracle access to a codeword $\Codeword \in \Code$, for some linear code $\Code \subseteq \Field^{\Domain}$ with exponential-size domain $\Domain$, and that they need to engage in some protocol that involves $\Codeword$. During the protocol, the prover may leak information about $\Codeword$ that is hard to compute (e.g., requires exponentially-many queries to $\Codeword$), and so would violate zero knowledge (as we see below, this is the case for protocols such as sumcheck).

Rather than directly invoking the protocol, the prover first sends to the verifier a random codeword $r \in \Code$ (as an oracle since $r$ has exponential size) and the verifier replies with a random field element $\rho \in \Field$; then the prover and verifier invoke the protocol on the new codeword $\Codeword' \DefineEqual \rho \Codeword + r \in \Code$ rather than $\Codeword$. Intuitively, running the protocol on $\Codeword'$ now does not leak information about $\Codeword$, because $\Codeword'$ is random in $\Code$ (up to resolvable technicalities). This \emph{random self-reducibility} makes sense for only some protocols, e.g., those where completeness is preserved for any choice of $\rho$ and soundness is broken for only a small fraction of $\rho$; but this will indeed be the case for the settings described below.

The aforementioned \emph{masking} technique was used by \cite{BenSassonCGV16} for codes with polynomial-size domains, but we use it for codes with exponential-size domains, which requires exponentially more efficient simulation techniques. Indeed, to prove perfect zero knowledge, a simulator must be able to reproduce, exactly, the view obtained by any malicious verifier that queries entries of $\Codeword'$, a uniformly random codeword in $\Code$; however, it is too expensive for the simulator to explicitly sample a random codeword and answer the verifier's queries according to it. Instead, the simulator must sample the ``local view'' that the verifier sees while querying $\Codeword'$ at a \emph{small} number of locations $\IdxSet \subseteq \Domain$.

But simulating local views of the form $\Restrict{\Codeword'}{\IdxSet}$ is reducible to detecting \emph{constraints}, i.e., codewords in the dual code $\Dual{\Code}$ whose support is contained in $\IdxSet$. Indeed, if no word in $\Dual{\Code}$ has support contained in $\IdxSet$ then $\Restrict{\Codeword'}{\IdxSet}$ is uniformly random; otherwise, each additional linearly independent constraint of $\Dual{\Code}$ with support contained in $\IdxSet$ further reduces the entropy of $\Restrict{\Codeword'}{\IdxSet}$ in a well-understood manner. (See \lemref{lem:efficient-codeword-simulator} for a formal statement.) In sum, succinct constraint detection enables us to ``implement'' \cite{GoldreichGN10,BogdanovW04} random codewords of $\Code$ despite $\Code$ having exponential size.

We now discuss two concrete protocols for which the aforementioned random self-reducibility applies, and for which we also have constructed suitably-efficient constraint detectors.

%%%%%%%%%%%%%%%%%%%%%%%%%%%%%%%%%%%%%%%%%%%%%%%%%%%%%%%%%%%%%%%%%%%%%%%%%%%%%%%%
\subsubsection{Perfect zero knowledge sumchecks}
\label{techniques:zk-sumcheck}

The celebrated sumcheck protocol \cite{LundFKN92} is \emph{not} zero knowledge. In the sumcheck protocol, the prover and verifier have oracle access to a low-degree $\SCVars$-variate polynomial $\SCPoly$ over a field $\Field$, and the prover wants to convince the verifier that $\sum_{\vec{\alpha} \in \SCSubset^{\SCVars}} \SCPoly(\vec{\alpha}) = 0$ for a given subset $\SCSubset$ of $\Field$. During the protocol, the prover communicates partial sums of $\SCPoly$, which are $\sharpP$-hard to compute and, as such, violate zero knowledge.

We now explain how to use random self-reducibility to make the sumcheck protocol \emph{perfect zero knowledge}, at the cost of moving from the Interactive Proof model to the Interactive PCP model.

%%%%%%%%%%%%%%%%%%%%%%%%%%%%%%%%%%%%%%%%
\parhead{IPCP sumcheck}
Consider the following tweak to the classical sumcheck protocol: rather than invoking sumcheck on $\SCPoly$ directly, the prover first sends to the verifier a random low-degree polynomial $\RandPoly$ that sums to zero on $\SCSubset^{\SCVars}$ (as an oracle), the verifier replies with a random field element $\rho$, and then the two invoke sumcheck on $\MaskedPoly \DefineEqual \rho\SCPoly + \RandPoly$ instead.

Completeness is clear because if $\sum_{\vec{\alpha} \in \SCSubset^{\SCVars}} \SCPoly(\vec{\alpha})=0$ and $\sum_{\vec{\alpha} \in \SCSubset^{\SCVars}} \RandPoly(\vec{\alpha})=0$ then $\sum_{\vec{\alpha} \in \SCSubset^{\SCVars}} (\rho\SCPoly + \RandPoly)(\vec{\alpha})=0$; soundness is also clear because if $\sum_{\vec{\alpha} \in \SCSubset^{\SCVars}} \SCPoly(\vec{\alpha}) \neq 0$ then $\sum_{\vec{\alpha} \in \SCSubset^{\SCVars}} (\rho\SCPoly + \RandPoly)(\vec{\alpha}) \neq 0$ with high probability over $\rho$, regardless of the choice of $\RandPoly$. (For simplicity, we ignore the fact that the verifier also needs to test that $\RandPoly$ has low degree.) We are thus left to show perfect zero knowledge, which turns out to be a much less straightforward argument.

%%%%%%%%%%%%%%%%%%%%%%%%%%%%%%%%%%%%%%%%
\parhead{The simulator}
Before we explain how to argue perfect zero knowledge, we first clarify what we mean by it: since the verifier has oracle access to $\SCPoly$ we cannot hope to `hide' it; nevertheless, we can hope to argue that the verifier, by participating in the protocol, does not learn anything about $\SCPoly$ beyond what the verifier can directly learn by querying $\SCPoly$ (and the fact that $\SCPoly$ sums to zero on $\SCSubset^{\SCVars}$). What we shall achieve is the following: an algorithm that simulates the verifier's view by making as many queries to $\SCPoly$ as the \emph{total} number of verifier queries to either $\SCPoly$ or $\RandPoly$.

On the surface, perfect zero knowledge seems easy to argue, because $\rho\SCPoly + \RandPoly$ seems random among low-degree $\SCVars$-variate polynomials that sum to zero on $\SCSubset^{\SCVars}$. More precisely, consider the simulator that samples a random low-degree polynomial $\MaskedPoly$ that sums to zero on $\SCSubset^{\SCVars}$ and uses it instead of $\rho\SCPoly + \RandPoly$ and answers the verifier queries as follows:
\begin{inparaenum}[(a)]
  \item whenever the verifier queries $\SCPoly(\vec{\alpha})$, respond by querying $\SCPoly(\vec{\alpha})$ and returning the true value;
  \item whenever the verifier queries $\RandPoly(\vec{\alpha})$, respond by querying $\SCPoly(\vec{\alpha})$ and returning $\MaskedPoly(\vec{\alpha}) - \rho\SCPoly(\vec{\alpha})$.
\end{inparaenum}
Observe that the number of queries to $\SCPoly$ made by the simulator equals the number of (mutually) distinct queries to $\SCPoly$ and $\RandPoly$ made by the verifier, as desired.

However, the above reasoning, while compelling, is insufficient. First, $\rho\SCPoly + \RandPoly$ is \emph{not} random because a malicious verifier can choose $\rho$ depending on queries to $\RandPoly$. Second, even if $\rho\SCPoly + \RandPoly$ were random (e.g., the verifier does not query $\RandPoly$ before choosing $\rho$), the simulator must run in polynomial time, both producing correctly-distributed `partial sums' of $\rho\SCPoly + \RandPoly$ and answering queries to $\RandPoly$, but sampling $\MaskedPoly$ alone requires exponential time. In this high level discussion we ignore the first problem (which nonetheless has to be tackled), and focus on the second.

At this point it should be clear from the discussion in \secref{techniques:masking-explanation} that the simulator does not have to sample $\MaskedPoly$ explicitly, but only has to perfectly simulate local views of it by leveraging the fact that it can keep state across queries. And doing so requires solving the succinct constraint detection problem for a suitable code $\Code$. In this case, it suffices to consider the code $\Code=\PSRM$, and our \thmref{thm:intro-psrm} guarantees the required constraint detector.

The above discussion omits several details, so we refer the reader to \secref{sec:zk-sumcheck} for further details.

%%%%%%%%%%%%%%%%%%%%%%%%%%%%%%%%%%%%%%%%%%%%%%%%%%%%%%%%%%%%%%%%%%%%%%%%%%%%%%%%
\subsubsection{Perfect zero knowledge proximity proofs for Reed--Solomon}
\label{techniques:pzk-iopp-bsrs}

Testing proximity of a codeword $\Codeword$ to a given linear code $\Code$ can be aided by a \emph{proximity proof} \cite{DinurR04,BenSassonGHSV06}, which is an auxiliary oracle $\pi$ that facilitates testing (e.g.,  $\Code$ is not locally testable). For example, testing proximity to the Reed--Solomon code, a crucial step towards achieving short PCPs, is aided via suitable proximity proofs \cite{BS08}.

From the perspective of zero knowledge, however, a proximity proof can be `dangerous': a few locations of $\pi$ can in principle leak a lot of information about the codeword $\Codeword$, and a malicious verifier could potentially learn a lot about $\Codeword$ with only a few queries to $\Codeword$ and $\pi$. The notion of zero knowledge for proximity proofs requires that this cannot happen: it requires the existence of an algorithm that simulates the verifier's view by making as many queries to $\Codeword$ as the \emph{total} number of verifier queries to either $\Codeword$ or $\pi$ \cite{IshaiW14}; intuitively, this means that any bit of the proximity proof $\pi$ reveals no more information than one bit of $\Codeword$.

We demonstrate again the use of random self-reducibility and show a general transformation that, under certain conditions, maps a PCP of proximity $(\PCPPProver,\PCPPVerifier)$ for a code $\Code$ to a corresponding $2$-round Interactive Oracle Proof of Proximity (IOPP) for $\Code$ that is \emph{perfect zero knowledge}.

%%%%%%%%%%%%%%%%%%%%%%%%%%%%%%%%%%%%%%%%
\parhead{IOP of proximity for $\Code$}
Consider the following IOP of Proximity: the prover and verifier have oracle access to a codeword $\Codeword$, and the prover wants to convince the verifier that $\Codeword$ is close to $\Code$; the prover first sends to the verifier a random codeword $r$ in $\Code$, and the verifier replies with a random field element $\rho$; the prover then sends the proximity proof $\Proof' \DefineEqual \Prover(\Codeword')$ that attests that $\Codeword' \DefineEqual \rho \Codeword + r$ is close to $\Code$. Note that this is a $2$-round IOP of Proximity for $\Code$, because completeness follows from the fact that $\Code$ is linear, while soundness follows because if $\Codeword$ is far from $\Code$, then so is $\rho \Codeword + r$ for every $r$ with high probability over $\rho$. But is the perfect zero knowledge property satisfied?

%%%%%%%%%%%%%%%%%%%%%%%%%%%%%%%%%%%%%%%%
\parhead{The simulator}
Without going into details, analogously to \secref{techniques:zk-sumcheck}, a simulator must be able to sample local views for random codewords from the code $\OtherCode \DefineEqual \{\, \Codeword \| \Prover(\Codeword) \,\}_{\Codeword \in \Code}$, so the simulator's efficiency reduces to the efficiency of constraint detection for $\OtherCode$. We indeed prove that if $\OtherCode$ has succinct constraint detection then the simulator works out. See \secref{sec:pzk-general-transformation} for further details.

%%%%%%%%%%%%%%%%%%%%%%%%%%%%%%%%%%%%%%%%
\parhead{The case of Reed--Solomon}
The above machinery allows us to derive a perfect zero knowledge IOP of Proximity for Reed--Solomon codes, thanks to our \thmref{thm:intro-bsrs}, which states that the family of linear codes comprised of evaluations of low-degree univariate polynomials concatenated with corresponding BS proximity proofs \cite{BS08} has succinct constraint detection; see \secref{sec:rs-pzk-iopp} for details. This is one of the building blocks of our construction of perfect zero knowledge IOPs for $\NEXP$, as described below in \secref{techniques:zk-nexp}.

%%%%%%%%%%%%%%%%%%%%%%%%%%%%%%%%%%%%%%%%%%%%%%%%%%%%%%%%%%%%%%%%%%%%%%%%%%%%%%%%
%%%%%%%%%%%%%%%%%%%%%%%%%%%%%%%%%%%%%%%%%%%%%%%%%%%%%%%%%%%%%%%%%%%%%%%%%%%%%%%%
\subsection{Achieving perfect zero knowledge beyond $\NP$}
\label{techniques:deriving-our-results}

We outline how to derive our results about perfect zero knowledge for $\sharpP$ and $\NEXP$.

%%%%%%%%%%%%%%%%%%%%%%%%%%%%%%%%%%%%%%%%%%%%%%%%%%%%%%%%%%%%%%%%%%%%%%%%%%%%%%%%
\subsubsection{Perfect zero knowledge for counting problems}
\label{techniques:zk-sharpp}

We provide intuition for the proof of \thmref{thm:intro-sharpp}, which states that the complexity class $\sharpP$ has Interactive PCPs that are perfect zero knowledge.

We first recall the classical (non zero knowledge) Interactive Proof for $\sharpP$ \cite{LundFKN92}. The language $\sharpSATLanguage$, which consists of pairs $(\Formula,\NumSats)$ where $\Formula$ is a $3$-CNF boolean formula and $\NumSats$ is the number of satisfying assignments of $\Formula$, is $\sharpP$-complete, and thus it suffices to construct an IP for it. The IP for $\sharpSATLanguage$ works as follows: the prover and verifier both \emph{arithmetize} $\Formula$ to obtain a low-degree multivariate polynomial $\Arithmetize{\Formula}$ and invoke the (non zero knowledge) sumcheck protocol on the claim ``$\sum_{\vec{\alpha} \in \Bits^{n}} \Arithmetize{\Formula}(\vec{\alpha}) = \NumSats$'', where arithmetic is over a large-enough prime field.

Returning to our goal, we obtain a perfect zero knowledge Interactive PCP by simply replacing the (non zero knowledge) IP sumcheck mentioned above with our perfect zero knowledge IPCP sumcheck, described in \secref{techniques:zk-sumcheck}. In \secref{sec:zk-sharpp} we provide further details, including proving that the zero knowledge guarantees of our sumcheck protocol suffice for this case.

%%%%%%%%%%%%%%%%%%%%%%%%%%%%%%%%%%%%%%%%%%%%%%%%%%%%%%%%%%%%%%%%%%%%%%%%%%%%%%%%
\subsubsection{Perfect zero knowledge for nondeterministic time}
\label{techniques:zk-nexp}

We provide intuition for the proof of \thmref{thm:intro-ntime}, which implies that the complexity class $\NEXP$ has Interactive Oracle Proofs that are perfect zero knowledge. Very informally, the proof consists of combining two building blocks:
\begin{inparaenum}[(i)]
  \item \cite{BenSassonCGV16}'s reduction from $\NEXP$ to \emph{randomizable} linear algebraic constraint satisfaction problems, and
  \item our construction of perfect zero knowledge IOPs of Proximity for Reed--Solomon codes, described in \secref{techniques:pzk-iopp-bsrs}.
\end{inparaenum}
Besides extending \cite{BenSassonCGV16}'s result from $\NP$ to $\NEXP$, our proof provides a conceptual simplification over \cite{BenSassonCGV16} by clarifying how the above two building blocks work together towards the final result. We now discuss this.

%%%%%%%%%%%%%%%%%%%%%%%%%%%%%%%%%%%%%%%%
\parhead{Starting point: \cite{BS08}}
Many PCP constructions consist of two steps:
\begin{inparaenum}[(1)]
  \item arithmetize the statement at hand (in our case, membership of an instance in some $\NEXP$-complete language) by reducing it to a ``PCP-friendly'' problem that looks like a \emph{linear-algebraic} constraint satisfaction problem (LACSP);
  \item design a tester that probabilistically checks witnesses for this LACSP.
\end{inparaenum}
In this paper, as in \cite{BenSassonCGV16}, we take \cite{BS08}'s PCPs for $\NEXP$ as a starting point, where the first step reduces $\NEXP$ to a ``univariate'' LACSP whose witnesses are codewords in a Reed--Solomon code of exponential degree that satisfy certain properties, and whose second step relies on suitable \emph{proximity proofs} \cite{DinurR04,BenSassonGHSV06} for that code. Thus, overall, the PCP consists of two oracles, one being the LACSP witness and the other being the corresponding BS proximity proof, and it is not hard to see that such a PCP is \emph{not} zero knowledge, because both the LACSP witness and its proximity proof reveal hard-to-compute information.

%%%%%%%%%%%%%%%%%%%%%%%%%%%%%%%%%%%%%%%%
\parhead{Step 1: sanitize the proximity proof}
We first address the problem that the BS proximity proof ``leaks'', by simply replacing it with our own perfect zero knowledge analogue. Namely, we replace it with our perfect zero knowledge $2$-round IOP of Proximity for Reed--Solomon codes, described in \secref{techniques:pzk-iopp-bsrs}. This modification ensures that there exists an algorithm that perfectly simulates the verifier's view by making as many queries to the LACSP witness as the \emph{total} number of verifier queries to \emph{either the LACSP witness or other oracles used to facilitate proximity testing}. At this point we have obtained a perfect zero knowledge $2$-round IOP of Proximity for $\NEXP$ (analogous to the notion of a zero knowledge PCP of Proximity \cite{IshaiW14}); this part is where, previously, \cite{BenSassonCGV16} were restricted to $\NP$ because their simulator only handled Reed--Solomon codes with \emph{polynomial} degree while our simulator is efficient even for such codes with \emph{exponential} degree. But we are not done yet: to obtain our goal, we also need to address the problem that the LACSP witness itself ``leaks'' when the verifier queries it, which we discuss next.

%%%%%%%%%%%%%%%%%%%%%%%%%%%%%%%%%%%%%%%%
\parhead{Step 2: sanitize the witness}
Intuitively, we need to inject randomness in the reduction from $\NEXP$ to LACSP because the prover ultimately sends an LACSP witness to the verifier as an oracle, which the verifier can query. This is precisely what \cite{BenSassonCGV16}'s reduction from $\NEXP$ to \emph{randomizable} LACSPs enables, and we thus use their reduction to complete our proof. Informally, given an a-priori query bound $\QueryBound$ on the verifier's queries, the reduction outputs a witness $w$ with the property that one can efficiently sample \emph{another} witness $w'$ whose entries are $\QueryBound$-wise independent. We can then simply use the IOP of Proximity from the previous step on this randomized witness. Moreover, since the efficiency of the verifier is polylogarithmic in $\QueryBound$, we can set $\QueryBound$ to be super-polynomial (e.g., exponential) to preserve zero knowledge against any polynomial number of verifier queries.

\medskip
\noindent
The above discussion is only a sketch and we refer the reader to \secref{sec:pzk-ntime} for further details. One aspect that we did not discuss is that an LACSP witness actually consists of two sub-witnesses, where one is a ``local'' deterministic function of the other, which makes arguing zero knowledge somewhat more delicate.

%%%%%%%%%%%%%%%%%%%%%%%%%%%%%%%%%%%%%%%%%%%%%%%%%%%%%%%%%%%%%%%%%%%%%%%%%%%%%%%%
%%%%%%%%%%%%%%%%%%%%%%%%%%%%%%%%%%%%%%%%%%%%%%%%%%%%%%%%%%%%%%%%%%%%%%%%%%%%%%%%
\subsection{Roadmap}
\label{sec:roadmap}

After providing formal definitions in \secref{sec:basic-notations}, the rest of the paper is organized as summarized by the table below.

\begin{center}
\small
\begin{tabu}{lll|lll}
  \S\ref{sec:partial-sums}
& \textbf{Theorem \ref{thm:intro-psrm}/\ref{thm:psrm-succinct-constraint-detection}}
& detecting constraints for $\PSRM$
& \S\ref{sec:bsrs-succinct-constraint-detection}
& \textbf{Theorem \ref{thm:intro-bsrs}/\ref{thm:bsrs-succinct-constraint-detection}}
& detecting constraints for $\BSCode$
\\[1pt]
& \quad$\big\downarrow$ &&& \quad$\big\downarrow$ & \\[1pt]
  \S\ref{sec:zk-sumcheck}
& \textbf{Theorem \ref{thm:sumcheck-ipcp}}
& PZK IPCP for sumcheck
& \S\ref{sec:pzk-codes}
& \textbf{Theorem \ref{thm:pzk-iopp-for-codes}}
& PZK IOP of Proximity for RS codes
\\[1pt]
& \quad$\big\downarrow$ &&& \quad$\big\downarrow$ & \\[1pt]
  \S\ref{sec:zk-sharpp}
& \textbf{Theorem \ref{thm:intro-sharpp}/\ref{thm:zk-sharpp}}
& PZK IPCP for $\sharpP$
& \S\ref{sec:pzk-ntime}
& \textbf{Theorem \ref{thm:intro-ntime}/\ref{thm:ntime}}
& PZK IOP for $\NEXP$
\end{tabu}
\end{center}

\doclearpage
%%%%%%%%%%%%%%%%%%%%%%%%%%%%%%%%%%%%%%%%%%%%%%%%%%%%%%%%%%%%%%%%%%%%%%%%%%%%%%%%
%%%%%%%%%%%%%%%%%%%%%%%%%%%%%%%%%%%%%%%%%%%%%%%%%%%%%%%%%%%%%%%%%%%%%%%%%%%%%%%%
%%%%%%%%%%%%%%%%%%%%%%%%%%%%%%%%%%%%%%%%%%%%%%%%%%%%%%%%%%%%%%%%%%%%%%%%%%%%%%%%
\section{Definitions}
\label{sec:definitions}

%%%%%%%%%%%%%%%%%%%%%%%%%%%%%%%%%%%%%%%%%%%%%%%%%%%%%%%%%%%%%%%%%%%%%%%%%%%%%%%%
%%%%%%%%%%%%%%%%%%%%%%%%%%%%%%%%%%%%%%%%%%%%%%%%%%%%%%%%%%%%%%%%%%%%%%%%%%%%%%%%
\subsection{Basic notations}
\label{sec:basic-notations}

%%%%%%%%%%%%%%%%%%%%%%%%%%%%%%%%%%%%%%%%
\parhead{Functions, distributions, fields}
We use $f \colon \Domain \to \Range$ to denote a function with domain $\Domain$ and range $\Range$; given a subset $\SubDomain$ of $\Domain$, we use $\Restrict{f}{\SubDomain}$ to denote the restriction of $f$ to $\SubDomain$. Given a distribution $\Distribution$, we write $x \gets \Distribution$ to denote that $x$ is sampled according to $\Distribution$. We denote by $\Field$ a finite field and by $\Field_{\FieldSize}$ the field of size $\FieldSize$; we say $\Field$ is a \emph{binary field} if its characteristic is $2$. Arithmetic operations over $\Field_{q}$ cost $\polylog q$ but we shall consider these to have unit cost (and inspection shows that accounting for their actual polylogarithmic cost does not change any of the stated results).

%%%%%%%%%%%%%%%%%%%%%%%%%%%%%%%%%%%%%%%%
\parhead{Distances}
A distance measure is a function $\Distance \colon \Alphabet^{n} \times \Alphabet^{n} \to [0,1]$ such that for all $x,y,z \in \Alphabet^{n}$:
\begin{inparaenum}[(i)]
  \item $\Distance(x,x) = 0$,
  \item $\Distance(x,y) = \Distance(y,x)$, and
  \item $\Distance(x,y) \leq \Distance(x,z) + \Distance(z,y)$.
\end{inparaenum}
We extend $\Distance$ to distances to sets: given $x \in \Alphabet^{n}$ and $S \subseteq \Alphabet^{n}$, we define $\Distance(x,S) \DefineEqual \min_{y \in S} \Distance(x,y)$ (or $1$ if $S$ is empty). We say that a string $x$ is $\epsilon$-close to another string $y$ if $\Distance(x,y) \leq \epsilon$, and $\epsilon$-far from $y$ if $\Distance(x,y) > \epsilon$; similar terminology applies for a string $x$ and a set $S$. Unless noted otherwise, we use the \emph{relative Hamming distance} over alphabet $\Alphabet$ (typically implicit): $\Distance(x,y) \DefineEqual \SetCardinality{\Set{i \,:\, x_{i} \neq y_{i}}}/n$.
% \HammingDistance{\Alphabet}

%%%%%%%%%%%%%%%%%%%%%%%%%%%%%%%%%%%%%%%%
\parhead{Languages and relations}
We denote by $\Relation$ a (binary ordered) relation consisting of pairs $(\Instance,\Witness)$, where $\Instance$ is the \emph{instance} and $\Witness$ is the \emph{witness}. We denote by $\GetLanguage{\Relation}$ the language corresponding to $\Relation$, and by $\Witnesses{\Relation}{\Instance}$ the set of witnesses in $\Relation$ for $\Instance$ (if $\Instance \not\in \GetLanguage{\Relation}$ then $\Witnesses{\Relation}{\Instance} \DefineEqual \emptyset$). As always, we assume that $\BitSize{\Witness}$ is bounded by some computable function of $\InstanceSize \DefineEqual \BitSize{\Instance}$; in fact, we are mainly interested in relations arising from nondeterministic languages: $\Relation \in \NTIME(\DeciderTime)$ if there exists a $\DeciderTime(\InstanceSize)$-time machine $\DeciderMachine$ such that $\DeciderMachine(\Instance,\Witness)$ outputs $1$ if and only if $(\Instance,\Witness) \in \Relation$. Throughout, we assume that $\DeciderTime(\InstanceSize) \geq \InstanceSize$. We say that $\Relation$ has relative distance $\ProximityParameter_{\Relation} \colon \Naturals \to [0,1]$ if $\ProximityParameter_{\Relation}(\InstanceSize)$ is the minimum relative distance among witnesses in $\Witnesses{\Relation}{\Instance}$ for all $\Instance$ of size $\InstanceSize$. Throughout, we assume that $\ProximityParameter_{\Relation}$ is a constant. \later\ale{TODO: remove assumption}

%%%%%%%%%%%%%%%%%%%%%%%%%%%%%%%%%%%%%%%%
\parhead{Polynomials}
We denote by $\PolynomialRing{\Field}{m}{\VariableX}$ the ring of polynomials in $m$ variables over $\Field$. Given a polynomial $P$ in $\PolynomialRing{\Field}{m}{\VariableX}$, $\IndividualDegree{P}[\VariableX_{i}]$ is the degree of $P$ in the variable $\VariableX_{i}$. We denote by $\PolynomialRingIndOne{\Field}{m}{\VariableX}{d}$ the subspace consisting of $P \in \PolynomialRing{\Field}{m}{\VariableX}$ with $\IndividualDegree{P}[\VariableX_{i}] < d$ for every $i \in \{1,\dots,m\}$.

%%%%%%%%%%%%%%%%%%%%%%%%%%%%%%%%%%%%%%%%
\parhead{Random shifts}
We later use a folklore claim about distance preservation for random shifts in linear spaces.

\begin{claim}
\label{claim:distance-preservation}
Let $n$ be in $\Naturals$, $\Field$ a finite field, $\LinearSpace$ an $\Field$-linear space in $\Field^{n}$, and $x,y \in \Field^{n}$. If $x$ is $\epsilon$-far from $\LinearSpace$, then $\alpha x+y$ is $\epsilon/2$-far from $\LinearSpace$, with probability $1-\SetCardinality{\Field}^{-1}$ over a random $\alpha \in \Field$. (Distances are relative Hamming distances.)
\end{claim}

%%%%%%%%%%%%%%%%%%%%%%%%%%%%%%%%%%%%%%%%%%%%%%%%%%%%%%%%%%%%%%%%%%%%%%%%%%%%%%%%
%%%%%%%%%%%%%%%%%%%%%%%%%%%%%%%%%%%%%%%%%%%%%%%%%%%%%%%%%%%%%%%%%%%%%%%%%%%%%%%%
\subsection{Single-prover proof systems}
\label{sec:models-of-proof-systems}

We use two types of proof systems that combine aspects of interactive proofs \cite{Babai85,GoldwasserMR89} and probabilistically checkable proofs \cite{BFLS91,AroraS98,AroraLMSS98}: \textbf{interactive PCPs} (IPCPs) \cite{KalaiR08} and \textbf{interactive oracle proofs} (IOPs) \cite{BenSassonCS16,ReingoldRR16}. We first describe IPCPs (\secref{sec:ipcp}) and then IOPs (\secref{sec:interactive-oracle-proofs}), which generalize the former.

%%%%%%%%%%%%%%%%%%%%%%%%%%%%%%%%%%%%%%%%%%%%%%%%%%%%%%%%%%%%%%%%%%%%%%%%%%%%%%%%
\subsubsection{Interactive probabilistically checkable proofs}
\label{sec:ipcp}

An \textbf{IPCP} \cite{KalaiR08} is a PCP followed by an IP. Namely, the prover $\Prover$ and verifier $\Verifier$ interact as follows: $\Prover$ sends to $\Verifier$ a probabilistically checkable proof $\Proof$; afterwards, $\Prover$ and $\Verifier^{\Proof}$ engage in an interactive proof. Thus, $\Verifier$ may read a few bits of $\Proof$ but must read subsequent messages from $\Prover$ in full. An \emph{IPCP system} for a relation $\Relation$ is thus a pair $\pair{\Prover}{\Verifier}$, where $\Prover,\Verifier$ are probabilistic interactive algorithms working as described, that satisfies naturally-defined notions of perfect completeness and soundness with a given error $\SoundnessError(\cdot)$; see \cite{KalaiR08} for details.
%\eli{inaccurate: it's not the same sending as further rounds; standard definition is that Verifier has oracle access to a proof oracle, which is separate from the Prover.}
%\ale{I think it's accurate; see Definition 1 of KR08.}

We say that an IPCP has $\NumRounds$ rounds if this ``PCP round'' is followed by a $(\NumRounds-1)$-round interactive proof. (That is, we count the PCP round towards round complexity, unlike \cite{KalaiR08}.) Beyond round complexity, we also measure how many bits the prover sends and how many the verifier reads: the \emph{proof length} $\ProofLength$ is the length of $\Proof$ in bits plus the number of bits in all subsequent prover messages; the \emph{query complexity} $\QueryComplexity$ is the number of bits of $\Proof$ read by the verifier plus the number of bits in all subsequent prover messages (since the verifier must read all of those bits).

In this work, we do not count the number of bits in the verifier messages, nor the number of random bits used by the verifier; both are bounded from above by the verifier's running time, which we do consider. Overall, we say that a relation $\Relation$ belongs to the complexity class $\IPCP[\NumRounds,\ProofLength,\QueryComplexity,\SoundnessError,\ProverTime,\VerifierTime]$ if there is an IOP system for $\Relation$ in which:
\begin{inparaenum}[(1)]
  \item the number of rounds is at most $\NumRounds(\InstanceSize)$;
  \item the proof length is at most $\ProofLength(\InstanceSize)$;
  \item the query complexity is at most $\QueryComplexity(\InstanceSize)$;
  \item the soundness error is $\SoundnessError(\InstanceSize)$;
  \item the prover algorithm runs in time $\ProverTime(\InstanceSize)$;
  \item the verifier algorithm runs in time $\VerifierTime(\InstanceSize)$.
\end{inparaenum}

%%%%%%%%%%%%%%%%%%%%%%%%%%%%%%%%%%%%%%%%%%%%%%%%%%%%%%%%%%%%%%%%%%%%%%%%%%%%%%%%
\subsubsection{Interactive oracle proofs}
\label{sec:interactive-oracle-proofs}

An \textbf{IOP} \cite{BenSassonCS16,ReingoldRR16} is a ``multi-round PCP''. That is, an IOP generalizes an interactive proof as follows: whenever the prover sends to the verifier a message, the verifier does not have to read the message in full but may probabilistically query it. In more detail, a $\NumRounds$-round IOP comprises $\NumRounds$ rounds of interaction. In the $i$-th round of interaction: the verifier sends a message $\Message_{i}$ to the prover; then the prover replies with a message $\OMessage_{i}$ to the verifier, which the verifier can query in this and later rounds (via oracle queries). After the $\NumRounds$ rounds of interaction, the verifier either accepts or rejects.

An \emph{IOP system} for a relation $\Relation$ with soundness error $\SoundnessError$ is thus a pair $\pair{\IOPProver}{\IOPVerifier}$, where $\IOPProver,\IOPVerifier$ are probabilistic interactive algorithms working as described, that satisfies the following properties. (See \cite{BenSassonCS16} for more details.)
\begin{description}

  \item[\textnormal{\emph{Completeness:}}]
For every instance-witness pair $(\Instance,\Witness)$ in the relation $\Relation$, $\Pr[\Interact{\IOPProver(\Instance,\Witness)}{\IOPVerifier(\Instance)}=1]=1$.

  \item[\textnormal{\emph{Soundness:}}]
For every instance $\Instance$ not in $\Relation$'s language and unbounded malicious prover $\Malicious{\IOPProver}$, $\Pr[\Interact{\Malicious{\IOPProver}}{\IOPVerifier(\Instance)}=1] \leq \SoundnessError(\InstanceSize)$.

\end{description}

Beyond round complexity, we also measure how many bits the prover sends and how many the verifier reads: the \emph{proof length} $\ProofLength$ is the total number of bits in all of the prover's messages, and the \emph{query complexity} $\QueryComplexity$ is the total number of bits read by the verifier across all of the prover's messages. Considering all of these parameters, we say that a relation $\Relation$ belongs to the complexity class $\IOP[\NumRounds,\ProofLength,\QueryComplexity,\SoundnessError,\ProverTime,\VerifierTime]$ if there is an IOP system for $\Relation$ in which:
\begin{inparaenum}[(1)]
  \item the number of rounds is at most $\NumRounds(\InstanceSize)$;
  \item the proof length is at most $\ProofLength(\InstanceSize)$;
  \item the query complexity is at most $\QueryComplexity(\InstanceSize)$;
  \item the soundness error is $\SoundnessError(\InstanceSize)$;
  \item the prover algorithm runs in time $\ProverTime(\InstanceSize)$;
  \item the verifier algorithm runs in time $\VerifierTime(\InstanceSize)$.
\end{inparaenum}

%%%%%%%%%%%%%%%%%%%%%%%%%%%%%%%%%%%%%%%%
\parhead{IOP vs.\ IPCP}
An IPCP (see \secref{sec:ipcp}) is a special case of an IOP because an IPCP verifier must read in full all of the prover's messages except the first one (while an IOP verifier may query any part of any prover message). The above complexity measures are consistent with those defined for IPCPs.

%%%%%%%%%%%%%%%%%%%%%%%%%%%%%%%%%%%%%%%%%%%%%%%%%%%%%%%%%%%%%%%%%%%%%%%%%%%%%%%%
\subsubsection{Restrictions and extensions}
\label{sec:restrictions-and-extensions}

The definitions below are about IOPs, but IPCPs inherit all of these definitions because they are a special case of IOP.

%%%%%%%%%%%%%%%%%%%%%%%%%%%%%%%%%%%%%%%%
\parhead{Adaptivity of queries}
An IOP system is \emph{non-adaptive} if the verifier queries are non-adaptive, i.e., the queried locations depend only on the verifier's inputs.

%%%%%%%%%%%%%%%%%%%%%%%%%%%%%%%%%%%%%%%%
\parhead{Public coins}
An IOP system is \emph{public coin} if each verifier message $\Message_{i}$ is chosen uniformly and independently at random, and all of the verifier queries happen after receiving the last prover message.

%%%%%%%%%%%%%%%%%%%%%%%%%%%%%%%%%%%%%%%%
\parhead{Proximity}
An \emph{IOP of proximity} extends the definition of an IOP in the same way that a PCP of proximity extends that of a PCP \cite{DinurR04,BenSassonGHSV06}. An \emph{IOPP system} for a relation $\Relation$ with soundness error $\SoundnessError$ and proximity parameter $\ProximityParameter$ is a pair $(\IOPPProver,\IOPPVerifier)$ that satisfies the following properties.
\begin{description}

  \item[\textnormal{\emph{Completeness:}}]
For every instance-witness pair $(\Instance,\Witness)$ in the relation $\Relation$, $\Pr[\Interact{\IOPPProver(\Instance,\Witness)}{\IOPPVerifier^{\Witness}(\Instance)}=1]=1$.

  \item[\textnormal{\emph{Soundness:}}]
For every instance-witness pair $(\Instance,\Witness)$ with $\DistanceMeasure(\Witness,\Witnesses{\Relation}{\Instance}) \geq \ProximityParameter(\InstanceSize)$ and unbounded malicious prover $\Malicious{\IOPPProver}$, $\Pr[\Interact{\Malicious{\IOPPProver}}{\IOPPVerifier^{\Witness}(\Instance)}=1] \leq \SoundnessError(\InstanceSize)$.

\end{description}
Similarly to above, a relation $\Relation$ belongs to the complexity class $\IOPP[\NumRounds,\ProofLength,\QueryComplexity,\SoundnessError,\ProximityParameter,\ProverTime,\VerifierTime]$ if there is an IOPP system for $\Relation$ with the corresponding parameters. \ale{LATER: (As in Section X, we always assume that $\ProximityParameter$ is less than $\frac{1}{2}\ProximityParameter_{\Relation}$.)}

%%%%%%%%%%%%%%%%%%%%%%%%%%%%%%%%%%%%%%%%
\parhead{Promise relations}
A \emph{promise relation} is a relation-language pair $(\Relation^{\yes}, \Language^{\no})$ with $\GetLanguage{\Relation^{\yes}} \cap \Language^{\no} = \emptyset$. An IOP for a promise relation is the same as an IOP for the (standard) relation $\Relation^{\yes}$, except that soundness need only hold for $\Instance \in \Language^{\no}$. An IOPP for a promise relation is the same as an IOPP for the (standard) relation $\Relation^{\yes}$, except that soundness need only hold for $\Instance \in \GetLanguage{\Relation^{\yes}} \cup \Language^{\no}$.

%%%%%%%%%%%%%%%%%%%%%%%%%%%%%%%%%%%%%%%%%%%%%%%%%%%%%%%%%%%%%%%%%%%%%%%%%%%%%%%%
\subsubsection{Prior constructions}
\label{sec:prior-constructions}

In this paper we give new IPCP and IOP constructions that achieve perfect zero knowledge for various settings. Below we summarize known constructions in these two models.

%%%%%%%%%%%%%%%%%%%%%%%%%%%%%%%%%%%%%%%%
\parhead{IPCPs}
Prior work obtains IPCPs with proof length that depends on the witness size rather than computation size \cite{KalaiR08,GKR08}, and IPCPs with statistical zero knowledge \cite{GoyalIMS10} (see \secref{sec:zero-knowledge} for more details).

%%%%%%%%%%%%%%%%%%%%%%%%%%%%%%%%%%%%%%%%
\parhead{IOPs}
Prior work obtains IOPs with perfect zero knowledge for $\NP$ \cite{BenSassonCGV16}, IOPs with small proof length and query complexity \cite{BenSassonCGRS16}, and an amortization theorem for ``unambiguous'' IOPs \cite{ReingoldRR16}. Also, \cite{BenSassonCS16} show how to compile public-coin IOPs into non-interactive proofs in the random oracle model.

%%%%%%%%%%%%%%%%%%%%%%%%%%%%%%%%%%%%%%%%%%%%%%%%%%%%%%%%%%%%%%%%%%%%%%%%%%%%%%%%
%%%%%%%%%%%%%%%%%%%%%%%%%%%%%%%%%%%%%%%%%%%%%%%%%%%%%%%%%%%%%%%%%%%%%%%%%%%%%%%%
\subsection{Zero knowledge}
\label{sec:zero-knowledge}

We define the notion of zero knowledge for IOPs and IPCPs achieved by our constructions: \emph{perfect zero knowledge via straightline simulators}. This notion is quite strong not only because it unconditionally guarantees perfect simulation of the verifier's view but also because straightline simulation implies desirable properties such as composability. We now provide some context and then give formal definitions.

At a high level, zero knowledge requires that the verifier's view can be efficiently simulated without the prover. Converting the informal statement into a mathematical one involves many choices, including choosing which verifier class to consider (e.g., the honest verifier? all polynomial-time verifiers?), the quality of the simulation (e.g., is it identically distributed to the view? statistically close to it? computationally close to it?), the simulator's dependence on the verifier (e.g., is it non-uniform? or is the simulator universal?), and others. The definitions below consider two variants: perfect simulation via universal simulators against either unbounded-query or bounded-query verifiers.

Moreover, in the case of universal simulators, one distinguishes between a non-blackbox use of the verifier, which means that the simulator takes the verifier's code as input, and a blackbox use of it, which means that the simulator only accesses the verifier via a restricted interface; we consider this latter case. Different models of proof systems call for different interfaces, which grant carefully-chosen ``extra powers'' to the simulator (in comparison to the prover) so to ensure that efficiency of the simulation does not imply the ability to efficiently decide the language. For example: in ZK IPs, the simulator may rewind the verifier; in ZK PCPs, the simulator may adaptively answer oracle queries. In ZK IPCPs and ZK IOPs (our setting), the natural definition would allow a blackbox simulator to rewind the verifier \emph{and also} to adaptively answer oracle queries. The definitions below, however, consider only simulators that are straightline \cite{FeigeS89,DworkS98}, that is they do not rewind the verifier, because our constructions achieve this stronger notion.

Recall that protocols that are perfectly secure with a straightline simulator (in the standalone model) remain secure under concurrent general composition \cite{KushilevitzLR10}. In particular, the above notion of zero knowledge is preserved under such composition and, for example, one can use parallel repetition to reduce soundness error without impacting zero knowledge. Indeed, straightline simulation facilitates many applications, e.g., verifiable secret sharing in \cite{IshaiW14}, and unconditional general universally-composable secure computation in \cite{GoyalIMS10}.

We are now ready to define the notion of perfect zero knowledge via straightline simulators. We first discuss the notion for IOPs, then for IOPs of proximity, and finally for IPCPs.

%%%%%%%%%%%%%%%%%%%%%%%%%%%%%%%%%%%%%%%%%%%%%%%%%%%%%%%%%%%%%%%%%%%%%%%%%%%%%%%%
\subsubsection{ZK for IOPs}
\label{sec:zk-iop}

Below is the definition of perfect zero knowledge (via straightline simulators) for IOPs.

\begin{definition}
\label{def:iop-view}
Let $A,B$ be algorithms and $x,y$ strings. We denote by $\IOPView{B(y)}{A(x)}$ the \defemph{view} of $A(x)$ in an interactive oracle protocol with $B(y)$, i.e., the random variable $(x,r,a_{1},\dots,a_{n})$ where $x$ is $A$'s input, $r$ is $A$'s randomness, and $a_{1},\dots,a_{n}$ are the answers to $A$'s queries into $B$'s messages.

%\parhead{IP:}
%Let $A,B$ be probabilistic interactive algorithms and $x$ a string. We denote by $\IPView{B}{A}{x}$ the view of $A$ in an interactive protocol with $B$ on common input $x$. Namely, the view is the random variable $(x,r,m_{1},\dots,m_{n})$ where $r$ is $A$'s randomness and $m_{1},\dots,m_{n}$ are $B$'s messages.

%\parhead{PCP:}
%Let $A$ be a probabilistic oracle algorithm, $\Pi$ a distribution over strings, and $x$ a string. We denote by $\PCPView{\Pi}{A}{x}$ the view of $A$ in an execution with input $x$ and an oracle sampled from $\Pi$. That is, the view is the random variable $(x,r,a_{1},\dots,a_{n})$ where $r$ is $A$'s randomness and $a_{1},\dots,a_{n}$ are the answers to $A$'s queries.
\end{definition}

Straightline simulators in the context of IPs were used in \cite{FeigeS89}, and later defined in \cite{DworkS98}. The definition below considers this notion in the context of IOPs, where the simulator also has to answer oracle queries by the verifier. Note that since we consider the notion of perfect zero knowledge, the definition of straightline simulation needs to allow the efficient simulator to work even with inefficient verifiers \cite{GoyalIMS10}.

\begin{definition}
\label{def:iop-access}
We say that an algorithm $B$ has \defemph{straightline access} to another algorithm $A$ if $B$ interacts with $A$, without rewinding, by exchanging messages with $A$ and also answering any oracle queries along the way. We denote by $B^{A}$ the concatenation of $A$'s random tape and $B$'s output. (Since $A$'s random tape could be super-polynomially large, $B$ cannot sample it for $A$ and then output it; instead, we restrict $B$ to not see it, and we prepend it to $B$'s output.)
\end{definition}

\begin{definition}
\label{def:zk-iop}
An IOP system $\pair{\IOPProver}{\IOPVerifier}$ for a relation $\Relation$ is perfect zero knowledge (via straightline simulators) \sunderline{against unbounded queries} (resp., \sunderline{against query bound $\QueryBound$}) if there exists a simulator algorithm $\IOPSimulator$ such that for every algorithm (resp., $\QueryBound$-query algorithm) $\Malicious{\IOPVerifier}$ and instance-witness pair $(\Instance,\Witness) \in \Relation$, $\IOPSimulator^{\Malicious{\IOPVerifier}} (\Instance)$ and $\IOPView{\IOPProver(\Instance,\Witness)}{\Malicious{\IOPVerifier}(\Instance)}$ are identically distributed. Moreover, $\Simulator$ must run in time $\poly(\BitSize{\Instance} + \QueryComplexity_{\Malicious{\IOPVerifier}}(\BitSize{\Instance}))$, where $\QueryComplexity_{\Malicious{\IOPVerifier}}(\cdot)$ is $\Malicious{\IOPVerifier}$'s query complexity.
\end{definition}

We say that a relation $\Relation$ belongs to the complexity class $\PZKIOP[\NumRounds,\ProofLength,\QueryComplexity,\SoundnessError,\ProverTime,\VerifierTime,\QueryBound]$ if there is an IOP system for $\Relation$, with the corresponding parameters, that is perfect zero knowledge with query bound $\QueryBound$; also, it belongs to the complexity class $\PZKIOP[\NumRounds,\ProofLength,\QueryComplexity,\SoundnessError,\ProverTime,\VerifierTime,\AnyBound]$ if the same is true with unbounded queries.

%%%%%%%%%%%%%%%%%%%%%%%%%%%%%%%%%%%%%%%%%%%%%%%%%%%%%%%%%%%%%%%%%%%%%%%%%%%%%%%%
\subsubsection{ZK for IOPs of proximity}
\label{sec:zk-iopp}

Below is the definition of perfect zero knowledge (via straightline simulators) for IOPs of proximity. It is a straightforward extension of the corresponding notion for PCPs of proximity, introduced in \cite{IshaiW14}.

\begin{definition}
\label{def:zk-iopp}
An IOPP system $\pair{\IOPPProver}{\IOPPVerifier}$ for a relation $\Relation$ is perfect zero knowledge (via straightline simulators) \sunderline{against unbounded queries} (resp., \sunderline{against query bound $\QueryBound$}) if there exists a simulator algorithm $\IOPSimulator$ such that for every algorithm (resp., $\QueryBound$-query algorithm) $\Malicious{\IOPVerifier}$ and instance-witness pair $(\Instance,\Witness) \in \Relation$, the following two random variables are identically distributed:
\begin{equation*}
\Big( \IOPPSimulator^{\Malicious{\IOPPVerifier},\Witness} (\Instance) \;,\; q_{\IOPPSimulator} \Big)
\quad\text{and}\quad
\Big( \IOPView{\IOPPProver(\Instance,\Witness)}{\Malicious{\IOPPVerifier}^{\Witness}(\Instance)} \;,\; q_{\Malicious{\IOPPVerifier}} \Big)
\enspace,
\end{equation*}
where $q_{\Simulator}$ is the number of queries to $\Witness$ made by $\Simulator$, and $q_{\Malicious{\Verifier}}$ is the number of queries to $\Witness$ or to prover messages made by $\Malicious{\Verifier}$. Moreover, $\Simulator$ must run in time $\poly(\BitSize{\Instance} + \QueryComplexity_{\Malicious{\IOPVerifier}}(\BitSize{\Instance}))$, where $\QueryComplexity_{\Malicious{\IOPVerifier}}(\cdot)$ is $\Malicious{\IOPVerifier}$'s query complexity.
\end{definition}

We say that a relation $\Relation$ belongs to the complexity class $\PZKIOPP[\NumRounds,\ProofLength,\QueryComplexity,\SoundnessError,\ProximityParameter,\ProverTime,\VerifierTime,\QueryBound]$ if there is an IOPP system for $\Relation$, with the corresponding parameters, that is perfect zero knowledge with query bound $\QueryBound$; also, it belongs to the complexity class $\PZKIOPP[\NumRounds,\ProofLength,\QueryComplexity,\SoundnessError,\ProximityParameter,\ProverTime,\VerifierTime,\AnyBound]$ if the same is true with unbounded queries.

\begin{remark}
Analogously to \cite{IshaiW14}, our definition of zero knowledge for IOPs of proximity requires that the number of queries to $\Witness$ by $\Simulator$ equals the total number of queries (to $\Witness$ or prover messages) by $\Malicious{\Verifier}$. Stronger notions are possible: ``the number of queries to $\Witness$ by $\Simulator$ equals the number of queries to $\Witness$ by $\Malicious{\Verifier}$''; or, even more, ``$\Simulator$ and $\Malicious{\Verifier}$ read the same locations of $\Witness$''. The definition above is sufficient for the applications of IOPs of proximity that we consider.
\end{remark}

%%%%%%%%%%%%%%%%%%%%%%%%%%%%%%%%%%%%%%%%%%%%%%%%%%%%%%%%%%%%%%%%%%%%%%%%%%%%%%%%
\subsubsection{ZK for IPCPs}
\label{sec:zk-ipcp}

The definition of perfect zero knowledge (via straightline simulators) for IPCPs follows directly from \defref{def:zk-iop} in \secref{sec:zk-iop} because IPCPs are a special case of IOPs. Ditto for IPCPs of proximity, whose perfect zero knowledge definition follows directly from \defref{def:zk-iopp} in \secref{sec:zk-iopp}. (For comparison, \cite{GoyalIMS10} define statistical zero knowledge IPCPs, also with straightline simulators.)

%%%%%%%%%%%%%%%%%%%%%%%%%%%%%%%%%%%%%%%%%%%%%%%%%%%%%%%%%%%%%%%%%%%%%%%%%%%%%%%%
%%%%%%%%%%%%%%%%%%%%%%%%%%%%%%%%%%%%%%%%%%%%%%%%%%%%%%%%%%%%%%%%%%%%%%%%%%%%%%%%
\subsection{Codes}
\label{sec:codes}

An error correcting code $\Code$ is a set of functions $\Codeword \colon \EvaluationDomain \to \Alphabet$, where $\EvaluationDomain,\Alphabet$ are finite sets known as the domain and alphabet; we write $\Code \subseteq \Alphabet^{\EvaluationDomain}$. The message length of $\Code$ is $\CodeMessageLength \DefineEqual \log_{\SetCardinality{\Alphabet}} \SetCardinality{\Code}$, its block length is $\CodeBlockLength \DefineEqual \SetCardinality{\EvaluationDomain}$, its rate is $\CodeRate \DefineEqual\CodeMessageLength/\CodeBlockLength$, its (minimum) distance is $\CodeAbsoluteDistance \DefineEqual \min \{ \AbsoluteHammingDistance(\Codeword,\OtherCodeword) \,:\, \Codeword,\OtherCodeword \in \Code \,,\, \Codeword \neq \OtherCodeword \}$ when $\AbsoluteHammingDistance$ is the (absolute) Hamming distance, and its (minimum) relative distance is $\CodeRelativeDistance \DefineEqual \CodeAbsoluteDistance/\CodeBlockLength$. At times we write $\CodeMessageLength(\Code),\CodeBlockLength(\Code),\CodeRate(\Code),\CodeAbsoluteDistance(\Code),\CodeRelativeDistance(\Code)$ to make the code under consideration explicit. All the codes we consider are linear codes, discussed next.

%%%%%%%%%%%%%%%%%%%%%%%%%%%%%%%%%%%%%%%%
\parhead{Linearity}
A code $\Code$ is \emph{linear} if $\Alphabet$ is a finite field and $\Code$ is a $\Alphabet$-linear space in $\Alphabet^{\EvaluationDomain}$. The dual code of $\Code$ is the set $\Dual{\Code}$ of functions $\OtherCodeword \colon \EvaluationDomain \to \Alphabet$ such that, for all $\Codeword \colon \EvaluationDomain \to \Alphabet$, $\InnerProduct{\OtherCodeword}{\Codeword} \DefineEqual \sum_{i \in \EvaluationDomain} \OtherCodeword(i)\Codeword(i)=0$. We denote by $\Dimension{\Code}$
the dimension of $\Code$; it holds that $\Dimension{\Code} + \Dimension{\Dual{\Code}} = \CodeBlockLength$ and $\Dimension{\Code}=\CodeMessageLength$ (dimension equals message length).

%%%%%%%%%%%%%%%%%%%%%%%%%%%%%%%%%%%%%%%%
\parhead{Code families}
A code family $\Class{\Code} = \{\Code_{\CodeIdx}\}_{\CodeIdx \in \Strings}$ has domain $\EvaluationDomain(\cdot)$ and alphabet $\Field(\cdot)$ if each code $\Code_{\CodeIdx}$ has domain $\EvaluationDomain(\CodeIdx)$ and alphabet $\Field(\CodeIdx)$. Similarly, $\Class{\Code}$ has message length $\CodeMessageLength(\cdot)$, block length $\CodeBlockLength(\cdot)$, rate $\CodeRate(\cdot)$, distance $\CodeAbsoluteDistance(\cdot)$, and relative distance $\CodeRelativeDistance(\cdot)$ if each code $\Code_{\CodeIdx}$ has message length $\CodeMessageLength(\CodeIdx)$, block length $\CodeBlockLength(\CodeIdx)$, rate $\CodeRate(\CodeIdx)$, distance $\CodeAbsoluteDistance(\CodeIdx)$, and relative distance $\CodeRelativeDistance(\CodeIdx)$. We also define $\CodeRate(\Class{\Code}) \DefineEqual \inf_{\CodeIdx \in \Naturals} \CodeRate(\CodeIdx)$ and $\CodeRelativeDistance(\Class{\Code}) \DefineEqual \inf_{\CodeIdx \in \Naturals} \CodeRelativeDistance(\CodeIdx)$.

\parhead{Reed--Solomon codes}
The Reed--Solomon (RS) code is the code consisting of evaluations of \emph{univariate} low-degree polynomials: given a field $\Field$, subset $\RSDomain$ of $\Field$, and positive integer $\RSDegree$ with $\RSDegree \leq \SetCardinality{\RSDomain}$, we denote by $\RSCode{\Field}{\RSDomain}{\RSDegree}$ the linear code consisting of evaluations $\Codeword \colon \RSDomain \to \Field$ over $\RSDomain$ of polynomials in $\Field^{<\RSDegree}[\VariableX]$. The code's message length is $\CodeMessageLength = \RSDegree$, block length is $\CodeBlockLength = \SetCardinality{\RSDomain}$, rate is $\CodeRate = \frac{\RSDegree}{\SetCardinality{\RSDomain}}$, and relative distance is $\CodeRelativeDistance = 1-\frac{\RSDegree-1}{\SetCardinality{\RSDomain}}$.

\parhead{Reed--Muller codes}
The Reed--Muller (RM) code is the code consisting of evaluations of \emph{multivariate} low-degree polynomials: given a field $\Field$, subset $\RMDomain$ of $\Field$, and positive integers $\RMVars,\RMDegree$ with $\RMDegree \leq \SetCardinality{\RMDomain}$, we denote by $\RMCode{\Field}{\RMDomain}{\RMVars}{\RMDegree}$ the linear code consisting of evaluations $\Codeword \colon \RMDomain^{\RMVars} \to \Field$ over $\RMDomain^{\RMVars}$ of polynomials in $\PolynomialRingIndOne{\Field}{\RMVars}{\VariableX}{\RMDegree}$ (i.e., we bound individual degrees rather than their sum). The code's message length is $\CodeMessageLength = \RMDegree^{\RMVars}$, block length is $\CodeBlockLength = \SetCardinality{\RMDomain}^{\RMVars}$, rate is $\CodeRate = (\frac{\RMDegree}{\SetCardinality{\RMDomain}})^{\RMVars}$, and relative distance is $\CodeRelativeDistance = (1-\frac{\RMDegree-1}{\SetCardinality{\RMDomain}})^{\RMVars}$.
%Also, one can verify that
%\begin{inparaenum}[(i)]
%  \item $\RSCode{\Field}{\RMDomain}{\RMDegree} = \RMCode{\Field}{\RMDomain}{1}{\RMDegree}$ and
%  \item $\RMCode{\Field}{\RMDomain}{\RMVars}{\RMDegree}$ is the $\RMVars$-wise tensor product \cite{Wolf65,WolfE63,Tanner81} of $\RSCode{\Field}{\RMDomain}{\RMDegree}$.
%\end{inparaenum}

\doclearpage
%%%%%%%%%%%%%%%%%%%%%%%%%%%%%%%%%%%%%%%%%%%%%%%%%%%%%%%%%%%%%%%%%%%%%%%%%%%%%%%%
%%%%%%%%%%%%%%%%%%%%%%%%%%%%%%%%%%%%%%%%%%%%%%%%%%%%%%%%%%%%%%%%%%%%%%%%%%%%%%%%
%%%%%%%%%%%%%%%%%%%%%%%%%%%%%%%%%%%%%%%%%%%%%%%%%%%%%%%%%%%%%%%%%%%%%%%%%%%%%%%%
\section{Succinct constraint detection}
\label{sec:succinct-constraint-detection}

We introduce the notion of \emph{succinct constraint detection} for linear codes. This notion plays a crucial role in constructing perfect zero knowledge simulators for super-polynomial complexity classes (such as $\sharpP$ and $\NEXP$), but we believe that this naturally-defined notion is also of independent interest. Given a linear code $\Code\subseteq \Field^{\Domain}$ we refer to its dual code $\Dual{\Code} \subseteq \Field^{\Domain}$ as the \emph{constraint space} of $\Code$. The \emph{constraint detection problem} corresponding to a family of linear codes $\CodeClass=\Set{\Code_{\CodeIdx}}_{\CodeIdx}$ with domain $\Domain(\cdot)$ and alphabet $\Field(\cdot)$ is the following:
\begin{center}
Given an index $\CodeIdx$ and subset $\IdxSet \subseteq \Domain(\CodeIdx)$, output a basis for $\{ z \in \Domain(\CodeIdx)^{I} : \; \forall\, \Codeword \in \Code_{\CodeIdx}\,,\, {\sum_{i \in I} z(i) \Codeword(i) = 0}  \}$.\footnote{In fact, the following weaker definition suffices for the applications in our paper: \emph{given an index $\CodeIdx$ and subset $\IdxSet \subseteq \Domain(\CodeIdx)$, output $\OtherCodeword \in \Field(\CodeIdx)^{\IdxSet}$ such that $\sum_{i \in \IdxSet} \OtherCodeword(i) \Codeword(i) = 0$ for all $\Codeword \in \Code_{\CodeIdx}$, or $\IndepTag$ if no such $\OtherCodeword$ exists}. We achieve the stronger definition, which is also easier to work with.}
\end{center}
If $\SetCardinality{\Domain(\CodeIdx)}$ is polynomial in $\BitSize{\CodeIdx}$ and a generating matrix for $\Code_{\CodeIdx}$ can be found in polynomial time, this problem can be solved in $\poly(\BitSize{\CodeIdx} + \SetCardinality{\IdxSet})$ time via Gaussian elimination; such an approach was implicitly taken by \cite{BenSassonCGV16} to construct a perfect zero knowledge simulator for an IOP for $\NP$. However, in our setting, $\SetCardinality{\Domain(\CodeIdx)}$ is \emph{exponential} in $\BitSize{\CodeIdx}$ and  $\SetCardinality{\IdxSet}$, and the aforementioned generic solution requires exponential time. With this in mind, we say $\CodeClass$ has \emph{succinct constraint detection} if there exists an algorithm that solves the constraint detection problem in $\poly(\BitSize{\CodeIdx} + \SetCardinality{\IdxSet})$ time when $\SetCardinality{\Domain(\CodeIdx)}$ is \emph{exponential} in $\BitSize{\CodeIdx}$. After defining succinct constraint detection in \secref{sec:definition-of-constraint-detection}, we proceed as follows.
\begin{itemize}

  \item In \secref{sec:partial-sums}, we construct a succinct constraint detector for the family of linear codes comprised of evaluations of partial sums of low-degree polynomials. The construction of the detector exploits derandomization techniques from algebraic complexity theory. Later on (in \secref{sec:zk-sumcheck}), we leverage this result to construct a perfect zero knowledge simulator for an IPCP for $\sharpP$.

  \item In \secref{sec:bsrs-succinct-constraint-detection}, we construct a succinct constraint detector for the family of evaluations of univariate polynomials concatenated with corresponding BS proximity proofs \cite{BS08}. The construction of the detector exploits the recursive structure of these proximity proofs. Later on (in \secref{sec:pzk-ntime}), we leverage this result to construct a perfect zero knowledge simulator for an IOP for $\NEXP$; this simulator can be interpreted as an analogue of \cite{BenSassonCGV16}'s simulator that runs \emph{exponentially faster} and thus enables us to ``scale up'' from $\NP$ to $\NEXP$.

\end{itemize}
Throughout this section we assume familiarity with terminology and notation about codes, introduced in \secref{sec:codes}. We assume for simplicity that $\BitSize{\CodeIdx}$, the number of bits used to represent $\CodeIdx$, is at least $\log \EvaluationDomain(\CodeIdx) + \log \Field(\CodeIdx)$; if this does not hold, then one can replace $\BitSize{\CodeIdx}$ with $\BitSize{\CodeIdx} + \log \EvaluationDomain(\CodeIdx) + \log \Field(\CodeIdx)$ throughout the section.

\begin{remark}[sparse representation]
In this section we make statements about vectors $v$ in $\Field^{\Domain}$ where the cardinality of the domain $\Domain$ may be super-polynomial. When such statements are computational in nature, we assume that $v$ is not represented as a list of $\SetCardinality{\Domain}$ field elements (which requires $\Omega(\SetCardinality{\Domain} \log \SetCardinality{\Field})$ bits) but, instead, assume that $v$ is represented as a list of the elements in $\Support{v}$ (and each element comes with its index in $\Domain$); this \emph{sparse} representation only requires $\Omega(\SetCardinality{\Support{v}} \cdot (\log \SetCardinality{\Domain} + \log \SetCardinality{\Field}))$ bits.
\end{remark}

%%%%%%%%%%%%%%%%%%%%%%%%%%%%%%%%%%%%%%%%%%%%%%%%%%%%%%%%%%%%%%%%%%%%%%%%%%%%%%%%
%%%%%%%%%%%%%%%%%%%%%%%%%%%%%%%%%%%%%%%%%%%%%%%%%%%%%%%%%%%%%%%%%%%%%%%%%%%%%%%%
\subsection{Definition of succinct constraint detection}
\label{sec:definition-of-constraint-detection}

Formally define the notion of a \emph{constraint detector}, and the notion of \emph{succinct constraint detection}.

\begin{definition}
\label{def:constraint-detector}
Let $\CodeClass = \{ \Code_{\CodeIdx} \}_{\CodeIdx}$ be a linear code family with domain $\EvaluationDomain(\cdot)$ and alphabet $\Field(\cdot)$. A \defemph{constraint detector} for $\CodeClass$ is an algorithm that, on input an index $\CodeIdx$ and subset $\IdxSet \subseteq \Domain(\CodeIdx)$, outputs a basis for the space
\begin{equation*}
\Big\{ \OtherCodeword \in \Domain(\CodeIdx)^{\IdxSet} : \, \forall\, \Codeword \in \Code_{\CodeIdx}\,, \sum_{i \in \IdxSet} \OtherCodeword(i) \Codeword(i) \Big\}
\enspace.
\end{equation*}
We say that $\CodeClass$ has a \defemph{$\CodeTime(\cdot, \cdot)$-time constraint detection} if there exists a detector for $\CodeClass$ running in time $\CodeTime(\CodeIdx, \ListSize)$; we also say that $\CodeClass$ has \defemph{succinct constraint detection} if it has $\poly(\BitSize{\CodeIdx} + \ListSize)$-time constraint detection.
\end{definition}

A constraint detector induces a corresponding probabilistic algorithm for `simulating' answers to queries to a random codeword; this is captured by the following lemma, the proof of which is in \appref{sec:proof-of-efficient-codeword-simulator}. We shall use such probabilistic algorithms in the construction of perfect zero knowledge simulators (see \secref{sec:zk-sumcheck} and \secref{sec:pzk-ntime}).

\begin{lemma}
\label{lem:efficient-codeword-simulator}
Let $\CodeClass = \{ \Code_{\CodeIdx} \}_{\CodeIdx}$ be a linear code family with domain $\EvaluationDomain(\cdot)$ and alphabet $\Field(\cdot)$ that has $\CodeTime(\cdot, \cdot)$-time constraint detection. Then there exists a probabilistic algorithm $\CodeSimAlgorithm$ such that, for every index $\CodeIdx$, set of pairs $S = \{(\alpha_{1}, \beta_{1}), \dots, (\alpha_{\ListSize}, \beta_{\ListSize})\} \subseteq \EvaluationDomain(\CodeIdx) \times \Field(\CodeIdx)$, and pair $(\alpha,\beta) \in \EvaluationDomain(\CodeIdx) \times \Field(\CodeIdx)$,
\begin{equation*}
\Pr\Big[
\CodeSimAlgorithm(\CodeIdx, S, \alpha) = \beta
\Big]
=
\Pr_{\Codeword \gets \Code_{\CodeIdx}}
\left[
\Codeword(\alpha) = \beta
\pST
\begin{array}{c}
\Codeword(\alpha_{1}) = \beta_{1} \\
\vdots \\
\Codeword(\alpha_{\ListSize}) = \beta_{\ListSize}
\end{array}
\right]\enspace.
\end{equation*}
Moreover $\CodeSimAlgorithm$ runs in time $\CodeTime(\CodeIdx, \ListSize) + \poly(\log \SetCardinality{\Field(\CodeIdx)} + \ListSize)$.
\end{lemma}

For the purposes of \emph{constructing} a constraint detector, the sufficient condition given in \lemref{lem:bsrs-dual-char} below is sometimes easier to work with. To state it we need to introduce two ways of restricting a code, and explain how these restrictions interact with taking duals; the interplay between these is delicate (see \remref{remark:expand-vs-dual}).

\begin{definition}
\label{def:puncturing}
Given a linear code $\Code \subseteq \Field^{\Domain}$ and a subset $\IdxSet \subseteq \Domain$, we denote by
\begin{inparaenum}[(i)]
  \item $\Puncture{\Code}{\IdxSet}$ the set consisting of the codewords $\Codeword \in \Code$ for which $\Support{\Codeword} \subseteq \IdxSet$, and
  \item $\Restrict{\Code}{\IdxSet}$ the restriction to $\IdxSet$ of codewords $\Codeword \in \Code$.
\end{inparaenum}
\end{definition}

Note that $\Puncture{\Code}{\IdxSet}$ and $\Restrict{\Code}{\IdxSet}$ are \emph{different notions}. Consider for example the $1$-dimensional linear code $\Code=\Set{00,11}$ in $\Field_{2}^{\Set{1,2}}$ and the subset $\IdxSet=\Set{1}$: it holds that $\Puncture{\Code}{\IdxSet}=\Set{00}$ and $\Restrict{\Code}{\IdxSet} = \Set{0,1}$. In particular, codewords in $\Puncture{\Code}{\IdxSet}$ are defined over $\Domain$, while codewords in $\Restrict{\Code}{\IdxSet}$ are defined over $\IdxSet$. Nevertheless, throughout this section, we sometimes compare vectors defined over different domains, with the implicit understanding that the comparison is conducted over the union of the relevant domains, by filling in zeros in the vectors' undefined coordinates. For example, we may write $\Puncture{\Code}{\IdxSet} \subseteq \Restrict{\Code}{\IdxSet}$ to mean that $\Set{00} \subseteq \Set{00,10}$ (the set obtained from $\Set{0,1}$ after filling in the relevant zeros).

\begin{claim}
\label{claim:bsrs-dual-of-projection}
Let $\Code$ be a linear code with domain $\Domain$ and alphabet $\Field$. For every $\IdxSet \subseteq \Domain$,
\begin{equation*}
\Dual{(\Restrict{\Code}{\IdxSet})} = \Puncture{(\Dual{\Code})}{\IdxSet} \enspace,
\end{equation*}
that is,
\begin{equation*}
\Big\{ \OtherCodeword \in \Domain(\CodeIdx)^{\IdxSet} : \, \forall\, \Codeword \in \Code_{\CodeIdx}\,, \sum_{i \in \IdxSet} \OtherCodeword(i) \Codeword(i) \Big\}
=
\Big\{
\OtherCodeword\in\Dual{\Code_{\CodeIdx}} : \Support{\OtherCodeword} \subseteq \IdxSet
\Big\}
\enspace.
\end{equation*}
\end{claim}

\begin{proof}
For the containment $\Puncture{(\Dual{\Code})}{\IdxSet} \subseteq \Dual{(\Restrict{\Code}{\IdxSet})}$: if $\OtherCodeword \in \Dual{\Code}$ and $\Support{\OtherCodeword} \subseteq \IdxSet$ then $\OtherCodeword$ lies in the dual of $\Restrict{\Code}{\IdxSet}$ because it suffices to consider the subdomain $\IdxSet$ for determining duality. For the reverse containment $\Puncture{(\Dual{\Code})}{\IdxSet} \supseteq \Dual{(\Restrict{\Code}{\IdxSet})}$: if $\OtherCodeword \in \Dual{(\Restrict{\Code}{\IdxSet})}$ then $\Support{\OtherCodeword} \subseteq \IdxSet$ (by definition) so that $\InnerProduct{\OtherCodeword}{\Codeword} = \InnerProduct{\OtherCodeword}{\Restrict{\Codeword}{\IdxSet}}$ for every $\Codeword \in \Code$, and the latter inner product equals $0$ because $\OtherCodeword$ is in the dual of $\Restrict{\Code}{\IdxSet}$; in sum $\OtherCodeword$ is dual to (all codewords in) $\Code$ and its support is contained in $\IdxSet$, so $\OtherCodeword$ belongs to $\Puncture{(\Dual{\Code})}{\IdxSet}$, as claimed.
\end{proof}

Observe that \clmref{claim:bsrs-dual-of-projection} tells us the constraint detection is equivalent to determining a basis of $\Dual{(\Restrict{\Code_{\CodeIdx}}{\IdxSet})} = \Puncture{(\Dual{\Code_{\CodeIdx}})}{\IdxSet}$. The following lemma asserts that if, given a subset $\IdxSet \subseteq \Domain$, we can find a set of constraints $\LocalSet$ in $\Dual{\Code}$ that spans $\Puncture{(\Dual{\Code})}{\IdxSet}$ then we can solve the constraint detection problem for $\Code$; we defer the proof of the lemma to \appref{sec:proof-of-bsrs-dual-char}.

\begin{lemma}
\label{lem:bsrs-dual-char}
Let $\CodeClass = \Set{\Code_{\CodeIdx}}_{\CodeIdx}$ be a linear code family with domain $\Domain(\cdot)$ and alphabet $\Field(\cdot)$. If there exists an algorithm that, on input an index $\CodeIdx$ and subset $\IdxSet \subseteq \Domain(\CodeIdx)$, outputs in $\poly(\BitSize{\CodeIdx} + \SetCardinality{\IdxSet})$ time a subset $\LocalSet \subseteq \Field(\CodeIdx)^{\Domain(\CodeIdx)}$ (in sparse representation) with $\Puncture{(\Dual{\Code_{\CodeIdx}})}{\IdxSet} \subseteq \Span(\LocalSet) \subseteq \Dual{\Code_{\CodeIdx}}$, then $\CodeClass$ has succinct constraint detection.
\end{lemma}

\begin{remark}
\label{remark:expand-vs-dual}
The following operations do \emph{not} commute:
\begin{inparaenum}[(i)]
  \item expanding the domain via zero padding (for the purpose of comparing vectors over different domains), and
  \item taking the dual of the code.
\end{inparaenum}
Consider for example the code $\Code = \Set{0} \subseteq \Field_2^{\Set{1}}$: its dual code is $\Dual{\Code} = \Set{0,1}$ and, when expanded to $\Field_{2}^{\Set{1,2}}$, the dual code is expanded to $\Set{(0,0),(1,0)}$; yet, when $\Code$ is expanded to $\Field_{2}^{\Set{1,2}}$ it produces the code $\Set{(0,0)}$ and its dual code is $\Set{(0,0),(1,0),(0,1),(1,1)}$. To resolve ambiguities (when asserting an equality as in \clmref{claim:bsrs-dual-of-projection}), we adopt the convention that expansion is done \emph{always last} (namely, as late as possible without having to compare vectors over different domains).
\end{remark}

%%%%%%%%%%%%%%%%%%%%%%%%%%%%%%%%%%%%%%%%%%%%%%%%%%%%%%%%%%%%%%%%%%%%%%%%%%%%%%%%
%%%%%%%%%%%%%%%%%%%%%%%%%%%%%%%%%%%%%%%%%%%%%%%%%%%%%%%%%%%%%%%%%%%%%%%%%%%%%%%%
\subsection{Partial sums of low-degree polynomials}
\label{sec:partial-sums}

We show that evaluations of partial sums of low-degree polynomials have succinct constraint detection (see \defref{def:constraint-detector}). In the following, $\Field$ is a finite field, $\SCVars,\SCDegree$ are positive integers, and $\SCSubset$ is a subset of $\Field$; also, $\PolynomialRingIndOne{\Field}{\SCVars}{\VariableX}{\SCDegree}$ denotes the subspace of $\PolynomialRing{\Field}{m}{\VariableX}$ consisting of those polynomials with individual degrees less than $\SCDegree$. Moreover, given $Q \in \PolynomialRingIndOne{\Field}{\SCVars}{\VariableX}{\SCDegree}$ and $\vec{\alpha} \in \Field^{\leq\SCVars}$ (vectors over $\Field$ of length at most $\SCVars$), we define $Q(\vec{\alpha}) \DefineEqual \sum_{\vec{\gamma} \in \SCSubset^{\SCVars - |\vec{\alpha}|}} Q(\vec{\alpha}, \vec{\gamma})$, i.e., the answer to a query that specifies only a suffix of the variables is the sum of the values obtained by letting the remaining variables range over $\SCSubset$. We begin by defining the code that we study, which extends the Reed--Muller code (see \secref{sec:codes}) with partial sums.

\begin{definition}
\label{def:psrm}
We denote by $\PSRMCode{\Field}{\SCVars}{\SCDegree}{\SCSubset}$ the linear code that comprises evaluations of partial sums of polynomials in $\PolynomialRingIndOne{\Field}{\SCVars}{\VariableX}{\SCDegree}$; more precisely, $\PSRMCode{\Field}{\SCVars}{\SCDegree}{\SCSubset} \DefineEqual \{ \Codeword_{Q} \}_{Q \in \PolynomialRingIndOne{\Field}{\SCVars}{\VariableX}{\SCDegree}}$ where $\Codeword_{Q} \colon \Field^{\leq \SCVars} \to \Field$ is the function defined by $\Codeword_{Q}(\vec{\alpha}) \DefineEqual \sum_{\vec{\gamma} \in \SCSubset^{\SCVars - |\vec{\alpha}|}} Q(\vec{\alpha}, \vec{\gamma})$ for each $\vec{\alpha} \in \Field^{\leq \SCVars}$.\footnote{Note that $\PSRMCode{\Field}{\SCVars}{\SCDegree}{\SCSubset}$ is indeed linear: for every $\Codeword_{Q_{1}}, \Codeword_{Q_2} \in \PSRMCode{\Field}{\SCVars}{\SCDegree}{\SCSubset}$, $a_{1}, a_{2} \in \Field$, and $\vec{\alpha} \in \Field^{\leq \SCVars}$, it holds that
$
a_{1} \Codeword_{Q_{1}}(\vec{\alpha}) + a_{2} \Codeword_{Q_2}(\vec{\alpha})
= a_{1} \sum_{\vec{\gamma} \in \SCSubset^{\SCVars - |\vec{\alpha}|}} Q_{1}(\vec{\alpha}, \vec{\gamma}) + a_{2}\sum_{\vec{\gamma} \in \SCSubset^{\SCVars - |\vec{\alpha}|}} Q_{2}(\vec{\alpha}, \vec{\gamma})
= \sum_{\vec{\gamma} \in \SCSubset^{\SCVars - |\vec{\alpha}|}} (a_{1} Q_{1} + a_{2} Q_{2})(\vec{\alpha}, \vec{\gamma})
= \Codeword_{a_{1} Q_{1} + a_{2} Q_{2}}(\vec{\alpha})
$.
But $\Codeword_{a_{1} Q_{1} + a_{2} Q_{2}} \in \PSRMCode{\Field}{\SCVars}{\SCDegree}{\SCSubset}$, since $\PolynomialRingIndOne{\Field}{\SCVars}{\VariableX}{\SCDegree}$ is a linear space.}
We denote by $\PSRM$ the linear code family indexed by tuples $\CodeIdx = (\Field,\SCVars,\SCDegree,\SCSubset)$ and where the $\CodeIdx$-th code equals $\PSRMCode{\Field}{\SCVars}{\SCDegree}{\SCSubset}$. (We represent indices $\CodeIdx$ so to ensure that $\BitSize{\CodeIdx} = \Theta(\log \SetCardinality{\Field} + \SCVars + \SCDegree + \SetCardinality{\SCSubset})$.)
\end{definition}

We prove that the linear code family $\PSRM$ has succinct constraint detection:

\begin{theorem}[formal statement of \ref{thm:intro-psrm}]
\label{thm:psrm-succinct-constraint-detection}
$\PSRM$ has $\poly(\log \SetCardinality{\Field} + \SCVars + \SCDegree + \SetCardinality{\SCSubset} + \ListSize)$-time constraint detection.
\end{theorem}

Combined with \lemref{lem:efficient-codeword-simulator}, the theorem above implies that there exists a probabilistic polynomial-time algorithm for answering queries to a codeword sampled at random from $\PSRM$, as captured by the following corollary.

\begin{corollary}
\label{cor:efficient-poly-simulator}
There exists a probabilistic algorithm $\CodeSimAlgorithm$ such that, for every finite field $\Field$, positive integers $\SCVars,\SCDegree$, subset $\SCSubset$ of $\Field$, subset $S = \{(\alpha_{1},\beta_{1}), \dots, (\alpha_{\ListSize}, \beta_{\ListSize})\} \subseteq \Field^{\leq\SCVars} \times \Field$, and $(\alpha,\beta) \in \Field^{\leq\SCVars} \times \Field$,
\begin{equation*}
\Pr\Big[
\CodeSimAlgorithm(\Field,\SCVars,\SCDegree,\SCSubset,S,\alpha) = \beta
\Big]
=
\Pr_{\RandPoly \gets \PolynomialRingIndOne{\Field}{\SCVars}{\VariableX}{\SCDegree}}
\left[
\RandPoly(\alpha) = \beta
\pST
\begin{array}{c}
\RandPoly(\alpha_{1}) = \beta_{1} \\
\vdots \\
\RandPoly(\alpha_{\ListSize}) = \beta_{\ListSize}
\end{array}
\right]\enspace.
\end{equation*}
Moreover $\CodeSimAlgorithm$ runs in time $\poly(\log \SetCardinality{\Field} + \SCVars + \SCDegree + \SetCardinality{\SCSubset} + \ListSize)$.
\end{corollary}

We sketch the proof of \thmref{thm:psrm-succinct-constraint-detection}, for the simpler case where the code is $\RMCode{\Field}{\SCVars}{\SCDegree}{\SCSubset}$ (i.e., without partial sums). We can view a polynomial $Q \in \PolynomialRingIndOne{\Field}{\SCVars}{\VariableX}{\SCDegree}$ as a vector over the monomial basis, with  an entry for each possible monomial $\VariableX_{1}^{i_{1}} \dots \VariableX_{\SCVars}^{i_{\SCVars}}$ (with $0 \leq i_{1}, \dots, i_{\SCVars} < \SCDegree$) containing the corresponding coefficient. The evaluation of $Q$ at a point $\vec{\alpha} \in \Field^{\SCVars}$ then equals the inner product of this vector with the vector $\phi_{\vec{\alpha}}$, in the same basis, whose entry for $\VariableX_{1}^{i_{1}} \dots \VariableX_{\SCVars}^{i_{\SCVars}}$ is equal to $\alpha_{1}^{i_{1}} \dots \alpha_{\SCVars}^{i_{\SCVars}}$. Given $\vec{\alpha}_{1}, \dots, \vec{\alpha}_{\ListSize}$, we could use Gaussian elimination on $\phi_{\vec{\alpha}_{1}}, \dots, \phi_{\vec{\alpha}_{\ListSize}}$ to check for linear dependencies, which would be equivalent to constraint detection for $\RMCode{\Field}{\SCVars}{\SCDegree}{\SCSubset}$.

However, we cannot afford to explicitly write down $\phi_{\vec{\alpha}}$, because it has $\SCDegree^{\SCVars}$ entries. Nevertheless, we can still implicitly check for linear dependencies, and we do so by reducing the problem, by building on and extending ideas of \cite{BogdanovW04}, to computing the nullspace of a certain set of polynomials, which can be solved via an algorithm of \cite{RazS05} (see also \cite{Kayal10}). The idea is to encode the entries of these vectors via a succinct description: a polynomial $\Phi_{\vec{\alpha}}$ whose coefficients (after expansion) are the entries of $\phi_{\vec{\alpha}}$. In our setting this polynomial has the particularly natural form:
\begin{equation*}
\Phi_{\vec{\alpha}}(\vec{\VariableX})
\DefineEqual
\prod_{i=1}^{\SCVars}
(1 + \alpha_{i} \VariableX_{i} + \alpha_{i}^{2} \VariableX_{i}^{2} + \cdots + \alpha_{i}^{\SCDegree-1} \VariableX_{i}^{\SCDegree-1})
\enspace;
\end{equation*}
note that the coefficient of each monomial equals its corresponding entry in $\phi_{\vec{\alpha}}$. Given this representation we can use standard polynomial identity testing techniques to find linear dependencies between these polynomials, which corresponds to linear dependencies between the original vectors. Crucially, we cannot afford any mistake, even with exponentially small probability, when looking for linear dependencies for otherwise we would not achieve perfect simulation; this is why the techniques we leverage rely on derandomization. We now proceed with the full proof.

\begin{proof}[Proof of \thmref{thm:psrm-succinct-constraint-detection}]
We first introduce some notation. Define $\LessThan{\SCDegree} \DefineEqual \{0,\dots,\SCDegree-1\}$. For vectors $\vec{\alpha} \in \Field^{\SCVars}$ and $\vec{a} \in \LessThan{\SCDegree}^{\SCVars}$, we define $\vec{\alpha}^{\vec{a}} \DefineEqual \prod_{i=1}^{\SCVars} \alpha_{i}^{a_{i}}$; similarly, for variables $\vec{\VariableX} = (\VariableX_{1}, \dots, \VariableX_{\SCVars})$, we define $\vec{\VariableX}^{\vec{a}} \DefineEqual \prod_{i=1}^{\SCVars} \VariableX_{i}^{a_{i}}$.

We identify $\PSRMCode{\Field}{\SCVars}{\SCDegree}{\SCSubset}$ with $\Field^{\LessThan{\SCDegree}^{\SCVars}}$; a codeword $\Codeword_{Q}$ then corresponds to a vector $\vec{Q}$ whose $\vec{a}$-th entry is the coefficient of the monomial $\vec{\VariableX}^{\vec{a}}$ in $Q$. For $\vec{\alpha} \in \Field^{\leq \SCVars}$, let
\begin{equation*}
\phi_{\vec{\alpha}}
\DefineEqual
\left(
  \vec{\alpha}^{\vec{a}}
    \sum_{\vec{\gamma} \in \SCSubset^{\SCVars - |\vec{\alpha}|}}
      \vec{\gamma}^{\vec{b}}
\right)_{\vec{a} \in \LessThan{\SCDegree}^{|\vec{\alpha}|} \,,\, \vec{b} \in \LessThan{\SCDegree}^{\SCVars - |\vec{\alpha}|}}
\enspace.
\end{equation*}
We can also view $\phi_{\vec{\alpha}}$ as a vector in $\Field^{\LessThan{\SCDegree}^{\SCVars}}$ by merging the indices, so that, for all $\vec{\alpha} \in \Field^{\leq \SCVars}$ and $\Codeword_{Q} \in \PSRMCode{\Field}{\SCVars}{\SCDegree}{\SCSubset}$,
\begin{align*}
   \Codeword_{Q}(\vec{\alpha})
&= \sum_{\vec{\gamma} \in \SCSubset^{\SCVars - |\vec{\alpha}|}} Q(\vec{\alpha}, \vec{\gamma})
= \sum_{\vec{\gamma} \in \SCSubset^{\SCVars - |\vec{\alpha}|}}
    \sum_{\vec{a} \in \LessThan{\SCDegree}^{|\vec{\alpha}|}}
      \sum_{\vec{b} \in \LessThan{\SCDegree}^{\SCVars - |\vec{\alpha}|}}
        \vec{Q}_{\vec{a},\vec{b}} \cdot \vec{\alpha}^{\vec{a}} \vec{\gamma}^{\vec{b}} \\
&= \sum_{\vec{a} \in \LessThan{\SCDegree}^{|\vec{\alpha}|}}
   \sum_{\vec{b} \in \LessThan{\SCDegree}^{\SCVars - |\vec{\alpha}|}}
     \vec{Q}_{\vec{a},\vec{b}} \cdot \vec{\alpha}^{\vec{a}}
       \sum_{\vec{\gamma} \in \SCSubset^{\SCVars - |\vec{\alpha}|}}
         \vec{\gamma}^{\vec{b}}
= \InnerProduct{\vec{Q}}{\phi_{\vec{\alpha}}}
\enspace.
\end{align*}
Hence for every $\vec{\alpha}_{1}, \dots, \vec{\alpha}_{\ListSize}, \vec{\alpha} \in \Field^{\leq \SCVars}$ and $a_{1}, \dots, a_{\ListSize} \in \Field$, the following statements are equivalent
\begin{inparaenum}[(i)]
  \item $\Codeword(\vec{\alpha}) = \sum_{i=1}^{\ListSize} a_{i} \Codeword(\vec{\alpha}_{i})$ for all $\Codeword \in \PSRMCode{\Field}{\SCVars}{\SCDegree}{\SCSubset}$;
  \item $\InnerProduct{\vec{f}}{\phi_{\vec{\alpha}}} = \sum_{i=1}^{\ListSize} a_{i} \InnerProduct{\vec{f}}{\phi_{\vec{\alpha}_{i}}}$ for all $\vec{f} \in \Field^{\LessThan{\SCDegree}^{\SCVars}}$
  \item $\phi_{\vec{\alpha}} = \sum_{i=1}^{\ListSize} a_{i} \phi_{\vec{\alpha}_{i}}$.
\end{inparaenum}
We deduce that constraint detection for $\PSRMCode{\Field}{\SCVars}{\SCDegree}{\SCSubset}$ is equivalent to the problem of finding $a_{1}, \dots, a_{\ListSize} \in \Field$ such that $\phi_{\vec{\alpha}} = \sum_{i=1}^{\ListSize} a_{i} \phi_{\vec{\alpha}_{i}}$, or returning $\IndepTag$ if no such $a_{1}, \dots, a_{\ListSize}$ exist.

However, the dimension of the latter vectors is $\SCDegree^{\SCVars}$, which may be much larger than $\poly(\log \SetCardinality{\Field} + \SCVars + \SCDegree + \SetCardinality{\SCSubset} + \ListSize)$, and so we cannot afford to ``explicitly'' solve the $\ListSize \times \SCDegree^{\SCVars}$ linear system. Instead, we ``succinctly'' solve it, by taking advantage of the special structure of the vectors, as we now describe. For $\vec{\alpha} \in \Field^{\SCVars}$, define the polynomial
\begin{equation*}
\Phi_{\vec{\alpha}}(\vec{\VariableX}) \DefineEqual \prod_{i=1}^{\SCVars} (1 + \alpha_{i} \VariableX_{i} + \alpha_{i}^{2} \VariableX_{i}^{2} + \cdots + \alpha_{i}^{\SCDegree-1} \VariableX_{i}^{\SCDegree-1}) \enspace.
\end{equation*}
Note that, while the above polynomial is computable via a small arithmetic circuit, its coefficients (once expanded over the monomial basis) correspond to the entries of the vector $\phi_{\vec{\alpha}}$. More generally, for $\vec{\alpha} \in \Field^{\leq \SCVars}$, we define the polynomial
\begin{equation*}
\Phi_{\vec{\alpha}}(\vec{\VariableX})
  \DefineEqual
\left(
  \prod_{i=1}^{|\vec{\alpha}|}
    (1 + \alpha_{i} \VariableX_{i} + \dots + \alpha_{i}^{\SCDegree-1} \VariableX_{i}^{\SCDegree-1})
\right)
\left(
  \prod_{i=1}^{\SCVars - |\vec{\alpha}|}
    \sum_{\gamma \in \SCSubset}
      (1 + \gamma \VariableX_{i+|\vec{\alpha}|} + \dots + \gamma^{\SCDegree-1} \VariableX_{i+|\vec{\alpha}|}^{\SCDegree-1})
\right)
\enspace.
\end{equation*}
Note that $\Phi_{\vec{\alpha}}$ is a product of univariate polynomials. To see that the above does indeed represent $\phi_{\vec{\alpha}}$, we rearrange the expression as follows:
\begin{align*}
\Phi_{\vec{\alpha}}(\vec{\VariableX})
&=
\left(
  \prod_{i=1}^{|\vec{\alpha}|}
     (1 + \alpha_{i} \VariableX_{i} + \dots + \alpha_{i}^{\SCDegree-1} \VariableX_{i}^{\SCDegree-1})
\right)
\left(
  \sum_{\vec{\gamma} \in \SCSubset^{\SCVars - |\vec{\alpha}|}}
    \prod_{i=1}^{\SCVars - |\vec{\alpha}|}
      (1 + \gamma_{i} \VariableX_{i+|\vec{\alpha}|} + \dots + \gamma_{i}^{\SCDegree-1} \VariableX_{i+|\vec{\alpha}|}^{\SCDegree-1})
\right)
\\
&=
\Phi_{\vec{\alpha}}(\VariableX_{1}, \dots, \VariableX_{|\vec{\alpha}|})
\left(
  \sum_{\vec{\gamma} \in \SCSubset^{\SCVars - |\vec{\alpha}|}}
    \Phi_{\vec{\gamma}}(\VariableX_{|\vec{\alpha}|+1}, \dots, \VariableX_{\SCVars})
\right)
\enspace;
\end{align*}
indeed, the coefficient of $\vec{\VariableX}^{\vec{a}, \vec{b}}$ for $\vec{a} \in \LessThan{\SCDegree}^{|\vec{\alpha}|}$ and $\vec{b} \in \LessThan{\SCDegree}^{\SCVars - |\vec{\alpha}|}$ is $\vec{\alpha}^{\vec{a}} \sum_{\vec{\gamma} \in \SCSubset^{\SCVars - |\vec{\alpha}|}}  \vec{\gamma}^{\vec{b}}$, as required.

Thus, to determine whether $\phi_{\alpha} \in \Span(\phi_{\alpha_{1}}, \dots, \phi_{\alpha_{\ListSize}})$, it suffices to determine whether $\Phi_{\alpha} \in \Span(\Phi_{\alpha_{1}}, \dots, \Phi_{\alpha_{\ListSize}})$. In fact, the linear dependencies are in correspondence: for $a_{1}, \dots, a_{\ListSize} \in \Field$, $\phi_{\alpha} = \sum_{i=1}^{\ListSize} a_{i} \phi_{\alpha_{i}}$ if and only if $\Phi_{\alpha} = \sum_{i=1}^{\ListSize} a_{i} \Phi_{\alpha_{i}}$. Crucially, each $\Phi_{\alpha_{i}}$ is not only in $\PolynomialRingIndOne{\Field}{\SCVars}{\VariableX}{\SCDegree}$ but is a product of $\SCVars$ univariate polynomials each represented via an $\Field$-arithmetic circuit of size $\poly(\SetCardinality{\SCSubset} + \SCDegree)$. We leverage this special structure and solve the above problem by relying on an algorithm of \cite{RazS05} that computes the nullspace for such polynomials (see also \cite{Kayal10}), as captured by the lemma below;%
\footnote{One could use polynomial identity testing to solve the above problem in probabilistic polynomial time; see \cite[Lemma 8]{Kayal10}. However, due to a nonzero probability of error, this suffices only to achieve statistical zero knowledge, but \emph{does not suffice to achieve perfect zero knowledge}.}
for completeness, we provide an elementary proof of the lemma in \appref{sec:proof-of-deterministic-algorithm}.

\begin{lemma}
\label{lem:polydep-det}
There exists a deterministic algorithm $\DetAlgorithm$ such that, on input a vector of $\SCVars$-variate polynomials $\vec{\Poly} = (\Poly_{1}, \dots, \Poly_{\ListSize})$ over $\Field$ where each polynomial has the form $\Poly_{k}(\vec{\VariableX}) = \prod_{i=1}^{\SCVars} \Poly_{k,i}(\VariableX_{i})$ and each $\Poly_{k,i}$ is univariate of degree less than $\SCDegree$ with $\SCDegree \leq \SetCardinality{\Field}$ and represented via an $\Field$-arithmetic circuit of size $\CircuitSize$, outputs a basis for the linear space $\Dual{\vec{\Poly}}  \DefineEqual \{ (a_{1}, \dots, a_{\ListSize}) \in \Field^{\ListSize} : \sum_{k=1}^{\ListSize} a_{k} \Poly_{k} \equiv 0 \}$. Moreover, $\DetAlgorithm$ runs in $\poly(\log \SetCardinality{\Field} + \SCVars + \SCDegree + \CircuitSize + \ListSize)$ time.
\end{lemma}

The above lemma immediately provides a way to construct a constraint detector for $\PSRM$: given as input an index $\CodeIdx = (\Field,\SCVars,\SCDegree,\SCSubset)$ and a subset $\IdxSet \subseteq \Domain(\CodeIdx)$, we construct the arithmetic circuit $\Phi_{\alpha}$ for each $\alpha \in \IdxSet$, and then run the algorithm $\DetAlgorithm$ on vector of circuits $(\Phi_{\alpha})_{\alpha \in \IdxSet}$, and directly output $\DetAlgorithm$'s result. The lemma follows.
\end{proof}

%%%%%%%%%%%%%%%%%%%%%%%%%%%%%%%%%%%%%%%%%%%%%%%%%%%%%%%%%%%%%%%%%%%%%%%%%%%%%%%%
%%%%%%%%%%%%%%%%%%%%%%%%%%%%%%%%%%%%%%%%%%%%%%%%%%%%%%%%%%%%%%%%%%%%%%%%%%%%%%%%
\subsection{Univariate polynomials with BS proximity proofs}
\label{sec:bsrs-succinct-constraint-detection}

We show that evaluations of univariate polynomials concatenated with corresponding BS proximity proofs \cite{BS08} have succinct constraint detection (see \defref{def:constraint-detector}). Recall that the Reed--Solomon code (see \secref{sec:codes}) is not locally testable, but one can test proximity to it with the aid of the quasilinear-size proximity proofs of Ben-Sasson and Sudan \cite{BS08}. These latter apply when low-degree univariate polynomials are evaluated over \emph{linear spaces}, so from now on we restrict our attention to Reed--Solomon codes of this form. More precisely, we consider Reed--Solomon codes $\RSCode{\Field}{\BSSpace}{\RSDegree}$ where $\Field$ is an extension field of a base field $\SubField$, $\BSSpace$ is a $\SubField$-linear subspace in $\Field$, and $\RSDegree = \SetCardinality{\BSSpace} \cdot \SetCardinality{\SubField}^{-\BSBalance}$ for some $\BSBalance \in \Naturals^+$. We then denote by $\BSCode[\SubField,\Field,\BSSpace,\BSBalance,\BSBaseDim]$ the code obtained by concatenating codewords in $\RSCode{\Field}{\BSSpace}{\SetCardinality{\BSSpace} \cdot \SetCardinality{\SubField}^{-\BSBalance}}$ with corresponding BS proximity proofs whose recursion terminates at ``base dimension'' $\BSBaseDim \in \{1,\dots,\Dimension{\BSSpace}\}$ (for completeness we include a formal definition of these in \appref{sec:bsrs-formal-definitions}); typically $\SubField,\BSBalance,\BSBaseDim$ are fixed to certain constants (e.g., \cite{BS08} fixes them to $\Field_{2},3,1$, respectively) but below we state the cost of constraint detection in full generality. The linear code family $\BSCode$ is indexed by tuples $\CodeIdx = (\SubField,\Field,\BSSpace,\BSBalance,\BSBaseDim)$ and the $\CodeIdx$-th code is $\BSCode[\SubField,\Field,\BSSpace,\BSBalance,\BSBaseDim]$, and our result about $\BSCode$ is the following:

\begin{theorem}[formal statement of \ref{thm:intro-bsrs}]
\label{thm:bsrs-succinct-constraint-detection}
$\BSCode$ has $\poly(\log \SetCardinality{\Field} + \Dimension{\BSSpace} + \SetCardinality{\SubField}^{\BSBalance} + \ListSize)$-time constraint detection.
\end{theorem}

The proof of the above theorem is technically involved, and we present it via several steps, as follows.
\begin{inparaenum}[(1)]
  \item In \secref{sec:code-covers} we introduce the notion of a \emph{code cover} and two key combinatorial properties of these: \emph{$\IdpParam$-locality} and \emph{$\IdpParam$-independence}.
  \item In \secref{sec:recursive-code-covers} we introduce the notion of a \emph{recursive} code cover and relate its combinatorial properties to those of (standard) code covers.
  \item In \secref{sec:efficient detectors} we show how to construct succinct constraint detectors starting from algorithms that detect constraints only `locally' for code covers and recursive code covers.
  \item In \secref{sec:proof-of-theorem-bsrs-detector} we show that $\BSCode$ has a recursive code cover with the requisite properties and thus implies, via the results of prior steps, a succinct constraint detector, as claimed.
\end{inparaenum}
Several sub-proofs are deferred to the appendices, and we provide pointers to these along the way.

%%%%%%%%%%%%%%%%%%%%%%%%%%%%%%%%%%%%%%%%
\parhead{The role of code covers}
We are interested in succinct constraint detection: solving the constraint detection problem for certain code families with exponentially-large domains (such as $\BSCode$). We now build some intuition about how code covers can, in some cases, facilitate this.

Consider the simple case where the code $\Code \subseteq \Field^{\Domain}$ is a direct sum of many small codes: there exists $\Cover =\Set{(\CDomain_{j},\CCode_{j})}_{j}$ such that $\Domain = \cup_{j} \CDomain_{j}$ and $\Code = \oplus_{j} \CCode_{j}$ where, for each $j$, $\CCode_{j}$ is a linear code in $\Field^{\CDomain_{j}}$ and the subdomain $\CDomain_{j}$ is small and disjoint from other subdomains. The detection problem for this case can be solved efficiently: use the generic approach of Gaussian elimination independently on each subdomain $\CDomain_{j}$.

Next consider a more general case where the subdomains are not necessarily disjoint: there exists $\Cover =\Set{(\CDomain_{j},\CCode_{j})}_{j}$ as above but we do not require that the $\CDomain_{j}$ form a partition of $\Domain$; we say that each $(\CDomain_{j},\CCode_{j})$ is a \emph{local view} of $\Code$ because $\CDomain_{j} \subseteq \Domain$ and $\CCode_{j} = \Restrict{\Code}{\CDomain_{j}}$, and we say that $\Cover$ is a \emph{code cover} of $\Code$. Now suppose that for each $j$ there exists an efficient constraint detector for $\CCode_{j}$ (which is defined on $\CDomain_{j}$); in this case, the detection problem can be solved efficiently at least for those subsets $\IndexSet$ that are contained in $\CDomain_{j}$ for some $j$. Generalizing further, we see that we can efficiently solve constraint detection for a code $\Code$ if there is a cover $\Cover =\Set{(\CDomain_{j},\CCode_{j})}_{j}$ such that, given a subset $\IndexSet \subseteq \Domain$,
\begin{inparaenum}[(i)]
\item $\IndexSet$ is contained in some subdomain $\CDomain_{j}$, and
\item constraint detection for $\CCode_{j}$ can be solved efficiently.
\end{inparaenum}

We build on the above ideas to derive analogous statements for recursive code covers, which arise naturally in the case of $\BSCode$. But note that recursive constructions are common in the PCP literature, and we believe that our cover-based techniques are of independent interest as, e.g., they are applicable to \emph{other} PCPs, including \cite{BFLS91,AroraS98}.

%%%%%%%%%%%%%%%%%%%%%%%%%%%%%%%%%%%%%%%%%%%%%%%%%%%%%%%%%%%%%%%%%%%%%%%%%%%%%%%%
\subsubsection{Covering codes with local views}
\label{sec:code-covers}

The purpose of this section is to formally define the notion of cover and certain combinatorial properties of these.

\begin{definition}
\label{def:bsrs-code-cover}
Let $\Code$ be a linear code with domain $\Domain$ and alphabet $\Field$. A \defemph{(local) view} of $\Code$ is a pair $(\CDomain,\CCode)$ such that $\CDomain \subseteq \Domain$ and $\Restrict{\Code}{\CDomain} = \CCode$. A \defemph{cover} of $\Code$ is a set of local views $\Cover =\Set{(\CDomain_{j},\CCode_{j})}_{j}$ of $\Code$ such that $\Domain = \cup_{j} \CDomain_{j}$. Also, we define a cover's domain intersection as $\InterDomain{\Cover} \DefineEqual \cup_{i \neq j} ( \CDomain_{i} \cap \CDomain_{j} )$ and, given a set $J$, we define $\CDomain_{J} \DefineEqual \cup_{j \in J} \CDomain_{j}$.
\end{definition}

\begin{example}[line cover of $\RM$]
\label{example:rm-line-cover}
Suppose for instance that $\Code$ is the Reed--Muller code $\RMCode{\Field}{\Field}{\RMVars}{\RMDegree}$: $\Code$ consists of evaluations over $\Domain = \Field^{\RMVars}$ of polynomials in $\PolynomialRingIndOne{\Field}{\RMVars}{\VariableX}{\RMDegree}$ (see \defref{sec:codes}). A cover of $\Code$ that is extensively studied in the PCP and property-testing literature is the one given by (axis-parallel) \emph{lines}. A line (in the $i$-th direction) is a set of $\SetCardinality{\Field}$ points that agree on all but one coordinate (the $i$-th one); and the \emph{line cover} of $\Code$ is thus $\Cover = \Set{(\CDomain_{\ell}, \CCode_{\ell})}$ where $\ell$ ranges over all (axis-parallel) lines and $\CCode_{\ell}$ is the Reed--Solomon code $\RSCode{\Field}{\Field}{\RSDegree}$ (see \defref{sec:codes}).

Observe that the domain intersection of the line cover equals $\Field^{\RMVars}$, which is also the domain $\Domain$ of the base code $\Code$. However, for $\BSCode$, we consider a cover whose domain intersection is a strict subset of $\Domain$ (see \appref{sec:bsrs-full-proof}).
\end{example}

Next, we specify a notion of \emph{locality} for covers. A partial assignment $\Codeword' \in \Field^{\Domain'}$ is \emph{locally consistent} with a cover  $\Cover =\Set{(\CDomain_{j},\CCode_{j})}_{j}$ if for every local view $(\CDomain_{j},\CCode_{j})$ with $\CDomain_{j} \subseteq \Domain'$ the restriction $\Restrict{\Codeword'}{\CDomain_{j}}$ is a codeword of $\CCode_{j}$. Then we say that a cover is \emph{$\IdpParam$-local} if any locally consistent assignment $\Codeword' \in \Field^{\Domain'}$, where $\Domain'$ is a union of at most $\IdpParam$ domains in the cover, can be extended to a ``globally consistent'' codeword $\Codeword$ of $\Code$.

\begin{definition}
\label{def:bsrs-local-cover}
Let $\Code$ be a linear code with domain $\Domain$ and alphabet $\Field$. Given $\IdpParam \in \Naturals$, a \defemph{$\IdpParam$-local cover} of $\Code$ is a cover $\Cover = \Set{ (\CDomain_{j},\CCode_{j}) }_{j}$ of $\Code$ such that: for every subset of view-indices $J$ of size at most $\IdpParam$ and every word $\Codeword' \in \Field^{\CDomain_{J}}$ with $\Restrict{\Codeword'}{\CDomain_{j}} \in \CCode_{j}$ (for every $j \in J$), the word $\Codeword'$ can be extended to some word $\Codeword$ in $\Code$, i.e., $\Codeword$ satisfies $\Restrict{\Codeword}{\CDomain_{J}} = \Codeword'$). (The trivial cover $\Cover = \Set{(\Domain,\Code)}$ of $\Code$ is $\IdpParam$-local for every $\IdpParam$.)
\end{definition}

The following definition significantly strengthens the previous one. Informally, a cover is \emph{$\IdpParam$-independent} if every partial assignment over a subdomain $\Domain'$ that is the union of $\IdpParam$ subdomains from the cover \emph{and} $\IdpParam$ auxiliary locations can be extended to a ``globally consistent'' codeword. We use this stronger notion in our main \lemref{lemma:bsrs-recursive-and-independent-implies-local}.

\begin{definition}
\label{def:bsrs-independent-cover}
Let $\Code$ be a linear code with domain $\Domain$ and alphabet $\Field$.
Given $\IdpParam \in \Naturals$, a \defemph{$\IdpParam$-independent cover} of $\Code$ is a $\IdpParam$-local cover $\Cover = \Set{ (\CDomain_{j},\CCode_{j}) }_{j}$ such that for every set $J$ of size at most $\IdpParam$, subdomain $\Domain' \subseteq \InterDomain{\Cover}$ of size at most $\IdpParam$, and $\Codeword' \in \Field^{\Domain' \cup \CDomain_{J}}$ with $\Restrict{\Codeword'}{\CDomain_{j}} \in \CCode_{j}$ for every $j \in J$, there exists $\Codeword \in \Code$ such that $\Restrict{\Codeword}{\Domain' \cup \CDomain_{J}} = \Codeword'$. (The trivial cover $\Cover=\Set{(\Domain,\Code)}$ of $\Code$ is $\IdpParam$-independent for every $\IdpParam\in\Naturals$ because it is $\IdpParam$-local and has $\InterDomain{\Cover}=\emptyset$.)
\end{definition}

\begin{example}
The line cover from \exmpref{example:rm-line-cover} is $\RMDegree$-local because any evaluation on at most $\RMDegree$ (axis-parallel) lines $\ell_{1},\dots,\ell_{\RMDegree}$ that is a polynomial of degree $(\RMDegree-1)$ along each of $\ell_{1},\dots,\ell_{\RMDegree}$ can be extended to a codeword of $\RMCode{\Field}{\Field}{\RMVars}{\RMDegree}$. Furthermore, it is $\RMDegree/2$-independent because any locally consistent assignment to $\RMDegree/2$ such lines and $\RMDegree/2$ points $p_{1},\dots,p_{\RMDegree/2}$ can be extended to a valid Reed--Muller codeword. To see this, observe that there exists an interpolating set $H_{1} \times \cdots \times H_{\RMVars}$ (with $\SetCardinality{H_i}=\RMDegree$) that contains $p_{1},\dots, p_{\RMDegree/2}$ and intersects each line in $\RMDegree$ points; we shall later use this observation for the bivariate case ($\RMVars=2$), a proof for that case is provided at \clmref{claim:bsrs-rm-folklore}.
\end{example}

%%%%%%%%%%%%%%%%%%%%%%%%%%%%%%%%%%%%%%%%%%%%%%%%%%%%%%%%%%%%%%%%%%%%%%%%%%%%%%%%
\subsubsection{Recursive covers and locality}
\label{sec:recursive-code-covers}

Proximity proofs for codes such as the Reed--Solomon code and the Reed--Muller code are typically obtained via techniques of proof composition \cite{AroraS98}. Informally, a problem is reduced to a set of smaller sub-problems of the same kind (which are usually interconnected), and a sub-proof is constructed for each sub-problem. This process leads to a proof for the original problem that is ``covered'' by the sub-proofs for the sub-problems, and naturally imply a cover of the proof by these sub-proofs. This process is then repeated recursively until the sub-problems are small enough for the verifier to check directly --- and in our case leads to the notion of \emph{recursive covers}, which we define below.

To support the definition of a recursive cover, we first introduce notation for rooted trees. Edges in a rooted tree $\TreeCover=(V,E)$ are directed from the root $\TreeRoot$ towards the leaves; the edge directed from $v$ to $u$ is denoted $(v,u)$; $v$ is the \emph{predecessor} of $u$ and $u$ the \emph{successor} of $v$; if there is a path from $v$ to $v'$ we say that $v$ is an \emph{ancestor} of $v'$; if there is no directed path between $v$ and $v'$ (in either direction) we say that the two vertices are \emph{disconnected}. The \emph{set of successors} of $v$ is denoted $\Children{\TreeCover,v}$. The \emph{depth} of a vertex $v$ in $\TreeCover$ is denoted $\TreeDepth{\TreeCover,v}$ and equals the number of edges on the path from $r$ to $v$. The depth of $\TreeCover$ is denoted $\TreeDepth{\TreeCover}$ and equals the maximum of $\TreeDepth{\TreeCover,v}$ as $v$ ranges in $V$. The $i$-th \emph{layer} of $\TreeCover$ is denoted $\TreeLayer{\TreeCover,i}$ and equals the set of $v \in V$ such that $\TreeDepth{\TreeCover,v}=i$. (Note that $\TreeDepth{\TreeCover,\TreeRoot}=0$ and $\TreeLayer{\TreeCover,0} = \Set{\TreeRoot}$.) An \emph{equidepth} tree is a tree in which all leaves have equal depth.

\begin{definition}
\label{def:bsrs-recursive-code-cover}
Let $\Code$ be a linear code with domain $\Domain$ and alphabet $\Field$. A \defemph{recursive cover} of $\Code$ is a directed rooted equidepth tree $\TreeCover$ of non-zero depth where each vertex $v$ is labeled by a view $(\CCode_{v},\CDomain_{v})$ such that:
\begin{inparaenum}[(i)]
\item $\CCode_{v}$ is a linear code with domain $\CDomain_{v}$ and alphabet $\Field$;
\item if $v$ is the root, then $(\CCode_{v},\CDomain_{v})=(\Code,\Domain)$; and
\item for every non-leaf $v$ the set $\CoverAtVertex{\TreeCover}{v} \DefineEqual \Set{(\CCode_{u},\CDomain_{u})}_{u \in \Children{\TreeCover,v}}$ is a cover of $\CCode_{v}$.
\end{inparaenum}
Furthermore we define the following notions:
\begin{itemize}

  \item Given $\CDepth \in \{0,\dots,\TreeDepth{\TreeCover}\}$, the \defemph{$\CDepth$-depth restriction} of $\TreeCover$ is $\CoverAtLayer{\TreeCover}{\CDepth} \DefineEqual \bigcup_{v\in \TreeLayer{\TreeCover,\CDepth}} \Set{(\CCode_{v}, \CDomain_{v})}$. (Note that $\CoverAtLayer{\TreeCover}{0}=\Set{(\Code, \Domain)}$.)
  %and $\CoverAtLayer{\TreeCover}{\TreeDepth{\TreeCover}}=\emptyset$

  \item Given $\IntParam \in \Naturals$, we say that $\TreeCover$ is \defemph{$\IntParam$-intersecting} if $\SetCardinality{\CDomain_{u} \cap \CDomain_{u'}} \leq \IntParam$ for every two disconnected vertices $u,v$.

  \item Given $\IdpParam \in \Naturals$, we say that $\TreeCover$ is \defemph{$\IdpParam$-independent} if $\CoverAtVertex{\TreeCover}{v}$ is a $\IdpParam$-independent cover of $\CCode_{v}$ for every non-leaf vertex $v$ in $\TreeCover$.
  %\eli{Assuming we take the trivial cover for leaves, we can (and probably should) extend the definition to all nodes, including leaves.}
\end{itemize}
\end{definition}

\begin{remark}
The above definition is restricted to equidepth trees, but can be extended to general trees as follows. Iteratively append to each leaf $v$ of non-maximal depth a single successor $u$ labeled by $(\CCode_{u},\CDomain_{u}) \DefineEqual (\CCode_{v},\CDomain_{v})$; this leads to a cover of $\CoverAtVertex{\TreeCover}{v}$ that is $0$-intersecting and $\IdpParam$-doubly independent for $\IdpParam$ that equals $\CCode_{v}$'s dual distance.
\end{remark}

%\begin{claim}
%\label{claim:bsrs-recursive-implies-small-intersection}
%If a recursive cover is locally $\IntParam$-intersecting then it is also $\IntParam$-intersecting.
%\end{claim}
%
%\begin{proof}
%Consider two disconnected vertices $u,v$ and let $w$ be their deepest common ancestor. Let $w_{u},w_{v}$ be the (unique) successors of $w$ such that $w_{u}$ is an ancestor of $u$ and $w_{v}$ is an ancestor of $v$. By assumption it holds that $\SetCardinality{\CDomain_{w_{u}} \cap \CDomain_{w_{v}}} \leq \IntParam$, and by \defref{def:bsrs-recursive-code-cover} it holds that $\CDomain_{u} \subseteq \CDomain_{w_u}$ and $\CDomain_{v} \subseteq \CDomain_{w_{v}}$. Therefore $\TreeCover$ is $\IntParam$-intersecting, as claimed.
%\end{proof}
%
%\eli{end of part that will be moved out}

Below we state the main lemma of this section. This lemma says that (given certain restrictions) if a recursive cover has the \emph{local} property of \emph{independence} (of some degree) at each internal vertex, then each of its layers has the \emph{global} property of \emph{locality} (of some degree) as a cover of the root. Later on (in \secref{sec:efficient detectors}) we show how cover locality is used to construct constraint detectors.

\begin{lemma}[main]
\label{lemma:bsrs-recursive-and-independent-implies-local}
Let $\Code$ be a linear code with domain $\Domain$ and alphabet $\Field$, and let $\TreeCover$ be a recursive cover of $\Code$ such that
\begin{inparaenum}[(i)]
  \item $\TreeCover$ is $\IntParam$-intersecting for $\IntParam>0$, and
  \item for every non-leaf vertex $v$ in $\TreeCover$ it holds that $\CoverAtVertex{\TreeCover}{v}$ is a $\IdpParam$-independent cover of $\CCode_{v}$.
\end{inparaenum}
Then, for every $\CDepth \in \{0,\dots,\TreeDepth{\TreeCover}\}$, $\CoverAtLayer{\TreeCover}{\CDepth}$ is a $\frac{\IdpParam}{\IntParam}$-local cover of $\Code$.
\end{lemma}

\begin{proof}
We prove the statement by induction on the non-negative integer $\CDepth$. The base case is when $\CDepth=0$, and holds because $\CoverAtLayer{\TreeCover}{0} = \Set{(\Domain,\Code)}$ is the trivial cover, thus it is a $\IdpParam'$-local cover of $\Code$ for any $\IdpParam' \ge 0$ and, in particular, a $\frac{\IdpParam}{\IntParam}$-local cover. We now assume the statement for $\CDepth<\TreeDepth{\TreeCover}$ and prove it for depth $\CDepth+1$.

Let $\CoverAtLayer{\TreeCover}{\CDepth} = \Set{(\CDomain_{j},\CCode_{j})}_{j}$ be the $\CDepth$-depth cover of $\Code$, and let $\CoverAtLayer{\TreeCover}{\CDepth+1} = \Set{(\CDomain_{i,j},\CCode_{i,j})}_{i,j}$ be the $(\CDepth+1)$-depth cover of $\Code$, where, for every $i$, $\CoverAtLayer{\TreeCover}{\CDepth+1}^{(i)} = \Set{ (\CDomain_{i,j},\CCode_{i,j})}_{j}$ is the cover of $\CCode_{i}$ (this can be ensured via suitable indexing). Let $J$ be a set of pairs $(i,j)$ of size at most $\frac{\IdpParam}{\IntParam}$, and let $\Codeword' \in \Field^{\CDomain_{J}}$ be such that $\Restrict{\Codeword'}{\CDomain_{i,j}} \in \CCode_{i,j}$ for every $(i,j) \in J$. We show that there exists $\Codeword \in \Code$ such that $\Restrict{\Codeword}{\CDomain_{J}} = \Codeword'$. Define $I \DefineEqual \Set{ i : \exists\, j \text{ s.t. } (i,j) \in J}$ and note that $\SetCardinality{I} \leq \SetCardinality{J} \leq \frac{\IdpParam}{\IntParam}$. By the inductive assumption, it suffices to show that there exists $\Codeword \in \Field^{\CDomain_{I}}$ such that
\begin{inparaenum}[(a)]
  \item $\Restrict{\Codeword}{\CDomain_{i,j}} = \Restrict{\Codeword'}{\CDomain_{i,j}}$ for every $(i,j) \in J$, and
  \item $\Restrict{\Codeword}{\CDomain_{i}} \in \CCode_{i}$ for every $i \in I$.
\end{inparaenum}

For simplicity assume $I=\Set{1,\dots,\SetCardinality{I}}$. We construct $\Codeword$ incrementally and view $\Codeword$ as belonging to $(\Field \cup \Set{\Undefined})^{\CDomain_{I}}$, i.e., it is a partial mapping from $\CDomain_{I}$ to $\Field$. Let $\defDomain{\Codeword} \DefineEqual \Set{\alpha \in \CDomain_{I} : \Codeword(\alpha) \neq \Undefined}$ denote the set of locations where $\Codeword$ is defined. Initialize $\defDomain{\Codeword} = \CDomain_{J}$ and $\Restrict{\Codeword}{\CDomain_{J}} = \Codeword'$; then, for increasing $i=1,\dots,\SetCardinality{I}$, iteratively extend $\Codeword$ to be defined (also) over $\CDomain_{i}$, eventually obtaining $\Codeword \in \Field^{\CDomain_{I}}$. In the $i$-th iteration (that handles $\CDomain_{i}$), it is sufficient to prove the existence of a codeword $\Codeword_{i} \in \CCode_{i}$ such that $\Restrict{\Codeword_{i}}{\CDomain_{i} \cap \defDomain{\Codeword}} = \Restrict{\Codeword}{\CDomain_{i} \cap \defDomain{\Codeword}}$. If such a codeword exists then we shall define $\Codeword$ on $\CDomain_{i}$ by $\Restrict{\Codeword}{\CDomain_{i}} = \Codeword_{i}$, thus eventually reaching $\Codeword$ that satisfies the stated requirements.

To show that during the $i$-th iteration the desired $\Codeword_{i}$ exists, partition the elements of $\CDomain_{i}\cap\defDomain{\Codeword}$ into two sets: $\SetA \DefineEqual \cup_{j \text{ s.t. } (i,j) \in J} \CDomain_{i,j}$ and $\SetB \DefineEqual (\CDomain_{i}\cap\defDomain{\Codeword}) \setminus \SetA$. Note that $\SetB \subseteq \cup_{i' \neq i} ( \CDomain_{i} \cap \CDomain_{i'} )$, because defining $\Codeword(\alpha)$ for any $\alpha \in \SetB$ can be done only:
\begin{inparaenum}[(i)]
  \item in the initialization phase, so that $\alpha \in \CDomain_{i',j'} \subseteq \CDomain_{i'}$ for some $(i',j') \in J$ with $i' \neq i$ (as otherwise $\alpha \in \SetA$); or
  \item in a previous iteration, so that $\alpha \in \CDomain_{i'}$ for some $i' <i$ (as $\CDomain_{i'}$ was already handled and $\Codeword$ is already defined on all of $\CDomain_{i'}$).
\end{inparaenum}
The above implies that $\SetB \subseteq \InterDomain{\CoverAtLayer{\TreeCover}{\CDepth+1}^{(i)}}$ and the assumption that $\TreeCover$ is $\IntParam$-intersecting implies
\begin{equation*}
\SetCardinality{\SetB}
\leq \sum_{i' \ne i} \SetCardinality{\CDomain_{i} \cap \CDomain_{i'}}
\leq \IntParam \cdot  \SetCardinality{I}
\leq \IntParam\cdot \SetCardinality{J}
\leq \IdpParam
\enspace.
\end{equation*}
Similarly, note that for every fixed $i \in I$ the number of pairs $(i,j) \in J$ is at most $\SetCardinality{J} \leq \IdpParam$. By assumption $\CCode_{i}$ has a $\IdpParam$-independent cover and thus we conclude (via \defref{def:bsrs-independent-cover}) that the desired $\Codeword_{i}$ exists, as required.
\end{proof}

%%%%%%%%%%%%%%%%%%%%%%%%%%%%%%%%%%%%%%%%%%%%%%%%%%%%%%%%%%%%%%%%%%%%%%%%%%%%%%%%
\subsubsection{From recursive covers to succinct constraint detection}
\label{sec:efficient detectors}

The purpose of this section is to establish sufficient conditions for succinct constraint detection by leveraging covers with small-enough views and large-enough locality. First, in \defref{def:bsrs-cover-based-detector} and \lemref{lemma:bsrs-cover-based-implies-succinct}, we define \emph{cover-based constraint detection} and prove that it implies succinct constraint detection; informally, we consider the case when a code has a sequence of covers where view size and locality reduce together, and prove that we can locally detect constraints in a number of views that is proportional to the constraint's weight and each view's size is proportional to the constraint's weight, by choosing the right cover from the sequence. Then, in \defref{def:bsrs-recursive-implies-cover-based} and \lemref{lem:bsrs-recursive-implies-local}, we extend our discussion to recursive code covers by defining \emph{recursive-cover-based constraint detection} and establishing that it implies the previous notion. We conclude (in \corref{cor:bsrs-recursive-implies-succinct}) that recursive-cover-based constraint detection implies succinct constraint detection.

\begin{definition}
\label{def:bsrs-cover-based-detector}
Let $\CodeClass = \Set{\Code_{\CodeIdx}}_{\CodeIdx}$ be a linear code family with domain $\Domain(\cdot)$ and alphabet $\Field(\cdot)$. We say that $\CodeClass$ has \defemph{cover-based constraint detection} if there exists an algorithm that, given an index $\CodeIdx$ and subset $\IndexSet \subseteq \Domain(\CodeIdx)$, outputs in $\poly(\BitSize{\CodeIdx} + \SetCardinality{\IndexSet})$ time a subset $\LocalSet \subseteq \Field(\CodeIdx)^{\Domain(\CodeIdx)}$ for which there exists a subset $\Cover'$ of some $\SetCardinality{\IndexSet}$-local cover $\Cover$ of $\Code_{\CodeIdx}$, and the following holds:
\begin{inparaenum}[(i)]
  \item $\SetCardinality{\Cover'} \leq \SetCardinality{\IndexSet}$;
  \item $\IndexSet \subseteq (\cup_{(\CDomain,\CCode) \in \Cover'}\, \CDomain)$;
  \item \label{step:ambiguous-equality} $\Span(\LocalSet) = \Span(\cup_{(\CDomain,\CCode) \in \Cover'}\, \Dual{\CCode})$.
\end{inparaenum}
\end{definition}

\begin{lemma}
\label{lemma:bsrs-cover-based-implies-succinct}
Let $\CodeClass = \Set{\Code_{\CodeIdx}}_{\CodeIdx}$ be a linear code family with domain $\Domain(\cdot)$ and alphabet $\Field(\cdot)$. If $\CodeClass$ has cover-based constraint detection then $\CodeClass$ has succinct constraint detection.
\end{lemma}

To prove this lemma we require a technical claim, the proof
of which is deferred to
\appref{sec:proof-of-independent-dual}.

\begin{claim}
\label{claim:bsrs-independent-dual}
Let $\Code$ be a linear code with domain $\Domain$ and alphabet $\Field$, let $\Cover = \Set{(\CDomain_{j},\CCode_{j})}_{j}$ be a $\IdpParam$-local cover of $\Code$. For any set $J$ of size at most $\IdpParam$ it holds $\Span( \cup_{j \in J} \Dual{\CCode_{j}}) = \Puncture{(\Dual{\Code})}{\left(\cup_{j \in J}\CDomain_{j}\right)}$.
\end{claim}

\begin{proof}[Proof of \lemref{lemma:bsrs-cover-based-implies-succinct}]
By \lemref{lem:bsrs-dual-char}, it suffices to show an algorithm that, on input an index $\CodeIdx$ and subset $\IdxSet \subseteq \Domain(\CodeIdx)$, outputs a subset $\LocalSet \subseteq \Field(\CodeIdx)^{\Domain(\CodeIdx)}$ with $\Puncture{(\Dual{\Code_{\CodeIdx}})}{\IdxSet} \subseteq \Span(\LocalSet) \subseteq \Dual{\Code_{\CodeIdx}}$ in $\poly(\BitSize{\CodeIdx} + \SetCardinality{\IdxSet})$ time. We take this algorithm to be the one guaranteed by \defref{def:bsrs-cover-based-detector}. To see correctness, let $\CDomain_{\Cover'} \DefineEqual \cup_{(\CDomain,\CCode) \in \Cover'}\, \CDomain$, and note that \defref{def:bsrs-cover-based-detector} and \clmref{claim:bsrs-independent-dual} imply that $\Span(\LocalSet) = \Puncture{(\Dual{\Code_{\CodeIdx}})}{\CDomain_{\Cover'}}$ and $\Puncture{(\Dual{\Code_{\CodeIdx}})}{\IdxSet} \subseteq \Puncture{(\Dual{\Code_{\CodeIdx}})}{\CDomain_{\Cover'}} \subseteq \Dual{\Code_{\CodeIdx}}$, as required.
\end{proof}

Next we show that, under certain conditions, code families with recursive covers imply a sequence of covers that we can use to construct cover-based constraint detectors. Combined with \lemref{lemma:bsrs-cover-based-implies-succinct}, this result is key for establishing a connection from certain proximity proof constructions to succinct constraint detectors.

\begin{definition}
\label{def:bsrs-recursive-implies-cover-based}
Let $\CodeClass = \Set{\Code_{\CodeIdx}}_{\CodeIdx}$ be a linear code family with domain $\Domain(\cdot)$ and alphabet $\Field(\cdot)$. We say that $\CodeClass$ has \defemph{recursive-cover-based constraint detection} if:
\begin{itemize}

  \item there exists $\IntParam \in \Naturals$ such that, for every index $\CodeIdx$, $\Code_{\CodeIdx}$ has a $\IntParam$-intersecting recursive cover $\TreeCover_{\CodeIdx}$;

  \item there exists an algorithm that, given an index $\CodeIdx$ and subset $\IndexSet \subseteq \Domain(\CodeIdx)$, outputs in $\poly(\BitSize{\CodeIdx} + \SetCardinality{\IndexSet})$ time a subset $\LocalSet \subseteq \Field(\CodeIdx)^{\Domain(\CodeIdx)}$ for which there exist $\CDepth \in \{0,\dots,\TreeDepth{\TreeCover_{\CodeIdx}}\}$ and $U \subseteq \TreeLayer{\TreeCover_{\CodeIdx},\CDepth}$ such that:
  \begin{inparaenum}[(i)]
    \item for every vertex $v$ in $\TreeCover_{\CodeIdx}$ with $\TreeDepth{\TreeCover_{\CodeIdx},v} < \CDepth$, the cover $\CoverAtVertex{\TreeCover}{\CodeIdx,v}$ is $\IntParam \SetCardinality{\IndexSet}$-independent;
    \item $\SetCardinality{U} \leq \SetCardinality{\IndexSet}$;
    \item $\IndexSet \subseteq (\cup_{u \in U} \CDomain_{u})$;
    \item \mbox{$\Span(\LocalSet) = \Span(\cup_{u \in U} \Dual{\CCode_{u}})$.}
  \end{inparaenum}

\end{itemize}
\end{definition}

\begin{lemma}
\label{lem:bsrs-recursive-implies-local}
Let $\CodeClass = \Set{\Code_{\CodeIdx}}_{\CodeIdx}$ be a linear code family with domain $\Domain(\cdot)$ and alphabet $\Field(\cdot)$. If $\CodeClass$ has recursive-cover-based constraint detection, then $\CodeClass$ has cover-based constraint detection.
\end{lemma}

\begin{proof}
The definition of recursive-cover-based detection says that there exist
\begin{inparaenum}[(a)]
  \item $\IntParam \in \Naturals$ such that, for every index $\CodeIdx$, $\Code_{\CodeIdx}$ has a $\IntParam$-intersecting recursive cover $\TreeCover_{\CodeIdx}$, and
  \item an algorithm satisfying certain properties.
\end{inparaenum}
We show that this algorithm meets the requirements for being a cover-based constraint detector (see \defref{def:bsrs-cover-based-detector}). Consider any index $\CodeIdx$ and subset $\IndexSet \subseteq \Domain(\CodeIdx)$, and let $\LocalSet$ be the output of the algorithm. Let $\CDepth \in \{0,\dots,\TreeDepth{\TreeCover_{\CodeIdx}}\}$ and $U \subseteq \TreeLayer{\TreeCover_{\CodeIdx},\CDepth}$ be the objects associated to $\LocalSet$ (guaranteed by the definition of recursive-cover-based constraint detection). Let $\Cover \DefineEqual \CoverAtLayer{\TreeCover_{\CodeIdx}}{\CDepth}$ (i.e., $\Cover$ is the $\CDepth$-depth restriction of $\TreeCover_{\CodeIdx}$) and $\Cover' \DefineEqual \Set{ (\CDomain_{u},\CCode_{u})}_{u \in U}$; it suffices to show that $\Cover$ is $\SetCardinality{\IndexSet}$-local. The claim follows directly by the assumption on $\CDepth$ and \lemref{lemma:bsrs-recursive-and-independent-implies-local}, because $\CoverAtVertex{\TreeCover}{\CodeIdx,v}$ is $\IntParam \SetCardinality{\IndexSet}$-independent for every vertex $v$ in $\TreeCover_{\CodeIdx}$ with $\TreeDepth{\TreeCover_{\CodeIdx},v} < \CDepth$,  and thus $\Cover = \CoverAtLayer{\TreeCover_{\CodeIdx}}{\CDepth}$ is indeed a $\SetCardinality{\IndexSet}$-local cover of $\Code$.
\end{proof}

\begin{corollary}
\label{cor:bsrs-recursive-implies-succinct}
Let $\CodeClass = \Set{\Code_{\CodeIdx}}_{\CodeIdx}$ be a linear code family with domain $\Domain(\cdot)$ and alphabet $\Field(\cdot)$. If $\CodeClass$ has recursive-cover-based constraint detection, then $\CodeClass$ has succinct constraint detection.
\end{corollary}

\begin{proof}
Follows directly from \lemref{lem:bsrs-recursive-implies-local} (recursive-cover-based constraint detection implies cover-based constraint detection) and \lemref{lemma:bsrs-cover-based-implies-succinct} (cover-based constraint detection implies succinct constraint detection).
\end{proof}

%%%%%%%%%%%%%%%%%%%%%%%%%%%%%%%%%%%%%%%%%%%%%%%%%%%%%%%%%%%%%%%%%%%%%%%%%%%%%%%%
\subsubsection{Proof of \thmref{thm:bsrs-succinct-constraint-detection}}
\label{sec:proof-of-theorem-bsrs-detector}

The purpose of this section is to prove \thmref{thm:bsrs-succinct-constraint-detection}. By \corref{cor:bsrs-recursive-implies-succinct}, it suffices to argue that the linear code family $\BSCode$ has recursive-cover-based constraint detection (see \defref{def:bsrs-recursive-implies-cover-based}).

Recall that we consider Reed--Solomon codes $\RSCode{\Field}{\BSSpace}{\RSDegree}$ where $\Field$ is an extension field of a base field $\SubField$, $\BSSpace$ is a $\SubField$-linear subspace in $\Field$, and $\RSDegree = \SetCardinality{\BSSpace} \cdot \SetCardinality{\SubField}^{-\BSBalance}$ for some $\BSBalance \in \Naturals$; and we denote by $\BSCode[\SubField,\Field,\BSSpace,\BSBalance,\BSBaseDim]$ the code obtained by concatenating codewords in $\RSCode{\Field}{\BSSpace}{\SetCardinality{\BSSpace} \cdot \SetCardinality{\SubField}^{-\BSBalance}}$ with corresponding \cite{BS08} proximity proofs with ``base dimension'' $\BSBaseDim \in \{1,\dots,\Dimension{\BSSpace}\}$ (see \appref{sec:bsrs-formal-definitions} for details). The linear code family $\BSCode$ is indexed by tuples $\CodeIdx = (\SubField,\Field,\BSSpace,\BSBalance,\BSBaseDim)$ and the $\CodeIdx$-th code is $\BSCode[\SubField,\Field,\BSSpace,\BSBalance,\BSBaseDim]$.

We represent indices $\CodeIdx$ so that $\log \SetCardinality{\Field} + \Dimension{\BSSpace} + \SetCardinality{\SubField}^{\BSBalance} \leq \poly(\BitSize{\CodeIdx})$. The base field $\SubField$ and extension field $\Field$ require $O(\log \SetCardinality{\SubField})$ and $O(\log \SetCardinality{\Field})$ bits to represent; the subspace $\BSSpace$ requires $O(\Dimension{\BSSpace})$ elements in $\Field$ to represent; and the two integers $\BSBalance$ and $\BSBaseDim$ require $O(\log \BSBalance)$ and $O(\log \BSBaseDim)$ bits to represent. In addition, we add $\SetCardinality{\SubField}^{\BSBalance}$ arbitrary bits of padding. Overall, we obtain that $\BitSize{\CodeIdx} = \Theta(\log \SetCardinality{\SubField} + \log \SetCardinality{\Field} + \log \SetCardinality{\Field} \cdot \Dimension{\BSSpace} + \log \BSBalance + \log \BSBaseDim + \SetCardinality{\SubField}^{\BSBalance}) = \Theta(\log \SetCardinality{\Field} \cdot \Dimension{\BSSpace} + \SetCardinality{\SubField}^{\BSBalance})$.

The main claim in this section is the following (and does not rely on fixing $\SubField,\BSBalance$).

\begin{lemma}
\label{lem:bsrs-has-recursive-code-cover}
Define the depth function $\CDepth(\SubField,\BSSpace,\BSBalance,a) \DefineEqual \BsOtherDepth{\BSSpace}{a}$. The linear code family $\BSCode$ satisfies the following properties.
\begin{itemize}

  \item For every index $\CodeIdx = (\SubField,\Field,\BSSpace,\BSBalance,\BSBaseDim)$, $\BSCode[\SubField,\Field,\BSSpace,\BSBalance,\BSBaseDim]$ has a $1$-intersecting recursive cover $\TreeCover_{\CodeIdx}$. Also, for every positive integer $m$ and non-leaf vertex $v$ in $\TreeCover_{\CodeIdx}$ with $\TreeDepth{\TreeCover_{\CodeIdx},v} < \CDepth(\SubField,\BSSpace,\BSBalance,m)$, the cover $\CoverAtVertex{\TreeCover}{\CodeIdx,v}$ is $m$-independent.

  \item There exists an algorithm that, given an index $\CodeIdx = (\SubField,\Field,\BSSpace,\BSBalance,\BSBaseDim)$ and subset $\IndexSet \subseteq \Domain(\CodeIdx)$, outputs in time $\poly(\log \SetCardinality{\Field} + \Dimension{\BSSpace} + \SetCardinality{\SubField}^{\BSBalance} + \SetCardinality{\IndexSet})$ a subset $\LocalSet \subseteq \Field^{\Domain(\CodeIdx)}$ for which there exist $U \subseteq \TreeLayer{\TreeCover_{\CodeIdx},\CDepth(\SubField,\BSSpace,\BSBalance,\SetCardinality{\IndexSet})}$ such that:
\begin{inparaenum}[(i)]
  \item $\SetCardinality{U} \leq \SetCardinality{\IndexSet}$;
  \item $\IndexSet \subseteq (\cup_{u \in U} \CDomain_{u})$;
  \item $\Span(\LocalSet) = \Span(\cup_{u \in U} \Dual{\CCode_{u}})$.
\end{inparaenum}

\end{itemize}
\end{lemma}

Given the above lemma, we can complete the proof of \thmref{thm:bsrs-succinct-constraint-detection}, as explained below. We defer the (long and technical) proof of the lemma to \appref{sec:bsrs-full-proof}, and instead end this section with an overview of that proof.

\begin{proof}[Proof of \thmref{thm:bsrs-succinct-constraint-detection}]
The proof follows from \lemref{lem:bsrs-has-recursive-code-cover} above and from \corref{cor:bsrs-recursive-implies-succinct}, as we now explain.

\corref{cor:bsrs-recursive-implies-succinct} states that if a linear code family $\CodeClass$ has recursive-cover-based constraint detection (see \defref{def:bsrs-recursive-implies-cover-based}), then $\CodeClass$ has succinct constraint detection (see \defref{def:constraint-detector}). Also recall that the definition of recursive-cover-based detection requires having a $\IntParam$-intersecting recursive cover for each code in the class, and an algorithm satisfying certain properties.

Observe that \lemref{lem:bsrs-has-recursive-code-cover} guarantees that every code in $\BSCode$ has a $1$-intersecting recursive code and, moreover, guarantees the existence of an algorithm whose output satisfies the required properties. We are left to argue that the algorithm runs in time $\poly(\BitSize{\CodeIdx} + \SetCardinality{\IndexSet})$. But this immediately follows from the running time stated in \lemref{lem:bsrs-has-recursive-code-cover} and the fact that $\log \SetCardinality{\Field} + \Dimension{\BSSpace} + \SetCardinality{\SubField}^{\BSBalance} \leq \poly(\BitSize{\CodeIdx})$.
\end{proof}

%%%%%%%%%%%%%%%%%%%%%%%%%%%%%%%%%%%%%%%%
\parhead{Overview of \lemref{lem:bsrs-has-recursive-code-cover}'s proof.}
We assume familiarity with the linear code family $\BSCode$ from \cite{BS08}; for completeness, we provide formal definitions and notations in \appref{sec:bsrs-formal-definitions}. Recall that the Reed--Solomon code is not locally testable, but one can test proximity to it with the aid of BS proximity proofs \cite{BS08}; the linear code family $\BSCode$ consists of the concatenation of Reed--Solomon codes with BS corresponding proximity proofs.

The construction of the aforementioned proximity proofs is \emph{recursive}, with each step in the recursion reducing both the evaluation domain size $\SetCardinality{\BSSpace}$ and the degree $\RSDegree$ to (approximately) their square roots. Namely, testing proximity of a codeword $\Codeword$ to $\RSCode{\Field}{\BSSpace}{\RSDegree}$ is reduced to testing proximity of $\Theta(\sqrt{\SetCardinality{\BSSpace}})$ codewords $\Set{\Codeword_{i}}_{i}$ to $\Set{\RSCode{\Field}{\BSSpace_{i}}{\RSDegree_{i}}}_{i}$, where $\SetCardinality{\BSSpace_{i}},\RSDegree_{i} = \Theta(\sqrt{\SetCardinality{\BSSpace}})$ for each $i$. This step is then recursively applied (by way of proof composition \cite{AroraS98}) to each codeword $\Codeword_{i}$, until the domain size is ``small enough''.

The first part of the proof of \lemref{lem:bsrs-has-recursive-code-cover} consists of various combinatorial claims (see \appref{sec:bsrs-combinatorial}). First, we observe that the union of the domains of the codewords $\Codeword_{i}$ covers (and, actually, slightly expands) the domain of the original codeword $\Codeword$; this holds recursively, and induces a recursive cover $\TreeCover$ (see \defref{def:bsrs-tree-cover}). We prove that $\TreeCover$ is $1$-intersecting (see \clmref{claim:bsrs-1-intersecting}) and that, for every vertex $v$ in $\TreeCover$ of depth at most $\CDepth$, the cover $\TreeCover_{v}$ is $(\SetCardinality{\BSSpace}^{2^{-\CDepth-1}} \cdot \SetCardinality{\SubField}^{-\BSBalance -2})$-independent, which implies the stated independence property about $\TreeCover_{v}$ (see \clmref{cor:bsrs-low-depth-big-independence}). The core of the argument for this second claim is to show that the code $\CCode_{v}$ equals $\BSCode[\SubField,\Field,\BSSpace_{v},\BSBalance,\BSBaseDim]$ for some subspace $\BSSpace_{v}$ such that $\Dimension{\BSSpace} \cdot 2^{-\CDepth} \leq \Dimension{\OtherBSSpace} \leq \Dimension{\BSSpace} \cdot 2^{-\CDepth} + 2\BSBalance$ (see \clmref{claim:bsrs-bounded-basis}).

The second part of the proof of \lemref{lem:bsrs-has-recursive-code-cover} consists of establishing the computational efficiency of certain tasks related to the recursive cover (see \appref{sec:bsrs-complexity}). Specifically, we bound the time required to compute a spanning set for covers in $\TreeCover$ (see \clmref{claim:bsrs-efficient-basis}). After a few more observations, we are able to conclude the proof.

\doclearpage
%%%%%%%%%%%%%%%%%%%%%%%%%%%%%%%%%%%%%%%%%%%%%%%%%%%%%%%%%%%%%%%%%%%%%%%%%%%%%%%%
%%%%%%%%%%%%%%%%%%%%%%%%%%%%%%%%%%%%%%%%%%%%%%%%%%%%%%%%%%%%%%%%%%%%%%%%%%%%%%%%
%%%%%%%%%%%%%%%%%%%%%%%%%%%%%%%%%%%%%%%%%%%%%%%%%%%%%%%%%%%%%%%%%%%%%%%%%%%%%%%%
\section{Sumcheck with perfect zero knowledge}
\label{sec:zk-sumcheck}

We show how to obtain an IPCP for sumcheck that is \emph{perfect zero knowledge against unbounded queries}.

%%%%%%%%%%%%%%%%%%%%%%%%%%%%%%%%%%%%%%%%
\parhead{Sumcheck}
The sumcheck protocol \cite{LundFKN92,Shamir92} is an IP for the claim ``$\sum_{\vec{\alpha} \in \SCSubset^{\SCVars}} \SCPoly(\vec{\alpha})=0$'', where $\SCPoly$ is a polynomial in $\PolynomialRingIndOne{\Field}{\SCVars}{\VariableX}{\SCDegree}$ and $\SCSubset$ is a subset of $\Field$. The prover and verifier have input $(\Field,\SCVars,\SCDegree,\SCSubset)$ and oracle access to (the evaluation table on $\Field^{\SCVars}$ of) $\SCPoly$. The sumcheck protocol has soundness error $1-(1-\frac{\SCDegree}{\SetCardinality{\Field}})^{\SCVars}$; the prover runs in space $\poly(\log \SetCardinality{\Field} + \SCVars + \SCDegree + \SetCardinality{\SCSubset})$ and the verifier in time $\poly(\log \SetCardinality{\Field} + \SCVars + \SCDegree + \SetCardinality{\SCSubset})$; the number of rounds is $\SCVars$; finally, the protocol is public coin and the verifier queries $\SCPoly$ only at one random point.

%%%%%%%%%%%%%%%%%%%%%%%%%%%%%%%%%%%%%%%%
\parhead{Leakage}
The sumcheck protocol is \emph{not} zero knowledge: a verifier, by interacting with the honest prover, learns partial sums of $\SCPoly$, in addition to the fact that ``$\sum_{\vec{\alpha} \in \SCSubset^{\SCVars}} \SCPoly(\vec{\alpha})=0$'' is true. Assuming one way functions, one \emph{can} make any interactive proof, including the sumcheck protocol, to be (computational) zero knowledge \cite{GoldwasserMR89,ImpagliazzoY87,BenOrGGHKMR88}; moreover, one-way functions are necessary for obtaining zero knowledge IPs for non-trivial languages \cite{OstrovskyW93}. As we do not wish to make intractability assumptions, we now turn to a different proof system model.

%%%%%%%%%%%%%%%%%%%%%%%%%%%%%%%%%%%%%%%%
\parhead{Perfect zero knowledge via IPCPs}
We obtain an IPCP for sumcheck that is perfect zero knowledge against unbounded queries. Namely, a malicious verifier has oracle access to a proof string $\Proof$ and also interacts with the prover, but learns no information about $\SCPoly$ beyond the fact that the statement about $\SCPoly$ is true, in the following sense. There exists an algorithm that perfectly simulates the verifier's view by making as many queries to $\SCPoly$ as the \emph{total} number of verifier queries to either $\SCPoly$ or the oracle $\Proof$. (Analogously to zero knowledge for proximity testers, a verifier may query $\SCPoly$ at any time, so any such information comes ``for free'' and, also, any query to $\Proof$ `counts' as a query to $\SCPoly$; see \secref{sec:zero-knowledge}.)

Our construction proceeds in two steps:
\begin{itemize}

  \item \emph{Step 1.}
  We modify the sumcheck protocol to make it perfect zero knowledge, but in a hybrid model where the prover and verifier have access to a random polynomial $\RandPoly \in \PolynomialRingIndOne{\Field}{\SCVars}{\VariableX}{\SCDegree}$ such that $\sum_{\vec{\alpha} \in \SCSubset^{\SCVars}} \RandPoly(\vec{\alpha})=0$. Crucially, soundness relies only on the fact that $\RandPoly$ is low-degree, but not the fact that it is random or sums to $0$. Also, the modified protocol does \emph{not} depend on a bound on the malicious verifier's queries, and thus maintains zero knowledge even against unbounded queries.

  \item \emph{Step 2.}
  We observe that in the IPCP model the prover can send an oracle proof string $\Proof$ that represents the evaluation table of $\RandPoly$, and the verifier can test that $\Proof$ is close to low-degree, and then use self correction to query it. This extension preserves the zero knowledge properties of the previous step.

\end{itemize}
The more interesting of the two steps is the first one, so we briefly discuss the intuition behind it. Our idea is that, rather than executing the sumcheck protocol on $\SCPoly$ directly, the prover and verifier execute it on $\rho\SCPoly + \RandPoly$, where $\rho$ is chosen at random by the verifier (after $\RandPoly$ is sampled). Completeness is clear because if $\sum_{\vec{\alpha} \in \SCSubset^{\SCVars}} \SCPoly(\vec{\alpha})=0$ and $\sum_{\vec{\alpha} \in \SCSubset^{\SCVars}} \RandPoly(\vec{\alpha})=0$ then $\sum_{\vec{\alpha} \in \SCSubset^{\SCVars}} (\rho\SCPoly + \RandPoly)(\vec{\alpha})=0$; soundness is also clear because if $\sum_{\vec{\alpha} \in \SCSubset^{\SCVars}} \SCPoly(\vec{\alpha}) \neq 0$ then $\sum_{\vec{\alpha} \in \SCSubset^{\SCVars}} (\rho\SCPoly + \RandPoly)(\vec{\alpha}) \neq 0$ with high probability over $\rho$ (regardless of $\sum_{\vec{\alpha} \in \SCSubset^{\SCVars}} \RandPoly(\vec{\alpha})$). We are thus left to show perfect zero knowledge, which turns out to be a much less straightforward argument.

On the surface, perfect zero knowledge appears easy to argue: simply note that $\rho\SCPoly + \RandPoly$ is random among all polynomials in $\PolynomialRingIndOne{\Field}{\SCVars}{\VariableX}{\SCDegree}$ that sum to zero on $\SCSubset^{\SCVars}$. However, this argument, while compelling, is not enough. First, $\rho\SCPoly + \RandPoly$ is \emph{not} random because a malicious verifier can choose $\rho$ depending on queries to $\RandPoly$; we discuss this issue further down below. Second, even if $\rho\SCPoly + \RandPoly$ were random (e.g., the verifier does not query $\RandPoly$ before choosing $\rho$), the simulator must run in polynomial time but it is not clear how that is possible, as we now explain.

Consider the following simulator:
\begin{inparaenum}[(1)]
  \item sample a random polynomial $\Simulated{\MaskedPoly} \in \PolynomialRingIndOne{\Field}{\SCVars}{\VariableX}{\SCDegree}$ subject to $\sum_{\vec{\alpha} \in \SCSubset^{\SCVars}} \Simulated{\MaskedPoly}(\vec{\alpha})=0$ and use it to simulate $\rho\SCPoly + \RandPoly$;
  \item whenever the verifier queries $\SCPoly(\vec{\alpha})$, respond by querying $\SCPoly(\vec{\alpha})$ and returning the true value;
  \item whenever the verifier queries $\RandPoly(\vec{\alpha})$, respond by querying $\SCPoly(\vec{\alpha})$ and returning $\Simulated{\MaskedPoly}(\vec{\alpha}) - \rho\SCPoly(\vec{\alpha})$.
\end{inparaenum}
One can argue that the simulator produces the correct distribution; moreover, the number of queries to $\SCPoly$ made by the simulator equals the number of (mutually) distinct queries to $\SCPoly$ and $\RandPoly$ made by the verifier, as desired.

But how does the simulator sample a random polynomial in $\PolynomialRingIndOne{\Field}{\SCVars}{\VariableX}{\SCDegree}$ in polynomial time? The size of the representation of such a polynomial is $\Omega(\SCDegree^{\SCVars})$, which is exponential. We get around this problem by exploiting the fact that the number of queries the verifier can make is polynomially bounded, and the simulator can keep state about the answers to past queries and `make up' on the fly the answer to a new query by resolving dependencies between queries. More precisely, we leverage our construction of a succinct constraint detector for evaluations of low-degree polynomials (see \secref{sec:partial-sums}), which itself relies on tools borrowed from algebraic complexity theory. The same detector also allows to simulate \emph{partial sums}, which the prover sends in the course of the sumcheck protocol itself.

Finally, we explain how we address the issue that the verifier may choose to query $\RandPoly$ before sending $\rho$. We handle this by first (implicitly) sampling a random polynomial $\Simulated{\RandPoly}$, and responding to each verifier query to $\RandPoly(\vec{\alpha})$ with $\Simulated{\RandPoly}(\vec{\alpha})$. Then, when the verifier sends $\rho$, we draw $\Simulated{\MaskedPoly}$ conditioned on the already-queried values for $\RandPoly$ being `correct'; i.e., for each point $\vec{\alpha}$ queried before $\rho$ is sent, we add the condition that $\Simulated{\MaskedPoly}(\vec{\alpha}) = \rho\SCPoly(\vec{\alpha}) + \Simulated{\RandPoly}(\vec{\alpha})$. We then continue as described above, and it is not too difficult to argue that this strategy yields the correct distribution.

\medskip
\noindent
We are now ready to turn the above discussions into formal definitions and proofs. First, we give the definition of the sumcheck relation and of a PZK IPCP system for sumcheck; then we state and prove the PZK Sumcheck Theorem.

\begin{definition}
\label{def:sumcheck-relation}
The sumcheck relation and its promise variant are defined as follows.
\begin{itemize}
  \item The \emph{sumcheck relation} is the relation $\SCRelation$ of instance-witness pairs $\big( (\Field,\SCVars,\SCDegree,\SCSubset,\SCConstant) , \SCPoly \big)$ such that
\begin{inparaenum}[(i)]
  \item $\Field$ is a finite field, $\SCSubset$ is a subset of $\Field$, $\SCConstant$ is an element of $\Field$, and $\SCVars,\SCDegree$ are positive integers with $\frac{\SCVars\SCDegree}{\SetCardinality{\Field}} < \frac{1}{2}$;
  \item $\SCPoly$ is in $\PolynomialRingIndOne{\Field}{\SCVars}{\VariableX}{\SCDegree}$ and sums to $\SCConstant$ on $\SCSubset^{\SCVars}$.
\end{inparaenum}

  \item The \emph{sumcheck promise relation} is the pair of relations $(\SCRelation^{\yes},\SCRelation^{\no})$ where $\SCRelation^{\yes} \DefineEqual \SCRelation$ and $\SCRelation^{\no}$ are the pairs $\big( (\Field,\SCVars,\SCDegree,\SCSubset,\SCConstant) , \SCPoly \big)$ such that $(\Field,\SCVars,\SCDegree,\SCSubset,\SCConstant)$ is as above and $\SCPoly$ is in $\PolynomialRingIndOne{\Field}{\SCVars}{\VariableX}{\SCDegree}$ but does \sunderline{not} \mbox{sum to $\SCConstant$ on $\SCSubset^{\SCVars}$.}

\end{itemize}
\end{definition}

\begin{definition}
\label{def:sumcheck-ipcp}
A \emph{PZK IPCP system for sumcheck} with soundness error $\SoundnessError$ is a pair of interactive algorithms $(\Prover,\Verifier)$ that satisfies the following properties.
\begin{itemize}

  \item \textsc{Completeness.}
  For every $\big( (\Field,\SCVars,\SCDegree,\SCSubset,\SCConstant) , \SCPoly \big) \in \SCRelation^{\yes}$,
  $
  \Pr\Big[
  	\Interact{\Prover^{\SCPoly}(\Field,\SCVars,\SCDegree,\SCSubset,\SCConstant)}
	         {\Verifier^{\SCPoly}(\Field,\SCVars,\SCDegree,\SCSubset,\SCConstant)}
	=1
  \Big]=1
  $.

  \item \textsc{Soundness.}
  For every $\big( (\Field,\SCVars,\SCDegree,\SCSubset,\SCConstant) , \SCPoly \big) \in \SCRelation^{\no}$ and malicious prover $\Malicious{\Prover}$, \mbox{
  $
  \Pr\Big[
    \Interact{\Malicious{\Prover}}
             {\Verifier^{\SCPoly}(\Field,\SCVars,\SCDegree,\SCSubset,\SCConstant)}
   =1
  \Big] \leq \SoundnessError
  $.}

  \item \textsc{Perfect zero knowledge.}
  There exists a straightline simulator $\Simulator$ such that, for every \mbox{$\big( (\Field,\SCVars,\SCDegree,\SCSubset,\SCConstant) , \SCPoly \big) \in \SCRelation^{\yes}$} and malicious verifier $\Malicious{\Verifier}$, the following two random variables are identically distributed
\begin{equation*}
\Big(\Simulator^{\Malicious{\Verifier},\SCPoly}(\Field,\SCVars,\SCDegree,\SCSubset,\SCConstant) \;,\; q_{\Simulator} \Big)
\quad\text{and}\quad
\Big(\IPCPView{\Prover^{\SCPoly}(\Field,\SCVars,\SCDegree,\SCSubset,\SCConstant)}{\Malicious{\Verifier}^{\SCPoly}} \;,\; q_{\Malicious{\Verifier}} \Big)
\enspace,
\end{equation*}
where $q_{\Simulator}$ is the number of queries to $\SCPoly$ made by $\Simulator$ and $q_{\Malicious{\Verifier}}$ is the number of queries to $\SCPoly$ or PCP oracle made by $\Malicious{\Verifier}$. Moreover, $\Simulator$ runs in time $\poly(\log \SetCardinality{\Field} + \SetCardinality{\SCSubset} + \SCVars + \QueryComplexity_{\Malicious{\Verifier}})$, where $\QueryComplexity_{\Malicious{\Verifier}}$ is $\Malicious{\Verifier}$'s query complexity.
\end{itemize}
\end{definition}

\begin{theorem}[PZK Sumcheck]
\label{thm:sumcheck-ipcp}
There exists a PZK public-coin IPCP system $(\Prover,\Verifier)$ for sumcheck with soundness error $\SoundnessError = O(\frac{\SCVars\SCDegree}{\SetCardinality{\Field}})$ and the following efficiency parameters.
\begin{itemize}[nolistsep]

  \item \emph{Oracle round:}
  $\Prover$ sends an oracle proof string $\Proof \colon {\Field}^{\SCVars} \to \Field$.

  \item \emph{Interactive proof:}
  after the oracle round, $\Prover$ and $\Verifier$ engage in an $\SCVars$-round interactive proof; across the interaction, the verifier sends to the prover $O(\SCVars)$ field elements, while the prover sends to the verifier $O(\SCVars \SCDegree)$ field elements.

  \item \emph{Queries:}
  after the interactive proof, $\Verifier$ non-adaptively queries $\Proof$ at $\poly(\log \SetCardinality{\Field} + \SCVars + \SCDegree)$ locations.

  \item \emph{Space and time:}
  $\Prover$ runs in space $\poly(\log \SetCardinality{\Field} + \SCVars + \SCDegree + \SetCardinality{\SCSubset})$, while $\Verifier$ in time $\poly(\log \SetCardinality{\Field} + \SCVars + \SCDegree + \SetCardinality{\SCSubset})$. (The prover's space complexity assumes that the randomness tape is two-way rather than one-way; see \remref{rem:prover-space} below.)

\end{itemize}
\end{theorem}

%%%%%%%%%%%%%%%%%%%%%%%%%%%%%%%%%%%%%%%%%%%%%%%%%%%%%%%%%%%%%%%%%%%%%%%%%%%%%%%%
%%%%%%%%%%%%%%%%%%%%%%%%%%%%%%%%%%%%%%%%%%%%%%%%%%%%%%%%%%%%%%%%%%%%%%%%%%%%%%%%
\subsection{Step 1}
\label{sec:zk-sumcheck-step1}

We construct a public-coin IP for sumcheck that is perfect zero knowledge, in the ``$\RandPoly$-hybrid'' model, where the prover and verifier have access to a uniformly random $\RandPoly \in\PolynomialRingIndOne{\Field}{\SCVars}{\VariableX}{\SCDegree}$ conditioned on $\sum_{\vec{\alpha} \in \SCSubset^{\SCVars}} \RandPoly(\vec{\alpha})=0$.

\begin{construction}
\label{con:sumcheck-ip}
The IP system $\pair{\IPSCProver}{\IPSCVerifier}$ is defined as follows. Both $\IPSCProver$ and $\IPSCVerifier$ receive a tuple $(\Field,\SCVars,\SCDegree,\SCSubset,\SCConstant)$ as common input, and two polynomials $\SCPoly,\RandPoly \in \PolynomialRingIndOne{\Field}{\SCVars}{\VariableX}{\SCDegree}$ as oracles. The interaction proceeds as follows:
\begin{enumerate}[nolistsep]
  \item $\IPSCVerifier$ draws a random element $\rho$ in $\Field$, and sends $\rho$ to $\IPSCProver$; then
  \item $\IPSCProver$ and $\IPSCVerifier$ run the sumcheck IP \cite{LundFKN92,Shamir92} on the statement ``$\sum_{\vec{\alpha} \in \SCSubset^{\SCVars}} \MaskedPoly(\vec{\alpha})=\rho\SCConstant$'' where $\MaskedPoly \DefineEqual \rho \SCPoly + \RandPoly$ (with $\IPSCProver$ playing the role of the prover and $\IPSCVerifier$ that of the verifier).
\end{enumerate}
\end{construction}
\noindent
Note that $\pair{\IPSCProver}{\IPSCVerifier}$ is public-coin, and satisfies the following efficiency properties.
\begin{itemize}[label=--]

  \item \sunderline{Communication:}
  The number of rounds is $\SCVars$, because the first message is from $\IPSCVerifier$ to $\IPSCProver$; then the standard sumcheck protocol starts with a message from the prover to the verifier, and is followed by $\SCVars-1$ full rounds. Across the interaction, $\IPSCProver$ sends $O(\SCVars)$ field elements to $\IPSCVerifier$, while $\IPSCProver$ sends $O(\SCVars \SCDegree)$ field elements to $\IPSCVerifier$.

  \item \sunderline{Queries:}
  $\IPSCVerifier$ queries $\SCPoly$ and $\RandPoly$ each at a single random point because, at the end of the sumcheck protocol, the verifier queries $\MaskedPoly$ at a random point $\vec{\gamma}$, and such a query can be ``simulated'' by querying $\SCPoly$ and $\RandPoly$ at $\vec{\gamma}$ and then using these answers, along with $\rho$, to compute the necessary value for $\MaskedPoly$.

  \item \sunderline{Space and time:}
  $\IPSCProver$ runs in space $\poly(\log \SetCardinality{\Field} + \SCVars + \SCDegree + \SetCardinality{\SCSubset})$, while $\IPSCVerifier$ in time $\poly(\log \SetCardinality{\Field} + \SCVars + \SCDegree + \SetCardinality{\SCSubset})$. (The prover's space complexity assumes that the randomness tape is two-way; see \remref{rem:prover-space} below.)

\end{itemize}
We now state and prove the completeness, soundness, and perfect zero knowledge properties.

\begin{lemma}
\label{lem:zk-sumcheck}
The IP system $\pair{\IPSCProver}{\IPSCVerifier}$ satisfies the following properties.
\begin{itemize}

  \item \textsc{Completeness.}
  For every $\big( (\Field,\SCVars,\SCDegree,\SCSubset,\SCConstant) , \SCPoly \big) \in \SCRelation^{\yes}$
  and $\RandPoly \in \PolynomialRingIndOne{\Field}{\SCVars}{\VariableX}{\SCDegree}$ with $\sum_{\vec{\alpha} \in \SCSubset^{\SCVars}} \RandPoly(\vec{\alpha})=0$,
  \begin{equation*}
  \Pr\Big[
  	\Interact{\IPSCProver^{\SCPoly,\RandPoly}(\Field,\SCVars,\SCDegree,\SCSubset,\SCConstant)}
	         {\IPSCVerifier^{\SCPoly,\RandPoly}(\Field,\SCVars,\SCDegree,\SCSubset,\SCConstant)}
	=1
  \Big]=1
  \enspace.
  \end{equation*}

  \item \textsc{Soundness.}
  For every $\big( (\Field,\SCVars,\SCDegree,\SCSubset,\SCConstant) , \SCPoly \big) \in \SCRelation^{\no}$, $\RandPoly \in \PolynomialRingIndOne{\Field}{\SCVars}{\VariableX}{\SCDegree}$, and malicious prover $\Malicious{\IOPProver}$,
  \begin{equation*}
  \Pr\Big[
    \Interact{\Malicious{\IOPProver}}
             {\IPSCVerifier^{\SCPoly,\RandPoly}(\Field,\SCVars,\SCDegree,\SCSubset,\SCConstant)}
   =1
  \Big] \leq \frac{\SCVars\SCDegree+1}{\SetCardinality{\Field}}
  \enspace.
  \end{equation*}

  \item \textsc{Perfect zero knowledge.}
  There exists a straightline simulator $\IPSCSimulator$ such that, for every $\big( (\Field,\SCVars,\SCDegree,\SCSubset,\SCConstant) , \SCPoly \big) \in \SCRelation^{\yes}$ and malicious verifier $\Malicious{\Verifier}$, the following two random variables are identically distributed
\begin{equation*}
\Big(\IPSCSimulator^{\Malicious{\Verifier},\SCPoly}(\Field,\SCVars,\SCDegree,\SCSubset,\SCConstant) \;,\; q_{\IPSCSimulator} \Big)
\quad\text{and}\quad
\Big(\IPCPView{\IPSCProver^{\SCPoly,\RandPoly}(\Field,\SCVars,\SCDegree,\SCSubset,\SCConstant)}{\Malicious{\Verifier}^{\SCPoly,\RandPoly}} \;,\; q_{\Malicious{\Verifier}} \Big)
\enspace,
\end{equation*}
where $\RandPoly$ is uniformly random in $\PolynomialRingIndOne{\Field}{\SCVars}{\VariableX}{\SCDegree}$ conditioned on $\sum_{\vec{\alpha} \in \SCSubset^{\SCVars}} \RandPoly(\vec{\alpha})=0$, $q_{\IPSCSimulator}$ is the number of queries to $\SCPoly$ made by $\IPSCSimulator$, and $q_{\Malicious{\Verifier}}$ is the number of queries to $\SCPoly$ or $\RandPoly$ made by $\Malicious{\Verifier}$. Moreover, $\IPSCSimulator$ runs in time $\poly(\log \SetCardinality{\Field} + \SCVars + \SCDegree + \SetCardinality{\SCSubset} + \QueryComplexity_{\Malicious{\Verifier}})$\, where $\QueryComplexity_{\Malicious{\Verifier}}$ is $\Malicious{\Verifier}$'s query complexity.
\end{itemize}
\end{lemma}

\begin{proof}
We argue first completeness, then soundness, and, finally, perfect zero knowledge.

%%%%%%%%%%%%%%%%%%%%%%%%%%%%%%%%%%%%%%%%
\parhead{Completeness}
If both $\SCPoly$ sums to $\SCConstant$ on $\SCSubset^{\SCVars}$ and $\RandPoly$ sums to $0$ on $\SCSubset^{\SCVars}$, then $\MaskedPoly \DefineEqual \rho \SCPoly + \RandPoly$ sums to $\rho\SCConstant$ on $\SCSubset^{\SCVars}$ for every choice of $\rho$. Then completeness follows from the completeness of standard sumcheck.

%%%%%%%%%%%%%%%%%%%%%%%%%%%%%%%%%%%%%%%%
\parhead{Soundness}
For every $\SCPoly,\RandPoly \in \PolynomialRingIndOne{\Field}{\SCVars}{\VariableX}{\SCDegree}$ with $\sum_{\vec{\alpha} \in \SCSubset^{\SCVars}} \SCPoly(\vec{\alpha}) \neq 0$. Then $\sum_{\vec{\alpha} \in \SCSubset^{\SCVars}} \MaskedPoly(\vec{\alpha})$ equals $\rho\SCConstant$ for at most one choice of $\rho$, namely, $(\sum_{\vec{\alpha} \in \SCSubset^{\SCVars}} \RandPoly(\vec{\alpha}))/(\SCConstant-\sum_{\vec{\alpha} \in \SCSubset^{\SCVars}} \SCPoly(\vec{\alpha}))$. Thus, except with probability $1/\SetCardinality{\Field}$, the sumcheck protocol is invoked on an incorrect claim, which incurs a soundness error of at most $\frac{\SCVars\SCDegree}{\SetCardinality{\Field}}$. The claimed soundness error follows by a union bound.

%%%%%%%%%%%%%%%%%%%%%%%%%%%%%%%%%%%%%%%%
\parhead{Perfect zero knowledge}
We begin by proving perfect zero knowledge via a straightline simulator $\SlowSimulator$ whose number of queries to $\SCPoly$ equals $q_{\Malicious{\Verifier}}$, but runs in time $\poly(\SetCardinality{\Field}^{\SCVars}+q_{\Malicious{\IOPVerifier}})$. After that, we explain how to modify $\SlowSimulator$ into another simulator $\IPSCSimulator$, with an identical output distribution, that runs in the faster time claimed in the lemma.

\begin{mdframed}
{\small
The simulator $\SlowSimulator$, given straightline access to $\Malicious{\IOPVerifier}$ and oracle access to $\SCPoly$, works as follows:
\medskip
\begin{enumerate}[nolistsep]

  \item \label{step:ipsc-sim-draw-R}
  Draw a uniformly random $\Simulated{\RandPoly} \in \PolynomialRingIndOne{\Field}{\SCVars}{\VariableX}{\SCDegree}$ conditioned on $\sum_{\vec{\alpha} \in \SCSubset^{\SCVars}} \Simulated{\RandPoly}(\vec{\alpha}) = 0$.

  \item \label{step:ipsc-sim-before-rho}
  Whenever $\Malicious{\IOPVerifier}$ queries $\SCPoly$ at $\vec{\gamma} \in \Field^{\SCVars}$, return $\SCPoly(\vec{\gamma})$; whenever $\Malicious{\IOPVerifier}$ queries $\RandPoly$ at $\vec{\gamma} \in \Field^{\SCVars}$, return $\Simulated{\RandPoly}(\vec{\gamma})$.

  \item \label{step:ipsc-sim-draw-Q}
  Receive $\Malicious{\rho}$ from $\Malicious{\IOPVerifier}$, and draw a uniformly random $\Simulated{\MaskedPoly} \in \PolynomialRingIndOne{\Field}{\SCVars}{\VariableX}{\SCDegree}$ conditioned on $\sum_{\vec{\alpha} \in \SCSubset^{\SCVars}} \Simulated{\MaskedPoly}(\vec{\alpha}) = \Malicious{\rho}\SCConstant$ and $\Simulated{\MaskedPoly}(\vec{\gamma}) = \Malicious{\rho} \SCPoly(\vec{\gamma}) + \Simulated{\RandPoly}(\vec{\gamma})$ for every coordinate $\vec{\gamma} \in \Field^{\SCVars}$ queried in \stepref{step:ipsc-sim-before-rho}. (This latter condition requires querying $\SCPoly$ at $\vec{\gamma}$ for every coordinate $\vec{\gamma} \in \Field^{\SCVars}$ queried to $\Simulated{\RandPoly}$ in \stepref{step:ipsc-sim-before-rho}.)

  \item \label{step:ipsc-sim-after-rho}
  Hereafter: whenever $\Malicious{\IOPVerifier}$ queries $\SCPoly$ at $\vec{\gamma} \in \Field^{\SCVars}$, return $\SCPoly(\vec{\gamma})$; whenever $\Malicious{\IOPVerifier}$ queries $\RandPoly$ at $\vec{\gamma} \in \Field^{\SCVars}$, return $\Simulated{\MaskedPoly}(\vec{\gamma}) - \Malicious{\rho} \SCPoly(\vec{\gamma})$. (In either case, a query to $\SCPoly$ is required.)

  \item \label{step:ipsc-sim-sumcheck}
  Run the sumcheck protocol with $\Malicious{\IOPVerifier}$ on $\Simulated{\MaskedPoly}$. (Note that $\Malicious{\IOPVerifier}$ may query $\SCPoly$ or $\RandPoly$ before, during, or after this protocol.)

  \item Output the view of the simulated $\Malicious{\IOPVerifier}$.

\end{enumerate}
}
\end{mdframed}
Note that $\SlowSimulator$ runs in time $\poly(\SetCardinality{\Field}^{\SCVars}+q_{\Malicious{\IOPVerifier}})$. Also, $\SlowSimulator$ makes one query to $\SCPoly$ for every query to $\SCPoly$ or $\RandPoly$ by $\Malicious{\IOPVerifier}$ (at least provided that $\Malicious{\IOPVerifier}$'s queries have no duplicates, which we can assume without loss of generality). Thus, overall, the number of queries to $\SCPoly$ by $\SlowSimulator$ is $q_{\Malicious{\IOPVerifier}}$. We now argue that $\SlowSimulator$'s output is identically distributed to $\Malicious{\IOPVerifier}$'s view when interacting with the honest prover $\IPSCProver$, for $\RandPoly$ random in $\PolynomialRingIndOne{\Field}{\SCVars}{\VariableX}{\SCDegree}$.
\begin{adjustwidth}{1cm}{1cm}
\begin{uclaim}
$\SlowSimulator^{\Malicious{\Verifier},\SCPoly} \equiv \IPCPView{\IPSCProver^{\SCPoly,\RandPoly}}{\Malicious{\Verifier}^{\SCPoly,\RandPoly}}$.
\end{uclaim}
\begin{proof}
Define the random variable $\MaskedPoly \DefineEqual \Malicious{\rho} \SCPoly + \RandPoly$, where $\Malicious{\rho}$ is chosen by $\Malicious{\IOPVerifier}$. Observe that there exists a (deterministic) function $\FView{\cdot}$ such that
\begin{equation*}
\IPCPView{\IPSCProver^{\SCPoly,\RandPoly}}{\Malicious{\Verifier}^{\SCPoly,\RandPoly}} = \FView{\MaskedPoly,\SCPoly,r}
\quad\text{and}\quad
\SlowSimulator^{\Malicious{\Verifier},\SCPoly} = \FView{\Simulated{\MaskedPoly},\SCPoly,\Randomness} \enspace,
\end{equation*}
where the random variable $\Randomness$ is $\Malicious{\IOPVerifier}$'s private randomness. Indeed,
\begin{inparaenum}[(i)]
  \item the messages sent and received by $\Malicious{\IOPVerifier}$ are identical to those when interacting with $\IPSCProver$ on $\MaskedPoly$ and $\Simulated{\MaskedPoly}$, respectively;
  \item $\Malicious{\IOPVerifier}$'s queries to $\SCPoly$ are answered honestly;
  \item $\Malicious{\IOPVerifier}$'s queries to $\RandPoly$ are answered by $\RandPoly = \MaskedPoly - \Malicious{\rho} \SCPoly$ and $\Simulated{\RandPoly} = \Simulated{\MaskedPoly} - \Malicious{\rho} \SCPoly$ respectively.
\end{inparaenum}
We are only left to argue that, for any choice of $\Randomness$, $\MaskedPoly$ and $\Simulated{\MaskedPoly}$ are identically distributed:
\begin{itemize}[label=--]

  \item $\MaskedPoly = \Malicious{\rho}\SCPoly + \RandPoly$ is uniformly random in $\PolynomialRingIndOne{\Field}{\SCVars}{\VariableX}{\SCDegree}$ conditioned on $\sum_{\vec{\alpha} \in \SCSubset^{\SCVars}} \MaskedPoly(\vec{\alpha}) = \Malicious{\rho}\SCConstant$, because $\RandPoly$ is uniformly random in $\PolynomialRingIndOne{\Field}{\SCVars}{\VariableX}{\SCDegree}$ conditioned on $\sum_{\vec{\alpha} \in \SCSubset^{\SCVars}} \RandPoly(\vec{\alpha}) = 0$ (and $\SCPoly$ is in $\PolynomialRingIndOne{\Field}{\SCVars}{\VariableX}{\SCDegree}$ and satisfies $\sum_{\vec{\alpha} \in \SCSubset^{\SCVars}} \SCPoly(\vec{\alpha}) = \SCConstant$); and

  \item $\Simulated{\MaskedPoly}$ is uniformly random in $\PolynomialRingIndOne{\Field}{\SCVars}{\VariableX}{\SCDegree}$ conditioned on $\sum_{\vec{\alpha} \in \SCSubset^{\SCVars}} \Simulated{\MaskedPoly}(\vec{\alpha}) = \Malicious{\rho}\SCConstant$, because $\Simulated{\MaskedPoly}$ is sampled at random in $\PolynomialRingIndOne{\Field}{\SCVars}{\VariableX}{\SCDegree}$ conditioned on $\sum_{\vec{\alpha} \in \SCSubset^{\SCVars}} \Simulated{\MaskedPoly}(\vec{\alpha}) = \Malicious{\rho}\SCConstant$ and $\Simulated{\MaskedPoly}(\vec{\gamma}_{i}) = \Simulated{\RandPoly}(\vec{\gamma}_{i}) + \Malicious{\rho} \SCPoly(\vec{\gamma}_{i})$ for some (adversarial) choice of $\vec{\gamma}_{1},\dots,\vec{\gamma}_{k}$. But $\Simulated{\RandPoly}$ is uniformly random in $\PolynomialRingIndOne{\Field}{\SCVars}{\VariableX}{\SCDegree}$ such that $\sum_{\vec{\alpha} \in \SCSubset^{\SCVars}} \Simulated{\RandPoly}(\vec{\alpha}) = 0$, so the latter condition says that $\Simulated{\MaskedPoly}$ matches a random polynomial on the set of points $\{\vec{\gamma}_{1},\dots,\vec{\gamma}_{k}\}$, giving the claimed distribution for $\Simulated{\MaskedPoly}$. \qedhere

\end{itemize}
\end{proof}
\end{adjustwidth}
We explain how to modify $\SlowSimulator$ so as to reduce the running time to $\poly(\log \SetCardinality{\Field} + \SCVars + \SCDegree + \SetCardinality{\SCSubset} + \QueryComplexity_{\Malicious{\Verifier}})$.

Note that $\SlowSimulator$'s inefficiency arises from sampling two random polynomials in $\PolynomialRingIndOne{\Field}{\SCVars}{\VariableX}{\SCDegree}$, namely $\Simulated{\RandPoly}$ and $\Simulated{\MaskedPoly}$, subject to certain constraints, and using them to answer $\Malicious{\Verifier}$'s messages and queries. We observe (and carefully justify below) that all information about $\Simulated{\RandPoly}$ and $\Simulated{\MaskedPoly}$ received by $\Malicious{\Verifier}$ is answers to queries of the form ``given $\vec{\gamma} \in \Field^{\leq\SCVars}$, return the value $A(\vec{\gamma}) \DefineEqual \sum_{\vec{\alpha} \in \SCSubset^{\SCVars - |\vec{\gamma}|}} A(\vec{\gamma}, \vec{\alpha})$'' for a random $A \in \PolynomialRingIndOne{\Field}{\SCVars}{\VariableX}{\SCDegree}$, possibly conditioned on previous such queries; when $\vec{\gamma}$ has length zero we use the symbol $\EmptyVector$, so that $A(\EmptyVector)$ denotes $\sum_{\vec{\alpha} \in \SCSubset^{\SCVars}} A(\vec{\alpha})$. The new simulator can use the algorithm $\CodeSimAlgorithm$ from our \corref{cor:efficient-poly-simulator} to adaptively answer such queries, without ever explicitly sampling the two polynomials.

We now argue that all information about $\Simulated{\RandPoly}$ and $\Simulated{\MaskedPoly}$ received by $\Malicious{\Verifier}$ from $\SlowSimulator$ can be viewed as queries of the above form, by discussing each step of $\SlowSimulator$.
\begin{itemize}

  \item In \stepref{step:ipsc-sim-draw-R}, $\SlowSimulator$ draws a uniformly random $\Simulated{\RandPoly} \in \PolynomialRingIndOne{\Field}{\SCVars}{\VariableX}{\SCDegree}$ conditioned on $\Simulated{\RandPoly}(\EmptyVector) = 0$.

  \item In \stepref{step:ipsc-sim-before-rho}, $\SlowSimulator$ answers any query $\vec{\gamma} \in \Field^{\SCVars}$ to $\RandPoly$ with $\Simulated{\RandPoly}(\vec{\gamma})$.

  \item In \stepref{step:ipsc-sim-draw-Q}, $\SlowSimulator$ draws a uniformly random $\Simulated{\MaskedPoly} \in \PolynomialRingIndOne{\Field}{\SCVars}{\VariableX}{\SCDegree}$ conditioned on $\Simulated{\MaskedPoly}(\bot) = \Malicious{\rho}\SCConstant$ and also on $\Simulated{\MaskedPoly}(\vec{\gamma}) = \Simulated{\RandPoly}(\vec{\gamma}) + \Malicious{\rho} \SCPoly(\vec{\gamma})$ for at most $q_{\Malicious{\Verifier}}$ points $\vec{\gamma} \in \Field^{\SCVars}$ (namely, the points corresponding to queries in \stepref{step:ipsc-sim-before-rho}).

  \item In \stepref{step:ipsc-sim-after-rho}, $\SlowSimulator$ replies any query $\vec{\gamma} \in \Field^{\SCVars}$ to $\RandPoly$ with $\Simulated{\MaskedPoly}(\vec{\gamma}) - \Malicious{\rho} \SCPoly(\vec{\gamma})$.

  \item In \stepref{step:ipsc-sim-sumcheck}, $\SlowSimulator$ runs the sumcheck protocol with $\Malicious{\Verifier}$ on $\Simulated{\MaskedPoly}$, which requires computing univariate polynomials of the form $\sum_{\vec{\alpha} \in \SCSubset^{\SCVars-|\theta|-1}} \Simulated{\MaskedPoly}(\vec{\theta}, \VariableX, \vec{\alpha}) \in \Field[\VariableX]$ for various choices of $\vec{\theta} \in \Field^{<\SCVars}$. Each of these polynomials has degree less than $\SCDegree$, and so can be obtained by interpolation from its evaluation at any $\SCDegree$ distinct points; each of these is the answer of a query to $\Simulated{\MaskedPoly}$ of the required form, with $\vec{\gamma} = (\vec{\theta}, \delta)$  for some $\delta \in \Field$. Overall, during the protocol, $\SlowSimulator$ only needs to query $\Simulated{\MaskedPoly}$ at $\SCVars\SCDegree$ points $\vec{\gamma} \in \Field^{\leq \SCVars}$.

\end{itemize}
In sum, we can modify $\SlowSimulator$ so that instead of explicitly sampling $\Simulated{\RandPoly}$ and $\Simulated{\MaskedPoly}$, it uses $\CodeSimAlgorithm$ to sample the answer for each query to $\Simulated{\MaskedPoly}$ or $\Simulated{\RandPoly}$, conditioning the uniform distribution on the answers to previous queries. Putting all of this together, we obtain the simulator $\IPSCSimulator$ described below, whose output is identically distributed to the output of $\SlowSimulator$.
\begin{mdframed}
{\small
The simulator $\IPSCSimulator$, given straightline access to $\Malicious{\IOPVerifier}$ and oracle access to $\SCPoly$, works as follows:
\medskip

\begin{enumerate}[nolistsep]

  \item Let $\AnsTable{\Simulated{\RandPoly}}$ be a subset of $\Field^{\leq \SCVars} \times \Field$ that records query-value pairs for $\Simulated{\RandPoly}$; initially, $\AnsTable{\Simulated{\RandPoly}}$ equals $\{(\EmptyVector, 0)\}$.

  \item \label{step:fast-ipsc-sim-before-rho}
  Whenever $\Malicious{\IOPVerifier}$ queries $\SCPoly$ at $\vec{\gamma} \in \Field^{\SCVars}$, return $\SCPoly(\vec{\gamma})$; whenever $\Malicious{\IOPVerifier}$ queries $\RandPoly$ at $\vec{\gamma} \in \Field^{\SCVars}$, return $\beta \DefineEqual \CodeSimAlgorithm(\Field,\SCVars,\SCDegree,\SCSubset,\AnsTable{\Simulated{\RandPoly}}, \vec{\gamma})$. In the latter case, add $(\vec{\gamma}, \beta)$ to $\AnsTable{\Simulated{\RandPoly}}$.

  \item \label{step:fast-ipsc-sim-draw-Q}
  Receive $\Malicious{\rho}$ from $\Malicious{\IOPVerifier}$, and compute $\AnsTable{\Simulated{\MaskedPoly}} \DefineEqual \{ (\vec{\gamma}, \beta + \Malicious{\rho} \SCPoly(\vec{\gamma})) \}_{ (\vec{\gamma}, \beta) \in \AnsTable{\Simulated{\RandPoly}} }$; this subset of $\Field^{\leq \SCVars} \times \Field$ records query-value pairs for $\Simulated{\MaskedPoly}$. Note that $\AnsTable{\Simulated{\MaskedPoly}}$ includes the pair $(\EmptyVector, \Malicious{\rho}\SCConstant)$ because $\SCPoly(\EmptyVector) = \SCConstant$ by assumption.

  \item \label{step:fast-ipsc-sim-after-rho}
  Hereafter: whenever $\Malicious{\IOPVerifier}$ queries $\SCPoly$ at $\vec{\gamma} \in \Field^{\SCVars}$, return $\SCPoly(\vec{\gamma})$; whenever $\Malicious{\IOPVerifier}$ queries $\RandPoly$ at $\vec{\gamma} \in \Field^{\SCVars}$, return $\beta' \DefineEqual \beta - \Malicious{\rho} \SCPoly(\vec{\gamma})$ where $\beta \DefineEqual \CodeSimAlgorithm(\Field,\SCVars,\SCDegree,\SCSubset,\AnsTable{\Simulated{\MaskedPoly}},\vec{\gamma})$. In the latter case, add $(\vec{\gamma}, \beta)$ to $\AnsTable{\Simulated{\MaskedPoly}}$.

  \item \label{step:fast-ipsc-sim-sumcheck}
  Run the sumcheck protocol with $\Malicious{\IOPVerifier}$ on $\Simulated{\MaskedPoly}$, by using the algorithm $\CodeSimAlgorithm$ and updating $\AnsTable{\Simulated{\MaskedPoly}}$ appropriately. (Note that $\Malicious{\IOPVerifier}$ may query $\SCPoly$ or $\RandPoly$ before, during, or after this protocol.)

  \item Output the view of the simulated $\Malicious{\IOPVerifier}$.

\end{enumerate}
}
\end{mdframed}
Note that $\IPSCSimulator$ makes the same number of queries to $\SCPoly$ as $\SlowSimulator$ does. Also, the number of pairs in $\AnsTable{\Simulated{\RandPoly}}$ is at most $q_{\Malicious{\Verifier}} + \SCVars\SCDegree + 1$; ditto for $\AnsTable{\Simulated{\MaskedPoly}}$. Since the algorithm $\CodeSimAlgorithm$ is called at most $q_{\Malicious{\Verifier}} + \SCVars\SCDegree$ times, the running time of $\IPSCSimulator$ is $\poly(\log \SetCardinality{\Field} + \SCVars + \SCDegree + \SetCardinality{\SCSubset} + \QueryComplexity_{\Malicious{\Verifier}})$, as required.
\end{proof}

%%%%%%%%%%%%%%%%%%%%%%%%%%%%%%%%%%%%%%%%%%%%%%%%%%%%%%%%%%%%%%%%%%%%%%%%%%%%%%%%
%%%%%%%%%%%%%%%%%%%%%%%%%%%%%%%%%%%%%%%%%%%%%%%%%%%%%%%%%%%%%%%%%%%%%%%%%%%%%%%%
\subsection{Step 2}
\label{sec:zk-sumcheck-step2}

The IP described and analyzed in \secref{sec:zk-sumcheck-step1} is in the ``$\RandPoly$-hybrid'' model. We now compile that IP into an IPCP, by using proximity testing and self-correction, thereby concluding the proof of the PZK Sumcheck Theorem.

\begin{proof}[Proof of \thmref{thm:sumcheck-ipcp}]
Construct an IPCP system $(\Prover,\Verifier)$ for sumcheck as follows:
\begin{itemize}

  \item The prover $\Prover$, given input $(\Field,\SCVars,\SCDegree,\SCSubset,\SCConstant)$ and oracle access to $\SCPoly$, samples a uniformly random polynomial $\RandPoly \in \PolynomialRingIndOne{\Field}{\SCVars}{\VariableX}{\SCDegree}$ conditioned on $\sum_{\vec{\alpha} \in \SCSubset^{\SCVars}} \RandPoly(\vec{\alpha})=0$ and sends its evaluation $\Proof \colon \Field^{\SCVars} \to \Field$ to the verifier $\Verifier$. Then $\Prover$ simulates $\IPSCProver^{\SCPoly,\RandPoly}(\Field,\SCVars,\SCDegree,\SCSubset,\SCConstant)$.

  \item The verifier $\Verifier$, after receiving a proof string $\Proof \colon \Field^{\SCVars} \to \Field$, simulates $\IPSCVerifier^{\SCPoly,\Proof}(\Field,\SCVars,\SCDegree,\SCSubset,\SCConstant)$ up to $\IPSCVerifier$'s single query $\vec{\alpha} \in \Field^{\SCVars}$ to $\Proof$ (which occurs after the interaction), which $\Verifier$ does not answer directly but instead answers as follows. First, $\Verifier$ checks that $\Proof$ is $\varrho$-close to the evaluation of a polynomial in $\PolynomialRingIndOne{\Field}{\SCVars}{\VariableX}{\SCDegree}$ by performing an individual-degree test with proximity parameter $\varrho \DefineEqual \frac{1}{8}$ and soundness error $\epsilon \DefineEqual \frac{\SCVars\SCDegree}{\SetCardinality{\Field}}$ \cite{GoldreichS06,GurR15}; then, $\Verifier$ computes $\Proof(\vec{\alpha})$ via self-correction with soundness error $\epsilon$ \cite{RubinfeldS96,AroraS03}, and replies with that value. Both procedures require $\poly(\log \SetCardinality{\Field} + \SCVars + \SCDegree)$ queries and time. Finally, $\Verifier$ rejects if $\IPSCVerifier$ rejects or the individual degree test rejects.

\end{itemize}
Completeness and perfect zero knowledge of $(\Prover,\Verifier)$ are inherited, in a straightforward way, from those of $(\IPSCProver,\IPSCVerifier)$. We now argue soundness. So consider an instance-witness pair $\big( (\Field,\SCVars,\SCDegree,\SCSubset,\SCConstant) , \SCPoly \big) \in \SCRelation^{\no}$ and a malicious prover $\Malicious{\Prover}$, and denote by $\Malicious{\Proof} \colon \Field^{\SCVars} \to \Field$ the proof string sent by $\Malicious{\Prover}$. We distinguish between the following two cases.
\begin{itemize}

  \item \emph{Case 1: $\Malicious{\Proof}$ is $\varrho$-far from evaluations of polynomials in $\PolynomialRingIndOne{\Field}{\SCVars}{\VariableX}{\SCDegree}$.}

  In this case, the low-degree test accepts with probability at most $\epsilon$.

  \item \emph{Case 2: $\Malicious{\Proof}$ is $\varrho$-close to evaluations of polynomials in $\PolynomialRingIndOne{\Field}{\SCVars}{\VariableX}{\SCDegree}$.}

  In this case, let $\Malicious{\RandPoly}$ be the unique polynomial in $\PolynomialRingIndOne{\Field}{\SCVars}{\VariableX}{\SCDegree}$ whose evaluation is $\varrho$-close to $\Malicious{\Proof}$; this polynomial exists because $\varrho$ is less than the unique decoding radius (of the corresponding Reed--Muller code), which equals $\frac{1}{2} (1 - \frac{\SCDegree-1}{\SetCardinality{\Field}})^{\SCVars}$, and is at most $\frac{1}{4}$ by the assumption that $\frac{\SCVars\SCDegree}{\SetCardinality{\Field}} < \frac{1}{2}$. By the soundness of $(\IPSCProver,\IPSCVerifier)$, the probability that $\IPSCVerifier^{\SCPoly,\Malicious{\RandPoly}}$ accepts is at most $\frac{\SCVars\SCDegree+1}{\SetCardinality{\Field}}$ (see \lemref{lem:zk-sumcheck}). However $\Verifier$ only has access to $\Malicious{\Proof}$, and uses self-correction on it to compute $\Malicious{\RandPoly}$ at the single location $\vec{\alpha} \in \Field^{\SCVars}$ required by $\IPSCVerifier$; the probability that the returned value is not correct is at most $\epsilon$. Hence, by a union bound, $\Verifier$ accepts with probability at most $\frac{\SCVars\SCDegree+1}{\SetCardinality{\Field}} + \epsilon$.

\end{itemize}
Overall, we deduce that $\Verifier$ accepts with probability at most $\max \{\epsilon \,,\, \frac{\SCVars\SCDegree+1}{\SetCardinality{\Field}} + \epsilon \} \leq 3 \frac{\SCVars\SCDegree}{\SetCardinality{\Field}}$.
\end{proof}

\begin{remark}[necessity of IPCP]
One may be tempted to ``flatten'' the IPCP used to prove \thmref{thm:sumcheck-ipcp}, by sending a single PCP that already contains all possible transcripts, relative to all possible $\rho$'s. Such a modification does indeed preserve completeness and soundness. (In fact, even a small subset of $\rho$'s is enough for constant soundness error, because only one $\rho$ in $\Field$ is ``bad''.) However, this modification does \emph{not} preserve zero knowledge: if a verifier learns, say, the partial sums $\alpha_{1} \DefineEqual \rho_{1}\SCPoly(\vec{\gamma}) + \RandPoly(\vec{\gamma})$ and $\alpha_{2} \DefineEqual \rho_{2}\SCPoly(\vec{\gamma}) + \RandPoly(\vec{\gamma})$ for $\rho_{1} \neq \rho_{2}$ and some $\vec{\gamma} \in \Field^{\leq\SCVars}$ then he also learns $\SCPoly(\vec{\gamma}) = \frac{\alpha_{1} - \alpha_{2}}{\rho_{1} - \rho_{2}}$, violating zero knowledge. (That said, the modification \emph{does} preserve \emph{honest-verifier} zero knowledge.)
\end{remark}

\begin{remark}[space complexity of the prover]
\label{rem:prover-space}
The prover in a zero knowledge protocol is a probabilistic function, and hence reads bits from its randomness tape. In the case of the above protocol, the prover $\Prover$ must sample the evaluation of a random polynomial $\RandPoly$ in $\PolynomialRingIndOne{\Field}{\SCVars}{\VariableX}{\SCDegree}$; the entropy of $\RandPoly$ is exponential and thus requires reading an exponential number of random bits from the randomness tape. (Beyond this, $\Prover$ requires no other randomness.)

It is easy to see that $\Prover$ can run in exponential time and space. However, if the prover has \emph{two}-way access to its random tape, $\Prover$ can run in exponential time and polynomial space: the prover treats the random tape as the coefficients of $\RandPoly$, computing an evaluation at a given point by reading the tape coefficient-by-coefficient and summing the contributions of each monomial.

Two-way access to the randomness tape is a relaxation of the standard definition, which permits only \emph{one}-way access to it \cite{Nisan93}; that is, random bits must be stored on the work tape in order to be accessed again. It is not known whether the relaxation makes polynomial-space machines more powerful for decision problems (the class is equivalent to ``almost''-$\PSPACE$ \cite{BookVW98}), nor do we know how to obtain a polynomial-space prover with only one-way access. Nevertheless, we believe that polynomial space with two-way access to the random tape is still quite meaningful, e.g., it yields standard polynomial space relative to a random oracle.
\end{remark}

\doclearpage
%%%%%%%%%%%%%%%%%%%%%%%%%%%%%%%%%%%%%%%%%%%%%%%%%%%%%%%%%%%%%%%%%%%%%%%%%%%%%%%%
%%%%%%%%%%%%%%%%%%%%%%%%%%%%%%%%%%%%%%%%%%%%%%%%%%%%%%%%%%%%%%%%%%%%%%%%%%%%%%%%
%%%%%%%%%%%%%%%%%%%%%%%%%%%%%%%%%%%%%%%%%%%%%%%%%%%%%%%%%%%%%%%%%%%%%%%%%%%%%%%%
\section{Perfect zero knowledge for counting problems}
\label{sec:zk-sharpp}

We prove that $\sharpP$ has an IPCP that is perfect zero knowledge against unbounded queries. (Recall that $\sharpP$ corresponds to all counting problems associated to decision problems in $\NP$.) We do so by constructing a suitable protocol for the counting problem associated to $\mathrm{3SAT}$, which is $\sharpP$-complete.

\begin{definition}
\label{def:sharp-3SAT}
Let $\sharpSATLanguage$ be the language of pairs $(\Formula,\NumSats)$ where $\Formula$ is a $3$-CNF boolean formula and $\NumSats$ is the number of satisfying assignments of $\Formula$. We denote by $\NumVars$ the number of variables and by $\NumClauses$ the number of clauses in $\Formula$.
\end{definition}

We construct a public-coin IPCP system for $\sharpSATLanguage$ that is perfect zero knowledge against unbounded queries, and has exponential proof length and polynomial query complexity. As in the non-ZK IP counterpart, the number of rounds is $O(\NumVars)$, the prover runs in space $\poly(\NumClauses)$ (with a caveat, see \remref{rem:prover-space}), and the verifier in time $\poly(\NumClauses)$.

\begin{theorem}[formal statement of \ref{thm:intro-sharpp}]
\label{thm:zk-sharpp}
There exists a IPCP system $\pair{\Prover}{\Verifier}$ that puts $\sharpSATLanguage$ in the complexity class
\begin{equation*}
{\small
\PZKIPCP
\left[
\begin{array}{llll}
\TextNumRounds
  & \NumRounds
  &=& O(\NumVars)
  \\
\TextProofLength
  & \ProofLength
  &=& \exp(\NumVars)
  \\
\TextQueryComplexity
  & \QueryComplexity
  &=& \poly(\NumClauses)
  \\
\TextSoundnessError
  & \SoundnessError
  &=& 1/2
  \\
\TextProverSpace
  & \ProverSpace
  &=& \poly(\NumClauses)
  \\
\TextVerifierTime
  & \VerifierTime
  &=& \poly(\NumClauses)
  \\
\TextQueryBound
  & \QueryBound
  &=& \AnyBound
\end{array}
\right]
}
\enspace.
\end{equation*}
Moreover, the verifier $\Verifier$ is public-coin and non-adaptive.
\end{theorem}

\begin{proof}
Let $\pair{\ZKSCProver}{\ZKSCVerifier}$ be the PZK IPCP system for sumcheck from \thmref{thm:sumcheck-ipcp}, and let $\ZKSCSimulator$ be any simulator attests to its perfect zero knowledge. We construct an IPCP system $\pair{\Prover}{\Verifier}$ for $\sharpSATLanguage$ as follows.
\begin{itemize}

  \item The prover $\Prover$, given an instance $(\Formula, \NumSats)$, finds a prime $q \in (2^{\NumVars},2^{2\NumVars}]$, computes the arithmetization $\Arithmetize{\Formula} \in \PolynomialRingIndOne{\Field_{q}}{\NumVars}{\VariableX}{3\NumClauses}$ of $\Formula$ and simulates $\ZKSCProver^{\SCPoly}(\Field,\SCVars,\SCDegree,\SCSubset,\SCConstant)$ with
$\Field \DefineEqual \Field_{q}$,
$\SCVars \DefineEqual \NumVars$,
$\SCDegree \DefineEqual 3\NumClauses$,
$\SCSubset \DefineEqual \Bits$,
$\SCConstant \DefineEqual \NumSats$,
and
$\SCPoly(\VariableX_{1}, \dots, \VariableX_{\NumVars}) \DefineEqual \Arithmetize{\Formula}(\VariableX_{1}, \dots, \VariableX_{\NumVars})$. (The prover $\Prover$ also communicates the prime $q$ to the $\Verifier$ verifier.)

    \item The verifier $\Verifier$, given an instance $(\Formula, \NumSats)$, also computes $\Field,\SCVars,\SCDegree,\SCSubset,\SCConstant,\SCPoly$, and then simulates $\ZKSCVerifier^{\SCPoly}(\Field,\SCVars,\SCDegree,\SCSubset,\SCConstant)$, and accepts if and only if the simulation accepts.

\end{itemize}
Note that the arithmetization of a $3$-CNF formula $\Formula$ can be computed in time $\poly(\NumClauses)$, and the claimed given efficiency parameters follow from \thmref{thm:sumcheck-ipcp}. We now argue completeness, then soundness, and finally perfect zero knowledge.

%%%%%%%%%%%%%%%%%%%%%%%%%%%%%%%%%%%%%%%%
\parhead{Completeness}
Completeness follows from the completeness of $\pair{\ZKSCProver}{\ZKSCVerifier}$ and the fact that if $(\Formula, \NumSats) \in \sharpSATLanguage$ then $\big( (\Field,\SCVars,\SCDegree,\SCSubset,\SCConstant) , \SCPoly \big) = ((\Field_{q}, \NumVars, 3\NumClauses,\Bits^{\NumVars},\NumSats), \Arithmetize{\Formula}) \in \SCRelation^{\yes}$.

%%%%%%%%%%%%%%%%%%%%%%%%%%%%%%%%%%%%%%%%
\parhead{Soundness}
Soundness follows from the soundness of $\pair{\ZKSCProver}{\ZKSCVerifier}$ and the fact that if $(\Formula, \NumSats) \notin \sharpSATLanguage$ then $\big( (\Field,\SCVars,\SCDegree,\SCSubset,\SCConstant) , \SCPoly \big) = ((\Field_{q}, \NumVars, 3\NumClauses, \Bits^{\NumVars}, \NumSats),\Arithmetize{\Formula}) \in \SCRelation^{\no}$.

%%%%%%%%%%%%%%%%%%%%%%%%%%%%%%%%%%%%%%%%
\parhead{Perfect zero knowledge}
We construct a simulator $\Simulator$ that provides perfect zero knowledge. Given an instance $(\Formula, \NumSats)$ and straightline access to a verifier $\Malicious{\Verifier}$, the simulator $\Simulator$ computes $\Field,\SCVars,\SCDegree,\SCSubset,\SCConstant,\SCPoly$ as above and simulates $\ZKSCSimulator^{\Malicious{\Verifier}, \SCPoly}(\Field,\SCVars,\SCDegree,\SCSubset,\SCConstant)$. By the perfect zero knowledge property of $\pair{\ZKSCProver}{\ZKSCVerifier}$, the simulator's output is identically distributed to $\IPCPView{\ZKSCProver^{\SCPoly}(\Field,\SCVars,\SCDegree,\SCSubset,\SCConstant)}{\Malicious{\Verifier}^{\SCPoly}}$; but note that $\Malicious{\Verifier}$ does not query any oracles outside of its interaction with $\ZKSCProver$ so that $\Malicious{\Verifier}^{\SCPoly} = \Malicious{\Verifier}$. By the construction of $\Prover$ above, this view equals to $\IPCPView{\Prover(\Formula, \NumSats)}{\Malicious{\Verifier}}$, as desired.
\end{proof}

\doclearpage
%%%%%%%%%%%%%%%%%%%%%%%%%%%%%%%%%%%%%%%%%%%%%%%%%%%%%%%%%%%%%%%%%%%%%%%%%%%%%%%%
%%%%%%%%%%%%%%%%%%%%%%%%%%%%%%%%%%%%%%%%%%%%%%%%%%%%%%%%%%%%%%%%%%%%%%%%%%%%%%%%
%%%%%%%%%%%%%%%%%%%%%%%%%%%%%%%%%%%%%%%%%%%%%%%%%%%%%%%%%%%%%%%%%%%%%%%%%%%%%%%%
\section{Perfect zero knowledge from succinct constraint detection}
\label{sec:pzk-codes}

We show how to obtain perfect zero knowledge $2$-round IOPs of Proximity for \emph{any} linear code that has proximity proofs with succinct constraint detection (\secref{sec:pzk-general-transformation}). Afterwards, we instantiate this general transformation for the case of Reed--Solomon codes (\secref{sec:rs-pzk-iopp}), whose proximity proofs we discussed in \secref{sec:bsrs-succinct-constraint-detection}.

%%%%%%%%%%%%%%%%%%%%%%%%%%%%%%%%%%%%%%%%%%%%%%%%%%%%%%%%%%%%%%%%%%%%%%%%%%%%%%%%
%%%%%%%%%%%%%%%%%%%%%%%%%%%%%%%%%%%%%%%%%%%%%%%%%%%%%%%%%%%%%%%%%%%%%%%%%%%%%%%%
\subsection{A general transformation}
\label{sec:pzk-general-transformation}

%%%%%%%%%%%%%%%%%%%%%%%%%%%%%%%%%%%%%%%%
\parhead{PCPs of proximity for codes}
A \emph{PCP of Proximity} \cite{DinurR04,BenSassonGHSV06} for a code family $\CodeClass = \{\Code_{\CodeIdx}\}_{\CodeIdx}$ is a pair $(\OfCode{\PCPPProver}, \OfCode{\PCPPVerifier})$ where for every index $\CodeIdx$ and $\Codeword \in \Field^{\Domain(\CodeIdx)}$: if $\Codeword \in \Code_{\CodeIdx}$ then $\OfCode{\PCPPVerifier}^{\Codeword,\Proof}(\CodeIdx)$ accepts with probability $1$ with $\Proof \DefineEqual \PCPPProver(\CodeIdx, \Codeword)$; if $\Codeword$ is `far' from $\Code_{\CodeIdx}$ then $\Codeword \in \Code_{\CodeIdx}$ rejects with high probability regardless of $\Proof$. We do not formally define this notion because it is a special case of a $1$-round IOP of Proximity (see \secref{sec:restrictions-and-extensions}) where the verifier message is empty; we use $\PCPP[\ProofLength,\QueryComplexity,\SoundnessError,\ProximityParameter,\ProverTime,\VerifierTime]$ to denote the corresponding complexity class. Note that, since both $\Proof$ and $\Codeword$ are provided as oracles to $\PCPPVerifier$, the query complexity of $\PCPPVerifier$ is the \emph{total} number of queries across both oracles.

%%%%%%%%%%%%%%%%%%%%%%%%%%%%%%%%%%%%%%%%
\parhead{Leakage from proximity proofs}
While proximity proofs facilitate local testing, they are a liability for zero knowledge: in principle even a single query to $\Proof$ may `summarize' information that needs many queries to $\Codeword$ to simulate. (This holds for BS proximity proofs \cite{BS08}, for instance.) Our construction facilitates local testing while avoiding this leakage.

%%%%%%%%%%%%%%%%%%%%%%%%%%%%%%%%%%%%%%%%
\parhead{Perfect zero knowledge IOPs of Proximity}
The notion of zero knowledge for IOPs of Proximity that we target is defined in \secref{sec:zk-iopp} and is analogous to \cite{IshaiW14}'s notion for PCPs of Proximity (a special case of our setting). Informally, it requiress an algorithm that simulates the verifier's view by making as many queries to $\Codeword$ as the \emph{total} number of verifier queries to either $\Codeword$ or any oracles sent by the prover; intuitively, this means that any bit of any message oracle reveals no more information than one bit of $\Codeword$.

%%%%%%%%%%%%%%%%%%%%%%%%%%%%%%%%%%%%%%%%
\parhead{A generic `masking' construction}
Suppose that the linear code family $\CodeClass = \{\Code_{\CodeIdx}\}_{\CodeIdx}$ has a PCP of Proximity. Consider the $2$-round IOP of Proximity that uses masking via random self-reducibility (similarly to \cite{BenSassonCGV16}) as follows. The prover and verifier have input $\CodeIdx$ and oracle access to a codeword $\Codeword$, and the prover wants to convince the verifier that $\Codeword$ is close to $\Code_{\CodeIdx}$. Rather than sending a proximity proof for $\Codeword$, the prover samples a random codeword $\RandCodeword \in \Code_{\CodeIdx}$ and sends it to the verifier; the verifier replies with a random field element $\rho$; the prover sends a proximity proof for the new codeword $\rho \Codeword + \RandCodeword$. Completeness follows from linearity of $\Code_{\CodeIdx}$; soundness follows from the fact that if $\Codeword$ is far from $\Code_{\CodeIdx}$ then so is the word $\rho \Codeword + \RandCodeword$ for every $\RandCodeword$ with high probability over $\rho$.

Perfect zero knowledge intuitively follows from the observation that $\rho \Codeword + \RandCodeword$ is essentially a random codeword (up to malicious choice of $\rho$). We formally prove this for the case where \emph{the linear code family consisting of the concatenation of codewords in $\CodeClass$ with corresponding proximity proofs has succinct constraint detection}.

\medskip
\noindent
We are now ready to turn the above discussions into formal definitions and proofs. Throughout, given a code family $\CodeClass$, we denote by $\GetRelation{\CodeClass}$ the relation consisting of all pairs $(\CodeIdx,\Codeword)$ such that $\Codeword \in \Code_{\CodeIdx}$.

\begin{definition}
\label{def:pcpp-linear}
Let $\CodeClass = \{ \Code_{\CodeIdx} \}_{\CodeIdx}$ be a linear code family with domain $\EvaluationDomain(\cdot)$ and alphabet $\Field(\cdot)$, and let $(\Prover,\Verifier)$ be a PCPP system for $\GetRelation{\CodeClass}$. We say that $(\Prover,\Verifier)$ is \defemph{linear} if $\Prover$ is deterministic and is linear in its input codeword: for every index $\CodeIdx$ there exists a matrix $A_{\CodeIdx}$ with entries in $\Field(\CodeIdx)$ such that $\Prover(\CodeIdx,\Codeword)=A_{\CodeIdx} \cdot \Codeword$ for all $\Codeword \in \Code_{\CodeIdx}$ (equivalently, the set $\{\, \Codeword \| \Prover(\CodeIdx,\Codeword) \,\}_{\Codeword \in \Code_{\CodeIdx}}$ is a linear code).
\end{definition}

\begin{definition}
\label{def:pcpp-constraint-detection}
Let $\CodeClass = \{ \Code_{\CodeIdx} \}_{\CodeIdx}$ be a linear code family with domain $\EvaluationDomain(\cdot)$ and alphabet $\Field(\cdot)$, and let $(\Prover,\Verifier)$ be a PCPP system for $\GetRelation{\CodeClass}$. We say that $(\Prover,\Verifier)$ has \defemph{$\CodeTime(\cdot, \cdot)$-time constraint detection} if $(\Prover,\Verifier)$ is linear and, moreover, the linear code family $\Class{\OtherCode} = \{ \OtherCode_{\CodeIdx} \} $ has $\CodeTime(\cdot, \cdot)$-time constraint detection, where $\OtherCode_{\CodeIdx} \DefineEqual \{\, \Codeword \| \Prover(\CodeIdx,\Codeword) \,\}_{\Codeword \in \Code_{\CodeIdx}}$; we also say that $(\Prover,\Verifier)$ has \defemph{succinct constraint detection} if the same holds with $\CodeTime(\CodeIdx, \ListSize) = \poly(\CodeIdx + \ListSize)$.
\end{definition}

\begin{theorem}
\label{thm:pzk-iopp-for-codes}
Let $\CodeClass = \{ \Code_{\CodeIdx} \}_{\CodeIdx}$ be a linear code family with domain $\EvaluationDomain(\cdot)$, alphabet $\Field(\cdot)$, block length $\CodeBlockLength(\cdot)$, and a $\CodeSamplingTime(\cdot)$-time sampler. Suppose that there exists a PCPP system $(\OfCode{\Prover},\OfCode{\Verifier})$ for $\GetRelation{\CodeClass}$ that (is linear and) has succinct constraint detection and puts $\GetRelation{\CodeClass}$ in the complexity class
\begin{equation*}
{\small
\PCPP
\left[
\begin{array}{ll}
\TextAnswerAlphabet        & \Field(\CodeIdx)                         \\
\TextProofLength           & \OfCode{\ProofLength}(\CodeIdx)          \\
\TextQueryComplexity       & \OfCode{\QueryComplexity}(\CodeIdx)      \\
\TextSoundnessError        & \OfCode{\SoundnessError}(\CodeIdx)       \\
\TextProximityParameter    & \OfCode{\ProximityParameter}(\CodeIdx)   \\
\TextProverTime            & \OfCode{\ProverTime}(\CodeIdx)           \\
\TextVerifierTime          & \OfCode{\VerifierTime}(\CodeIdx)         \\
\TextVerifierRandomness    & \OfCode{\VerifierRandomness}(\CodeIdx)   \\
\end{array}
\right]
\enspace.
}
\end{equation*}
Then there exists an IOPP system $(\Prover,\Verifier)$ that puts $\GetRelation{\CodeClass}$ in the complexity class
\begin{equation*}
{\small
\PZKIOPP
\left[
\begin{array}{llll}
\TextAnswerAlphabet
  & \Field(\CodeIdx)
  &~&
  \\
\TextNumRounds
  & \NumRounds(\CodeIdx)
  &=& 2
  \\
\TextProofLength
  & \ProofLength(\CodeIdx)
  &=& \OfCode{\ProofLength}(\CodeIdx) + \CodeBlockLength(\CodeIdx)
  \\
\TextQueryComplexity
  & \QueryComplexity(\CodeIdx)
  &=& 2\OfCode{\QueryComplexity}(\CodeIdx)
  \\
\TextSoundnessError
  & \SoundnessError(\CodeIdx)
  &=& \OfCode{\SoundnessError}(\CodeIdx) + \frac{1}{\SetCardinality{\Field(\CodeIdx)}}
  \\
\TextProximityParameter
  & \ProximityParameter(\CodeIdx)
  &=& 2\OfCode{\ProximityParameter}(\CodeIdx)
  \\
\TextProverTime
  & \ProverTime(\CodeIdx)
  &=& \OfCode{\ProverTime}(\CodeIdx) + \CodeSamplingTime(\CodeIdx) + O(\CodeBlockLength(\CodeIdx))
  \\
%\TextProverRandomness
%  & ?
%  &=& \log\SetCardinality{\Code_{\CodeIdx}} % TODO: depends on sampler!
%  \\
\TextVerifierTime
  & \VerifierTime(\CodeIdx)
  &=& \OfCode{\VerifierTime}(\CodeIdx) + O(\OfCode{\QueryComplexity}(\CodeIdx))
  \\
\TextVerifierRandomness
  & \VerifierRandomness(\CodeIdx)
  &=& \OfCode{\VerifierRandomness}(\CodeIdx) + \log \SetCardinality{\Field(\CodeIdx)}
  \\
\TextQueryBound
  & \QueryBound(\CodeIdx)
  &=& \AnyBound
\end{array}
\right]
\enspace.
}
\end{equation*}
\end{theorem}

\begin{remark}[the case of LTCs]
\label{rem:pzk-for-ltc}
It is tempting to apply \thmref{thm:pzk-iopp-for-codes} to the notable special case where the proximity proof is \emph{empty} (e.g., when $\CodeClass$ is locally testable so no proximity proofs are needed). However, in this case, the zero knowledge guarantee of our construction does not buy anything compared to when the verifier queries only the codeword (indeed, the verifier \emph{already} learns precisely the value of codeword at those positions which it queries and nothing else).
\end{remark}

\begin{construction}
\label{con:pzk-iopp}
The IOPP system $(\Prover,\Verifier)$ is defined as follows. The prover receives a pair $(\CodeIdx,\Codeword)$ as input, while the verifier receives the index $\CodeIdx$ as input and $\Codeword$ as an oracle. The interaction proceeds as follows:
\begin{enumerate}[nolistsep]
  \item $\Prover$ samples a random codeword $\RandCodeword$ in $\Code_{\CodeIdx}$ and sends $\RandCodeword$ to $\Verifier$;
  \item $\Verifier$ samples a random element $\rho$ in $\Field(\CodeIdx)$ and sends $\rho$ to $\Prover$;
  \item $\Prover$ computes the proof $\Proof \DefineEqual \OfCode{\Prover}(\CodeIdx,\rho \Codeword + \RandCodeword)$ and sends $\Proof$ to $\Verifier$;
  \item $\Verifier$ checks that $\OfCode{\Verifier}^{\MaskedCodeword,\Proof}(\CodeIdx)$ accepts, where $\MaskedCodeword \DefineEqual \rho \Codeword + \RandCodeword$ (any query $\alpha$ to $\MaskedCodeword$ is computed as $\rho \Codeword(\alpha) + \RandCodeword(\alpha)$).
\end{enumerate}
\end{construction}

The round complexity, proof length, query complexity, and prover and verifier complexities claimed in \thmref{thm:pzk-iopp-for-codes} follow in a straightforward way from \conref{con:pzk-iopp}. We now argue completeness and soundness (\clmref{clm:pzk-iopp-1}) and perfect zero knowledge (\clmref{clm:pzk-iopp-2}).

\begin{claim}
\label{clm:pzk-iopp-1}
The IOPP system $(\Prover,\Verifier)$ has completeness $1$ and soundness error $\SoundnessError(\CodeIdx) = \OfCode{\SoundnessError}(\CodeIdx) + \frac{1}{\SetCardinality{\Field(\CodeIdx)}}$.
\end{claim}

\begin{proof}
First we argue completeness. Suppose that the instance-witness pair $(\CodeIdx,\Codeword)$ is in the relation $\GetRelation{\CodeClass}$, i.e., that the word $\Codeword$ is in the code $\Code_{\CodeIdx}$. Then, the linearity of $\Code_{\CodeIdx}$ implies that, for every word $\OtherCodeword$ in $\Code_{\CodeIdx}$ and element $\rho$ in $\Field(\CodeIdx)$, the word $\rho \Codeword + \OtherCodeword$ is also in $\Code_{\CodeIdx}$. Thus completeness of $(\Prover,\Verifier)$ follows from the completeness of $(\OfCode{\Prover},\OfCode{\Verifier})$.

Next we argue soundness. Suppose that $\Codeword$ is $2\OfCode{\ProximityParameter}(\CodeIdx)$-far from $\Code_{\CodeIdx}$. For every word $\OtherCodeword$ in $\Field(\CodeIdx)^{\CodeBlockLength(\CodeIdx)}$ (not necessarily in $\Code_{\CodeIdx}$), there exists at most one $\rho$ in $\Field(\CodeIdx)$ such that $\rho \Codeword + \OtherCodeword$ is $\OfCode{\ProximityParameter}(\CodeIdx)$-close to $\Code_{\CodeIdx}$ (see \clmref{claim:distance-preservation}). The soundness of $(\OfCode{\Prover},\OfCode{\Verifier})$ implies that $\OfCode{\Verifier}$ accepts with probability at most $\OfCode{\SoundnessError}(\CodeIdx)$ if $\OfCode{\Verifier}$ is invoked on a word that is $\OfCode{\ProximityParameter}(\CodeIdx)$-far from $\Code_{\CodeIdx}$. Thus, since $\Verifier$ invokes $\OfCode{\Verifier}$ on the word $\rho \Codeword + \RandCodeword$, the probability that $\Verifier$ accepts is at most $\OfCode{\SoundnessError}(\CodeIdx) + \frac{1}{\SetCardinality{\Field(\CodeIdx)}}$.
\end{proof}

\begin{claim}
\label{clm:pzk-iopp-2}
There exists a straightline simulator $\Simulator$ such that, for every $(\CodeIdx,\Codeword) \in \GetRelation{\CodeClass}$ and malicious verifier $\Malicious{\Verifier}$, the following two random variables are identically distributed
\begin{equation*}
\Big(\Simulator^{\Malicious{\Verifier},\Codeword}(\CodeIdx) \;,\; q_{\Simulator} \Big)
\quad\text{and}\quad
\Big(\IOPView{\Prover^{\Codeword}(\CodeIdx)}{\Malicious{\Verifier}^{\Codeword}} \;,\; q_{\Malicious{\Verifier}} \Big)
\enspace,
\end{equation*}
where $q_{\Simulator}$ is the number of queries to $\Codeword$ made by $\Simulator$ and $q_{\Malicious{\Verifier}}$ is the number of queries to $\Codeword$ or to prover messages made by $\Malicious{\Verifier}$. Moreover, $\Simulator$ runs in time $\poly(\CodeIdx + \QueryComplexity_{\Malicious{\Verifier}})$, where $\QueryComplexity_{\Malicious{\Verifier}}$ is $\Malicious{\Verifier}$'s query complexity.
\end{claim}

\begin{proof}
We begin by proving perfect zero knowledge via a straightline simulator $\SlowSimulator$ whose number of queries to $\Codeword$ equals $q_{\Malicious{\Verifier}}$, but runs in time $\poly(\ProverTime(\CodeIdx)+q_{\Malicious{\IOPVerifier}})$. After that, we explain how to modify $\SlowSimulator$ into another simulator $\Simulator$, with an identical output distribution, that runs in the faster time claimed in the lemma.

\begin{mdframed}
{\small
The simulator $\SlowSimulator$, given straightline access to $\Malicious{\IOPVerifier}$ and oracle access to $\Codeword$, works as follows:
\medskip
\begin{enumerate}[nolistsep]

  \item \label{step:iopp-sim-draw-R}
  Draw a uniformly random $\Simulated{\RandCodeword} \in \Code_{\CodeIdx}$.

  \item \label{step:iopp-sim-before-rho}
  Whenever $\Malicious{\IOPVerifier}$ queries $\Codeword$ at $\alpha \in \EvaluationDomain(\CodeIdx)$, return $\Codeword(\alpha)$; whenever $\Malicious{\IOPVerifier}$ queries $\RandCodeword$ at $\alpha \in \EvaluationDomain(\CodeIdx)$, return $\Simulated{\RandCodeword}(\alpha)$.

  \item \label{step:iopp-sim-draw-Q}
  Receive $\Malicious{\rho}$ from $\Malicious{\IOPVerifier}$, and draw a uniformly random $\Simulated{\MaskedCodeword} \| \Simulated{\Proof} \in \OtherCode_{\CodeIdx}$ conditioned on $\Simulated{\MaskedCodeword}(\alpha) = \Malicious{\rho} \Codeword(\alpha) + \Simulated{\RandCodeword}(\alpha)$ for every coordinate $\alpha \in \EvaluationDomain(\CodeIdx)$ queried in \stepref{step:iopp-sim-before-rho}. (This latter condition requires querying $\Codeword$ at $\alpha$ for every coordinate $\alpha \in \EvaluationDomain(\CodeIdx)$ queried to $\Simulated{\RandCodeword}$ in \stepref{step:iopp-sim-before-rho}.)

  \item \label{step:iopp-sim-after-rho}
  Hereafter: whenever $\Malicious{\IOPVerifier}$ queries $\Codeword$ at $\alpha \in \EvaluationDomain(\CodeIdx)$, return $\Codeword(\alpha)$; whenever $\Malicious{\IOPVerifier}$ queries $\RandCodeword$ at $\alpha \in \EvaluationDomain(\CodeIdx)$, return $\Simulated{\MaskedCodeword}(\alpha) - \Malicious{\rho} \Codeword(\alpha)$. (In either case, a query to $\Codeword$ is required.)

  \item \label{step:iopp-sim-proof}
  Tell $\Malicious{\IOPVerifier}$ that the oracle $\Proof$ has been `sent'; whenever $\Malicious{\IOPVerifier}$ queries the $i$-th entry of $\Proof$, return the $i$-th entry of $\Simulated{\Proof}$. (Note that $\Malicious{\IOPVerifier}$ may query $\Codeword$ or $\RandCodeword$ before or after learning about $\Proof$.)

  \item Output the view of the simulated $\Malicious{\IOPVerifier}$.

\end{enumerate}
}
\end{mdframed}
Note that $\SlowSimulator$ runs in time $\poly(\ProverTime(\CodeIdx)+q_{\Malicious{\IOPVerifier}})$. Also, $\SlowSimulator$ makes one query to $\Codeword$ for every query to $\Codeword$ or $\RandCodeword$ by $\Malicious{\IOPVerifier}$ (at least provided that $\Malicious{\IOPVerifier}$'s queries have no duplicates, which we can assume without loss of generality), and zero queries to $\Codeword$ for every query to $\Proof$ by $\Malicious{\IOPVerifier}$. Thus, overall, the number of queries to $\Codeword$ by $\SlowSimulator$ is at most $q_{\Malicious{\IOPVerifier}}$; clearly, this number of queries can be padded to be equal to $q_{\Malicious{\IOPVerifier}}$. We now argue that $\SlowSimulator$'s output is identically distributed to $\Malicious{\IOPVerifier}$'s view when interacting with the honest prover $\IOPProver$, for $\RandCodeword$ random in $\Code_{\CodeIdx}$.
\begin{adjustwidth}{1cm}{1cm}
\begin{uclaim}
$\SlowSimulator^{\Malicious{\Verifier},\Codeword} \equiv \IPCPView{\IOPProver^{\Codeword}(\CodeIdx)}{\Malicious{\Verifier}^{\Codeword}}$.
\end{uclaim}
\begin{proof}
Define the random variable $\MaskedCodeword \DefineEqual \Malicious{\rho} \Codeword + \RandCodeword$, where $\Malicious{\rho}$ is chosen by $\Malicious{\IOPVerifier}$. Observe that there exists a (deterministic) function $\FView{\cdot}$ such that
\begin{equation*}
\IPCPView{\IOPProver^{\Codeword,\RandCodeword}}{\Malicious{\Verifier}^{\Codeword,\RandCodeword}} = \FView{\MaskedCodeword,\Codeword,r}
\quad\text{and}\quad
\SlowSimulator^{\Malicious{\Verifier},\Codeword} = \FView{\Simulated{\MaskedCodeword},\Codeword,\Randomness} \enspace,
\end{equation*}
where the random variable $\Randomness$ is $\Malicious{\IOPVerifier}$'s private randomness. Indeed,
\begin{inparaenum}[(i)]
  \item the messages sent and received by $\Malicious{\IOPVerifier}$ are identical to those when interacting with $\IOPProver$ on $\MaskedCodeword$ and $\Simulated{\MaskedCodeword}$, respectively;
  \item $\Malicious{\IOPVerifier}$'s queries to $\Codeword$ are answered honestly;
  \item $\Malicious{\IOPVerifier}$'s queries to $\RandCodeword$ are answered by $\RandCodeword = \MaskedCodeword - \Malicious{\rho} \Codeword$ and $\Simulated{\RandCodeword} = \Simulated{\MaskedCodeword} - \Malicious{\rho} \Codeword$ respectively;
  \item $\Malicious{\IOPVerifier}$'s queries to $\Proof$ are answered by $\OfCode{\Prover}(\CodeIdx,\MaskedCodeword)$ and $\OfCode{\Prover}(\CodeIdx,\Simulated{\MaskedCodeword})$ respectively.
\end{inparaenum}
We are only left to argue that, for any choice of $\Randomness$, $\MaskedCodeword$ and $\Simulated{\MaskedCodeword}$ are identically distributed:
\begin{itemize}[label=--]

  \item $\MaskedCodeword = \Malicious{\rho}\Codeword + \RandCodeword$ is uniformly random in $\Code_{\CodeIdx}$, because $\RandCodeword$ is uniformly random in $\Code_{\CodeIdx}$, $\Codeword$ is in $\Code_{\CodeIdx}$, and $\Code_{\CodeIdx}$ is linear; and

  \item $\Simulated{\MaskedCodeword}$ is uniformly random in $\Code_{\CodeIdx}$, because $\Simulated{\MaskedCodeword}$ is sampled at random in $\Code_{\CodeIdx}$ conditioned on $\Simulated{\MaskedCodeword}(\alpha_{i}) = \Simulated{\RandCodeword}(\alpha_{i}) + \Malicious{\rho} \Codeword(\alpha_{i})$ for some (adversarial) choice of $\alpha_{1},\dots,\alpha_{k}$. But $\Simulated{\RandCodeword}$ is uniformly random in $\Code_{\CodeIdx}$, so the latter condition says that $\Simulated{\MaskedCodeword}$ matches a random codeword on the set of points $\{\alpha_{1},\dots,\alpha_{k}\}$, giving the claimed distribution for $\Simulated{\MaskedCodeword}$. \qedhere

\end{itemize}
\end{proof}
\end{adjustwidth}

We explain how to modify $\SlowSimulator$ so as to reduce the running time to $\poly(\CodeIdx + \QueryComplexity_{\Malicious{\Verifier}})$. The inefficiency of $\SlowSimulator$ comes from sampling $\Simulated{\RandCodeword} \in \Code_{\CodeIdx}$ and $\Simulated{\MaskedCodeword} \| \Simulated{\Proof} \in \OtherCode_{\CodeIdx}$, which takes time $\poly(\CodeSamplingTime(\CodeIdx) + \CodeBlockLength(\CodeIdx))$ and $\poly(\OfCode{\ProverTime}(\CodeIdx))$ respectively, which need not be polynomial in $\CodeIdx$. In the following, we show how to not explicitly sample these codewords but, instead, adaptively sample them by relying on constraint detection.

First, note that if $\Class{\OtherCode}$ has constraint detection with a certain efficiency then so does $\CodeClass$ with the same efficiency. The theorem's hypothesis says that $\Class{\OtherCode}$ has succinct constraint detection; so both $\CodeClass$ and $\Class{\OtherCode}$ have succinct constraint detection. We then invoke \lemref{lem:efficient-codeword-simulator} to obtain probabilistic polynomial-time algorithms $\CodeSimAlgorithm_{\CodeClass}, \CodeSimAlgorithm_{\Class{\OtherCode}}$ for the code families $\CodeClass, \Class{\OtherCode}$ respectively. Using these algorithms, we write the more efficient simulator $\Simulator$, as follows.
\begin{mdframed}
{\small
Let $\EvaluationDomain^*(\CodeIdx)$ be the domain of $\OtherCode_{\CodeIdx}$: this is the disjoint union of $\EvaluationDomain(\CodeIdx)$ (the domain of $\Code_{\CodeIdx}$) and $[\OfCode{\ProofLength}(\CodeIdx)]$ (the domain of $\OfCode{\Prover}(\CodeIdx,\Code_{\CodeIdx})$). The simulator $\Simulator$, given straightline access to $\Malicious{\IOPVerifier}$ and oracle access to $\Codeword$, works as follows:
\medskip

\begin{enumerate}[nolistsep]

  \item Let $\AnsTable{\Simulated{\RandCodeword}}$ be a subset of $\EvaluationDomain(\CodeIdx) \times \Field(\CodeIdx)$ that records query-value pairs for $\Simulated{\RandCodeword}$; initially, $\AnsTable{\Simulated{\RandCodeword}}$ equals $\varnothing$.

  \item \label{step:fast-iopp-sim-before-rho}
  Whenever $\Malicious{\IOPVerifier}$ queries $\Codeword$ at $\alpha \in \EvaluationDomain(\CodeIdx)$, return $\Codeword(\alpha)$; whenever $\Malicious{\IOPVerifier}$ queries $\RandCodeword$ at $\alpha \in \EvaluationDomain(\CodeIdx)$, return $\beta \DefineEqual \CodeSimAlgorithm_{\CodeClass}(\CodeIdx, \AnsTable{\Simulated{\RandCodeword}}, \alpha)$. In the latter case, add $(\alpha, \beta)$ to $\AnsTable{\Simulated{\RandCodeword}}$.

  \item \label{step:fast-iopp-sim-draw-Q}
  Receive $\Malicious{\rho}$ from $\Malicious{\IOPVerifier}$, and compute $\AnsTable{} \DefineEqual \{ (\alpha, \beta + \Malicious{\rho} \Codeword(\alpha)) \}_{ (\alpha, \beta) \in \AnsTable{\Simulated{\RandCodeword}} }$; this subset of $\EvaluationDomain^{*}(\CodeIdx) \times \Field(\CodeIdx)$ records query-value pairs for $\Simulated{\MaskedCodeword} \| \Simulated{\Proof}$.

  \item \label{step:fast-iopp-sim-after-rho}
  Hereafter: whenever $\Malicious{\IOPVerifier}$ queries $\Codeword$ at $\alpha \in \EvaluationDomain(\CodeIdx)$, return $\Codeword(\alpha)$; whenever $\Malicious{\IOPVerifier}$ queries $\RandCodeword$ at $\alpha \in \EvaluationDomain(\CodeIdx)$, return $\beta' \DefineEqual \beta - \Malicious{\rho} \Codeword(\alpha)$ where $\beta \DefineEqual \CodeSimAlgorithm_{\Class{\OtherCode}}(\CodeIdx, \AnsTable{}, \alpha)$ and add $(\alpha, \beta)$ to $\AnsTable{}$.

  \item  \label{step:fast-iopp-sim-proof}
  Tell $\Malicious{\IOPVerifier}$ that the oracle $\Proof$ has been `sent'; note that we do not yet commit to any entries in the proof, save for those which are implied by the verifier's previous queries to $\Codeword$. Whenever $\Malicious{\IOPVerifier}$ queries the $i$-th location in $\Proof$, return $\Proof_{i} \DefineEqual \CodeSimAlgorithm_{\Class{\OtherCode}}(\CodeIdx, \AnsTable{}, i)$ and add $(i, \Proof_{i})$ to $\AnsTable{}$. (Note that $\Malicious{\IOPVerifier}$ may query $\Codeword$ or $\RandCodeword$ before or after learning about $\Proof$.)

  \item Output the view of the simulated $\Malicious{\IOPVerifier}$.

\end{enumerate}
}
\end{mdframed}
Note that $\Simulator$ makes the same number of queries to $\Codeword$ as $\SlowSimulator$ does. Also, the number of pairs in $\AnsTable{\Simulated{\RandCodeword}}$ is at most $q_{\Malicious{\Verifier}}$; ditto for $\AnsTable{}$. Since the algorithm $\CodeSimAlgorithm$ is called at most $q_{\Malicious{\Verifier}}$ times, the running time of $\Simulator$ is $\poly(\CodeIdx + \QueryComplexity_{\Malicious{\Verifier}})$, as required.
\end{proof}

We conclude with a lemma that says that succinct constraint detection is in some sense inherent to the ``masking'' approach used in \conref{con:pzk-iopp}.

\begin{lemma}
\label{lem:converse}
If $(\OfCode{\Prover},\OfCode{\Verifier})$ is a linear PCPP such that \conref{con:pzk-iopp} yields $(\Prover, \Verifier)$ with perfect zero knowledge, then $(\OfCode{\Prover},\OfCode{\Verifier})$ has a constraint detector that runs in probabilistic polynomial time. (In fact, the same statement holds even if the construction yields $(\Prover, \Verifier)$ with only statistical zero knowledge.)
\end{lemma}

\begin{proof}
The linear code $\OtherCode_{\CodeIdx} \DefineEqual \{\, \Codeword \| \Prover(\CodeIdx,\Codeword) \,\}_{\Codeword \in \Code_{\CodeIdx}}$ has domain $\EvaluationDomain^*(\CodeIdx) \DefineEqual \EvaluationDomain(\CodeIdx) \sqcup [\OfCode{\ProofLength}(\CodeIdx)]$, which is the disjoint union of $\EvaluationDomain(\CodeIdx)$ (the domain of $\Code_{\CodeIdx}$) and $[\OfCode{\ProofLength}(\CodeIdx)]$ (the domain of $\Prover(\CodeIdx,\Code_{\CodeIdx})$); see \defref{def:pcpp-constraint-detection}. Let $\IdxSet \subseteq \EvaluationDomain^*(\CodeIdx)$. We need, in probabilistic polynomial time, to output a basis for $\Dual{(\Restrict{\OtherCode_{\CodeIdx}}{\IdxSet})}$. Construct a malicious verifier $\Malicious{\Verifier}$ that works as follows:
\begin{enumerate}[nolistsep,leftmargin=20pt]

\item Receive $\RandCodeword \in \Code_{\CodeIdx}$ from $\Prover$.

\item Send $\rho = 0$ to $\Prover$.

\item Receive $\Proof$ from $\Prover$.

\item For each $i \in \IdxSet$: if $i \in \EvaluationDomain(\CodeIdx)$ then query $\RandCodeword$ at $i$, and if $i \in [\OfCode{\ProofLength}(\CodeIdx)]$ then query $\Proof$ at $i$; call the answer $\beta_{i}$.

\end{enumerate}
By PZK, there is a probabilistic polynomial time algorithm $\Simulator$ such that the output of $\Simulator^{\Malicious{\Verifier},\Codeword}(\CodeIdx)$ is identically distributed to $\IOPView{\Prover^{\Codeword}(\CodeIdx)}{\Malicious{\Verifier}^{\Codeword}}$, for every $\Codeword \in \Code_{\CodeIdx}$. Set $\Codeword$ to be the zero codeword, and suppose we run $\Simulator^{\Malicious{\Verifier},\Codeword}(\CodeIdx)$; this invocation makes $\Simulator$ sample answers $(\beta_{i})_{i \in \IdxSet} = (\RandCodeword'(i))_{i \in \IdxSet}$ for $\RandCodeword' = \RandCodeword \| \OfCode{\Prover}(\CodeIdx, \RandCodeword)$ uniformly random in $\OtherCode_{\CodeIdx}$.

Thus, to perform constraint detection in probabilistic polynomial time, we proceed as follows. We run $\Simulator^{\Malicious{\Verifier},\Codeword}(\CodeIdx)$ $k > \SetCardinality{\IdxSet}$ times, recording a vector $\vec{\beta}^{j} = (\beta_{i})_{i \in \IdxSet}$ at the $j$-th iteration. Let $B$ be the $k \times \SetCardinality{\IdxSet}$ matrix with rows $\vec{\beta}^{1}, \dots, \vec{\beta}^{k}$. Output a basis for the nullspace of $B$, which we can find in $\poly(\log \SetCardinality{\Field} + k + \SetCardinality{\IdxSet})$ time.

We now argue correctness of the above approach. First, for every $u \in \Field^{\IdxSet}$ such that $\sum_{i \in \IdxSet} u(i) \Codeword'(i) = 0$ for every $\Codeword' \in \OtherCode_{\CodeIdx}$, it holds that $u$ is in the nullspace of $B$, because codewords used to generate $B$ satisfy the same relation. Next, the probability that there exists $u \in \Field^{\IdxSet}$ in the nullspace of $B$ such that $\sum_{i \in \IdxSet} u(i) \Codeword'(i) \neq 0$ for some $\Codeword' \in \OtherCode_{\CodeIdx}$ is at most $1/\SetCardinality{\Field}^{k - \SetCardinality{\IdxSet}}$. Indeed, for every such $u$, $\Pr_{\RandCodeword' \gets \OtherCode_{\CodeIdx}}[\sum_{i \in \IdxSet} u(i) \RandCodeword'(i) = 0] \leq 1/\SetCardinality{\Field}$ (since $\OtherCode_{\CodeIdx}$ is a linear code), so the probability that $u$ is in the nullspace of $B$ is at most $1/\SetCardinality{\Field}^{k}$; we then obtain the claimed probability by a union bound. Overall, the probability that the algorithm answers incorrectly is at most $1/\SetCardinality{\Field}^{k - \SetCardinality{\IdxSet}}$.
\end{proof}

%%%%%%%%%%%%%%%%%%%%%%%%%%%%%%%%%%%%%%%%%%%%%%%%%%%%%%%%%%%%%%%%%%%%%%%%%%%%%%%%
%%%%%%%%%%%%%%%%%%%%%%%%%%%%%%%%%%%%%%%%%%%%%%%%%%%%%%%%%%%%%%%%%%%%%%%%%%%%%%%%
\subsection{Perfect zero knowledge IOPs of proximity for Reed--Solomon codes}
\label{sec:rs-pzk-iopp}

We have already proved that the linear code family $\BSCode$, which consists of low-degree univariate polynomials concatenated with corresponding BS proximity proofs \cite{BS08}, has succinct constraint detection. When combined with the results in \secref{sec:pzk-general-transformation}, we obtain IOPs of Proximity for Reed--Solomon codes, as stated in the corollary below.

\begin{definition}
\label{def:rs-from-bs}
We denote by $\ARS$ the linear code family indexed by tuples $\CodeIdx = (\Field,\BSSpace,\RSDegree)$, where $\Field$ is an extension field of $\Field_{2}$ and $\BSSpace$ is an $\Field_{2}$-linear subspace of $\Field$ with $\RSDegree \leq \SetCardinality{\BSSpace}/8$, and the $\CodeIdx$-th code consists of the codewords from the Reed--Solomon code $\RSCode{\Field}{\BSSpace}{\RSDegree}$.
\end{definition}

\begin{theorem}[\cite{BS08}]
\label{thm:bs08-parameters}
For every function $\ProximityParameter \colon \Bits^{*} \to (0,1)$, the linear code family $\ARS$ has PCPs of Proximity with soundness error $1/2$, proximity parameter $\ProximityParameter$, prover running time (and thus proof length) that is quasilinear in the block length $\CodeBlockLength(\CodeIdx)$, and verifier running time (and thus query complexity) that is polylogarithmic in $\CodeBlockLength(\CodeIdx)/\ProximityParameter(\CodeIdx)$.
\end{theorem}

\begin{corollary}
\label{cor:pzk-just-rs}
For every function $\ProximityParameter \colon \Bits^{*} \to (0,1)$, there exists an IOPP system that puts $\GetRelation{\ARS}$ in the complexity class
\begin{equation*}
{\footnotesize
\PZKIOPP
\left[
\begin{array}{llll}
\TextAnswerAlphabet
  & \Field(\CodeIdx)
  &~&
  \\
\TextNumRounds
  & \NumRounds(\CodeIdx)
  &=& 2
  \\
\TextProofLength
  & \ProofLength(\CodeIdx)
  &=& \Ot{\CodeBlockLength(\CodeIdx)}
  \\
\TextQueryComplexity
  & \QueryComplexity(\CodeIdx)
  &=& \polylog(\CodeBlockLength(\CodeIdx)/\ProximityParameter(\CodeIdx))
  \\
\TextSoundnessError
  & \SoundnessError(\CodeIdx)
  &=& 1/2
  \\
\TextProximityParameter
  & \ProximityParameter(\CodeIdx)
  &~&
  \\
\TextProverTime
  & \ProverTime(\CodeIdx)
  &=& \Ot{\CodeBlockLength(\CodeIdx)}
  \\
\TextVerifierTime
  & \VerifierTime(\CodeIdx)
  &=& \polylog(\CodeBlockLength(\CodeIdx)/\ProximityParameter(\CodeIdx))
  \\
\TextQueryBound
  & \QueryBound(\CodeIdx)
  &=& \AnyBound
\end{array}
\right]
\enspace.
}
\end{equation*}
\end{corollary}

\begin{proof}
Invoke \thmref{thm:pzk-iopp-for-codes} on the linear code family $\ARS$ with corresponding BS proximity proofs (\thmref{thm:bs08-parameters}). Indeed, the concatenation of codewords in $\ARS$ and proximity proofs yields the family $\BSCode$, which has succinct constraint detection by \thmref{thm:bsrs-succinct-constraint-detection}. (This last step omits a technical, but uninteresting, step: the proximity proofs from \thmref{thm:bsrs-succinct-constraint-detection} consider the case where the degree $\RSDegree$ \emph{equals} the special value $\SetCardinality{\BSSpace}/8$, rather than being bounded by it; but proximity proofs for smaller degree $\RSDegree$ are easily obtained from these, as explained in \cite{BS08}.)
\end{proof}

When constructing perfect zero knowledge IOPs for $\NEXP$ (\secref{sec:pzk-ntime}) we shall need perfect zero knowledge IOPs of Proximity not quite for the family $\ARS$ but for an extension of it that we denote by $\RSVRS$, and for which \cite{BS08} also gives PCPs of proximity. The analogous perfect zero knowledge result follows in a similar way, as explained below.

\begin{definition}
\label{def:rsvrs-from-bs}
Given a field $\Field$ of characteristic 2, $\GF{2}$-linear subspaces $\RSVRSVanishing, \RSVRSSubdomain \subseteq \Field$ with $\SetCardinality{\RSVRSVanishing} \leq \SetCardinality{\RSVRSSubdomain}/8$, and $\RSDegree_{0},\RSDegree_{1} \in \Naturals$ with $\RSDegree_{0},\RSDegree_{1} \leq \SetCardinality{\RSVRSSubdomain}/8$, we denote by $\RSVRSCode{\Field}{\RSVRSSubdomain}{\RSVRSVanishing}{\RSDegree_{0}}{\RSDegree_{1}}$ the linear code consisting of all pairs $(\Codeword_{0}, \Codeword_{1})$ where $\Codeword_{0} \in \RSCode{\Field}{\RSVRSSubdomain}{\RSDegree_{0}}$, $\Codeword_{1} \in \RSCode{\Field}{\RSVRSSubdomain}{\RSDegree_{1}}$, and $\Codeword_{1}(x) = 0$ for all $x \in \RSVRSVanishing$. We denote by $\RSVRS$ the linear code family indexed by tuples $\CodeIdx = (\Field,\BSSpace,\RSDegree_{0},\RSDegree_{1})$ for which the $\CodeIdx$-th code is $\RSVRSCode{\Field}{\RSVRSSubdomain}{\RSVRSVanishing}{\RSDegree_{0}}{\RSDegree_{1}}$. 
\end{definition}

\begin{theorem}[\cite{BS08}]
\label{thm:bs08-vsrs}
For every function $\ProximityParameter \colon \Bits^{*} \to (0,1)$, the linear code family $\RSVRS$ has PCPs of Proximity with soundness error $1/2$, proximity parameter $\ProximityParameter$, prover running time (and thus proof length) that is quasilinear in the block length $\CodeBlockLength(\CodeIdx)$, and verifier running time (and thus query complexity) that is polylogarithmic in $\CodeBlockLength(\CodeIdx)/\ProximityParameter(\CodeIdx)$.
\end{theorem}

\begin{proof}[Proof sketch]
A PCP of proximity for a codeword $(\Codeword_{0}, \Codeword_{1})$ to $\RSVRSCode{\Field}{\RSVRSSubdomain}{\RSVRSVanishing}{\RSDegree_{0}}{\RSDegree_{1}}$ consists of $(\pi_{0},\Codeword_{1}',\pi_{1})$, where
\begin{itemize}[nolistsep]
  \item $\pi_{0}$ is a PCP of proximity for $\Codeword_{0}$ to $\RSCode{\Field}{\RSVRSSubdomain}{\RSDegree_{0}}$;
  \item $\Codeword_{1}'$ is the evaluation of the polynomial obtained by dividing (the polynomial of) $\Codeword_{1}$ by the zero polynomial of $\RSVRSVanishing$;
  \item $\pi_{1}$ is a PCP of proximity for $\Codeword_{1}'$ to $\RSCode{\Field}{\RSVRSSubdomain}{\RSDegree_{1} - \SetCardinality{\RSVRSVanishing}}$.
\end{itemize}
The verifier, which has oracle access to $(\Codeword_{0}, \Codeword_{1})$ and $(\pi_{0},\Codeword_{1}',\pi_{1})$, checks both PCPs or proximity and then performs a consistency check between $\Codeword_{1}$ and $\Codeword_{1}'$. See \cite{BS08} for details.
\end{proof}

\begin{corollary}
\label{cor:rsvrs-pzk-iopp}
For every function $\ProximityParameter \colon \Bits^{*} \to (0,1)$, there exists an IOPP system that puts $\GetRelation{\RSVRS}$ in the complexity class
\begin{equation*}
{\footnotesize
\PZKIOPP
\left[
\begin{array}{llll}
\TextAnswerAlphabet
  & \Field(\CodeIdx)
  &~&
  \\
\TextNumRounds
  & \NumRounds(\CodeIdx)
  &=& 2
  \\
\TextProofLength
  & \ProofLength(\CodeIdx)
  &=& \Ot{\CodeBlockLength(\CodeIdx)}
  \\
\TextQueryComplexity
  & \QueryComplexity(\CodeIdx)
  &=& \polylog(\CodeBlockLength(\CodeIdx)/\ProximityParameter(\CodeIdx))
  \\
\TextSoundnessError
  & \SoundnessError(\CodeIdx)
  &=& 1/2
  \\
\TextProximityParameter
  & \ProximityParameter(\CodeIdx)
  &~&
  \\
\TextProverTime
  & \ProverTime(\CodeIdx)
  &=& \Ot{\CodeBlockLength(\CodeIdx)}
  \\
\TextVerifierTime
  & \VerifierTime(\CodeIdx)
  &=& \polylog(\CodeBlockLength(\CodeIdx)/\ProximityParameter(\CodeIdx))
  \\
\TextQueryBound
  & \QueryBound(\CodeIdx)
  &=& \AnyBound
\end{array}
\right]
\enspace.
}
\end{equation*}
\end{corollary}

\begin{proof}
Invoke \thmref{thm:pzk-iopp-for-codes} on the linear code family $\RSVRS$ with corresponding BS proximity proofs (\thmref{thm:bs08-vsrs}), which has succinct constraint detection as we now clarify. A codeword $(\Codeword_{0}, \Codeword_{1})$ has proximity proof $(\pi_{0},\Codeword_{1}',\pi_{1})$, and \thmref{thm:bsrs-succinct-constraint-detection} implies that $(\Codeword_{0},\pi_{0})$ and $(\Codeword_{1}',\pi_{1})$ have succinct constraint detection. But every coordinate of $\Codeword_{1}'$ is easy to compute from the same coordinate in $\Codeword_{1}$, and concatenating codewords preserves succinct constraint detection.
\end{proof}

\doclearpage
%%%%%%%%%%%%%%%%%%%%%%%%%%%%%%%%%%%%%%%%%%%%%%%%%%%%%%%%%%%%%%%%%%%%%%%%%%%%%%%%
%%%%%%%%%%%%%%%%%%%%%%%%%%%%%%%%%%%%%%%%%%%%%%%%%%%%%%%%%%%%%%%%%%%%%%%%%%%%%%%%
%%%%%%%%%%%%%%%%%%%%%%%%%%%%%%%%%%%%%%%%%%%%%%%%%%%%%%%%%%%%%%%%%%%%%%%%%%%%%%%%
\section{Perfect zero knowledge for nondeterministic time}
\label{sec:pzk-ntime}

We prove that $\NEXP$ has $2$-round IOPs that are perfect zero knowledge against unbounded queries. We do so by constructing a suitable IOP system for $\NTIME(\DeciderTime)$ against query bound $\QueryBound$, for each time function $\DeciderTime$ and query bound function $\QueryBound$, where the verifier runs in time polylogarithmic in both $\DeciderTime$ and $\QueryBound$. Crucially, the simulator runs in time $\poly(\Malicious{q} + \log \DeciderTime+ \log \QueryBound)$, where $\Malicious{q}$ is the actual number of queries made by the malicious verifier; this exponential improvement over \cite{BenSassonCGV16}, where the simulator runs in time $\poly(\DeciderTime+\QueryBound)$, enables us to ``go up to $\NEXP$''.

\begin{theorem}[formal statement of \thmref{thm:intro-ntime}]
\label{thm:ntime}
For every constant $d>0$,
time bound function $\DeciderTime \colon \Naturals \to \Naturals$ with
$\InstanceSize\leq \DeciderTime(\InstanceSize) \leq 2^{\InstanceSize^{d}}$, and
query bound function $\QueryBound \colon \Naturals\to \Naturals$ with
$\QueryBound(\InstanceSize)\leq  2^{\InstanceSize^{d}}$,
there exists an IOP system $\pair{\Prover}{\Verifier}$ that makes $\NTIME(\DeciderTime)$ a subset of the complexity class
\begin{equation*}
\PZKIOP
\left[
\begin{array}{llll}
\TextAnswerAlphabet
  & \Field(\InstanceSize)
  &=& \Field_{2}
  \\
\TextNumRounds
  & \NumRounds(\InstanceSize)
  &=& 2
  \\
\TextProofLength
  & \ProofLength(\InstanceSize)
  &=& \Ot{\DeciderTime(\InstanceSize)+\QueryBound(\InstanceSize)}
  \\
\TextQueryComplexity
  & \QueryComplexity(\InstanceSize)
  &=& \polylog(\DeciderTime(\InstanceSize)+\QueryBound(\InstanceSize))
  \\
\TextSoundnessError
  & \SoundnessError(\InstanceSize)
  &=& 1/2
  \\
\TextProverTime
  & \ProverTime(\InstanceSize)
  &=& \poly(\InstanceSize) \cdot \Ot{\DeciderTime(\InstanceSize)+\QueryBound(\InstanceSize)}
  \\
\TextVerifierTime
  & \VerifierTime(\InstanceSize)
  &=& \poly(\InstanceSize + \log (\DeciderTime(\InstanceSize)+\QueryBound(\InstanceSize)))
  \\
\TextQueryBound
  & \QueryBound(\InstanceSize)
  &&
\end{array}
\right]
\enspace.
\end{equation*}
Moreover, the verifier $\Verifier$ is public-coin and non-adaptive.
\end{theorem}

Our proof is similar to that of \cite{BenSassonCGV16}, and the only major difference is that  \cite{BenSassonCGV16}'s simulator \emph{explicitly} samples random codewords, while we rely on succinct constraint detection to do so \emph{implicitly}. Indeed, the reduction from $\NTIME(\DeciderTime)$ generates codewords of size $\Ot{\DeciderTime}$, which means that sampling random codewords of that size is infeasible when $\DeciderTime$ is super-polynomial. We structure our argument in three steps, highlighting the essential components that implicitly underlie \cite{BenSassonCGV16}'s `monolithic' argument; we view this as a conceptual contribution of our work.

%%%%%%%%%%%%%%%%%%%%%%%%%%%%%%%%%%%%%%%%
\parhead{Step 1 (\secref{sec:lacsp-definition})}
We construct perfect zero knowledge IOPs of Proximity for \emph{linear algebraic constraint satisfaction problems} (LACSPs) \cite{BenSassonCGV16}, a family of constraint satisfaction problems whose domain and range are linear codes. An instance $\Instance$ of LACSP is specified by a function $\LocalMap$ and a pair of codes $\MakeZero{\Code}$, $\MakeOne{\Code}$; a witness $\Witness$ for $\Instance$ is a pair $(\MakeZero{\Codeword}, \MakeOne{\Codeword})$ such that $\MakeZero{\Codeword} \in \MakeZero{\Code}$, $\MakeOne{\Codeword} \in \MakeOne{\Code}$, and $\LocalMap(\MakeZero{\Codeword}) = \MakeOne{\Codeword}$. A natural approach to construct a perfect zero knowledge IOPP for this relation is the following: if we are given a perfect zero knowledge IOP of Proximity for the relation $\GetRelation{\MakeZero{\Code} \times \MakeOne{\Code}}$, then the verifier can test proximity of $\Witness=(\MakeZero{\Codeword}, \MakeOne{\Codeword})$ to $\MakeZero{\Code} \times \MakeOne{\Code}$ and then sample a random index $j$ and check that $\LocalMap(\MakeZero{\Codeword})[j] = \MakeOne{\Codeword}[j]$. In order for the verifier's strategy to make sense, we require $\LocalMap$ to
\begin{inparaenum}[(i)]
  \item satisfy a distance condition with respect to $\MakeZero{\Code}, \MakeOne{\Code}$, namely, that $\MakeOne{\Code} \cup \LocalMap(\MakeZero{\Code})$ has large relative distance;
  \item be `local', which means that computing $\LocalMap(\MakeZero{\Codeword})[j]$ requires examining only a few indices of $\MakeZero{\Codeword}$; and
  \item be `evasive', which means that  if $\MakeZero{\Malicious{\Codeword}}$ is close to some $\MakeZero{\Codeword} \in \MakeZero{\Code}$, then $\LocalMap(\MakeZero{\Malicious{\Codeword}})$ is close to $\LocalMap(\MakeZero{\Codeword})$.
\end{inparaenum}
All of this implies that if $(\MakeZero{\Malicious{\Codeword}}, \MakeOne{\Malicious{\Codeword}})$ is far from any valid witness but close to $\MakeZero{\Code} \times \MakeOne{\Code}$, we know that $\LocalMap(\MakeZero{\Malicious{\Codeword}})$ is far from $\MakeOne{\Malicious{\Codeword}}$, so that examining a random index $j$ gives good soundness.

%%%%%%%%%%%%%%%%%%%%%%%%%%%%%%%%%%%%%%%%
\parhead{Step 2 (\secref{sec:rlacsp-definition})}
We build on the above result to derive perfect zero knowledge IOPs for a subfamily of LACSPs called \emph{randomizable LACSPs} (RLACSPs) \cite{BenSassonCGV16}. The key difference between this protocol and the IOP of Proximity described above is that in the ``proximity setting'', the verifier, and thus also the simulator, has oracle access to the witness, while in the ``non-proximity setting'' the witness is \emph{sent} to the verifier but the simulator must make do without it; in particular, merely sending the witness $(\MakeZero{\Codeword}, \MakeOne{\Codeword})$ is \emph{not} zero knowledge. We thus rely on the randomizability property of RLACSPs to generate witnesses from a $\Randomizability$-wise independent distribution, where $\Randomizability$ is larger than the query bound $\QueryBound$. In particular, while the simulator runs in time polynomial in the actual number of queries made by a verifier, it runs in time \emph{polylogarithmic} in $\Randomizability$, and thus we can set $\QueryBound$ to be super-polynomial in order to obtain unbounded-query zero knowledge against polynomial-time verifiers.

%%%%%%%%%%%%%%%%%%%%%%%%%%%%%%%%%%%%%%%%
\parhead{Step 3 (\secref{sec:pzk-iopp-ntime})}
We derive \thmref{thm:ntime} (perfect zero knowledge IOPs for $\NTIME(\DeciderTime)$) by combining:
\begin{inparaenum}[(1)]
  \item the aforementioned result for RLACSPs;
  \item \cite{BenSassonCGV16}'s reduction from $\NTIME$ to RLACSPs;
  \item a perfect zero knowledge IOP of Proximity for a suitable choice of $\MakeZero{\Code} \times \MakeOne{\Code}$, which we derived in \secref{sec:rs-pzk-iopp}.
\end{inparaenum}
This last component is the one that makes use of succinct constraint detection, and relies on the technical innovations of our work.

%%%%%%%%%%%%%%%%%%%%%%%%%%%%%%%%%%%%%%%%%%%%%%%%%%%%%%%%%%%%%%%%%%%%%%%%%%%%%%%%
%%%%%%%%%%%%%%%%%%%%%%%%%%%%%%%%%%%%%%%%%%%%%%%%%%%%%%%%%%%%%%%%%%%%%%%%%%%%%%%%
\subsection{Perfect zero knowledge IOPs of proximity for LACSPs}
\label{sec:lacsp-definition}

A constraint satisfaction problem asks whether, for a given ``local'' function $\LocalMap$, there exists an input $\Assignment$ such that $\LocalMap(\Assignment)$ is an accepting output. For example, in the case of 3SAT with $\LocalMapDomainSize$ variables and $\LocalMapRangeSize$ clauses, the function $\LocalMap$ maps $\Bits^{\LocalMapDomainSize}$ to $\Bits^{\LocalMapRangeSize}$, and $\LocalMap(\Assignment)$ indicates which clauses are satisfied by $\Assignment \in \Bits^{\LocalMapDomainSize}$; hence $\Assignment$ yields an accepting output if (and only if) $\LocalMap(\Assignment)=1^{\LocalMapRangeSize}$. Below we introduce a family of constraint satisfaction problems whose domain and range are linear-algebraic objects, namely, linear codes.

We begin by providing the notion of locality that we use for $\LocalMap$, along with a measure of $\LocalMap$'s ``pseudorandomness''.

\begin{definition}
\label{def:local-map-properties}
Let $\LocalMap \colon \LocalMapAlphabet^{\LocalMapDomainSize} \to \LocalMapAlphabet^{\LocalMapRangeSize}$ be a function. We say that $\LocalMap$ is \defemph{$\MapLocality$-local} if for every $j\in [\LocalMapRangeSize]$ there exists $\IndexSet_{j} \subseteq [\LocalMapDomainSize]$ with $\SetCardinality{\IndexSet_{j}} = \MapLocality$ such that $\LocalMap(\Assignment)[j]$ (the $j$-th coordinate of $\LocalMap(\Assignment)$) depends only on $\Restrict{\Assignment}{\IndexSet_{j}}$ (the restriction of $\Assignment$ to $\IndexSet_{j}$). Moreover, we say that $\LocalMap$ is \defemph{$s$-evasive} if for every $\IndexSet\subseteq [\LocalMapDomainSize]$ the probability that $\IndexSet_{j}$ intersects $\IndexSet$ for a uniform $j \in [\LocalMapRangeSize]$ is at most $s\cdot \frac{\SetCardinality{\IndexSet}}{\LocalMapDomainSize}$.
\end{definition}

For example, if $\LocalMap$ is a 3SAT formula then $\LocalMap$ is 3-local because $\IndexSet_{j}$ equals the variables appearing in clause $j$; moreover, $\LocalMap$ is $s$-evasive if and only if every variable $x_{i}$ appears in at most a fraction $s/\LocalMapDomainSize$ of the clauses (i.e., the evasiveness property corresponds to the fraction of clauses in which a variable appears). Also, a natural case where $\LocalMap$ is $\MapLocality$-evasive is when the elements of $I_{j}$ are individually uniform in $[\LocalMapDomainSize]$ when $j$ is uniform in $[\LocalMapRangeSize]$.

\begin{definition}
\label{def:local-efficiency}
Let $\LocalMap \colon \LocalMapAlphabet^{\LocalMapDomainSize} \to \LocalMapAlphabet^{\LocalMapRangeSize}$ be a function. We say that $\LocalMap$ is \defemph{$\MapEfficiency$-efficient} if there is a $\MapEfficiency$-time algorithm that, given $j$ and $\Restrict{\Assignment}{\IndexSet_{j}}$, computes the set $\IndexSet_{j}$ and value $\LocalMap(\Assignment)[j]$.
\end{definition}

The above definition targets succinctly-described languages. For example, a succinct 3SAT instance is given by a circuit of size $S$ that, on input $j$, outputs a description of the $j$-th clause; the definition is then satisfied with $\MapEfficiency = O(S)$.

\begin{definition}[LACSP]
\label{def:linear-algebraic-csp}
Let $\MakeZero{\Code}(\InstanceSize), \MakeOne{\Code}(\InstanceSize)$ be (descriptions of) linear codes over $\Field(\InstanceSize)$ with block length $\CodeBlockLength(\InstanceSize)$ and relative distance $\CodeRelativeDistance(\InstanceSize)$. The promise relation of \defemph{linear algebraic CSPs} (LACSPs)
\begin{equation*}
\LACSPRelationP
\end{equation*}
considers instance-witness pairs $(\Instance,\Witness)$ of the following form.
\begin{itemize}

  \item An instance $\Instance$ is a tuple $(1^{\InstanceSize},\LocalMap)$ where:

  \begin{itemize}[nolistsep]

    \item $\LocalMap \colon \Field(\InstanceSize)^{\CodeBlockLength(\InstanceSize)}\to \Field(\InstanceSize)^{\CodeBlockLength(\InstanceSize)}$ is $\MapLocality(\InstanceSize)$-local, $\MapLocality(\InstanceSize)$-evasive, and $\MapEfficiency(\InstanceSize)$-efficient;
     \item $\MakeOne{\Code}(\InstanceSize) \cup \LocalMap(\MakeZero{\Code}(\InstanceSize))$ has relative distance at least $\CodeRelativeDistance(\InstanceSize)$ (though may not be a linear space).

  \end{itemize}

  \item A witness $\Witness$ is a tuple $(\MakeZero{\Assignment},\MakeOne{\Assignment})$ where $\MakeZero{\Assignment},\MakeOne{\Assignment} \in \Field(\InstanceSize)^{\CodeBlockLength(\InstanceSize)}$.
  
\end{itemize}
The yes-relation $\LACSPRelation^{\yes}$ consists of all pairs $(\Instance, \Witness)$ as above where the instance $\Instance$ and witness $\Witness$ jointly satisfy the following:
  $\MakeZero{\Assignment} \in \MakeZero{\Code}(\InstanceSize)$,
  $\MakeOne{\Assignment} \in \MakeOne{\Code}(\InstanceSize)$, and
  $\LocalMap(\MakeZero{\Assignment}) = \MakeOne{\Assignment}$.
  (In particular, a witness $\Witness=(\MakeZero{\Codeword}, \LocalMap(\MakeZero{\Codeword}))$ with $\MakeZero{\Codeword} \in \MakeZero{\Code}(\InstanceSize)$ satisfies $\Instance$ if and only if $\LocalMap(\MakeZero{\Codeword}) \in \MakeOne{\Code}(\InstanceSize)$.) The no-language consists of all instances $\Instance$ as above where $\Instance \notin \GetLanguage{\LACSPRelation^{\yes}}$.
\end{definition}

\begin{remark}
In \cite{BenSassonCGV16} the codes $\MakeZero{\Code}$ and $\MakeOne{\Code}$ are allowed to have distinct block lengths while, for simplicity, we assume that they have the same block length; this restriction does not change any of their, or our, results.
\end{remark}

We are now ready to give perfect zero knowledge IOPs of proximity for LACSPs.

\begin{theorem}
\label{thm:ioppzk-construction}
Suppose that there exists an IOPP system $(\CodesProver,\CodesVerifier)$ that puts $\LACodeRelation$ in the complexity class
\begin{equation*}
\PZKIOPP[\Field,\IOPPize{\NumRounds},\IOPPize{\ProofLength},\IOPPize{\QueryComplexity},\IOPPize{\ProximityParameter},\IOPPize{\SoundnessError},\IOPPize{\ProverTime},\IOPPize{\VerifierTime},\AnyBound]
\enspace.
\end{equation*}
Then there exists an IOPP system $(\LAProver,\LAVerifier)$ that puts $\LACSPRelationQ$ in the complexity class
\begin{equation*}
{\small
\PZKIOPP
\left[
\begin{array}{llll}
\TextAnswerAlphabet
  &   \Field(\InstanceSize)
  &~& \\
\TextNumRounds
  &   \NumRounds(\InstanceSize)
  &=& \IOPPize{\NumRounds}(\InstanceSize) \\
\TextProofLength
  &    \ProofLength(\InstanceSize)
  &=&  \IOPPize{\ProofLength}(\InstanceSize) \\
\TextQueryComplexity
  &   \QueryComplexity(\InstanceSize)
  &=& \IOPPize{\QueryComplexity}(\InstanceSize) + \LAize{\MapLocality}(\InstanceSize) + 1 \\
\TextSoundnessError
  &   \SoundnessError(\InstanceSize)
  &=& \max \big\{\IOPPize{\SoundnessError}(\InstanceSize)\,,\, 1 - \LAize{\CodeRelativeDistance}(\InstanceSize) + 2\IOPPize{\ProximityParameter}(\InstanceSize) \cdot (\LAize{\MapLocality}(\InstanceSize) + 1)\big\} \\
\TextProximityParameter
  &   \ProximityParameter(\InstanceSize)
  &=& \IOPPize{\ProximityParameter}(\InstanceSize) \\
\TextProverTime
  &   \ProverTime(\InstanceSize)
  &=& \IOPPize{\ProverTime}(\InstanceSize) \\
\TextVerifierTime
  &   \VerifierTime(\InstanceSize)
  &=& \IOPPize{\VerifierTime}(\InstanceSize) + \RLAize{\MapEfficiency}(\InstanceSize)  \\
%\TextVerifierRandomness & \VerifierRandomness &=& \IOPPize{\VerifierRandomness}(\RLAize{\CodeBlockLength}) + \log \RLAize{\CodeBlockLength} \\
\TextQueryBound
  &   \QueryBound(\InstanceSize)
  &=& \AnyBound
\end{array}
\right]
\enspace.
}
\end{equation*}
Moreover, if $(\CodesProver,\CodesVerifier)$ is non-adaptive (respectively, public-coin) then so is $(\LAProver,\LAVerifier)$.
\end{theorem}

\begin{proof}
We construct the IOPP system $(\LAProver,\LAVerifier)$ for $\Relation$, where the prover receives $(\Instance,\Witness) = \big((1^{\InstanceSize},\LocalMap),(\MakeZero{\Assignment},\MakeOne{\Assignment})\big)$ as input while the verifier receives $\Instance$ as input and $\Witness$ as an oracle, as follows:
\begin{enumerate}[nolistsep]

  \item $\LAProver$ and $\LAVerifier$ invoke the IOPP system $(\CodesProver,\CodesVerifier)$ to prove that $(\MakeZero{\Codeword}, \MakeOne{\Codeword}) \in \MakeZero{\Code} \times \MakeOne{\Code}$;

  \item $\LAVerifier$ chooses a random $j \in [\LAize{\CodeBlockLength}]$ and checks that $\LocalMap(\MakeZero{\Codeword})[j] = \MakeOne{\Codeword}[j]$;

  \item $\LAVerifier$ rejects if and only if $\CodesVerifier$ rejects or the above check fails.

\end{enumerate}

%%%%%%%%%%%%%%%%%%%%%%%%%%%%%%%%%%%%%%%%
\parhead{Completeness}
If $(\Instance,\Witness) \in \LACSPRelation^{\yes}$, then
\begin{inparaenum}[(i)]
  \item $\MakeZero{\Codeword} \in \MakeZero{\Code}, \MakeOne{\Codeword} \in \MakeOne{\Code}$, so $\CodesVerifier$ always accepts, and
  \item $\LocalMap(\MakeZero{\Codeword}) = \MakeOne{\Codeword}$ so the consistency check succeeds for every $j \in [\LAize{\CodeBlockLength}]$.
\end{inparaenum}
We deduce that $\LAVerifier$ always accepts. 

%%%%%%%%%%%%%%%%%%%%%%%%%%%%%%%%%%%%%%%%
\parhead{Soundness}
Suppose that $\Instance \in \GetLanguage{\LACSPRelation^{\yes}} \cup \LACSPLanguage^{\no}$ and $\Malicious{\Witness}$ are such that $\DistanceMeasure(\Malicious{\Witness},\Witnesses{\LACSPRelation^{\yes}}{\Instance}) \geq \IOPPize{\ProximityParameter}$. Writing $\Malicious{\Witness} = (\MakeZero{\Malicious{\Codeword}},\MakeOne{\Malicious{\Codeword}})$, we argue as follows.
\begin{itemize}

\item \emph{Case 1: $\DistanceMeasure(\Malicious{\Witness},\MakeZero{\Code} \times \MakeOne{\Code}) \geq \IOPPize{\ProximityParameter}$.}
In this case $\CodesVerifier$ rejects with probability at least $1 - \IOPPize{\SoundnessError}$.

\item \emph{Case 2: $\DistanceMeasure(\Malicious{\Witness},\MakeZero{\Code} \times \MakeOne{\Code}) < \IOPPize{\ProximityParameter}$.}
There exist codewords $\MakeZero{\Codeword} \in \MakeZero{\Code}$ and $\MakeOne{\Codeword} \in \MakeOne{\Code}$ such that $(\MakeZero{\Codeword}, \MakeOne{\Codeword})$ is $\ProximityParameter$-close to $(\MakeZero{\Malicious{\Codeword}},\MakeOne{\Malicious{\Codeword}})$ for $\ProximityParameter < \IOPPize{\ProximityParameter}$. By assumption, $\DistanceMeasure(\Malicious{\Witness},\Witnesses{\LACSPRelation^{\yes}}{\Instance}) \geq \IOPPize{\ProximityParameter}$, so in particular $(\MakeZero{\Codeword}, \MakeOne{\Codeword})$ cannot be in $\Witnesses{\LACSPRelation^{\yes}}{\Instance}$, and $\LocalMap(\MakeZero{\Codeword}) \neq \MakeOne{\Codeword}$. Since $\MakeOne{\Code} \cup \LocalMap(\MakeZero{\Code})$ has relative distance at least $\LAize{\CodeRelativeDistance}$, $\DistanceMeasure(\LocalMap(\MakeZero{\Codeword}), \MakeOne{\Codeword}) \geq \LAize{\CodeRelativeDistance}$.
Observe that since $\MakeZero{\Code}$ and $\MakeOne{\Code}$ have the same block length, $\DistanceMeasure(\MakeZero{\Malicious{\Codeword}}, \MakeZero{\Codeword}) \leq 2 \IOPPize{\ProximityParameter}$ and $\DistanceMeasure(\MakeOne{\Malicious{\Codeword}}, \MakeOne{\Codeword}) \leq 2 \IOPPize{\ProximityParameter}$. Thus since $\LocalMap$ is $\LAize{\MapLocality}$-evasive, the probability that the set of coordinates $I \DefineEqual \{ i \in [\LAize{\CodeBlockLength}] : \MakeZero{\Codeword}[i] \neq \MakeZero{\Malicious{\Codeword}}[i]\}$ intersects with $\IndexSet_{j}$ for random $j \in [\LAize{\CodeBlockLength}]$ is at most $2 \IOPPize{\ProximityParameter} \LAize{\MapLocality}$, so $\DistanceMeasure(\LocalMap(\MakeZero{\Codeword}),\LocalMap(\MakeZero{\Malicious{\Codeword}})) \leq 2 \IOPPize{\ProximityParameter} \LAize{\MapLocality}$.
Using the triangle inequality, we deduce that
\begin{equation*}
\DistanceMeasure(\LocalMap(\MakeZero{\Malicious{\Codeword}}),\MakeOne{\Malicious{\Codeword}}) \geq \LAize{\CodeRelativeDistance} - 2\IOPPize{\ProximityParameter} (\LAize{\MapLocality} + 1),
\end{equation*}
which means the consistency check rejects with probability at least $\LAize{\CodeRelativeDistance} - 2\IOPPize{\ProximityParameter} (\LAize{\MapLocality} + 1)$.
\end{itemize}
It follows that $\LAVerifier$ accepts with probability at most $\max \{\IOPPize{\SoundnessError}, 1 - \LAize{\CodeRelativeDistance} + 2\IOPPize{\ProximityParameter} (\LAize{\MapLocality} + 1)\}$.

%%%%%%%%%%%%%%%%%%%%%%%%%%%%%%%%%%%%%%%%
\parhead{Perfect zero knowledge}
We can choose the simulator $\LASimulator$ for $(\LAProver,\LAVerifier)$ to equal any simulator $\CodesSimulator$ that fulfills the perfect zero knowledge guarantee of $(\CodesProver,\CodesVerifier)$. Indeed, the behavior of $\LAProver$ is the same as $\CodesProver$, and so the view of any malicious verifier $\Malicious{\IOPPVerifier}$ when interacting with $\LAProver$ is identical to its view when interacting with $\CodesProver$.
\end{proof}

%%%%%%%%%%%%%%%%%%%%%%%%%%%%%%%%%%%%%%%%%%%%%%%%%%%%%%%%%%%%%%%%%%%%%%%%%%%%%%%%
%%%%%%%%%%%%%%%%%%%%%%%%%%%%%%%%%%%%%%%%%%%%%%%%%%%%%%%%%%%%%%%%%%%%%%%%%%%%%%%%
\subsection{Perfect zero knowledge IOPs for RLACSPs}
\label{sec:rlacsp-definition}

The above discussion achieves perfect zero knowledge for LACSPs, ``up to queries to the witness''. We now explain how to simulate these queries as well, without any knowledge of the witness, for a special class of LACSPs called \emph{randomizable LACSPs}. For these, the prover can randomize a given witness $(\MakeZero{\Codeword}, \LocalMap(\MakeZero{\Codeword}))$ by sampling a random $\IOPProverNonce$ in a $\Randomizability$-wise independent subcode $\Subcode$ of $\MakeZero{\Code}$, and use the new `shifted' witness $(\MakeZero{\Codeword} + \IOPProverNonce, \LocalMap(\MakeZero{\Codeword} + \IOPProverNonce))$ instead of the original one. We now define the notion of randomizable LACSPs, and then show how to construct perfect zero knowledge IOPs for these, against bounded-query verifiers and where the the query bound depends on $\Randomizability$.

\begin{definition}[randomizability]
\label{def:randomizable-instances}
An instance $\Instance=(1^{\InstanceSize},\LocalMap)$ is \defemph{$\Randomizability(\InstanceSize)$-randomizable in time $\RandomizeTime(\InstanceSize)$} (with respect to code families $\MakeZero{\Code}(\InstanceSize),\MakeOne{\Code}(\InstanceSize)$) if:
\begin{inparaenum}[(i)]
  \item there exists a $\Randomizability(\InstanceSize)$-wise independent subcode $\Subcode \subseteq \MakeZero{\Code}(\InstanceSize)$ such that if $(\MakeZero{\Codeword}, \LocalMap(\MakeZero{\Codeword}))$ satisfies $\Instance$, then, for every $\MakeZero{\Codeword'}$ in $\Subcode + \MakeZero{\Codeword} \DefineEqual \ConditionalSet{\SubcodeCodeword + \MakeZero{\Codeword}}{\SubcodeCodeword \in \Subcode}$, the witness $(\MakeZero{\Codeword'}, \LocalMap(\MakeZero{\Codeword'}))$ also satisfies $\Instance$; and
  \item one can sample, in time $\RandomizeTime(\InstanceSize)$, three uniformly random elements in $\Subcode,\MakeZero{\Code}(\InstanceSize),\MakeOne{\Code}(\InstanceSize)$ respectively.
\end{inparaenum}
\end{definition}

\begin{definition}[RLACSP]
\label{def:randomizable-lacsp}
The promise relation of \defemph{randomizable linear algebraic CSPs} (RLACSPs) is
\begin{equation*}
\RLACSPRelationP
\end{equation*}
where $\RLACSPRelation^{\yes}$ is obtained by restricting $\LACSPRelation$ to instances that are $\Randomizability$-randomizable in time $\RandomizeTime$, and $\RLACSPLanguage^{\no} \DefineEqual \LACSPLanguage^{\no}$.
\end{definition}

\begin{theorem}
\label{thm:iopzk-construction}
Suppose that there exists an IOPP system $(\CodesProver,\CodesVerifier)$ that puts $\LACodeRelation$ in the complexity class
\begin{equation*}
\PZKIOPP[\Field,\IOPPize{\NumRounds},\IOPPize{\ProofLength},\IOPPize{\QueryComplexity},\IOPPize{\ProximityParameter},\IOPPize{\SoundnessError},\IOPPize{\ProverTime},\IOPPize{\VerifierTime},\AnyBound]
\enspace.
\end{equation*}
Then there exists an IOP system $(\RLAProver,\RLAVerifier)$ that puts $\RLACSPRelationQ$ (with $\RLAize{\MapEfficiency}$ polynomially bounded) in the complexity class
\begin{equation*}
{\small
\PZKIOP
\left[
\begin{array}{llll}
\TextAnswerAlphabet
  &   \Field(\InstanceSize)
  &~& \\
\TextNumRounds
  &   \NumRounds(\InstanceSize)
  &=& \IOPPize{\NumRounds}(\InstanceSize) \\
\TextProofLength
  &   \ProofLength(\InstanceSize)
  &=& \IOPPize{\ProofLength}(\InstanceSize) +  \RLAize{\CodeBlockLength}(\InstanceSize) \\
\TextQueryComplexity
  &   \QueryComplexity(\InstanceSize)
  &=& \IOPPize{\QueryComplexity}(\InstanceSize) + \RLAize{\MapLocality}(\InstanceSize) + 1 \\
\TextSoundnessError
  &   \SoundnessError(\InstanceSize)
  &=& \max \big\{\IOPPize{\SoundnessError}(\InstanceSize)\,,\,1-\RLAize{\CodeRelativeDistance}(\InstanceSize)+2\cdot \IOPPize{\ProximityParameter}(\InstanceSize) \cdot(\RLAize{\MapLocality}(\InstanceSize)+1)\big\} \\
\TextProverTime
  &   \ProverTime(\InstanceSize)
  &=& \IOPPize{\ProverTime}(\InstanceSize) +\RLAize{\MapEfficiency}(\InstanceSize) \cdot \RLAize{\CodeBlockLength}(\InstanceSize) + \RLAize{\RandomizeTime}(\InstanceSize) \\
\TextVerifierTime
  &   \VerifierTime(\InstanceSize)
  &=& \IOPPize{\VerifierTime}(\InstanceSize) + \RLAize{\MapEfficiency}(\InstanceSize)  \\
%\TextVerifierRandomness & \VerifierRandomness &=& \IOPPize{\VerifierRandomness}(\RLAize{\CodeBlockLength}) + \log \RLAize{\CodeBlockLength} \\
\TextQueryBound
  &   \QueryBound (\InstanceSize)
  &=& \RLAize{\Randomizability}(\InstanceSize) / \RLAize{\MapLocality}(\InstanceSize)
\end{array}
\right]
\enspace.
}
\end{equation*}
Moreover, if $(\CodesProver,\CodesVerifier)$ is non-adaptive (respectively, public-coin) then so is $(\RLAProver,\RLAVerifier)$.
\end{theorem}

\begin{proof}
Let $(\LAProver,\LAVerifier)$ be the IOPP system for $\LACSPRelation$ guaranteed by \thmref{thm:ioppzk-construction}. We construct the IOP system $(\RLAProver,\RLAVerifier)$ for $(\RLACSPRelation^{\yes}, \RLACSPLanguage^{\no})$, where the prover receives $(\Instance,\Witness) = \big((1^{\InstanceSize},\LocalMap),(\MakeZero{\Assignment},\MakeOne{\Assignment})\big)$ as input while the verifier receives $\Instance$ as input, as follows:
\begin{enumerate}

  \item
  The prover $\RLAProver$ parses the witness $\Witness$ as $(\MakeZero{\IOPProverAssignment},\MakeOne{\IOPProverAssignment}) \in \MakeZero{\Code} \times \MakeOne{\Code}$, samples a random $\IOPProverNonce \in \Subcode$ (the subcode of $\MakeZero{\Code}$ for which $\RLAize{\Randomizability}$-randomizability holds), sets $\MakeZero{\IOPProverRandAssignment} \DefineEqual \IOPProverNonce + \MakeZero{\IOPProverAssignment}$ and $\MakeOne{\IOPProverRandAssignment} \DefineEqual \LocalMap(\MakeZero{\IOPProverRandAssignment})$, and sends $\RLAWitness \DefineEqual (\MakeZero{\IOPProverRandAssignment},\MakeOne{\IOPProverRandAssignment})$ to $\RLAVerifier$.
 
  \item
  In parallel to the above interaction, the prover $\RLAProver$ and verifier $\RLAVerifier$ invoke the IOPP system $(\LAProver,\LAVerifier)$ on the input $\Instance$ and new ``prover-randomized'' witness $\RLAWitness$. The verifier $\RLAVerifier$ accepts if and only if $\LAVerifier$ does.

\end{enumerate}
The claimed efficiency parameters immediately follow by construction. We now show that $(\RLAProver,\RLAVerifier)$ satisfies completeness, soundness, and perfect zero-knowledge.

%%%%%%%%%%%%%%%%%%%%%%%%%%%%%%%%%%%%%%%%
\parhead{Completeness}
Suppose that $(\Instance,\Witness) = \big((1^{\InstanceSize},\LocalMap),(\MakeZero{\IOPProverAssignment},\MakeOne{\IOPProverAssignment})\big)$ is in the relation $\RLACSPRelation^{\yes}$, so that $\MakeZero{\IOPProverAssignment} \in \MakeZero{\Code}$, $\MakeOne{\IOPProverAssignment} \in \MakeOne{\Code}$, and $\LocalMap(\MakeZero{\IOPProverAssignment}) = \MakeOne{\IOPProverAssignment}$. By randomizability (\defref{def:randomizable-lacsp}), since $\MakeZero{\IOPProverRandAssignment} \in \Subcode + \MakeZero{\IOPProverAssignment}$, we deduce that $\RLAWitness = (\MakeZero{\IOPProverRandAssignment}, \MakeOne{\IOPProverRandAssignment}) = (\MakeZero{\IOPProverRandAssignment}, \LocalMap(\MakeZero{\IOPProverRandAssignment}))$ satisfies $\Instance$, and so $(\Instance,\RLAWitness) \in \LACSPRelation^{\yes}$. Completeness then follows by the completeness of $(\LAProver,\LAVerifier)$.

%%%%%%%%%%%%%%%%%%%%%%%%%%%%%%%%%%%%%%%%
\parhead{Soundness}
Suppose that $\Instance = (1^{\InstanceSize},\LocalMap)$ is in the language $\RLACSPLanguage^{\no} = \LACSPLanguage^{\no}$, so that $\Witnesses{\LACSPRelation^{\yes}}{\Instance}=\emptyset$. Regardless of what `witness' $\RLAWitness$ is sent by $\RLAProver$, it holds that $\DistanceMeasure(\RLAWitness,\Witnesses{\LACSPRelation^{\yes}}{\Instance}) = \DistanceMeasure(\RLAWitness,\emptyset) = 1 \geq \IOPPize{\ProximityParameter}$, so that the soundness of $(\LAProver,\LAVerifier)$ implies that $\LAVerifier$, and thus $\RLAVerifier$, accepts with probability at most $\max \{\IOPPize{\SoundnessError}, 1 - \LAize{\CodeRelativeDistance} + 2\IOPPize{\ProximityParameter} \cdot (\LAize{\MapLocality} + 1)\}$.

%%%%%%%%%%%%%%%%%%%%%%%%%%%%%%%%%%%%%%%%
\parhead{Perfect zero knowledge}
Let $\Malicious{\Verifier}$ be any verifier that makes at most $\QueryBound \DefineEqual \RLAize{\Randomizability} / \RLAize{\MapLocality}$ queries, and let $\LASimulator$ be the perfect zero knowledge simulator for $(\LAProver,\LAVerifier)$. We construct a simulator $\RLASimulator$ for $(\RLAProver,\RLAVerifier)$ as follows:

\begin{adjustwidth}{0.5cm}{0.5cm}
$\RLASimulator^{\Malicious{\Verifier}}(\Instance)$:
\begin{enumerate}
 \item $\RLASimulator$ initializes two empty strings $\MakeZero{\SimRLAWitness}$ and $\MakeOne{\SimRLAWitness}$ which will be partially filled during the simulation.
 \item $\RLASimulator$ invokes $\LASimulator^{\Malicious{\Verifier},(\MakeZero{\SimRLAWitness},\MakeOne{\SimRLAWitness})}$, and during the execution answers oracle queries to $(\MakeZero{\SimRLAWitness},\MakeOne{\SimRLAWitness})$ in the following way.
 \begin{enumerate}[nolistsep]
 
  \item If $\LASimulator$ queries $\MakeZero{\SimRLAWitness}$ at a location $j\in [\RLAize{\CodeBlockLength}]$:
  if $\MakeZero{\SimRLAWitness}[j]$ is already defined then return that value; otherwise sample a random $a\in \Field(\InstanceSize)$, set $\MakeZero{\SimRLAWitness}[j]\DefineEqual a$, and reply with $\MakeZero{\SimRLAWitness}[j]$.
 
  \item If $\LASimulator$ queries $\MakeOne{\SimRLAWitness}$ at a location $j\in [\RLAize{\CodeBlockLength}]$:
  if $\MakeOne{\SimRLAWitness}[j]$ is already defined then return that value; otherwise compute the set of indices $I_{j}\subseteq [\RLAize{\CodeBlockLength}]$ that $\LocalMap(\cdot)_{j}$ depends on; then `query' the values of $\MakeZero{\SimRLAWitness}[i]$ for  all $i\in I_{j}$ as in the previous step; then update $\MakeOne{\SimRLAWitness}[j]\DefineEqual \LocalMap(\MakeZero{\SimRLAWitness})[j]$ and reply with $\MakeOne{\SimRLAWitness}[j]$.

\end{enumerate}
\end{enumerate}
\end{adjustwidth}
Observe that $\RLASimulator$ runs in time $\poly(\BitSize{\Instance} + q_{\Malicious{\Verifier}} + \RLAize{\MapEfficiency})$, where $q_{\Malicious{\Verifier}}$ denotes the actual number of queries made by $\Malicious{\Verifier}$ and $\RLAize{\MapEfficiency}$ is $\LocalMap$'s efficiency (see \defref{def:local-efficiency}). Since $\RLAize{\MapEfficiency}$ is polynomially bounded, $\RLASimulator$ runs in time $\poly(\BitSize{\Instance} + \QueryComplexity_{\Malicious{\Verifier}})$, as required.

We must show that $\IOPView{\RLAProver(\Instance,\Witness)}{\Malicious{\Verifier}(\Instance)}$ and $\RLASimulator^{\Malicious{\IOPVerifier}}(\Instance)$ are identically distributed. Recall that $\RLAProver(\Instance,\Witness)$ samples $\RLAWitness$ and then invokes $\LAProver(\Instance,\RLAWitness)$; viewing $\RLAWitness$ as a random variable, we get that $\IOPView{\RLAProver(\Instance,\Witness)}{\Malicious{\Verifier}(\Instance)} \equiv \IOPPView{\LAProver(\Instance,\RLAWitness)}{\Malicious{\Verifier}(\Instance)}$. By $(\LAProver,\LAVerifier)$'s perfect zero knowledge guarantee, we also know that
\begin{equation*}
\big(\IOPPView{\LAProver(\Instance,\RLAWitness)}{\Malicious{\Verifier}(\Instance)}, q_{\Malicious{\Verifier}}\big)
\equiv
\big(\LASimulator^{\Malicious{\Verifier},\RLAWitness}(\Instance), q_{\LASimulator}\big)
\enspace.
\end{equation*}
We are left to show that $\LASimulator^{\Malicious{\Verifier},\RLAWitness}(\Instance) \equiv \RLASimulator^{\Malicious{\Verifier}}(\Instance)$.

By the query bound, we know that $\LASimulator$ makes at most $\RLAize{\Randomizability}/\RLAize{\MapLocality}$ queries to $\RLAWitness$. By construction of $\RLASimulator$, this causes at most $\RLAize{\Randomizability}$ entries in $\MakeZero{\SimRLAWitness}$ to be `defined', since $\SetCardinality{\IndexSet_{j}} \leq \RLAize{\MapLocality}$ for all $j \in [\RLAize{\CodeBlockLength}]$ (by $\LocalMap$'s locality); let $\RLAEntries \subseteq [\RLAize{\CodeBlockLength}]$ be these entries. Since $\MakeOne{\Codeword} = \LocalMap(\MakeZero{\Codeword})$, all of the responses to $\LASimulator$'s queries are determined by $\Restrict{\MakeZero{\IOPProverRandAssignment}}{\RLAEntries}$. While $\RLAEntries$ is itself dependent on $\MakeZero{\IOPProverRandAssignment}$ (as $\Malicious{\Verifier}$'s queries may be adaptive), this does not affect the distribution of the string $\Restrict{\MakeZero{\IOPProverRandAssignment}}{\RLAEntries}$ because $\SetCardinality{\RLAEntries} \leq \RLAize{\Randomizability}$ and $\MakeZero{\IOPProverRandAssignment}$ is drawn from a $\RLAize{\Randomizability}$-wise independent distribution. We deduce that there exists a deterministic function $\FView{\cdot}$ such that $\LASimulator$'s queries to $\RLAWitness$ are answered by $\FView{\Restrict{\MakeZero{\IOPProverRandAssignment}}{\RLAEntries}}$ in the `real' execution, and $\RLASimulator$ answers the same queries with $\FView{U}$ where $U$ is uniformly random in $\Field^{\RLAEntries}$. But $\MakeZero{\IOPProverRandAssignment}$ is $\SetCardinality{\RLAEntries}$-wise independent, so that $\Restrict{\MakeZero{\IOPProverRandAssignment}}{\RLAEntries} \equiv U$, and thus $\LASimulator^{\Malicious{\Verifier},\RLAWitness}(\Instance) \equiv \RLASimulator^{\Malicious{\Verifier}}(\Instance)$.
\end{proof}

%%%%%%%%%%%%%%%%%%%%%%%%%%%%%%%%%%%%%%%%%%%%%%%%%%%%%%%%%%%%%%%%%%%%%%%%%%%%%%%%
%%%%%%%%%%%%%%%%%%%%%%%%%%%%%%%%%%%%%%%%%%%%%%%%%%%%%%%%%%%%%%%%%%%%%%%%%%%%%%%%
\subsection{Putting things together}
\label{sec:pzk-iopp-ntime}

We are almost ready to prove \thmref{thm:ntime}, the main theorem of this section. The last missing piece is a suitable reduction from $\NTIME(\DeciderTime)$ to $\RLACSPRelation$, the promise relation of RLACSPs. Below, we state a special case of \cite[Thm.~7.9]{BenSassonCGV16}, which provides the reduction that we need.

\begin{theorem}[$\NTIME \to \RLACSPRelation$]
\label{thm:ntime-to-lacsp}
For every $\DeciderTime,\Randomizability \colon \Naturals \to \Naturals$, constant $\CodeRelativeDistance \in (0,1)$, and $\Relation \in \NTIME(\DeciderTime)$ there exist algorithms $\InstanceMap,\WitnessMap,\ExtractorMap$ satisfying the following conditions:
\begin{itemize}

  \item \textsc{Efficient reduction.}
For every instance $\Instance$, letting $\Instance' \DefineEqual \InstanceMap(\Instance)$:
\begin{itemize}[nolistsep]
  \item if $\Instance \in \GetLanguage{\Relation}$ then $\Instance' \in \GetLanguage{\RLACSPRelation^{\yes}}$;
  \item if $\Instance \not\in \GetLanguage{\Relation}$ then $\Instance' \in \RLACSPLanguage^{\no}$;
  \item for every witness $\Witness$, if $(\Instance,\Witness) \in \Relation$ then $(\Instance',\WitnessMap(\Instance,\Witness)) \in \RLACSPRelation^{\yes}$;
  \item for every witness $\Witness'$, if $(\Instance',\Witness') \in \RLACSPRelation^{\yes}$ then $(\Instance,\ExtractorMap(\Instance,\Witness')) \in \Relation$.
\end{itemize}
Moreover, $\InstanceMap$ runs in time $\poly(\InstanceSize + \log (\DeciderTime(\InstanceSize)+\Randomizability(\InstanceSize)))$ and $\WitnessMap,\ExtractorMap$ run in time $\poly(\InstanceSize) \cdot \Ot{\DeciderTime(\InstanceSize)+\Randomizability(\InstanceSize)}$.

  \item \textsc{Randomizable linear algebraic CSP.}
The promise relation $(\RLACSPRelation^{\yes},\RLACSPLanguage^{\no})$ has the parameters:
\begin{equation*}
(\RLACSPRelation^{\yes},\RLACSPLanguage^{\no})
\left[
\begin{array}{llll}
\TextField                & \Field                &=& \Field_{2^{\log(\DeciderTime+\Randomizability)+O(\log \log (\DeciderTime+\Randomizability))}} \\
\TextCodeZero             & \MakeZero{\Code}      &~& \\
\TextCodeOne              & \MakeOne{\Code}       &~& \\
\TextCodeBlockLength      & \CodeBlockLength      &=& \Ot{\DeciderTime+\Randomizability} \\
\TextCodeRelativeDistance & \CodeRelativeDistance & & \\
\TextMapLocality          & \MapLocality          &=& \polylog{\DeciderTime} \\
\TextMapEfficiency        & \MapEfficiency        &=& \poly(\InstanceSize + \log \DeciderTime) \\
\TextRandomizability      & \Randomizability      & & \\
\TextRandomizeTime        & \RandomizeTime        &=& \Ot{\DeciderTime+\Randomizability} \\
\end{array}
\right]
\enspace.
\end{equation*}

(The hidden constants depend on the choice of $\CodeRelativeDistance$; see \cite[Thm.~7.9]{BenSassonCGV16} for the dependence on $\CodeRelativeDistance$.)

\item \textsc{Additive Reed--Solomon codes.}
$\LACodeRelation$ is a subfamily of $\RSVRS$.

\end{itemize}
\end{theorem}

\begin{proof}[Proof of \thmref{thm:ntime}]
The theorem directly follows by having the prover and verifier reduce the given relation in $\NTIME(\DeciderTime)$ to $(\RLACSPRelation^{\yes},\RLACSPLanguage^{\no})$, following \thmref{thm:ntime-to-lacsp}, and then invoking \thmref{thm:iopzk-construction} with the perfect zero knowledge IOP of Proximity for $\RSVRS$ from \corref{cor:rsvrs-pzk-iopp}.
\end{proof}

\doclearpage
%%%%%%%%%%%%%%%%%%%%%%%%%%%%%%%%%%%%%%%%%%%%%%%%%%%%%%%%%%%%%%%%%%%%%%%%%%%%%%%%
%%%%%%%%%%%%%%%%%%%%%%%%%%%%%%%%%%%%%%%%%%%%%%%%%%%%%%%%%%%%%%%%%%%%%%%%%%%%%%%%
%%%%%%%%%%%%%%%%%%%%%%%%%%%%%%%%%%%%%%%%%%%%%%%%%%%%%%%%%%%%%%%%%%%%%%%%%%%%%%%%
\section*{Acknowledgements}
\label{sec:acks}

Work of E.\ Ben-Sasson, A.\ Gabizon, and M. Riabzev was supported by the Israel Science Foundation (grant 1501/14). Work of M.\ A.\ Forbes was supported by the NSF, including NSF CCF-1617580, and the DARPA Safeware program; it was also partially completed when the author was at Princeton University, supported by the Princeton Center for Theoretical Computer Science.

%%%%%%%%%%%%%%%%%%%%%%%%%%%%%%%%%%%%%%%%%%%%%%%%%%%%%%%%%%%%%%%%%%%%%%%%%%%%%%%%
%%%%%%%%%%%%%%%%%%%%%%%%%%%%%%%%%%%%%%%%%%%%%%%%%%%%%%%%%%%%%%%%%%%%%%%%%%%%%%%%
%%%%%%%%%%%%%%%%%%%%%%%%%%%%%%%%%%%%%%%%%%%%%%%%%%%%%%%%%%%%%%%%%%%%%%%%%%%%%%%%
\appendix
%%%%%%%%%%%%%%%%%%%%%%%%%%%%%%%%%%%%%%%%%%%%%%%%%%%%%%%%%%%%%%%%%%%%%%%%%%%%%%%%
%%%%%%%%%%%%%%%%%%%%%%%%%%%%%%%%%%%%%%%%%%%%%%%%%%%%%%%%%%%%%%%%%%%%%%%%%%%%%%%%
%%%%%%%%%%%%%%%%%%%%%%%%%%%%%%%%%%%%%%%%%%%%%%%%%%%%%%%%%%%%%%%%%%%%%%%%%%%%%%%%

\doclearpage
%%%%%%%%%%%%%%%%%%%%%%%%%%%%%%%%%%%%%%%%%%%%%%%%%%%%%%%%%%%%%%%%%%%%%%%%%%%%%%%%
%%%%%%%%%%%%%%%%%%%%%%%%%%%%%%%%%%%%%%%%%%%%%%%%%%%%%%%%%%%%%%%%%%%%%%%%%%%%%%%%
%%%%%%%%%%%%%%%%%%%%%%%%%%%%%%%%%%%%%%%%%%%%%%%%%%%%%%%%%%%%%%%%%%%%%%%%%%%%%%%%
\section{Prior work on single-prover unconditional zero knowledge}
\label{sec:prior-work-summary}

We summarize prior work on \emph{single}-prover proof systems that achieve zero knowledge unconditionally. First, the complexity classes of PZK IPs and SZK IPs are contained in $\AM \cap \coAM$ \cite{Fortnow87,AielloH91}, so they do not contain $\NP$ unless the polynomial hierarchy collapses \cite{BoppanaHZ87}; thus, IPs have strong limitations. Next, we discuss other single-prover proof systems: PCPs and IPCPs; all prior work for these is about \emph{statistical} zero knowledge (SZK), via simulators that are straightline (which is needed in many of the cryptographic applications explored in these works).

%%%%%%%%%%%%%%%%%%%%%%%%%%%%%%%%%%%%%%%%
\parhead{SZK PCP for $\NEXP$}
\cite{KilianPT97} obtain PCPs for $\NEXP$ that are SZK against unbounded queries; the PCP has exponential length, the honest verifier makes a polynomial number of queries, and malicious verifiers can make any polynomial number of queries. Their construction has two steps:
\begin{inparaenum}[(1)]
  \item transform a given PCP into a new one that is PZK against (several independent copies of) the honest verifier;
  \item transform the latter PCP into a new one that is SZK against malicious verifiers.
\end{inparaenum}
The first step uses secret sharing and builds on techniques of \cite{DworkFKNS92}; the second uses \emph{locking schemes}, which are information-theoretic PCP-analogues of commitment schemes. Subsequent work simplifies the steps: \cite{IshaiW14} use MPC techniques to simplify the first step; and \cite{IshaiMS12,IshaiMSX15} give a simple construction of locking schemes, by obtaining a non-interactive PCP-analogue of \cite{Naor91}'s commitment scheme.

%%%%%%%%%%%%%%%%%%%%%%%%%%%%%%%%%%%%%%%%
\parhead{SZK PCP for $\NP$ against \sunderline{unbounded} queries}
A PCP must have super-polynomial length if it ensures SZK against any polynomial number of malicious queries: if not, a malicious verifier could read the entire PCP, in which case zero knowledge is impossible for non-trivial languages \cite{GoldreichO94}. If one allows the prover to be inefficient, then invoking \cite{KilianPT97}'s result for any language in $\NEXP$, including $\NP$ languages, suffices. Yet, in the case of $\NP$, one can still aim for \emph{oracle efficiency}: the prover outputs a succinct representation of the oracle, i.e., a polynomial-size circuit that, given an index, outputs the value at that index. However, \cite{IshaiMS12,MahmoodyX13,IshaiMSX15} show that languages with oracle-efficient PCPs that are SZK against unbounded queries are contained in the complexity class of SZP IPs, which is unlikely to contain $\NP$.

%%%%%%%%%%%%%%%%%%%%%%%%%%%%%%%%%%%%%%%%
\parhead{SZK PCP for $\NP$ against \sunderline{bounded} queries}
\cite{KilianPT97} obtain PCPs for $\NP$ that are SZK against $\QueryBound$ malicious queries, for a given polynomially-bounded function $\QueryBound$. The construction is analogous to the one for $\NEXP$, but with different parameter choices. (The simplifications in \cite{IshaiMS12,IshaiMSX15,IshaiW14} also apply to this case.)

Subsequently, \cite{IshaiW14} consider the case of zero knowledge PCPs \emph{of proximity}; they obtain PCPPs for $\NP$ that are SZK against $\QueryBound$ malicious queries. Like \cite{KilianPT97}, their construction has two steps:
\begin{inparaenum}[(1)]
  \item use MPC techniques to transform a given PCPP into a new one that is PZK against (several independent copies of) the honest verifier;
  \item use locking schemes to transform the latter PCPP into a new one that is SZK against malicious verifiers.
\end{inparaenum}

%%%%%%%%%%%%%%%%%%%%%%%%%%%%%%%%%%%%%%%%
\parhead{SZK IPCP for $\NP$ against \sunderline{unbounded} queries}
For an IPCP to ensure SZK against any polynomial number of queries, the prover must send a PCP with super-polynomial length: if not, a malicious verifier could read the entire PCP, forcing the IPCP model to ``collapse'' to IP (recall that the complexity class of SZK IPs is unlikely to contain $\NP$). As in the PCP model, one may still aim for oracle efficiency, and this time no limitations apply because a positive result is known: \cite{GoyalIMS10} obtain oracle-efficient IPCPs for $\NP$ that are SZK against unbounded queries. Their construction is analogous to \cite{KilianPT97}'s, but relies on \emph{interactive locking schemes} in the IPCP model, rather than non-interactive ones in the PCP model; this circumvents the impossibility result for oracle-efficient PCPs.

\doclearpage
%%%%%%%%%%%%%%%%%%%%%%%%%%%%%%%%%%%%%%%%%%%%%%%%%%%%%%%%%%%%%%%%%%%%%%%%%%%%%%%%
%%%%%%%%%%%%%%%%%%%%%%%%%%%%%%%%%%%%%%%%%%%%%%%%%%%%%%%%%%%%%%%%%%%%%%%%%%%%%%%%
%%%%%%%%%%%%%%%%%%%%%%%%%%%%%%%%%%%%%%%%%%%%%%%%%%%%%%%%%%%%%%%%%%%%%%%%%%%%%%%%
\section{Proof of \lemref{lem:efficient-codeword-simulator}}
\label{sec:proof-of-efficient-codeword-simulator}

The algorithm $\CodeSimAlgorithm$, given $(\CodeIdx, S, \alpha)$, where $S = \{(\alpha_{1}, \beta_{1}), \dots, (\alpha_{\ListSize}, \beta_{\ListSize})\} \subseteq \EvaluationDomain(\CodeIdx) \times \Field(\CodeIdx)$ and $\alpha \in \EvaluationDomain(\CodeIdx)$, works as follows:
\begin{inparaenum}[(1)]
\item run $\CodeClass$'s constraint detector on input $(\CodeIdx, \{ \alpha_{1}, \dots, \alpha_{\ListSize}, \alpha \})$;
\item if the detector outputs an empty basis or a basis $z_{1}, \dots, z_{d}$ where $z_{i}(\alpha) = 0$ for all $i$, then output a random element in $\Field(\CodeIdx)$;
\item if the detector outputs some basis element $z_{j}$ where $z_{j}(\alpha) \neq 0$, then output $-\sum_{i=1}^{\ListSize} \frac{z_{j}(\alpha_{i})}{z_{j}(\alpha)} \beta_{i}$.
\end{inparaenum}
The stated time complexity of $\CodeSimAlgorithm$ is clear from its construction. We now argue correctness. Define the probability
\begin{equation*}
p \DefineEqual \Pr\limits_{\Codeword \gets \Code_{\CodeIdx}}
\left[
\Codeword(\alpha) = \beta
\,
\middle\vert
\,
\begin{array}{c}
\Codeword(\alpha_{1}) = \beta_{1} \\
\vdots \\
\Codeword(\alpha_{\ListSize}) = \beta_{\ListSize}
\end{array}
\right]
\enspace.
\end{equation*}
\begin{mdframed}
{\small
\begin{uclaim}
\begin{inparaenum}[(A)]
\item \label{cond:lindep} If there exist $a_{1}, \dots, a_{\ListSize} \in \Field(\CodeIdx)$ such that $\Codeword(\alpha) = \sum_{i=1}^{\ListSize} a_{i} \Codeword(\alpha_{i})$ for all $\Codeword \in \Code_{n}$ \emph{(\condref{cond:lindep})}, then $p = 1$ if $\beta = \sum_{i=1}^{\ListSize} a_{i} \beta_{i}$ and $p=0$ otherwise.
\item If no such $a_{1}, \dots, a_{\ListSize}$ exist, then $p = \frac{1}{\SetCardinality{\Field(\CodeIdx)}}$.
\end{inparaenum}
\end{uclaim}
\begin{proof}[Proof of claim]
If \condref{cond:lindep} holds, then, for any $\Codeword \in \Code_{n}$ such that $\Codeword(\alpha_{1}) = \beta_{1}, \dots, \Codeword(\alpha_{\ListSize}) = \beta_{\ListSize}$, it holds that $\Codeword(\alpha) = \sum_{i=1}^{\ListSize} a_{i} \Codeword(\alpha_{i}) = \sum_{i=1}^{\ListSize} a_{i} \beta_{i}$, which proves the first part of the claim.

Next, let $d \DefineEqual \Dimension{\Code_{n}}$ and let $\Codeword_{1}, \dots, \Codeword_{d}$ be a basis of $\Code_{n}$. Define $\phi_{\alpha} \DefineEqual (\Codeword_{1}(\alpha), \dots, \Codeword_{d}(\alpha))$. We argue that \condref{cond:lindep} holds if and only if $\phi_{\alpha} \in \Span(\phi_{\alpha_{1}}, \dots, \phi_{\alpha_{\ListSize}})$:
\begin{itemize}

  \item Suppose that \condref{cond:lindep} holds. Then $\Codeword_{j}(\alpha) = \sum_{i=1}^{\ListSize} a_{i} \Codeword_{j}(\alpha_{i})$ for every $j \in \{1, \dots, d\}$. Since $\Codeword_{j}(\alpha)$ is the $j$-th coordinate of $\phi_{\alpha}$, it also holds that $\phi_{\alpha} = \sum_{i=1}^{\ListSize} a_{i} \phi_{\alpha_{i}}$, so that $\phi_{\alpha} \in \Span(\phi_{\alpha_{1}}, \dots, \phi_{\alpha_{\ListSize}})$.

  \item Suppose that $\phi_{\alpha} \in \Span(\phi_{\alpha_{1}}, \dots, \phi_{\alpha_{\ListSize}})$. Then there exist $a_{1}, \dots, a_{\ListSize}$ such that $\phi_\alpha = \sum_{i=1}^{\ListSize} a_{i} \phi_{\alpha_{i}}$. For any $\Codeword \in \Code_{n}$, we can write $\Codeword = \sum_{j=1}^{d} b_{j} \Codeword_{j}$ (for some $b_{j}$'s), so that $\Codeword(\alpha) = \sum_{j=1}^d b_{j} \Codeword_{j}(\alpha) = \InnerProduct{\Codeword}{\phi_{\alpha}} = \sum_{i=1}^{\ListSize} a_{i} \InnerProduct{\Codeword}{\phi_{\alpha_{i}}} = \sum_{i=1}^{\ListSize} a_{i} \Codeword(\alpha_{i})$.

\end{itemize}
Thus, the negation of \condref{cond:lindep} is equivalent to $\phi_{\alpha} \notin \Span(\phi_{\alpha_{1}}, \dots, \phi_{\alpha_{\ListSize}})$, which we now assume to prove the second part of the claim, as follows.

Let $\Phi \in \Field(\CodeIdx)^{\ListSize \times d}$ be the matrix whose rows are $\phi_{\alpha_{1}}, \dots, \phi_{\alpha_{\ListSize}}$, and let $\Codeword'_{1}, \dots, \Codeword'_{k}$ be a basis for $\Phi$'s nullspace. Let $\Phi'$ be the matrix $\Phi$ augmented with the row $\phi_{\alpha}$. Note that $\rank(\Phi') = \rank(\Phi) + 1$, so the nullspace of $\Phi'$ has dimension $k - 1$, which implies that there exists $j \in \{1, \dots, k\}$ such that $\InnerProduct{\Codeword'_{j}}{\phi_{\alpha}} \neq 0$. Also note that, for every $\Codeword \in \Code_{n}$ such that $\Codeword(\alpha_{1}) = \beta_{1}, \dots, \Codeword(\alpha_{\ListSize}) = \beta_{\ListSize}$ and $r \in \Field(\CodeIdx)$, the codeword $\Codeword + r \Codeword'_{j}$ satisfies the same equations as $\Codeword$ does. Therefore, if $\Codeword$ is drawn uniformly randomly from $\Code_{n}$ such that $\Codeword(\alpha_{1}) = \beta_{1}, \dots, \Codeword(\alpha_{\ListSize}) = \beta_{\ListSize}$, then $\Codeword + r \Codeword'_{j}$ for $r$ uniformly random in $\Field(\CodeIdx)$ is identically distributed to $\Codeword$. We conclude that $\Pr[\Codeword(\alpha) = \beta] = \Pr[(\Codeword + r \Codeword'_{j})(\alpha) = \beta] = \Pr[r = \frac{\beta - \InnerProduct{\Codeword}{\phi_{\alpha}}}{\InnerProduct{\Codeword'_{j}}{\phi_{\alpha}}}] = \frac{1}{\SetCardinality{\Field(\CodeIdx)}}$, since $\InnerProduct{\Codeword'_{j}}{\phi_{\alpha}} \neq 0$.
\end{proof}
}
\end{mdframed}
By the definition of constraint detection, $a_{1}, \dots, a_{\ell}$ as above exist if and only if there exists $z$ in the space output by the constraint detector such that $z(\alpha) = 1$. If the constraint detector outputs $z_{1}, \dots, z_{d}$ such that $z_{i}(\alpha) = 0$ for all $i$, then clearly the space contains no such vector. Otherwise, let $j$ be such that $z_{j}(\alpha) \neq 0$; then $a_{i} = -z_{j}(\alpha_{i})/z_{j}(\alpha)$ for $i = 1, \dots, \ListSize$ is a solution. Hence this distribution equals that of $\CodeSimAlgorithm$'s output, and moreover fully describes the probability distribution of $\Codeword(\alpha)$. The lemma follows.

\doclearpage
%%%%%%%%%%%%%%%%%%%%%%%%%%%%%%%%%%%%%%%%%%%%%%%%%%%%%%%%%%%%%%%%%%%%%%%%%%%%%%%%
%%%%%%%%%%%%%%%%%%%%%%%%%%%%%%%%%%%%%%%%%%%%%%%%%%%%%%%%%%%%%%%%%%%%%%%%%%%%%%%%
%%%%%%%%%%%%%%%%%%%%%%%%%%%%%%%%%%%%%%%%%%%%%%%%%%%%%%%%%%%%%%%%%%%%%%%%%%%%%%%%
\section{Proof of \lemref{lem:bsrs-dual-char}}
\label{sec:proof-of-bsrs-dual-char}

By \clmref{claim:bsrs-dual-of-projection}, it suffices to show an algorithm that computes a basis of $\Puncture{(\Dual{\Code_{\CodeIdx}})}{\IdxSet}$ in $\poly(\BitSize{\CodeIdx} + \SetCardinality{\IdxSet})$ time. So consider the algorithm that, on input an index $\CodeIdx$ and subset $\IdxSet \subseteq \Domain(\CodeIdx)$, works as follows. First, invoke the hypothesis to compute the set $W$; since vectors are represented sparsely we conclude that $\SetCardinality{W},\SetCardinality{\Support{W}} \leq \poly(\BitSize{\CodeIdx} + \SetCardinality{\IdxSet})$. (Recall that $\Support{W} \DefineEqual \cup_{\OtherCodeword \in W} \Support{\OtherCodeword}$.) We may assume $W$ is linearly independent; otherwise, make it thus via Gaussian elimination which runs in time $\poly(\SetCardinality{W}+\SetCardinality{\Support{W}})$. Similarly, the bound on $\SetCardinality{W}$ and $\SetCardinality{\Support{W}}$ implies that a basis $W'$ for the subspace $\Puncture{W}{\IdxSet}$ can be found in time $\poly(\BitSize{\CodeIdx} + \SetCardinality{\IdxSet})$, and we let $W'$ be the output of our algorithm.

To argue correctness it suffices to show that $\Span(W')=\Puncture{(\Dual{\Code_{\CodeIdx}})}{\IdxSet}$. We first argue $\Span(W') \subseteq \Puncture{(\Dual{\Code_{\CodeIdx}})}{\IdxSet}$, so let $\OtherCodeword' \in \Span(W')$, which can be represented as $\OtherCodeword' = \sum_{\lambda \in \Lambda} a_{\lambda} \sum_{\OtherCodeword \in W} \lambda(\OtherCodeword) \cdot \OtherCodeword$; note that $\OtherCodeword' \in \Span(W) \subseteq \Dual{\Code_{\CodeIdx}}$ and $\Support{\OtherCodeword'} \subseteq \Support{W} = \IdxSet \cup \Complement{\IdxSet} $. Hence, it suffices to show that $\Support{\OtherCodeword'} \cap \Complement{\IdxSet} = \emptyset$; but this is true by the choice of $\Lambda$, because $M \cdot \lambda = 0$ for every $\lambda \in \Lambda$, so that $\sum_{\OtherCodeword \in W} \lambda(\Codeword) \cdot \OtherCodeword(\alpha) = 0$ for every $\alpha \in \Complement{\IdxSet}$ (by $M$'s definition), so that $\OtherCodeword'(\alpha) = \sum_{\lambda \in \Lambda} a_{\lambda} \sum_{\OtherCodeword \in W} \lambda(\OtherCodeword) \cdot \OtherCodeword(\alpha) = 0$ for every $\alpha \in \Complement{\IdxSet}$, as required.

We next argue that $\Span(W') \supseteq \Puncture{(\Dual{\Code_{\CodeIdx}})}{\IdxSet}$, and for this it suffices to show that any $\Codeword \in \Span(W)$ having representation $\Codeword = \sum_{\OtherCodeword \in W} a_{\OtherCodeword} \cdot \OtherCodeword$ such that $\vec{a} \DefineEqual \left( a_{\OtherCodeword} \right)_{\OtherCodeword \in W} \notin \Span(\Lambda)$ can not be in $\Puncture{(\Dual{\Code_{\CodeIdx}})}{\IdxSet}$. This follows by the definition of $\Lambda$, because for any $\vec{a} \notin \Span(\Lambda)$ there exists $\alpha \in \Complement{\IdxSet}$ such that $\Codeword(\alpha) = \sum_{\OtherCodeword \in W} a_{\OtherCodeword} \cdot \OtherCodeword(\alpha) \neq 0$, so that $\Codeword \notin \Puncture{(\Dual{\Code_{\CodeIdx}})}{\IdxSet}$.

\doclearpage
%%%%%%%%%%%%%%%%%%%%%%%%%%%%%%%%%%%%%%%%%%%%%%%%%%%%%%%%%%%%%%%%%%%%%%%%%%%%%%%%
%%%%%%%%%%%%%%%%%%%%%%%%%%%%%%%%%%%%%%%%%%%%%%%%%%%%%%%%%%%%%%%%%%%%%%%%%%%%%%%%
%%%%%%%%%%%%%%%%%%%%%%%%%%%%%%%%%%%%%%%%%%%%%%%%%%%%%%%%%%%%%%%%%%%%%%%%%%%%%%%%
\section{Proof of \lemref{lem:polydep-det}}
\label{sec:proof-of-deterministic-algorithm}

For completeness, we give an elementary proof of \lemref{lem:polydep-det}, by simplifying the proof of \cite[Thm.~10]{Kayal10} for polynomials of the form we require; note that \cite{RazS05} and \cite{BogdanovW04} also use similar techniques. We first introduce some notation. We consider a polynomial $\Poly \in \PolynomialRingIndOne{\Field}{\SCVars}{\VariableX}{\SCDegree}$ equivalently as a univariate polynomial of degree less than $\SCDegree$ in $\VariableX_{1}$ with coefficients in $ \PolynomialRingIndTwo{\Field}{\SCVars}{\VariableX}{\SCDegree}$, and let $\partial^{j}_{1} \Poly$ be the coefficient of $\VariableX_{1}^{j}$ in this representation. Define $\partial^{j}_{1} \vec{\Poly} \DefineEqual (\partial^{j}_{1} \Poly_{1}, \dots, \partial^{j}_{1} \Poly_{\ListSize})$. In general, given an arbitrary arithmetic circuit representing a polynomial $\Poly$, it is not clear how to efficiently compute a circuit representing $\partial^{j}_{1} \Poly$, because $\Poly$ may have exponentially many monomials. Nevertheless, for circuits of the required form, this computation is trivial.

\begin{mdframed}
{\small
\begin{uclaim}
Let $\vec{\Poly} \DefineEqual (\Poly_{1},\dots,\Poly_{\ListSize})$ be a vector of polynomials in $\PolynomialRingIndOne{\Field}{\SCVars}{\VariableX}{\SCDegree}$. If $\SCDegree \leq \SetCardinality{\Field}$ then $\Dual{\vec{\Poly}} = \bigcap_{j=0}^{\SCDegree-1} \Dual{(\partial^{j}_{1} \vec{\Poly})}$.
\end{uclaim}

\begin{proof}
When $\SCDegree \leq \SetCardinality{\Field}$, $\Poly \in \PolynomialRingIndOne{\Field}{\SCVars}{\VariableX}{\SCDegree} \equiv 0$ if and only if all of its coefficients are zero when written as a formal sum. Then one direction of the set equality follows straightforwardly from the linearity of $\partial^{j}_{1}$, namely, $\Dual{\vec{\Poly}} \subseteq \bigcap_{j=0}^{\SCDegree-1} \Dual{(\partial^{j}_{1} \vec{\Poly})}$. For the other direction, we argue as follows. Fix some $(a_{1}, \dots, a_{\ListSize}) \in \bigcap_{j=0}^{\SCDegree-1} \Dual{(\partial^{j}_{1} \vec{\Poly})}$ and let $\OtherPoly \DefineEqual \sum_{k=1}^{\ListSize} a_{k} \Poly_{k}$; we have that $\partial^{j}_{1} \OtherPoly \equiv 0$ for all $j \in \{0, \dots, \SCDegree-1\}$, by linearity. But $\OtherPoly = \sum_{j=0}^{d-1} (\partial^{j}_{1} \OtherPoly) \VariableX_{1}^{j}$ by definition, so $\OtherPoly \equiv 0$, and thus $(a_{1}, \dots, a_{\ListSize}) \in \Dual{\vec{\Poly}}$.
\end{proof}
}
\end{mdframed}
Thus to compute a basis of $\Dual{\vec{\Poly}}$ it suffices to compute the intersection of the bases of $\Dual{(\partial^{j}_{1} \vec{\Poly})}$ for all $j \in \{0, \dots, d-1\}$. The naive approach yields an exponential-time algorithm since we reduce the problem to $d$ subproblems of roughly the same size. Observe, however, that for $\Poly_{k}$ of the specified form,
\begin{equation*}
 \partial^{j}_{1} \Poly_{k}
=  c_{k,j} \OtherPoly_{k}
\quad
\text{ where }
\OtherPoly_{k} \DefineEqual \left( \prod_{i=2}^{\SCVars} \Poly_{k,i}(\VariableX_{i}) \right)
\enspace,
\end{equation*}
for constants $c_{k,j}$ computable in time $\poly(\CircuitSize)$. Let $\vec{\OtherPoly} \DefineEqual (\OtherPoly_{1}, \dots, \OtherPoly_{\ListSize})$ and let $\Dual{\OtherPoly} \in \Field^{\ListSize \times b}$ be a basis for $\Dual{\vec{\OtherPoly}}$; note that $b \leq \ListSize$. Let $\vec{a} \DefineEqual (a_{1}, \dots, a_{\ListSize}) \in \Field^{\ListSize}$, and observe that for each $j$, $\sum_{k=1}^{\ListSize} a_{k} \partial^{j}_{1} \Poly_{k} \equiv 0$ if and only if $\sum_{k=1}^{\ListSize} a_{k} c_{k,j} \OtherPoly_{k} \equiv 0$, or equivalently, $(a_{1} c_{1,j}, \dots, a_{\ListSize} c_{\ListSize,j}) \in \Dual{\vec{\OtherPoly}}$. Hence $(a_{1}, \dots, a_{\ListSize}) \in \bigcap_{j=0}^{\SCDegree} \Dual{(\partial^{j}_{1} \vec{\Poly})}$ if and only if for each $j$ there exists $\vec{v_{j}} \in \Field^{b}$ such that $\Dual{\OtherPoly} \vec{v_{j}} = (a_{1} c_{1,j}, \dots, a_{\ListSize} c_{\ListSize,j})$. This is a system of linear equations in $\vec{a}, \vec{v_{0}}, \dots, \vec{v_{\SCDegree-1}}$ of size $\poly(\ListSize + \SCDegree + b)$, and hence we can compute a basis for its solution space in time $\poly(\log \SetCardinality{\Field} + \SCDegree + \ListSize + b)$. Restricting this basis to $\vec{a}$ yields a basis for $\Dual{\vec{\Poly}}$.

If $\Poly_{1}, \dots, \Poly_n$ are univariate then we can easily determine a basis for $\Dual{\vec{\Poly}}$ in deterministic polynomial time (by Gaussian elimination). Otherwise, if the $\Poly_{i}$ are $\SCVars$-variate, we use the procedure above to reduce computing $\Dual{\vec{\Poly}}$ to computing $\Dual{\vec{\OtherPoly}}$ for some $\vec{\OtherPoly} = (\OtherPoly_{1}, \dots, \OtherPoly_{\ListSize})$ where the $\OtherPoly_{i}$ are $(m-1)$-variate. This algorithm terminates in time $\poly(\log \SetCardinality{\Field} + \SCVars + \SCDegree + \CircuitSize + \ListSize)$.

\doclearpage
%%%%%%%%%%%%%%%%%%%%%%%%%%%%%%%%%%%%%%%%%%%%%%%%%%%%%%%%%%%%%%%%%%%%%%%%%%%%%%%%
%%%%%%%%%%%%%%%%%%%%%%%%%%%%%%%%%%%%%%%%%%%%%%%%%%%%%%%%%%%%%%%%%%%%%%%%%%%%%%%%
%%%%%%%%%%%%%%%%%%%%%%%%%%%%%%%%%%%%%%%%%%%%%%%%%%%%%%%%%%%%%%%%%%%%%%%%%%%%%%%%
\section{Proof of \clmref{claim:bsrs-independent-dual}}
\label{sec:proof-of-independent-dual}

First we show that $\Span(\cup_{j \in J} \Dual{\CCode_{j}}) \subseteq \Puncture{(\Dual{\Code})}{\left( \cup_{j \in J} \CDomain_{j} \right)}$. For every $j \in J$ and $\OtherCodeword \in \Dual{\CCode_{j}}$, it holds that $\Support{\OtherCodeword} \subseteq \CDomain_{j}$; therefore, for every $\OtherCodeword \in \Span(\cup_{j \in J} \Dual{\CCode_{j}})$, it holds that $\Support{\OtherCodeword} \subseteq \cup_{j \in J} \CDomain_{j}$; thus it is suffices to show that, for every $\OtherCodeword \in \cup_{j \in J} \Dual{\CCode_{j}}$ and $\Codeword \in \Code$, it holds that $\InnerProduct{\Codeword}{\OtherCodeword} = 0$. But this holds because for every $\OtherCodeword \in \cup_{j \in J} \Dual{\CCode_{j}}$ there exists $j \in J$ such that $\OtherCodeword \in\Dual{\CCode_{j}}$ and $\Restrict{\Code}{\CDomain_{j}} = \CCode_{j}$ so that $\InnerProduct{\Codeword}{\OtherCodeword} = \InnerProduct{\Restrict{\Codeword}{\CDomain_{j}}}{\OtherCodeword} = 0$, as required.

Next we show that $\Span(\cup_{j \in J} \Dual{\CCode_{j}}) \supseteq \Puncture{(\Dual{\Code})}{\left( \cup_{j \in J} \CDomain_{j} \right)}$, which is equivalent to $\Span(\cup_{j \in J} \Dual{\CCode_{j}}) \supseteq \Dual{(\Restrict{\Code}{\cup_{j \in J} \CDomain_{j}})}$ by \clmref{claim:bsrs-dual-of-projection}. Recall that for any two linear spaces $U,V$ it holds that $U \subseteq V$ if and only if $\Dual{U} \supseteq \Dual{V}$, thus it is sufficient to show that $\Dual{\Span(\cup_{j \in J} \Dual{\CCode_{j}})} \subseteq \Restrict{\Code}{\cup_{j \in J} \CDomain_{j}}$, i.e., that every $\Codeword\in \Dual{\Span(\cup_{j \in J} \Dual{\CCode_{j}})}$ can be extended to $w'\in\Code$. This latter statement holds because $\Restrict{\Dual{\Span(\cup_{j \in J} \Dual{\CCode_{j}})}}{\CDomain_{j}} \subseteq \Dual{(\Dual{\CCode_{j}})} = \CCode_{j}$ for every $j \in J$, and thus $\Restrict{\Codeword}{\CDomain_{j}}\in \CCode_{j}$. Recalling $\SetCardinality{J} \leq \IdpParam$ implies, by \defref{def:bsrs-local-cover}, that $\Codeword$ can be extended to a codeword $\Codeword' \in \Code$, as claimed.

\doclearpage
%%%%%%%%%%%%%%%%%%%%%%%%%%%%%%%%%%%%%%%%%%%%%%%%%%%%%%%%%%%%%%%%%%%%%%%%%%%%%%%%
%%%%%%%%%%%%%%%%%%%%%%%%%%%%%%%%%%%%%%%%%%%%%%%%%%%%%%%%%%%%%%%%%%%%%%%%%%%%%%%%
%%%%%%%%%%%%%%%%%%%%%%%%%%%%%%%%%%%%%%%%%%%%%%%%%%%%%%%%%%%%%%%%%%%%%%%%%%%%%%%%
\section{Definition of the linear code family $\BSCode$}
\label{sec:bsrs-formal-definitions}

In this section we define the linear code family $\BSCode$, which consists of evaluations of univariate polynomials concatenated with corresponding BS proximity proofs \cite{BS08}. The definition is quite technical, and we refer the interested reader to \cite{BS08} for a discussion of why it enables proximity testing. We begin with notation used later.

\begin{definition}
\label{def:bsrs-basic-notation}
Given a field $\Field$, a subfield $\SubField \subseteq \Field$, a $\SubField$-linear space $\BSSpace \subseteq \Field$ with a basis $(b_{1},b_{2},\dots,b_{\BSSpaceDimSuperLocal})$, a positive integer $\BSBalance$, and a positive integer $\BSBaseDim > 2\BSBalance$, we make the following definitions.
\begin{itemize}

\item Four subspaces of $\BSSpace$ and a subset of $\BSSpace$:
\begin{align*}
\BSSpace_{0}[\SubField,\Field,\BSSpace,\BSBalance]
  \DefineEqual& \Span_\SubField(b_{1},b_2,\dots, b_{\lfloor \BSSpaceDimSuperLocal/2 \rfloor}) \\
\BSSpace'_{0}[\SubField,\Field,\BSSpace,\BSBalance]
  \DefineEqual& \Span_\SubField(b_{1},b_2,\dots, b_{\lfloor \BSSpaceDimSuperLocal/2 \rfloor + \BSBalance - 1}) \\
\BSSpace_{1}[\SubField,\Field,\BSSpace,\BSBalance]
  \DefineEqual& \Span_\SubField(b_{\lfloor \BSSpaceDimSuperLocal/2 \rfloor + 1}, \dots , b_\BSSpaceDimSuperLocal) \\
\forall\, \beta \in \BSSpace_{1}[\SubField,\Field,\BSSpace,\BSBalance]\,,\, \BSSpace_\beta[\SubField,\Field,\BSSpace,\BSBalance]
  \DefineEqual& \Span_\SubField(b_{1},b_2,\dots, b_{\lfloor \BSSpaceDimSuperLocal/2 \rfloor + \BSBalance - 1}, \beta') \\
\forall\, \beta \in \BSSpace_{1}[\SubField,\Field,\BSSpace,\BSBalance]\,,\, R_\beta[\SubField,\Field,\BSSpace,\BSBalance]
  \DefineEqual& \BSSpace_\beta \setminus (\BSSpace_{0} + \beta)
\end{align*}
where $\beta' \DefineEqual b_{\lfloor \BSSpaceDimSuperLocal/2 \rfloor + \BSBalance}$ if $\beta \in \BSSpace'_{0}$ and $\beta' \DefineEqual \beta$ otherwise.

\item The vanishing polynomial of $\BSSpace_{0}$: $\VanishingPoly{\BSSpace_{0}}[\SubField,\Field,\BSSpace,\BSBalance](\VariableX) \DefineEqual \prod_{\alpha \in \BSSpace_{0}} (\VariableX - \alpha)$.

\item The following domains:
\begin{align*}
\MakeF{\Domain}[\SubField,\Field,\BSSpace,\BSBalance] \DefineEqual& \Set{ (\alpha ,\VanishingPoly{\BSSpace_{0}}(\beta)) : \beta \in \BSSpace_{1} , \alpha \in \BSSpace_\beta } \\
\MakeP{\Domain}[\SubField,\Field,\BSSpace,\BSBalance] \DefineEqual& \Set{ (\alpha ,\VanishingPoly{\BSSpace_{0}}(\beta)) : \beta \in \BSSpace_{1} , \alpha \in R_\beta } \\
\MakeBox{\Domain}[\SubField,\Field,\BSSpace,\BSBalance] \DefineEqual& \left( \Set{\SymbolRS} \times \BSSpace \right) \sqcup \left( \Set{\SymbolProx} \times \MakeP{\Domain} \right)
\end{align*}
where we use the symbols `$\SymbolRS$' and `$\SymbolProx$' to distinguish different parts of the disjoint union.

\item The bijection $\phi[\SubField,\Field,\BSSpace,\BSBalance] \colon \MakeF{\Domain} \to \MakeBox{\Domain}$ is defined by
$\phi(\alpha, \beta) \DefineEqual
\begin{cases}
(\SymbolProx,(\alpha,\beta)) & (\alpha, \beta) \in \MakeP{\Domain} \\
(\SymbolRS,\alpha) & \text{otherwise}
\end{cases}
$.

\item Given $\Codeword \in \Field^{\MakeBox{\Domain}[\SubField,\Field,\BSSpace,\BSBalance]}$, the bivariate function $f_{\Codeword} \colon \MakeF{\Domain}[\SubField,\Field,\BSSpace,\BSBalance] \to \Field$ is defined by $f_{\Codeword}(\alpha,\beta) \DefineEqual \Codeword(\phi(\alpha,\beta))$.

\item The fractional degree $\rho[\SubField,\Field,\BSBalance] \DefineEqual \SetCardinality{\SubField}^{-\BSBalance}$.

\item The domain $\MakeProx{\Domain^{\BSCode}}[\SubField,\Field,\BSSpace,\BSBalance,\BSBaseDim]$ implied by the recursion below:
\begin{itemize}

\item if $\Dimension{\BSSpace} \leq \BSBaseDim$ then
$\MakeProx{\Domain^{\BSCode}}[\SubField,\Field,\BSSpace,\BSBalance,\BSBaseDim] \DefineEqual \MakeP{\Domain}[\SubField,\Field,\BSSpace,\BSBalance]$;

\item if $\Dimension{\BSSpace} > \BSBaseDim$ then
\begin{align*}
\MakeProx{\Domain^{\BSCode}}[\SubField,\Field,\BSSpace,\BSBalance,\BSBaseDim]
\DefineEqual \MakeP{\Domain}[\SubField,\Field,\BSSpace,\BSBalance]
\;\bigsqcup\;&
\Big( \sqcup_{\alpha \in \BSSpace_{0}'} \Set{(\SymbolCol,\alpha)} \times \MakeProx{\Domain^{\BSCode}}[\SubField,\Field,\VanishingPoly{\BSSpace_{0}}(\BSSpace_{1}),\BSBalance,\BSBaseDim] \Big) \\
\;\bigsqcup\;&
\Big( \sqcup_{\beta \in \BSSpace_{1}} \Set{(\SymbolRow,\beta)} \times \MakeProx{\Domain^{\BSCode}}[\SubField,\Field,\BSSpace_\beta,\BSBalance,\BSBaseDim] \Big)
\enspace.
\end{align*}
where we use the symbols `$\SymbolCol$' and `$\SymbolRow$' to distinguish different parts of the disjoint union.

\end{itemize}

  \item The domain $\Domain^{\BSCode}[\SubField,\Field,\BSSpace,\BSBalance,\BSBaseDim] \DefineEqual \left( \Set{\SymbolRS} \times \BSSpace \right) \sqcup \left( \Set{\SymbolProx} \times \MakeProx{\Domain^{\BSCode}}[\SubField,\Field,\BSSpace,\BSBalance,\BSBaseDim] \right)$.

  \item Given $\alpha \in \BSSpace'_{0}$, the embedding
  $\phi_{\SymbolCol,\alpha} \colon \Domain^{\BSCode}[\SubField,\Field,\VanishingPoly{\BSSpace_{0}}(\BSSpace_{1}),\BSBalance,\BSBaseDim] \hookrightarrow \Domain^{\BSCode}[\SubField,\Field,\BSSpace,\BSBalance,\BSBaseDim]$ is defined by
  \begin{equation*}
  \phi_{\SymbolCol,\alpha}(x) \DefineEqual
  \begin{cases}
  (\SymbolProx,((\SymbolCol,\alpha),x)) & x \in \Set{\SymbolProx} \times \MakeProx{\Domain^{\BSCode}}[\SubField,\Field,\VanishingPoly{\BSSpace_{0}}(\BSSpace_{1}),\BSBalance,\BSBaseDim] \\
  \phi(\alpha,\beta) & x = (\Set{\SymbolRS},\VanishingPoly{\BSSpace_{0}}(\beta))
  \end{cases}
  \end{equation*}
  We denote by $\Domain_{\SymbolCol,\alpha}$ the image of $\phi_{\SymbolCol,\alpha}$.

  \item Given $\beta \in \BSSpace_{1}$, the embedding
  $\phi_{\SymbolRow,\beta} \colon \Domain^{\BSCode}[\SubField,\Field,\BSSpace_{\beta},\BSBalance,\BSBaseDim] \hookrightarrow \Domain^{\BSCode}[\SubField,\Field,\BSSpace,\BSBalance,\BSBaseDim]$ is defined by
  \begin{equation*}
  \phi_{\SymbolRow,\beta}(x) \DefineEqual
  \begin{cases}
  (\SymbolProx,((\SymbolRow,\beta),x)) & x \in \Set{\SymbolProx} \times \MakeProx{\Domain^{\BSCode}}[\SubField,\Field,\BSSpace_{\beta},\BSBalance,\BSBaseDim] \\
  \phi(\alpha,\beta) & x = (\Set{\SymbolRS},\alpha)
  \end{cases}
  \end{equation*}
  We denote by $\Domain_{\SymbolRow,\beta}$ the image of $\phi_{\SymbolRow,\beta}$.

  \item Given $\alpha \in \BSSpace'_{0}$,
  $\psi_{\SymbolCol,\alpha} \colon \Field^{\Domain^{\BSCode}[\SubField,\Field,\BSSpace,\BSBalance,\BSBaseDim]} \to \Field^{\Domain^{\BSCode}[\SubField,\Field,\VanishingPoly{\BSSpace_{0}}(\BSSpace_{1}),\BSBalance,\BSBaseDim]}$
  is the projection of $\Domain^{\BSCode}[\SubField,\Field,\BSSpace,\BSBalance,\BSBaseDim]$
  on
  $\Domain_{\SymbolCol,\alpha}$ with indices renamed to elements of
  $\Domain^{\BSCode}[\SubField,\Field,\VanishingPoly{\BSSpace_{0}}(\BSSpace_{1}),\BSBalance,\BSBaseDim]$.
  Formally, $\psi_{\SymbolCol,\alpha}(\Codeword) = \Codeword'$ if and only if
  $\Codeword'(\phi_{\SymbolCol,\alpha}(x)) = \Codeword(x)$ for all
  $x \in \Domain^{\BSCode}[\SubField,\Field,\VanishingPoly{\BSSpace_{0}}(\BSSpace_{1}),\BSBalance,\BSBaseDim]$.

  \item Given $\beta \in \BSSpace_{1}$,
  $\psi_{\SymbolRow,\beta} \colon \Field^{\Domain^{\BSCode}[\SubField,\Field,\BSSpace,\BSBalance,\BSBaseDim]} \to \Field^{\Domain^{\BSCode}[\SubField,\Field,\BSSpace_{\beta},\BSBalance,\BSBaseDim]}$
  is the projection of $\Domain^{\BSCode}[\SubField,\Field,\BSSpace,\BSBalance,\BSBaseDim]$
  on
  $\Domain_{\SymbolRow,\beta}$ with indices renamed to elements of
  $\Domain^{\BSCode}[\SubField,\Field,\BSSpace_{\beta},\BSBalance,\BSBaseDim]$.
  Formally, $\psi_{\SymbolRow,\beta}(\Codeword) = \Codeword'$ if and only if
  $\Codeword'(\phi_{\SymbolRow,\beta}(x)) = \Codeword(x)$ for all
  $x \in \Domain^{\BSCode}[\SubField,\Field,\BSSpace_{\beta},\BSBalance,\BSBaseDim]$.

\end{itemize}
\end{definition}

The following definition considers a code that extends the evaluation of a univariate polynomial with a bivariate function that represents the polynomial over a specially-chosen set.

\begin{definition}[$\RSBox$]
\label{def:bs-proof-of-proximity}
Given a field $\Field$, a subfield $\SubField \subseteq \Field$, a $\SubField$-linear space $\BSSpace \subseteq \Field$, and a positive integer $\BSBalance$, the code $\RSBox[\SubField,\Field,\BSSpace,\BSBalance]$ consists of all $\Codeword \in \Field^{\MakeBox{\Domain}[\SubField,\Field,\BSSpace,\BSBalance]}$ such that $f_{\Codeword} \colon \MakeF{\Domain} \to \Field$ is an evaluation of a low degree polynomial: there exists a polynomial $g \in \Field[X,Y]$ such that:
\begin{inparaenum}[(i)]
\item $\IndividualDegree{g}[X] < \SetCardinality{\BSSpace_{0}}$,
\item $\IndividualDegree{g}[Y] < \SetCardinality{\BSSpace_{1}} \cdot \rho[\SubField,\Field,\BSBalance]$,
\item $\Restrict{g}{\MakeF{\Domain}} = f_{\Codeword}$.
\end{inparaenum}
\end{definition}

\noindent
Ben-Sasson and Sudhan \cite{BS08} show that:
\begin{itemize}

  \item $v \in \RSCode{\Field}{\BSSpace}{\SetCardinality{\BSSpace} \cdot \rho}$ if and only if there exists $\Codeword \in \RSBox[\SubField,\Field,\BSSpace,\BSBalance]$ such that $\Restrict{\Codeword}{\Set{\SymbolRS} \times \BSSpace} = v$;

  \item $\Codeword \in \RSBox[\SubField,\Field,\BSSpace,\BSBalance]$ if and only if
\begin{itemize}[nolistsep]

\item for every $\alpha \in \BSSpace'_{0}$, $\Restrict{f_{\Codeword}}{\Set{\alpha} \times \VanishingPoly{\BSSpace_{0}}(\BSSpace_{1})} \in \RSCode{\Field}{\VanishingPoly{\BSSpace_{0}}(\BSSpace_{1})}{\SetCardinality{\BSSpace_{1}} \cdot \rho}$ (with the standard mapping between domains) and

\item for every $\beta \in \BSSpace_{1}$, $\Restrict{f_{\Codeword}}{\BSSpace_\beta \times \Set{\VanishingPoly{\BSSpace_{0}}(\beta)}} \in \RSCode{\Field}{\BSSpace_\beta}{\SetCardinality{\BSSpace_{0}}}$ (with the standard mapping between domains).

\end{itemize}
\end{itemize}
The above equivalences illustrate the `quadratic reduction' from testing that $\Codeword \in \Field^{\BSSpace}$ is a codeword of $\RSCode{\Field}{\BSSpace}{\SetCardinality{\BSSpace}\cdot \rho}$ to a set of $\Theta(\sqrt{\SetCardinality{\BSSpace}})$ problems of testing membership in codes of the form $\RSCode{\Field}{\BSSpace'}{\RSDegree'}$ with $\SetCardinality{\BSSpace'},\RSDegree' = \Theta(\sqrt{\SetCardinality{\BSSpace}})$.

The code from \defref{def:bs-proof-of-proximity} corresponds to one step of the recursive construction of \cite{BS08}. We now build on that definition, and recursively define the linear code family $\BSCode$.

\begin{definition}[$\BSCode$]
\label{def:bsrs-family}
Given a field $\Field$, a subfield $\SubField \subseteq \Field$, a $\SubField$-linear space $\BSSpace \subseteq \Field$, a positive integer $\BSBalance$, and a positive integer $\BSBaseDim > 2\BSBalance$, the code $\BSCode[\SubField,\Field,\BSSpace,\BSBalance,\BSBaseDim]$ consists of all words $\Codeword \in \Field^{\Domain^{\BSCode}[\SubField,\Field,\BSSpace,\BSBalance,\BSBaseDim]}$ satisfying the following. If $\Dimension{\BSSpace} \leq \BSBaseDim$ then $\Codeword \in \RSBox[\SubField,\Field,\BSSpace,\BSBalance]$. If $\Dimension{\BSSpace} > \BSBaseDim$ the following holds:
\begin{inparaenum}[(1)]
\item for every $\alpha \in \BSSpace'_{0}$ there exists $\Codeword_{\alpha} \in \BSCode[\SubField,\Field,\VanishingPoly{\BSSpace_{0}}(\BSSpace_{1}),\BSBalance,\BSBaseDim]$ such that $\Codeword_{\alpha}(\phi_{\SymbolCol,\alpha}(x)) = \Codeword(x)$ for every $x \in \Domain[\SubField,\Field,\VanishingPoly{\BSSpace_{0}}(\BSSpace_{1}),\BSBalance,\BSBaseDim]$;
\item for every $\beta \in \BSSpace_{1}$ there exists $\Codeword_{\beta} \in \BSCode[\SubField,\Field,\BSSpace_{\beta},\BSBalance,\BSBaseDim]$ such that $\Codeword_{\beta}(\phi_{\SymbolRow,\beta}(x)) = \Codeword(x)$ for every $x \in \Domain[\SubField,\Field,\BSSpace_{\beta},\BSBalance,\BSBaseDim]$.
\end{inparaenum}
\end{definition}

We conclude this section with two claims about $\BSCode$ that we use in later sections. We omit the proof of the first claim (and refer the interested reader to \cite{BS08}), and prove the second claim based on the first one.

\begin{claim}
\label{claim:bsrs-code-cover-proof-defined-by-RS}
For every codeword $\Codeword \in \RSCode{\Field}{\BSSpace}{\SetCardinality{\BSSpace}\cdot \CodeRate}$, positive integer $\BSBalance$, and positive integer $\BSBaseDim > 2\BSBalance$, there exists a unique $\Proof_{\Codeword}$ such that $\Codeword \circ \Proof_{\Codeword} \in \BSCode[\SubField,\Field,\BSSpace,\BSBalance,\BSBaseDim]$.
\end{claim}

\begin{claim}
\label{claim:bsrs-code-cover-views}
The following two statements hold for the code $\BSCode[\SubField,\Field,\BSSpace,\BSBalance,\BSBaseDim]$:
\begin{itemize}[nolistsep]

\item for every $\alpha \in \BSSpace'_{0}$ and $\Codeword' \in \BSCode[\SubField,\Field,\VanishingPoly{\BSSpace_{0}}(\BSSpace_{1}),\BSBalance,\BSBaseDim]$ there exists $\Codeword \in \BSCode[\SubField,\Field,\BSSpace,\BSBalance,\BSBaseDim]$ such that
\mbox{$\psi_{\SymbolCol,\alpha}(\Codeword) = \Codeword'$;}

\item for every $\beta \in \BSSpace_{1}$ and $\Codeword' \in \BSCode[\SubField,\Field,\BSSpace_{\beta},\BSBalance,\BSBaseDim]$ there exists $\Codeword \in \BSCode[\SubField,\Field,\BSSpace,\BSBalance,\BSBaseDim]$ such that $\psi_{\SymbolRow,\beta}(\Codeword) = \Codeword'$.

\end{itemize}
\end{claim}

\begin{proof}
The proofs for the two statements are similar, so we only give the proof for the first statement.
Let $\Codeword' \in \BSCode[\SubField,\Field,\VanishingPoly{\BSSpace_{0}}(\BSSpace_{1}),\BSBalance,\BSBaseDim]$, and define $\Codeword_{\SymbolRS} \DefineEqual \Restrict{\Codeword'}{\Set{\SymbolRS} \times \VanishingPoly{\BSSpace_{0}}(\BSSpace_{1})}$; observe that $\Codeword_{\SymbolRS}$ in $\RSCode{\Field}{\VanishingPoly{\BSSpace_{0}}(\BSSpace_{1})}{\SetCardinality{\BSSpace_{1}} \cdot \CodeRate}$. By \clmref{claim:bsrs-code-cover-proof-defined-by-RS}, $\Codeword'$ is uniquely determined by $\Codeword_{\SymbolRS}$, thus it suffices to show that there exists $\MakeBox{\Codeword} \in \RSBox[\SubField,\Field,\BSSpace,\BSBalance]$ such that $\Restrict{f_{\MakeBox{\Codeword}}}{\Set{\alpha} \times \VanishingPoly{\BSSpace_{0}}(\BSSpace_{1})} = \Codeword_{\SymbolRS}$. By definition of $\RSBox$, it suffices to show that there exists a bivariate polynomial $g \in \Field[X,Y]$ such that:
\begin{inparaenum}[(i)]
\item $\IndividualDegree{g}[X] < \SetCardinality{\BSSpace_{0}}$,
\item $\IndividualDegree{g}[Y] < \SetCardinality{\BSSpace_{1}} \cdot \CodeRate$,
\item $\Restrict{g}{\Set{\alpha} \times \VanishingPoly{\BSSpace_{0}}(\BSSpace_{1})} = \Codeword_{\SymbolRS} \in \RSCode{\Field}{\VanishingPoly{\BSSpace_{0}}(\BSSpace_{1})}{\SetCardinality{\BSSpace_{1}} \cdot \CodeRate}$.
\end{inparaenum}
The existence of such $g$ follows by considering a suitable interpolating set (see, e.g., \appref{sec:proof-of-BSRS-recursive-cover-independence}).
\end{proof}

\clearpage
%%%%%%%%%%%%%%%%%%%%%%%%%%%%%%%%%%%%%%%%%%%%%%%%%%%%%%%%%%%%%%%%%%%%%%%%%%%%%%%%
%%%%%%%%%%%%%%%%%%%%%%%%%%%%%%%%%%%%%%%%%%%%%%%%%%%%%%%%%%%%%%%%%%%%%%%%%%%%%%%%
%%%%%%%%%%%%%%%%%%%%%%%%%%%%%%%%%%%%%%%%%%%%%%%%%%%%%%%%%%%%%%%%%%%%%%%%%%%%%%%%
\section{Proof of \lemref{lem:bsrs-has-recursive-code-cover}}
\label{sec:bsrs-full-proof}

In this section we prove \lemref{lem:bsrs-has-recursive-code-cover}. In \appref{sec:bsrs-combinatorial} we define the recursive cover and prove its combinatorial properties; in \appref{sec:bsrs-complexity} we prove that a spanning set for the duals of codes in this cover can be computed efficiently; in \appref{sec:bsrs-conclusion}, we put these together to conclude the proof.

%%%%%%%%%%%%%%%%%%%%%%%%%%%%%%%%%%%%%%%%%%%%%%%%%%%%%%%%%%%%%%%%%%%%%%%%%%%%%%%%
%%%%%%%%%%%%%%%%%%%%%%%%%%%%%%%%%%%%%%%%%%%%%%%%%%%%%%%%%%%%%%%%%%%%%%%%%%%%%%%%
\subsection{The recursive cover and its combinatorial properties}
\label{sec:bsrs-combinatorial}

We define a recursive cover for $\BSCode$ and then prove certain combinatorial properties for it. The definition relies on the definition of another cover, which we now introduce.

\begin{definition}
\label{def:native-cover}
The \defemph{native cover} $\Cover[\SubField,\Field,\BSSpace,\BSBalance,\BSBaseDim]$ of $\BSCode[\SubField,\Field,\BSSpace,\BSBalance,\BSBaseDim]$ is defined as follows:
\begin{itemize}

  \item if $\Dimension{\BSSpace} \leq \BSBaseDim$ then the cover contains only the trivial view $( \Domain^{\BSCode}[\SubField,\Field,\BSSpace,\BSBalance,\BSBaseDim], \BSCode[\SubField,\Field,\BSSpace,\BSBalance,\BSBaseDim])$;

  \item if $\Dimension{\BSSpace} > \BSBaseDim$ then the cover contains

\begin{itemize}[nolistsep]

  \item the view $(\BSCode[\SubField,\Field,\VanishingPoly{\BSSpace_{0}}(\BSSpace_{1}),\BSBalance,\BSBaseDim], \Domain_{\SymbolCol,\alpha})$ for every $\alpha \in \BSSpace'_{0}$, and

  \item the view $(\BSCode[\SubField,\Field,\BSSpace_{\beta},\BSBalance,\BSBaseDim], \Domain_{\SymbolRow,\beta})$ for every $\beta \in \BSSpace_{1}$.

\end{itemize}

\end{itemize}
\end{definition}

We now prove that the native cover is indeed a cover.

\begin{claim}
\label{claim:bsrs-code-cover}
The native cover of $\BSCode[\SubField,\Field,\BSSpace,\BSBalance,\BSBaseDim]$ is a code cover (see \defref{def:bsrs-code-cover}).
\end{claim}

\begin{proof}
From \clmref{claim:bsrs-code-cover-views} we know that:
\begin{itemize}[nolistsep]

\item for every $\alpha \in \BSSpace'_{0}$, the restriction of $\BSCode[\SubField,\Field,\BSSpace,\BSBalance,\BSBaseDim]$ to $\Domain_{\SymbolCol,\alpha}$ equals $\BSCode[\SubField,\Field,\VanishingPoly{\BSSpace_{0}}(\BSSpace_{1}),\BSBalance,\BSBaseDim]$;

\item for every $\beta \in \BSSpace_{1}$, the restriction of $\BSCode[\SubField,\Field,\BSSpace,\BSBalance,\BSBaseDim]$ to $\Domain_{\SymbolRow,\beta}$ equals $\BSCode[\SubField,\Field,\BSSpace_{\beta},\BSBalance,\BSBaseDim]$.

\end{itemize}
Therefore, it suffices to show that $\MakeBox{\Domain} \subseteq ( \cup_{\alpha \in \BSSpace'_{0}} \Domain_{\SymbolCol,\alpha} ) \cup ( \cup_{\beta \in \BSSpace_{1}} \Domain_{\SymbolRow,\beta} )$. So let $x$ be an index in $\MakeBox{\Domain}$.
\begin{itemize}[nolistsep]

\item If there exists $\alpha \in \BSSpace'_{0}$ such that $x \in \Set{\SymbolProx} \times \Set{(\SymbolCol,\alpha)} \times \MakeP{\Domain}^{\BSCode}[\SubField,\Field,\VanishingPoly{\BSSpace_{0}}(\BSSpace_{1}),\BSBalance,\BSBaseDim]$, then $x \in \Domain_{\SymbolCol,\alpha}$.

\item If there exists $\beta \in \BSSpace_{1}$ such that $x \in \Set{\SymbolProx} \times \Set{(\SymbolRow,\beta)} \times \MakeP{\Domain}^{\BSCode}[\SubField,\Field,\BSSpace_{\beta},\BSBalance,\BSBaseDim]$, then $x \in \Domain_{\SymbolRow,\beta}$.

\item If $x \in \Set{\SymbolProx} \times \MakeP{\Domain}[\SubField,\Field,\BSSpace,\BSBalance]$, then there exist $\beta \in \BSSpace_{1}$ and $\alpha \in R_\beta$ such that $x = (\SymbolProx,(\alpha,\VanishingPoly{\BSSpace_{0}}(\beta)))$, so $x \in \Domain_{\SymbolRow,\beta}$.

\item If $x \in \Set{\SymbolRS} \times \BSSpace$, then there exist $\beta \in \BSSpace_{1}$
and $\alpha \in \BSSpace_{\beta}$ such that $\phi[\SubField,\Field,\BSSpace,\BSBalance](\alpha,\VanishingPoly{\BSSpace_{0}}(\beta)) = x$, so $x \in \Domain_{\SymbolRow,\beta}$. \qedhere

\end{itemize}
\end{proof}

The recursive cover of $\BSCode$ is recursively defined based on the native cover of $\BSCode$.

\begin{definition}
\label{def:bsrs-tree-cover}
The recursive cover $\TreeCover[\SubField,\Field,\BSSpace,\BSBalance,\BSBaseDim]$ of $\BSCode[\SubField,\Field,\BSSpace,\BSBalance,\BSBaseDim]$ is the tree of depth $\BsDepth{\BSSpace}{\BSBaseDim}$ where, for every non-leaf vertex $v$ labeled by $(\CDomain,\BSCode[\SubField,\Field,\OtherBSSpace,\BSBalance,\BSBaseDim])$, the vertex $v$ has $\SetCardinality{\Cover[\SubField,\Field,\OtherBSSpace,\BSBalance,\BSBaseDim]}$ successors, all labeled by elements of $\Cover[\SubField,\Field,\OtherBSSpace,\BSBalance,\BSBaseDim]$ with the natural embedding of their domains into $\CDomain$.
\end{definition}

\begin{claim}
\label{claim:bsrs-1-intersecting}
The recursive cover of $\BSCode[\SubField,\Field,\BSSpace,\BSBalance,\BSBaseDim]$ is $1$-intersecting (see \defref{def:bsrs-recursive-code-cover}).
\end{claim}

\begin{proof}
We must show that for every two disconnected vertices $u,v$ it holds that $\SetCardinality{\CDomain_{u} \cap \CDomain_{v}} \leq 1$. It suffices to do so for every two distinct \emph{siblings} $u,v$, because if $a$ is an ancestor of $b$ then $\CDomain_{a}$ contains $\CDomain_{b}$. Hence, we only need to show that for every two distinct views $(\CDomain,\CCode),(\CDomain',\CCode')$ in the native cover $\Cover[\SubField,\Field,\BSSpace,\BSBalance,\BSBaseDim]$, it holds that $\SetCardinality{\CDomain \cap \CDomain'} \leq 1$. First we observe that for every $\alpha_{1} \neq \alpha_2 \in \BSSpace'_{0}$ and $\beta_{1} \neq \beta_{2} \in \BSSpace_{1}$, the following sets are disjoint by definition:
\begin{itemize}[nolistsep]

\item $\Set{\SymbolProx} \times \Set{(\SymbolCol,\alpha_{1})} \times \MakeP{\Domain}^{\BSCode}[\SubField,\Field,\VanishingPoly{\BSSpace_{0}}(\BSSpace_{1}),\BSBalance,\BSBaseDim]$,

\item $\Set{\SymbolProx} \times \Set{(\SymbolCol,\alpha_2)} \times \MakeP{\Domain}^{\BSCode}[\SubField,\Field,\VanishingPoly{\BSSpace_{0}}(\BSSpace_{1}),\BSBalance,\BSBaseDim]$,

\item $\Set{\SymbolProx} \times \Set{(\SymbolRow,\beta_{1})} \times \MakeP{\Domain}^{\BSCode}[\SubField,\Field,\BSSpace_{\beta_{1}},\BSBalance,\BSBaseDim]$,

\item $\Set{\SymbolProx} \times \Set{(\SymbolRow,\beta_{2})} \times \MakeP{\Domain}^{\BSCode}[\SubField,\Field,\BSSpace_{\beta_{2}},\BSBalance,\BSBaseDim]$.

\end{itemize}
Thus it is enough to show that:
\begin{itemize}[nolistsep]

\item Any two columns are distinct: $\phi(\alpha_{1},\beta_{1}) \neq \phi(\alpha_2,\beta_{2})$ for every $\alpha_{1} \neq \alpha_2 \in \BSSpace'_{0}$ and $\beta_{1},\beta_{2} \in \VanishingPoly{\BSSpace_{0}}(\BSSpace_{1})$.

\item Any two rows are distinct: $\phi(\alpha_{1},\beta_{1}) \neq \phi(\alpha_2,\beta_{2})$ for every $\beta_{1} \neq \beta_{2} \in \VanishingPoly{\BSSpace_{0}}(\BSSpace_{1})$, $\alpha_{1} \in \BSSpace_{\beta_{1}}$, and $\alpha_2 \in \BSSpace_{\beta_{2}}$.

\item The intersection of any row and column has at most one element: $\phi(\alpha,\beta') \neq \phi(\alpha',\beta)$ for every $\alpha,\alpha' \in \BSSpace'_{0}$ and $\beta,\beta' \in \VanishingPoly{\BSSpace_{0}}(\BSSpace_{1})$ with $(\alpha',\beta') \neq (\alpha,\beta)$.

\end{itemize}
But all the above follow from the fact that $\phi[\SubField,\Field,\BSSpace,\BSBalance]$ is a bijection and, thus, an injection.
\end{proof}

The next claim establishes a connection between the depth of a vertex $v$ in the recursive cover and the independence of the cover $\TreeCover_{v}$ of the code $\CCode_{v}$.

\begin{claim}
\label{cor:bsrs-low-depth-big-independence}
For every vertex $v$ in $\TreeLayer{\TreeCover,\CDepth}$, the cover $\TreeCover_{v}$ is $(\SetCardinality{\SubField}^{\Dimension{\BSSpace} \cdot 2^{-\CDepth - 1} -\BSBalance -2})$-independent. In particular, by assignment, it holds that, for every positive integer $m$ and every non-leaf vertex $v$ in $\TreeCover[\SubField,\Field,\BSSpace,\BSBalance,\BSBaseDim]$ with depth less than $\BsOtherDepth{\BSSpace}{m}$, the cover $\TreeCover_{v}$ is $m$-independent.
\end{claim}

The proof of the above claim directly follows from \clmref{claim:bsrs-cover-independence} and \clmref{claim:bsrs-bounded-basis}, stated and proved below. The first of these two claims connects the depth of a vertex $v$ and the dimension of a space $\BSSpace_{v}$ such that $\CCode_{v} = \BSCode[\Field,\SubField,\BSSpace_{v},\BSBalance,\BSBaseDim]$ (this claim is used separately also for establishing computational properties in in \appref{sec:bsrs-complexity}).

\begin{claim}
\label{claim:bsrs-bounded-basis}
If $v \in \TreeLayer{\TreeCover[\SubField,\Field,\BSSpace,\BSBalance,\BSBaseDim],\CDepth}$ then $\CCode_{v} = \BSCode[\SubField,\Field,\OtherBSSpace,\BSBalance,\BSBaseDim]$ for some $\OtherBSSpace$ such that
\begin{equation*}
\Dimension{\BSSpace} \cdot 2^{-\CDepth} \leq \Dimension{\OtherBSSpace} \leq \Dimension{\BSSpace} \cdot 2^{-\CDepth} + 2\BSBalance
\enspace.
\end{equation*}
\end{claim}

\begin{proof}
The proof is by induction on $\CDepth$. The base case $\CDepth = 0$ follows directly from the definition; so we now assume the claim for $\CDepth -1$ and prove it for $\CDepth$. Let $v \in \TreeLayer{\TreeCover,\CDepth}$ be a vertex of depth $\CDepth$, and let $u \in \TreeLayer{\TreeCover,\CDepth-1}$ be $v$'s predecessor. By the inductive assumption, $\CCode_{u} = \BSCode[\SubField,\Field,\BSSpace_{u},\BSBalance,\BSBaseDim]$ for some $\BSSpace_{u}$ such that $\Dimension{\BSSpace} \cdot 2^{-(\CDepth-1)} \leq \Dimension{\BSSpace_{u}} \leq \Dimension{\BSSpace} \cdot 2^{-(\CDepth-1)} + 2\BSBalance$.

First we argue that $\TreeCover_{u}$ is not the trivial (singleton) cover. For this, it suffices to show that $\Dimension{\BSSpace_{u}} > \BSBaseDim$. But this follows from the inductive assumption, since $\TreeDepth{\TreeCover,u} < \BsDepth{\BSSpace}{\BSBaseDim}$, so that $\Dimension{\BSSpace_{u}} \ge \Dimension{\BSSpace} \cdot 2^{-(\BsDepth{\BSSpace}{\BSBaseDim}-1)} \ge 2 \BSBaseDim$.

Recall $\CCode_{v} = \BSCode[\SubField,\Field,\BSSpace_{v},\BSBalance,\BSBaseDim]$ for some space $\BSSpace_{v}$; we are thus left to show that $\Dimension{\BSSpace_{u}} \cdot 2^{-1} \leq \Dimension{\BSSpace_{v}} \leq \Dimension{\BSSpace_{u}} \cdot 2^{-1} + \BSBalance$. We do so by giving two cases, based on the form of $\BSSpace_{v}$:
\begin{inparaenum}[(a)]
\item if $\BSSpace_{v} = \VanishingPoly{\BSSpace_{0}[\SubField,\Field,\BSSpace_{u},\BSBalance,\BSBaseDim]}(\BSSpace_{1}[\SubField,\Field,\BSSpace_{u},\BSBalance,\BSBaseDim])$
then
$\Dimension{\BSSpace_{v}} = \Dimension{\BSSpace_{1}[\SubField,\Field,\BSSpace_{u},\BSBalance,\BSBaseDim]} =  \lceil \frac{\Dimension{\BSSpace_{u}}}{2} \rceil$;
\item if there exists $\beta \in \BSSpace_{1}[\SubField,\Field,\BSSpace_{u},\BSBalance,\BSBaseDim]$
such that $\BSSpace_{v} = \BSSpace_\beta[\SubField,\Field,\BSSpace_{u},\BSBalance,\BSBaseDim]$
then $\Dimension{\BSSpace_{v}} = \Dimension{\BSSpace_{0}[\SubField,\Field,\BSSpace_{u},\BSBalance,\BSBaseDim]} + \BSBalance = \lfloor \frac{\Dimension{\BSSpace_{u}}}{2} \rfloor + \BSBalance$.
\end{inparaenum}
In either case $\Dimension{\BSSpace_{u}} \cdot 2^{-1} \leq \Dimension{\BSSpace_{v}} \leq \Dimension{\BSSpace_{u}} \cdot 2^{-1} + \BSBalance$, and the claim follows.
\end{proof}

\begin{claim}
\label{claim:bsrs-cover-independence}
The native cover $\Cover[\SubField,\Field,\BSSpace,\BSBalance,\BSBaseDim]$ is $\SetCardinality{\SubField}^{\frac{\Dimension{\BSSpace}}{2} - \BSBalance -2}$-independent.
\end{claim}

\begin{proof}
Recalling \defref{def:bsrs-independent-cover},
fix arbitrary subsets $\VariantA{\Domain} \subseteq (\Set{\SymbolCol} \times \BSSpace'_{0} ) \sqcup (\Set{\SymbolRow} \times \BSSpace_{1} )$ and $\VariantB{\Domain} \subseteq \MakeBox{\Domain}[\SubField,\Field,\BSSpace,\BSBalance]$ both of size at most $\SetCardinality{\SubField}^{\frac{\Dimension{\BSSpace}}{2} - \BSBalance -2}$, and define $\CDomain \DefineEqual \VariantB{\Domain} \cup ( \cup_{(\SymbolCol,\alpha) \in \VariantA{\Domain}} \Domain_{\SymbolCol,\alpha} ) \cup ( \cup_{(\SymbolRow,\beta) \in \VariantA{\Domain}} \Domain_{\SymbolRow,\beta})$. Let $\Codeword' \in \Field^{\Domain^{\BSCode}}$ be such that:
\begin{inparaenum}[(i)]
\item for every $(\SymbolCol,\alpha) \in \VariantA{\Domain}$ it holds that $\psi_{\SymbolCol,\alpha}(\Codeword') \in \BSCode[\SubField,\Field,\VanishingPoly{\BSSpace_{0}}(\BSSpace_{1}),\BSBalance,\BSBaseDim]$; and
\item for every $(\SymbolRow,\beta) \in \VariantA{\Domain}$ it holds that $\psi_{\SymbolRow,\beta}(\Codeword') \in \BSCode[\SubField,\Field,\BSSpace_{\beta},\BSBalance,\BSBaseDim]$.
\end{inparaenum}
We need to show that there exists $\Codeword \in \BSCode[\SubField,\Field,\BSSpace,\BSBalance,\BSBaseDim]$ such that $\Restrict{\Codeword}{\CDomain} = \Restrict{\Codeword'}{\CDomain}$.

In fact, it suffices to show that there exists $\MakeBox{\Codeword} \in \RSBox[\SubField,\Field,\BSSpace,\BSBalance]$ such that $\Restrict{\MakeBox{\Codeword}}{\CDomain \cap \MakeBox{\Domain}} = \Restrict{\Codeword'}{\CDomain \cap \MakeBox{\Domain}}$, because \clmref{claim:bsrs-code-cover-proof-defined-by-RS} implies there exists a unique codeword $\Codeword \in \BSCode[\SubField,\Field,\BSSpace,\BSBalance,\BSBaseDim]$ such that $\Restrict{\Codeword}{\MakeBox{\Domain}} = \MakeBox{\Codeword}$ and $\Restrict{\Codeword}{\CDomain} = \Restrict{\Codeword'}{\CDomain}$.

Thus, we now argue that there exists $\MakeBox{\Codeword} \in \RSBox[\SubField,\Field,\BSSpace,\BSBalance]$ such that the following holds.
\begin{itemize}

\item For every $(\SymbolCol,\alpha) \in \VariantA{\Domain}$ and $\beta \in \VanishingPoly{\BSSpace_{0}}(\BSSpace_{1})$, it holds that $f_{\MakeBox{\Codeword}}(\alpha,\beta) = \Codeword'(\phi(\alpha,\beta)) = \left(\psi_{\SymbolCol,\alpha}(\Codeword')\right)(\SymbolRS,\beta)$. In particular, $\Restrict{f_{\MakeBox{\Codeword}}}{\Set{\alpha} \times \VanishingPoly{\BSSpace_{0}}(\BSSpace_{1})} \in \RSCode{\Field}{\VanishingPoly{\BSSpace_{0}}(\BSSpace_{1})}{\SetCardinality{\BSSpace_{1}}\cdot \CodeRate}$,
and let $p_{\SymbolCol,\alpha}$ be its univariate low degree extension to $\Field$.

\item For every $(\SymbolRow,\beta) \in \VariantA{\Domain}$ and $\alpha \in \BSSpace_{\beta}$, it holds that $f_{\MakeBox{\Codeword}}(\alpha,\VanishingPoly{\BSSpace_{0}}(\beta)) = \Codeword'(\phi(\alpha,\VanishingPoly{\BSSpace_{0}}(\beta))) = \left(\psi_{\SymbolRow,\beta}(\Codeword')\right)(\SymbolRS,\alpha)$. In particular, $\Restrict{f_{\MakeBox{\Codeword}}}{\BSSpace_{\beta} \times \Set{\VanishingPoly{\BSSpace_{0}}(\BSSpace_{1})}} \in \RSCode{\Field}{\BSSpace_{\beta}}{\SetCardinality{\BSSpace_{0}}}$, and let $p_{\SymbolRow,\beta}$ be its univariate low degree extension to $\Field$.

\item For every $(\alpha,\beta) \in \VariantB{\Domain}$, it holds that $f_{\MakeBox{\Codeword}}(\alpha,\VanishingPoly{\BSSpace_{0}}(\beta)) = \Codeword'(\phi(\alpha,\VanishingPoly{\BSSpace_{0}}(\beta)))$.

\end{itemize}
By \defref{def:bs-proof-of-proximity} it suffices to show that there exists a bivariate polynomial $g \in \Field[X,Y]$ such that:
\begin{inparaenum}[(i)]
  \item $\IndividualDegree{g}[X] < \SetCardinality{\BSSpace_{0}}$;
  \item $\IndividualDegree{g}[Y] < \SetCardinality{\BSSpace_{1}} \cdot \rho$;
  \item $\Restrict{g}{X=\alpha} = p_{\SymbolCol,\alpha}$ for every $(\SymbolCol,\alpha) \in \VariantA{\Domain}$;
  \item $\Restrict{g}{Y=\VanishingPoly{\BSSpace_{0}}(\beta)} = p_{\SymbolRow,\beta}$ for every $(\SymbolRow,\beta) \in \VariantA{\Domain}$;
  \item $g(\alpha,\beta) = \Codeword'(\phi(\alpha,\VanishingPoly{\BSSpace_{0}}(\beta)))$ for every $(\alpha,\beta) \in \VariantB{\Domain}$.
\end{inparaenum}
But notice that
$
\SetCardinality{\VariantA{\Domain}} + \SetCardinality{\VariantB{\Domain}}
\leq 2 \cdot \SetCardinality{\SubField}^{\frac{\Dimension{\BSSpace}}{2} - \BSBalance -2}
=    \SetCardinality{\SubField}^{\frac{\Dimension{\BSSpace}}{2} - \BSBalance -1}
<    \min\Set{\SetCardinality{\BSSpace_{0}}, \SetCardinality{\BSSpace_{1}}\cdot\rho}
$, because
\begin{inparaenum}[(a)]
  \item $\log_{\SetCardinality{\SubField}}( \SetCardinality{\BSSpace_{0}}) = \Dimension{\BSSpace_{0}} \geq \frac{\Dimension{\BSSpace}}{2} -1 $, and
  \item $\log_{\SetCardinality{\SubField}}( \SetCardinality{\BSSpace_{1}} \cdot \rho) = \Dimension{\BSSpace_{1}} - \BSBalance \geq \frac{\Dimension{\BSSpace}}{2} -\BSBalance $.
\end{inparaenum}
The claim follows by considering a suitable interpolating set (see \secref{sec:proof-of-BSRS-recursive-cover-independence}).
\end{proof}

%%%%%%%%%%%%%%%%%%%%%%%%%%%%%%%%%%%%%%%%%%%%%%%%%%%%%%%%%%%%%%%%%%%%%%%%%%%%%%%%
%%%%%%%%%%%%%%%%%%%%%%%%%%%%%%%%%%%%%%%%%%%%%%%%%%%%%%%%%%%%%%%%%%%%%%%%%%%%%%%%
\subsection{Computing spanning sets of dual codes in the recursive cover}
\label{sec:bsrs-complexity}

We prove that spanning sets for duals of codes in the recursive cover can be computed efficiently; this is the key fact that we later use to argue that the algorithm required by \lemref{lem:bsrs-has-recursive-code-cover} satisfies the stated time complexity.

\begin{claim}
\label{claim:bsrs-efficient-basis}
For every positive integer $m$ and vertex $v$ in $\TreeCover[\SubField,\Field,\BSSpace,\BSBalance,\BSBaseDim]$ of depth at least $\log_{2} \Dimension{\BSSpace} - \log_{2} \log_{\SetCardinality{\SubField}} m$, a spanning set of $\Dual{\CCode_{v}}$ can be computed in time $\poly(\log_{2} \SetCardinality{\Field} + \SetCardinality{\SubField}^{\BSBalance} + m)$.
\end{claim}

The above claim directly follows from \clmref{claim:bsrs-bounded-domain} and \clmref{claim:bsrs-efficient-basis-aux}, stated and proved below.

\begin{claim}
\label{claim:bsrs-bounded-domain}
For every positive integer $m$ and vertex $v$ in $\TreeCover[\SubField,\Field,\BSSpace,\BSBalance,\BSBaseDim]$ of depth at least $\log_{2} \Dimension{\BSSpace} - \log_{2} \log_{\SetCardinality{\SubField}} m$, $\SetCardinality{\CDomain_{v}} \leq \poly(m + \SetCardinality{\SubField}^{\BSBalance})$.
\end{claim}

\begin{proof}
Ben-Sasson and Sudan \cite{BS08} show that the block length of $\BSCode[\SubField,\Field,\BSSpace,\BSBalance,\BSBaseDim]$ is $\tilde{O}_{\SubField,\BSBalance,\BSBaseDim}(\SetCardinality{\BSSpace})$ for any fixed $\SubField,\BSBalance,\BSBaseDim$. One can verify that, if we do not fix these parameters, the block length is $\tilde{O}(\SetCardinality{\BSSpace} \cdot \SetCardinality{\SubField}^{\BSBalance})$. Next, observe that $\CCode_{v} = \BSCode[\SubField,\Field,\BSSpace_{v},\BSBalance,\BSBaseDim]$ for some $\BSSpace_{v}$ such that $\Dimension{\BSSpace_{v}} \leq \log_{\SetCardinality{\SubField}} m + 2\BSBalance$ (\clmref{claim:bsrs-bounded-basis}); in this case, the aforementioned bound becomes $\poly(m + \SetCardinality{\SubField}^{\BSBalance})$.
\end{proof}

\ignore[explicit proof]{
\begin{proof}
By \clmref{claim:bsrs-bounded-basis} $\CCode_{v} = \Code_{(\Field,\BSSpace_{v})}$ for some $\BSSpace_{v}$ such that $\Dimension{\BSSpace_{v}} \leq \log_{\SetCardinality{\SubField}}(m) + 2\BSBalance$,
thus $\SetCardinality{\BSSpace_{v}} \leq m + \SetCardinality{\SubField}^{2\BSBalance}$.

It suffices to show $\SetCardinality{\Domain^{\BSCode}[\SubField,\Field,\BSSpace,\BSBalance,\BSBaseDim]} \leq \left( \SetCardinality{\BSSpace}\cdot \SetCardinality{\SubField}^{\BSBalance} \right)^2$, we show it using induction on $\Dimension{\BSSpace}$:
\begin{itemize}

\item If $\Dimension{\BSSpace} \leq \BSBaseDim$ then
$\SetCardinality{\Domain^{\BSCode}[\SubField,\Field,\BSSpace,\BSBalance,\BSBaseDim]} =
\SetCardinality{\MakeBox{\Domain}[\SubField,\Field,\BSSpace,\BSBalance]} =
\SetCardinality{\BSSpace}\cdot \SetCardinality{\SubField}^{\BSBalance} \le
\left( \SetCardinality{\BSSpace}\cdot \SetCardinality{\SubField}^{\BSBalance} \right)^2$.

\item If $\Dimension{\BSSpace} > \BSBaseDim$ then
\begin{align}
\SetCardinality{\Domain^{\BSCode}[\SubField,\Field,\BSSpace,\BSBalance,\BSBaseDim]} \leq &
\SetCardinality{\MakeBox{\Domain}[\SubField,\Field,\BSSpace,\BSBalance]}  \label{prf:bsrs-dom-size-dom-def}\\
+& \sum_{\alpha \in \BSSpace'_{0}} \SetCardinality{\MakeP{\Domain}^{\BSCode}[\SubField,\Field,\VanishingPoly{\BSSpace_{0}}(\BSSpace_{1}),\BSBalance,\BSBaseDim]}
+ \sum_{\beta \in \BSSpace_{1}} \SetCardinality{\MakeP{\Domain}^{\BSCode}[\SubField,\Field,\BSSpace_{\beta},\BSBalance,\BSBaseDim]}\\
\leq & \SetCardinality{\BSSpace}\cdot \SetCardinality{\SubField}^{\BSBalance} +
\SetCardinality{\BSSpace'_{0}} \cdot \left( \SetCardinality{\BSSpace_{1}} \cdot \SetCardinality{\SubField}^{\BSBalance} \right)^2 +
\SetCardinality{\BSSpace_{1}} \cdot \left( \SetCardinality{\BSSpace'_{0}} \cdot \SetCardinality{\SubField}^{\BSBalance+1} \right)^2 \label{prf:bsrs-dom-size-ind}\\
\leq & \SetCardinality{\BSSpace}\cdot \SetCardinality{\SubField}^{\BSBalance} +
\left( \SetCardinality{\BSSpace_{0}} \cdot \SetCardinality{\SubField}^{\BSBalance-1} \right) \cdot \left( \SetCardinality{\BSSpace_{1}} \cdot \SetCardinality{\SubField}^{\BSBalance} \right)^2 +
\SetCardinality{\BSSpace_{1}} \cdot \left( \left(\SetCardinality{\BSSpace_{0}} \cdot \SetCardinality{\SubField}^{\BSBalance-1} \right) \cdot \SetCardinality{\SubField}^{\BSBalance+1} \right)^2\\
= & \SetCardinality{\BSSpace} \cdot \SetCardinality{\SubField}^{\BSBalance} \cdot
\left(1+ \SetCardinality{\BSSpace_{1}}\cdot \SetCardinality{\SubField}^{2\BSBalance-1} + \SetCardinality{\BSSpace_{0}} \cdot \SetCardinality{\SubField}^{\BSBalance} \right) \label{prf:bsrs-dom-size-L0L1=L}\\
\leq & \SetCardinality{\BSSpace} \cdot \SetCardinality{\SubField}^{\BSBalance} \cdot
\left(1 + \sqrt{\SetCardinality{\BSSpace}}\cdot \SetCardinality{\SubField}^{2\BSBalance} + \sqrt{\SetCardinality{\BSSpace}} \cdot \SetCardinality{\SubField}^{\BSBalance} \right) \label{prf:bsrs-dom-size-L0,L1=sqrL}\\
\leq & \SetCardinality{\BSSpace} \cdot \SetCardinality{\SubField}^{2\BSBalance} \cdot
\sqrt{\SetCardinality{\BSSpace}}\cdot \left(  \SetCardinality{\SubField}^{\BSBalance} + 2 \right)\\
\leq & \SetCardinality{\BSSpace} \cdot \SetCardinality{\SubField}^{2\BSBalance} \cdot
\sqrt{\SetCardinality{\BSSpace}}\cdot \left(  \SetCardinality{\SubField}^{\BSBalance+1}\right)\\
\leq & \SetCardinality{\BSSpace} \cdot \SetCardinality{\SubField}^{2\BSBalance} \cdot
\sqrt{\SetCardinality{\BSSpace}}\cdot \left(  \SetCardinality{\SubField}^{\frac{\BSBaseDim+1}{2}}\right) \label{prf:bsrs-dom-size-k,m}\\
\leq & \SetCardinality{\BSSpace}^2 \cdot \SetCardinality{\SubField}^{2\BSBalance} \label{prf:bsrs-dom-size-L,k}
\end{align}
Where \eqnref{prf:bsrs-dom-size-dom-def} follows by definition of $\Domain^{\BSCode}[\SubField,\Field,\BSSpace,\BSBalance,\BSBaseDim]$,
\eqnref{prf:bsrs-dom-size-ind} follows by the inductive assumption,
\eqnref{prf:bsrs-dom-size-L0L1=L} follows by the fact that $\SetCardinality{\BSSpace} = \SetCardinality{\BSSpace_{0}} \cdot \SetCardinality{\BSSpace_{1}}$,
\eqnref{prf:bsrs-dom-size-L0,L1=sqrL} follows by the facts that
$\SetCardinality{\BSSpace_{0}} \leq \sqrt{\SetCardinality{\BSSpace}}$
and
$\SetCardinality{\BSSpace_{1}} \leq \sqrt{\SetCardinality{\BSSpace}}\cdot \SetCardinality{\SubField}$,
\eqnref{prf:bsrs-dom-size-k,m} follows by the requirement $\BSBaseDim > 2\BSBalance$,
and \eqnref{prf:bsrs-dom-size-L,k} follows by the assumption $\Dimension{\BSSpace} \geq \BSBaseDim+1$.

\end{itemize}
\end{proof}
}%ignore

\begin{claim}
\label{claim:bsrs-efficient-basis-aux}
A spanning set for $\Dual{\BSCode[\SubField,\Field,\BSSpace,\BSBalance,\BSBaseDim]}$ can be found
in time $\poly \left( \log_{2}\SetCardinality{\Field} + \SetCardinality{\Domain^{\BSCode}[\SubField,\Field,\BSSpace,\BSBalance,\BSBaseDim]} \right)$.
\end{claim}

\begin{proof}
We show an algorithm that constructs the desired spanning set in the stated time complexity. First, a spanning set for $\Dual{\RSCode{\Field}{S}{d}}$ can be found in time $\poly(\log_{2} \SetCardinality{\Field} + \SetCardinality{S})$, for any finite field $\Field$, subset $S \subseteq \Field$, and degree bound $d < \SetCardinality{S}$. Hence, a spanning set for $\Dual{\RSBox[\SubField,\Field,\BSSpace,\BSBalance]}$
can be found in time $(\log_{2} \SetCardinality{\Field} \cdot \SetCardinality{\BSSpace} \cdot \SetCardinality{\SubField}^{\BSBalance})^{c}$ for some $c>0$.

We argue by induction on $\Dimension{\BSSpace}$ that a spanning set for $\Dual{\BSCode[\SubField,\Field,\BSSpace,\BSBalance,\BSBaseDim]}$ can be found in time $(\log_{2} \SetCardinality{\Field} \cdot \SetCardinality{\BSSpace} \cdot \SetCardinality{\SubField}^{\BSBalance})^{c}$. We rely the property that the code $\BSCode[\SubField,\Field,\BSSpace,\BSBalance,\BSBaseDim]$ is covered by
\begin{equation*}
\Set{(\BSCode[\SubField,\Field,\VanishingPoly{\BSSpace_{0}}(\BSSpace_{1}),\BSBalance,\BSBaseDim],\Domain_{\SymbolCol,\alpha})}_{\alpha \in \BSSpace'_{0}}
\cup
\Set{(\BSCode[\SubField,\Field,\BSSpace_{\beta},\BSBalance,\BSBaseDim],\Domain_{\SymbolRow,\beta})}_{\beta \in \BSSpace_{1}}
\end{equation*}
and the property that $\Codeword \in \BSCode[\SubField,\Field,\BSSpace,\BSBalance,\BSBaseDim]$ if an only if:
\begin{itemize}
\item $\psi_{\SymbolCol,\alpha}(\Codeword) \in \BSCode[\SubField,\Field,\VanishingPoly{\BSSpace_{0}}(\BSSpace_{1}),\BSBalance,\BSBaseDim]$ for every $\alpha \in \BSSpace'_{0}$ and
\item $\psi_{\SymbolRow,\beta}(\Codeword) \in \BSCode[\SubField,\Field,\BSSpace_{\beta},\BSBalance,\BSBaseDim]$ for every $\beta \in \BSSpace_{1}$.
\end{itemize}
Thus $\Dual{\BSCode[\SubField,\Field,\BSSpace,\BSBalance,\BSBaseDim]}$ is spanned by the duals of codes in its cover and, in particular, is spanned by their spanning sets; in sum, it suffices to construct a spanning set for its cover.

In light of the above, we can bound the construction time of a spanning set for $\Dual{\BSCode[\SubField,\Field,\BSSpace,\BSBalance,\BSBaseDim]}$ as follows:
\begin{itemize}

\item If $\Dimension{\BSSpace} \leq \BSBaseDim$, the claim follows as the code in this case simply equals $\RSBox[\SubField,\Field,\BSSpace,\BSBalance]$.

\item If $\Dimension{\BSSpace} > \BSBaseDim$, the time to construct a spanning set is at most the time to construct spanning sets for the cover:
\begin{align}
& \SetCardinality{\BSSpace'_{0}} \cdot (\log_{2} \SetCardinality{\Field} \cdot \SetCardinality{\BSSpace_{1}} \cdot \SetCardinality{\SubField}^{\BSBalance})^{c} +
\SetCardinality{\BSSpace_{1}} \cdot (\log_{2} \SetCardinality{\Field} \cdot \SetCardinality{\BSSpace'_{0}} \cdot \SetCardinality{\SubField}^{\BSBalance+1})^{c} \label{prf:bsrs-spanning-time-ind}\\
=\; &
\SetCardinality{\BSSpace_{0}} \cdot \SetCardinality{\SubField}^{\BSBalance-1} \cdot (\log_{2} \SetCardinality{\Field} \cdot \SetCardinality{\BSSpace_{1}} \cdot \SetCardinality{\SubField}^{\BSBalance})^{c} +
\SetCardinality{\BSSpace_{1}} \cdot (\log_{2} \SetCardinality{\Field} \cdot \SetCardinality{\BSSpace_{0}} \cdot \SetCardinality{\SubField}^{2\BSBalance})^{c} \label{prf:bsrs-spanning-time-def-L'0}\\
=\; &
\SetCardinality{\BSSpace} \cdot \left( \log_{2} \SetCardinality{\Field} \cdot \SetCardinality{\SubField}^{\BSBalance} \right)^{c} \cdot \left( \SetCardinality{\SubField}^{\BSBalance-1} \cdot \SetCardinality{\BSSpace_{1}}^{c-1} + \SetCardinality{\SubField}^{c\BSBalance} \cdot \SetCardinality{\BSSpace_{0}}^{c-1} \right)
\label{prf:bsrs-spanning-time-L0L1=L}\\
\leq\; &
\SetCardinality{\BSSpace} \cdot \left( \log_{2} \SetCardinality{\Field} \cdot \SetCardinality{\SubField}^{\BSBalance} \right)^{c} \cdot \left( \SetCardinality{\SubField}^{\BSBalance + c -2} \cdot \SetCardinality{\BSSpace}^{\frac{c-1}{2}} + \SetCardinality{\SubField}^{c\BSBalance} \cdot \SetCardinality{\BSSpace}^{\frac{c-1}{2}} \right)
\label{prf:bsrs-spanning-time-def-L0L1}\\
=\; &
\SetCardinality{\BSSpace}^{\frac{c+1}{2}} \cdot \left( \log_{2} \SetCardinality{\Field} \cdot \SetCardinality{\SubField}^{\BSBalance} \right)^{c} \cdot \left( \SetCardinality{\SubField}^{\BSBalance+c-2} + \SetCardinality{\SubField}^{c\BSBalance} \right) \\
\leq\; &
(\log_{2} \SetCardinality{\Field} \cdot \SetCardinality{\BSSpace} \cdot \SetCardinality{\SubField}^{\BSBalance})^{c}
\enspace. \label{prf:bsrs-spanning-time-last-inequality}
\end{align}
Above,
\eqnref{prf:bsrs-spanning-time-ind} is by the inductive assumption,
\eqnref{prf:bsrs-spanning-time-def-L'0} is by definition of $\BSSpace'_{0}$,
\eqnref{prf:bsrs-spanning-time-L0L1=L} is by the fact that $\SetCardinality{\BSSpace_{0}} \cdot \SetCardinality{\BSSpace_{1}} = \SetCardinality{\BSSpace}$,
and \eqnref{prf:bsrs-spanning-time-def-L0L1} is by definition of $\BSSpace_{0}, \BSSpace_{1}$.
We are left to show \eqnref{prf:bsrs-spanning-time-last-inequality}, and this follows from the fact that:
\begin{inparaenum}[(i)]
  \item $\Dimension{\BSSpace} > \BSBaseDim$,
  \item$\BSBaseDim > 2\BSBalance$ by definition, and
  \item we can choose $c$ to be large enough (namely, so that $\SetCardinality{\SubField}^{\BSBalance+c-2} + \SetCardinality{\SubField}^{c\BSBalance} \leq \SetCardinality{\BSSpace}^{\frac{c-1}{2}}$ holds).
\end{inparaenum}
\qedhere
\end{itemize}
\end{proof}

%%%%%%%%%%%%%%%%%%%%%%%%%%%%%%%%%%%%%%%%%%%%%%%%%%%%%%%%%%%%%%%%%%%%%%%%%%%%%%%%
%%%%%%%%%%%%%%%%%%%%%%%%%%%%%%%%%%%%%%%%%%%%%%%%%%%%%%%%%%%%%%%%%%%%%%%%%%%%%%%%
\subsection{Putting things together}
\label{sec:bsrs-conclusion}

\begin{proof} [Proof of \lemref{lem:bsrs-has-recursive-code-cover}]
Define the depth function $\CDepth(\SubField,\BSSpace,\BSBalance,a) \DefineEqual \BsOtherDepth{\BSSpace}{a}$. We argue the two conditions in the lemma. First, for every index $\CodeIdx = (\SubField,\Field,\BSSpace,\BSBalance,\BSBaseDim)$, $\TreeCover[\SubField,\Field,\BSSpace,\BSBalance,\BSBaseDim]$ is a $1$-intersecting recursive cover of $\BSCode[\SubField,\Field,\BSSpace,\BSBalance,\BSBaseDim]$ (by \clmref{claim:bsrs-1-intersecting}). Moreover, for every positive integer $m$ and non-leaf vertex $v$ in $\TreeCover$ with $\TreeDepth{\TreeCover,v} < \CDepth(\SubField,\BSSpace,\BSBalance,m)$, the cover $\TreeCover_{v}$ is $m$-independent (by \clmref{cor:bsrs-low-depth-big-independence}).

Second, consider the algorithm that, given an index $\CodeIdx = (\SubField,\Field,\BSSpace,\BSBalance,\BSBaseDim)$ and subset $\IndexSet \subseteq \Domain^{\BSCode}[\SubField,\Field,\BSSpace,\BSBalance,\BSBaseDim]$, works as follows:
\begin{inparaenum}[(1)]
\item for every $\alpha \in \IndexSet$ choose an arbitrary vertex $v_{\alpha}$ in $\TreeLayer{\TreeCover,\CDepth(\SubField,\BSSpace,\BSBalance,\SetCardinality{\IndexSet})}$
such that $\alpha \in \CDomain_{v_{\alpha}}$, and then set $U \DefineEqual \Set{v_{\alpha}}_{\alpha \in \IndexSet}$;
\item compute a spanning set $\LocalSet_{v}$ set for $\Dual{\CCode_{v}}$;
\item return $\LocalSet \DefineEqual \cup_{u \in U} \LocalSet_{u}$.
\end{inparaenum}
This algorithm satisfies the required properties. First, it runs in time $\poly(\log_{2} \SetCardinality{\Field} + \Dimension{\BSSpace} + \SetCardinality{\SubField}^{\BSBalance} + \SetCardinality{\IndexSet})$ because a spanning set set for $\Dual{\CCode_{u}}$ can be computed in time $\poly(\log_{2} \SetCardinality{\Field} + \SetCardinality{\SubField}^{\BSBalance} + \SetCardinality{\IndexSet})$ (by \clmref{claim:bsrs-efficient-basis}). Next, its output $\LocalSet$ meets the requirements:
\begin{itemize}[nolistsep]

\item $U \subseteq \TreeLayer{\TreeCover[\SubField,\Field,\BSSpace,\BSBalance,\BSBaseDim],\CDepth(\SubField,\BSSpace,\BSBalance,\SetCardinality{\IndexSet})}$;

\item $\SetCardinality{U} \leq \SetCardinality{\IndexSet}$, by definition of $U$;

\item $\IndexSet \subseteq (\cup_{u \in U} \CDomain_{u})$, by definition of $U$;

\item $\Span(\LocalSet) = \Span(\cup_{u \in U} \LocalSet_{u}) = \Span(\cup_{u \in U} \Dual{\CCode_{u}})$, by definition of $\LocalSet$ and $\LocalSet_{u}$.

\end{itemize}
This completes the proof of \lemref{lem:bsrs-has-recursive-code-cover}.
\end{proof}

\doclearpage
%%%%%%%%%%%%%%%%%%%%%%%%%%%%%%%%%%%%%%%%%%%%%%%%%%%%%%%%%%%%%%%%%%%%%%%%%%%%%%%%
%%%%%%%%%%%%%%%%%%%%%%%%%%%%%%%%%%%%%%%%%%%%%%%%%%%%%%%%%%%%%%%%%%%%%%%%%%%%%%%%
%%%%%%%%%%%%%%%%%%%%%%%%%%%%%%%%%%%%%%%%%%%%%%%%%%%%%%%%%%%%%%%%%%%%%%%%%%%%%%%%
\section{Folklore claim on interpolating sets}
\label{sec:proof-of-BSRS-recursive-cover-independence}

\begin{claim}
\label{claim:bsrs-rm-folklore}
Let $\Field$ be a field, let $\dCols, \dRows \in \Naturals$, and consider three sets $\SCols,\SRows \subseteq \Field$ and $\SPoints \subseteq \Field\times\Field$ such that $\SetCardinality{\SCols} + \SetCardinality{\SRows} + \SetCardinality{\SPoints} \leq \min\Set{\dCols, \dRows}$.
Let $f \colon \Field\times\Field \to \Field$ be a function such that:
\begin{itemize}

\item for every $\alpha \in \SCols$ there exists $\gCol{\alpha} \in \Field^{<\dRows}[x]$ such that $f(\alpha,\beta) = \gCol{\alpha}(\beta)$ for every $\beta \in \Field$;

\item for every $\beta \in \SRows$ there exists $\gRow{\beta} \in \Field^{<\dCols}[x]$ such that $f(\alpha,\beta) = \gRow{\beta}(\alpha)$ for every $\alpha \in \Field$.

\end{itemize}
Then there exists $\gExists \in \Field[X,Y]$ such that:
\begin{inparaenum}[(i)]
\item $\IndividualDegree{\gExists}[X] < \dRows$;
\item $\IndividualDegree{\gExists}[Y] < \dCols$;
\item $\Restrict{\gExists}{X=\alpha} = \gCol{\alpha}$ for every $\alpha \in \SCols$;
\item $\Restrict{\gExists}{Y=\beta} = \gRow{\beta}$ for every $\beta \in \SRows$;
\item $\gExists(\alpha,\beta) = f(\alpha,\beta)$ for every $(\alpha,\beta) \in \SPoints$.
\end{inparaenum}
\end{claim}

\begin{proof}
Any rectangle $\Domain_{X} \times \Domain_{Y} \subseteq \Field \times \Field$ with $\SetCardinality{\Domain_{X}} = \dRows$ and $\SetCardinality{\Domain_{Y}} = \dCols$
is an interpolating set: for every $\Codeword \in \Field^{\Domain_{X} \times \Domain_{Y}}$ there exists a unique $\gExists \in \Field[X,Y]$ such that:
\begin{inparaenum}[(i)]
\item $\IndividualDegree{\gExists}[X] < \dRows$;
\item $\IndividualDegree{\gExists}[Y] < \dCols$;
\item $\Restrict{\gExists}{\Domain_{X} \times \Domain_{Y}} = w$.
\end{inparaenum}
Define
\begin{equation*}
\Domain_{X} \DefineEqual \SCols \cup \Set{\alpha : \exists\, \beta \text{ s.t. } (\alpha,\beta) \in \SPoints}
\quad\text{and}\quad
\Domain_{Y} \DefineEqual \SRows \cup \Set{\beta : \exists\, \alpha \text{ s.t. } (\alpha,\beta) \in \SPoints}
\enspace.
\end{equation*}
Note that $\SetCardinality{\Domain_{X}} \leq \dRows$ and $\SetCardinality{\Domain_{Y}} \leq \dCols$; if either is strictly smaller, extend it arbitrarily to match the upper bound.

Choose $\Codeword \in \Field^{\Domain_{X} \times \Domain_{Y}}$ to be a word that satisfies:
\begin{inparaenum}[(i)]
  \item $\Codeword(\alpha,\beta) = f(\alpha,\beta)$ for every $\alpha \in \SCols$ and $\beta \in \Domain_{Y}$;
  \item $\Codeword(\alpha,\beta) = f(\alpha,\beta)$ for every $\beta \in \SRows$ and $\alpha \in \Domain_{X}$;
  \item $\Codeword(\alpha,\beta) = f(\alpha,\beta)$ for every $(\alpha,\beta) \in \SPoints$.
\end{inparaenum}
Denote by $\gExists_{\Codeword} \in \Field[X,Y]$ the unique ``low degree extension'' of $w$; we show that $\gExists_{\Codeword}$ satisfies the requirements of the claim.

The degree bounds and the equivalence on $\SPoints$ follows by definition of $\gExists_{\Codeword}$; thus it suffices to show equivalence of $\gExists_{\Codeword}$ with $f$ when restricted to the required rows and columns.
\begin{itemize}

  \item For every $\alpha \in \SCols$: it holds by definition of $\gExists_{\Codeword}$ that $\Restrict{\gExists_{\Codeword}}{\Set{\alpha}\times\Domain_{Y}} = \Restrict{f}{\Set{\alpha}\times\Domain_{Y}}$; moreover, $\Restrict{\gExists_{\Codeword}}{\Set{\alpha}\times\Field}$ and $\Restrict{f}{\Set{\alpha}\times\Field}$ are evaluations of polynomials of degree less than $\SetCardinality{\Domain_{Y}}$, which implies that $\Restrict{\gExists_{\Codeword}}{\Set{\alpha}\times\Field}=\Restrict{f}{\Set{\alpha}\times\Field}$.

  \item For every $\beta \in \SRows$: it holds by definition of $\gExists_{\Codeword}$ that $\Restrict{\gExists_{\Codeword}}{\Domain_{X}\times\Set{\beta}} = \Restrict{f}{\Domain_{X}\times\Set{\beta}}$; moreover, $\Restrict{\gExists_{\Codeword}}{\Field\times\Set{\beta}}$ and $\Restrict{f}{\Field\times\Set{\beta}}$ are evaluations of polynomials of degree less than $\SetCardinality{\Domain_{X}}$, which implies that $\Restrict{\gExists_{\Codeword}}{\Field\times\Set{\beta}}=\Restrict{f}{\Field\times\Set{\beta}}$. \qedhere

\end{itemize}
\end{proof}

\ignore[MV ZK]{
%%%%%%%%%%%%%%%%%%%%%%%%%%%%%%%%%%%%%%%%%%%%%%%%%%%%%%%%%%%%%%%%%%%%%%%%%%%%%%%%
%%%%%%%%%%%%%%%%%%%%%%%%%%%%%%%%%%%%%%%%%%%%%%%%%%%%%%%%%%%%%%%%%%%%%%%%%%%%%%%%
%%%%%%%%%%%%%%%%%%%%%%%%%%%%%%%%%%%%%%%%%%%%%%%%%%%%%%%%%%%%%%%%%%%%%%%%%%%%%%%%
\doclearpage
\section{\ale{deprecated} Reducing Zero-Knowledge sumCheck with oracles to Zero-Knowledge sumCheck}
We assume for simplicity that $\F$ has characteristic two,
though it doesn't seem too essential.

Suppose Prover $P$ wishes to show that a multivariate polynomial
$f$ of degree at most $D$
satisfies
\[\sumOnRange{f}=0\]
$f$ will typically be a composition of a
polynomial $Q$ of degree $e$ known to both $P$ and $V$
and a polynomial $A\in \polys{d}$ that only $P$ knows and wants to keep secret.
Specifically, $f$ will look something like
\[f(\x) = Q\circ A \DefineEqual  Q(A(N_{1}(\x)),\dots,A(N_q(\x)))
\]
for some fixed, efficiently computable, functions $N_{1},\dots,N_q:\sumrange \to \sumrange$
known to both $P$ and $V$.
Assume for simplicity $\deg{N_{i}}\leq 1$.

Let $d\DefineEqual |H|-1$.
Let \Z be the polynomial of degree $d\cdot m$, and individual degree $d$, that is zero on \sumrange.
$P$ computes $A'\DefineEqual A + \Z\cdot R$;
where $R$ is a random element of \polys{K}.
Note that $A'\in\polys{d'}$
for $d'\DefineEqual d+K$.
$P$ will send an oracle of \eval{A'}{\shiftedRange}
where  $H'\subseteq \F$ be a subset disjoint from $H$ of size $|H'|=2D'+K$.
Note that $(A')$'s values are $K$-wise independent outside of \sumrange.
Denote
\[f(\x) = Q\circ A' = Q(A'(N_{1}(\x)),\dots,A'(N_q(\x)))
\]
Denote $D'\DefineEqual D+K$.

Assume $H\subseteq \F$ is an $\F_2$-linear subspace with $|H|/|\F|\leq 1/10$, and $H'\DefineEqual \F\setminus H$.
Denote $S\DefineEqual H\cup H'$.
Note that $V$ is also an $\F_2$-subspace.
Assume the neighbors $\Set{N_{i}}$ are \emph{$H$-stable},
in the sense that
%\begin{itemize}
% \item for each $i\in [q]$ and $\x\in H^m$.
%$N_{i}(\x)\in H^m$,
%\item
for each $i\in [q]$ and $\x\in \F^m \setminus H^m$, $N_{i}(\x)\in H'^m$
%\end{itemize}

\begin{enumerate}
\item $P$
chooses random $g\in \polys{D'}$,with the constraint that $\sumOnRange{g}=0$.and
sends \eval{f}{\shiftedRange}
and \eval{g}{\shiftedRange}.
\item $V$ checks that $A'$ is $\eps$-close to \polys{d'} on $H'^m$.
It does this using less than $K$ queries to $A'$ with the help of RS-PCPPs
(Here looks like there will be some sub-protocol for this).

\item $V$ chooses random $\rho\in \F$.

\item $P$ send an oracle to $h\DefineEqual \rho\cdot f + g$ on $\F^m$.
\item $V$ and $P$ run the ZK version of the sumcheck protocol to check that
\[\sumOnRange{h}=0.\]
\item $V$ checks that $h$ is $\eps$-close to \polys{D'} both
on $H'^m$ and on $\F^m$.

\item Now $V$ performs a ``consistency check'':
He chooses random $\x \in H'^m$ and checks, using the oracles to $g,h$ and $A'$, that
\[h(\x) = \rho\cdot f(\x) +g(\x)\]
\end{enumerate}

\paragraph{Soundness:}
Suppose that $V$ accepts w.h.p..
Let $B$ be the unique element of $\polys{d'}$ $\eps$-close to $A'$.
Let $f'\DefineEqual Q\circ B$.
Suppose that $\sumOnRange{f'}\neq 0$.
Then for any fixed $g$ (even \emph{not} of the required degree)
$\sumOnRange{\rho\cdot f' + g}=\rho\cdot \sumOnRange{f'} + \sumOnRange{g}$ can be zero for at most one choice
of $\rho$, specifically, $\rho \DefineEqual \frac{\sumOnRange{g}}{\sumOnRange{f'}}$.
Thus, w.h.p., $\sumOnRange{\rho\cdot f' +g }\neq 0$.
Note that as $|f-f'|\leq \eps\cdot q$,
$|\rho\cdot f' +g-\rho\cdot f+ g|\leq \eps\cdot q$.
So, $\rho\cdot f'+g$ is $\eps\cdot q$-close to $\rho\cdot f +g$.

Given that the consistency check passed, $h$ is $\eps$-close
to $\rho\cdot f + g$ on $H'^m$,
and thus $\eps\cdot q$-close to $\rho\cdot f' + g$ on $H'^m$.
Given that the proximity test passed, $h$ is $\eps$-close to some $h'\in \polys{D'}$.
From proximity and consistency tests passing, together with the triangle inequality, we know that
\[\distOn{H'^m}{h'}{\rho\cdot f' + g}\leq \eps'\DefineEqual  \eps\cdot (q+2)\]
Denote $\eta = |H'^m|/|\F^m|=(1-|H|/|\F|)^{-m}$.
As $\distOn{H'^m}{h'}{\rho\cdot f' + g}\leq \eps' $,
for any $\rho\cdot f' +g \neq P\in \polys{D'}$,
$\distOn{\F^m}{h'}{P}\geq \eta\cdot (1-\eps' - D'/|H'|)$.

Assuming $\eps$ is smaller than this,
we must have $\distOn{\F^m}{h'}{\rho\cdot f'+ g}\leq \eps$ and thus $h'=\rho\cdot f' +g$, assuming $\eps$ is smaller than the unique distance $1- |D'|/|\F|$ of
\polys{D'} (as function $\F^m \to \F$.
Assuming the proximity testing of $h$ on $\F^m$ passed, $h$
is $\eps$-close to some $h''\in \polys{D'}$ on $\F^m$.
We assume that $\eps$ is small enough so that similar calculation to above shows
that $h'=h''$.
Thus, $\sumOnRange{h'} = \sumOnRange{\rho \cdot f' +g}\neq 0$,
and with high probability the sumcheck verifier will reject.

\paragraph{Zero-Knowledge up to $K/q$ queries:}
Simulation:
We assume for simplicity that $V$ queries the oracles \emph{only after sending $\rho$ and
receiving $g$ and $h$}
Without this assumption, simulation might be a bit more of a headache,
as in the final DZK paper.
\begin{itemize}
\item Using the needed components - a ZK sumcheck, and a ZK IOPP,
all parts of the view can be simulated by queries to $h,A'$ and $g$.
 \item To answer queries to $g$
 \item $A'$ is built on the fly, by answering new queries to $A'$ are randomly.
 \item New queries to $g$ are answered by first
 perform the following `update' procedure':
 Say $g(\x)$ is requested.
 First $A'$ is updated randomly, on all values
 on which $f(\x)$ depends, i.e., $A'(N_{1}(\x)),\dots,A'(N_q(\x))$,
 then $f(\x)$ is computed,
 then $h(\x)$ is sampled,
 and then $g(\x)$ is computed as $g(\x) = h(\x)-\rho\cdot f(\x)$.
\end{itemize}

}%ignore

%%%%%%%%%%%%%%%%%%%%%%%%%%%%%%%%%%%%%%%%%%%%%%%%%%%%%%%%%%%%%%%%%%%%%%%%%%%%%%%%
%%%%%%%%%%%%%%%%%%%%%%%%%%%%%%%%%%%%%%%%%%%%%%%%%%%%%%%%%%%%%%%%%%%%%%%%%%%%%%%%
%%%%%%%%%%%%%%%%%%%%%%%%%%%%%%%%%%%%%%%%%%%%%%%%%%%%%%%%%%%%%%%%%%%%%%%%%%%%%%%%
\clearpage
\bibliographystyle{alphaurl}
{\small\bibliography{references}}

\newcommand{\etalchar}[1]{$^{#1}$}
\begin{thebibliography}{BGKW88}

\bibitem[AH91]{AielloH91}
William Aiello and Johan H{\aa}stad.
\newblock Statistical zero-knowledge languages can be recognized in two rounds.
\newblock {\em Journal of Computer and System Sciences}, 42(3):327--345, 1991.
\newblock Preliminary version appeared in FOCS~'87.

\bibitem[ALM{\etalchar{+}}98]{AroraLMSS98}
Sanjeev Arora, Carsten Lund, Rajeev Motwani, Madhu Sudan, and Mario Szegedy.
\newblock Proof verification and the hardness of approximation problems.
\newblock {\em Journal of the ACM}, 45(3):501--555, 1998.
\newblock Preliminary version in FOCS~'92.

\bibitem[AS98]{AroraS98}
Sanjeev Arora and Shmuel Safra.
\newblock Probabilistic checking of proofs: a new characterization of {NP}.
\newblock {\em Journal of the ACM}, 45(1):70--122, 1998.
\newblock Preliminary version in FOCS~'92.

\bibitem[AS03]{AroraS03}
Sanjeev Arora and Madhu Sudan.
\newblock Improved low-degree testing and its applications.
\newblock {\em Combinatorica}, 23(3):365--426, 2003.
\newblock Preliminary version appeared in STOC~'97.

\bibitem[Bab85]{Babai85}
L\'{a}szl\'{o} Babai.
\newblock Trading group theory for randomness.
\newblock In {\em Proceedings of the 17th Annual ACM Symposium on Theory of
  Computing}, STOC~'85, pages 421--429, 1985.

\bibitem[BCG{\etalchar{+}}16]{BenSassonCGRS16}
Eli {Ben-Sasson}, Alessandro Chiesa, Ariel Gabizon, Michael Riabzev, and
  Nicholas Spooner.
\newblock Short interactive oracle proofs with constant query complexity, via
  composition and sumcheck, 2016.
\newblock Crypto ePrint 2016/324.

\bibitem[BCGV16]{BenSassonCGV16}
Eli {Ben-Sasson}, Alessandro Chiesa, Ariel Gabizon, and Madars Virza.
\newblock Quasilinear-size zero knowledge from linear-algebraic {PCP}s.
\newblock In {\em Proceedings of the 13th Theory of Cryptography Conference},
  TCC~'16-A, pages 33--64, 2016.

\bibitem[BCS16]{BenSassonCS16}
Eli {Ben-Sasson}, Alessandro Chiesa, and Nicholas Spooner.
\newblock Interactive oracle proofs.
\newblock In {\em Proceedings of the 14th Theory of Cryptography Conference},
  TCC~'16-B, pages ??--??, 2016.

\bibitem[BFL91]{BabaiFL91}
L{\'{a}}szl{\'{o}} Babai, Lance Fortnow, and Carsten Lund.
\newblock Non-deterministic exponential time has two-prover interactive
  protocols.
\newblock {\em Computational Complexity}, 1:3--40, 1991.
\newblock Preliminary version appeared in FOCS~'90.

\bibitem[BFLS91]{BFLS91}
L\'{a}szl\'{o} Babai, Lance Fortnow, Leonid~A. Levin, and Mario Szegedy.
\newblock Checking computations in polylogarithmic time.
\newblock In {\em Proceedings of the 23rd Annual ACM Symposium on Theory of
  Computing}, STOC~'91, pages 21--32, 1991.

\bibitem[BGG{\etalchar{+}}88]{BenOrGGHKMR88}
Michael {Ben-Or}, Oded Goldreich, Shafi Goldwasser, Johan H{\aa}stad, Joe
  Kilian, Silvio Micali, and Phillip Rogaway.
\newblock Everything provable is provable in zero-knowledge.
\newblock In {\em Proceedings of the 8th Annual International Cryptology
  Conference}, CRYPTO~'89, pages 37--56, 1988.

\bibitem[BGH{\etalchar{+}}06]{BenSassonGHSV06}
Eli {Ben-Sasson}, Oded Goldreich, Prahladh Harsha, Madhu Sudan, and Salil~P.
  Vadhan.
\newblock Robust {PCP}s of proximity, shorter {PCP}s, and applications to
  coding.
\newblock {\em SIAM Journal on Computing}, 36(4):889--974, 2006.

\bibitem[BGK{\etalchar{+}}10]{BenSassonGKSV10}
Eli {Ben-Sasson}, Venkatesan Guruswami, Tali Kaufman, Madhu Sudan, and Michael
  Viderman.
\newblock Locally testable codes require redundant testers.
\newblock {\em SIAM Journal on Computing}, 39(7):3230--3247, 2010.

\bibitem[BGKW88]{BenOrGKW88}
Michael {Ben-Or}, Shafi Goldwasser, Joe Kilian, and Avi Wigderson.
\newblock Multi-prover interactive proofs: how to remove intractability
  assumptions.
\newblock In {\em Proceedings of the 20th Annual ACM Symposium on Theory of
  Computing}, STOC~'88, pages 113--131, 1988.

\bibitem[BHR05]{BenSassonHR05}
Eli {Ben-Sasson}, Prahladh Harsha, and Sofya Raskhodnikova.
\newblock Some {3CNF} properties are hard to test.
\newblock {\em SIAM Journal on Computing}, 35(1):1--21, 2005.

\bibitem[BHZ87]{BoppanaHZ87}
Ravi~B. Boppana, Johan H\r{a}stad, and Stathis Zachos.
\newblock Does co-{NP} have short interactive proofs?
\newblock {\em Information Processing Letters}, 25(2):127--132, 1987.

\bibitem[BM88]{BabaiM88}
L{\'{a}}szl{\'{o}} Babai and Shlomo Moran.
\newblock Arthur-merlin games: {A} randomized proof system, and a hierarchy of
  complexity classes.
\newblock {\em Journal of Computer and System Sciences}, 36(2):254--276, 1988.

\bibitem[BS08]{BS08}
Eli {Ben-Sasson} and Madhu Sudan.
\newblock Short {PCP}s with polylog query complexity.
\newblock {\em SIAM Journal on Computing}, 38(2):551--607, 2008.
\newblock Preliminary version appeared in STOC~'05.

\bibitem[BVW98]{BookVW98}
Ronald~V. Book, Heribert Vollmer, and Klaus~W. Wagner.
\newblock Probabilistic type-2 operators and ``almost''-classes.
\newblock {\em Computational Complexity}, 7(3):265--289, 1998.

\bibitem[BW04]{BogdanovW04}
Andrej Bogdanov and Hoeteck Wee.
\newblock A stateful implementation of a random function supporting parity
  queries over hypercubes.
\newblock In {\em Proceedings of the 7th International Workshop on
  Approximation Algorithms for Combinatorial Optimization Problems, and of the
  8th International Workshop on Randomization and Computation},
  APPROX-RANDOM~'04, pages 298--309, 2004.

\bibitem[DFK{\etalchar{+}}92]{DworkFKNS92}
Cynthia Dwork, Uriel Feige, Joe Kilian, Moni Naor, and Shmuel Safra.
\newblock Low communication 2-prover zero-knowledge proofs for {NP}.
\newblock In {\em Proceedings of the 11th Annual International Cryptology
  Conference}, CRYPTO~'92, pages 215--227, 1992.

\bibitem[DR04]{DinurR04}
Irit Dinur and Omer Reingold.
\newblock Assignment testers: Towards a combinatorial proof of the {PCP}
  theorem.
\newblock In {\em Proceedings of the 45th Annual IEEE Symposium on Foundations
  of Computer Science}, FOCS~'04, pages 155--164, 2004.

\bibitem[DS98]{DworkS98}
Cynthia Dwork and Amit Sahai.
\newblock Concurrent zero-knowledge: Reducing the need for timing constraints.
\newblock In {\em Proceedings of the 18th Annual International Cryptology
  Conference}, CRYPTO~'98, pages 442--457, 1998.

\bibitem[FGL{\etalchar{+}}96]{FGLSS96}
Uriel Feige, Shafi Goldwasser, Laszlo Lov\'{a}sz, Shmuel Safra, and Mario
  Szegedy.
\newblock Interactive proofs and the hardness of approximating cliques.
\newblock {\em Journal of the ACM}, 43(2):268--292, 1996.
\newblock Preliminary version in FOCS~'91.

\bibitem[For87]{Fortnow87}
Lance Fortnow.
\newblock The complexity of perfect zero-knowledge (extended abstract).
\newblock In {\em Proceedings of the 19th Annual ACM Symposium on Theory of
  Computing}, STOC~'87, pages 204--209, 1987.

\bibitem[FRS88]{FortnowRS88}
Lance Fortnow, John Rompel, and Michael Sipser.
\newblock On the power of multi-prover interactive protocols.
\newblock In {\em Theoretical Computer Science}, pages 156--161, 1988.

\bibitem[FS89]{FeigeS89}
Uriel Feige and Adi Shamir.
\newblock Zero knowledge proofs of knowledge in two rounds.
\newblock In {\em Proceedings of the 9th Annual International Cryptology
  Conference}, CRYPTO~'89, pages 526--544, 1989.

\bibitem[GGN10]{GoldreichGN10}
Oded Goldreich, Shafi Goldwasser, and Asaf Nussboim.
\newblock On the implementation of huge random objects.
\newblock {\em SIAM Journal on Computing}, 39(7):2761--2822, 2010.
\newblock Preliminary version appeared in FOCS~'03.

\bibitem[GIMS10]{GoyalIMS10}
Vipul Goyal, Yuval Ishai, Mohammad Mahmoody, and Amit Sahai.
\newblock Interactive locking, zero-knowledge {PCP}s, and unconditional
  cryptography.
\newblock In {\em Proceedings of the 30th Annual Conference on Advances in
  Cryptology}, CRYPTO'10, pages 173--190, 2010.

\bibitem[GKR08]{GKR08}
Shafi Goldwasser, Yael~Tauman Kalai, and Guy~N. Rothblum.
\newblock Delegating computation: Interactive proofs for {M}uggles.
\newblock In {\em Proceedings of the 40th Annual ACM Symposium on Theory of
  Computing}, STOC~'08, pages 113--122, 2008.

\bibitem[GMR89]{GoldwasserMR89}
Shafi Goldwasser, Silvio Micali, and Charles Rackoff.
\newblock The knowledge complexity of interactive proof systems.
\newblock {\em SIAM Journal on Computing}, 18(1):186--208, 1989.
\newblock Preliminary version appeared in STOC~'85.

\bibitem[GO94]{GoldreichO94}
Oded Goldreich and Yair Oren.
\newblock Definitions and properties of zero-knowledge proof systems.
\newblock {\em Journal of Cryptology}, 7(1):1--32, December 1994.

\bibitem[GR15]{GurR15}
Tom Gur and Ron~D. Rothblum.
\newblock Non-interactive proofs of proximity.
\newblock In {\em Proceedings of the 6th Innovations in Theoretical Computer
  Science Conference}, ITCS~'15, pages 133--142, 2015.

\bibitem[GS06]{GoldreichS06}
Oded Goldreich and Madhu Sudan.
\newblock Locally testable codes and {PCP}s of almost-linear length.
\newblock {\em Journal of the ACM}, 53:558--655, July 2006.
\newblock Preliminary version in STOC~'02.

\bibitem[IMS12]{IshaiMS12}
Yuval Ishai, Mohammad Mahmoody, and Amit Sahai.
\newblock On efficient zero-knowledge {PCP}s.
\newblock In {\em Proceedings of the 9th Theory of Cryptography Conference on
  Theory of Cryptography}, TCC~'12, pages 151--168, 2012.

\bibitem[IMSX15]{IshaiMSX15}
Yuval Ishai, Mohammad Mahmoody, Amit Sahai, and David Xiao.
\newblock On zero-knowledge {PCP}s: Limitations, simplifications, and
  applications, 2015.
\newblock Available at
  \url{http://www.cs.virginia.edu/~mohammad/files/papers/ZKPCPs-Full.pdf}.

\bibitem[IW14]{IshaiW14}
Yuval Ishai and Mor Weiss.
\newblock Probabilistically checkable proofs of proximity with zero-knowledge.
\newblock In {\em Proceedings of the 11th Theory of Cryptography Conference},
  TCC~'14, pages 121--145, 2014.

\bibitem[IWY16]{IshaiWY16}
Yuval Ishai, Mor Weiss, and Guang Yang.
\newblock Making the best of a leaky situation: Zero-knowledge {PCP}s from
  leakage-resilient circuits.
\newblock In {\em Proceedings of the 13th Theory of Cryptography Conference},
  TCC~'16-A, pages 3--32, 2016.

\bibitem[IY87]{ImpagliazzoY87}
Russell Impagliazzo and Moti Yung.
\newblock Direct minimum-knowledge computations.
\newblock In {\em Proceedings of the 7th Annual International Cryptology
  Conference}, CRYPTO~'87, pages 40--51, 1987.

\bibitem[Kay10]{Kayal10}
Neeraj Kayal.
\newblock Algorithms for arithmetic circuits, 2010.
\newblock ECCC TR10-073.

\bibitem[KI04]{KabanetsI04}
Valentine Kabanets and Russell Impagliazzo.
\newblock Derandomizing polynomial identity tests means proving circuit lower
  bounds.
\newblock {\em Computational Complexity}, 13(1-2):1--46, 2004.

\bibitem[KLR10]{KushilevitzLR10}
Eyal Kushilevitz, Yehuda Lindell, and Tal Rabin.
\newblock Information-theoretically secure protocols and security under
  composition.
\newblock {\em SIAM Journal on Computing}, 39(5):2090--2112, 2010.
\newblock Preliminary version appeared in STOC~'06.

\bibitem[KPT97]{KilianPT97}
Joe Kilian, Erez Petrank, and G\'{a}bor Tardos.
\newblock Probabilistically checkable proofs with zero knowledge.
\newblock In {\em Proceedings of the 29th Annual ACM Symposium on Theory of
  Computing}, STOC~'97, pages 496--505, 1997.

\bibitem[KR08]{KalaiR08}
Yael Kalai and Ran Raz.
\newblock Interactive {PCP}.
\newblock In {\em Proceedings of the 35th International Colloquium on Automata,
  Languages and Programming}, ICALP~'08, pages 536--547, 2008.

\bibitem[LFKN92]{LundFKN92}
Carsten Lund, Lance Fortnow, Howard~J. Karloff, and Noam Nisan.
\newblock Algebraic methods for interactive proof systems.
\newblock {\em Journal of the ACM}, 39(4):859--868, 1992.

\bibitem[LS95]{LapidotS95}
Dror Lapidot and Adi Shamir.
\newblock A one-round, two-prover, zero-knowledge protocol for {NP}.
\newblock {\em Combinatorica}, 15(2):204--214, 1995.

\bibitem[MX13]{MahmoodyX13}
Mohammad Mahmoody and David Xiao.
\newblock Languages with efficient zero-knowledge {PCP}s are in {SZK}.
\newblock In {\em Proceedings of the 10th Theory of Cryptography Conference},
  TCC~'13, pages 297--314, 2013.

\bibitem[Nao91]{Naor91}
Moni Naor.
\newblock Bit commitment using pseudorandomness.
\newblock {\em Journal of Cryptology}, 4(2):151--158, 1991.
\newblock Preliminary version appeared in CRYPTO~'89.

\bibitem[Nis93]{Nisan93}
Noam Nisan.
\newblock On read-once vs. multiple access to randomness in logspace.
\newblock {\em Theoretical Computer Science}, 107(1):135--144, 1993.

\bibitem[Ost91]{Ostrovsky91}
Rafail Ostrovsky.
\newblock One-way functions, hard on average problems, and statistical
  zero-knowledge proofs.
\newblock In {\em Proceedings of the 6th Annual Structure in Complexity Theory
  Conference}, CoCo~'91, pages 133--138, 1991.

\bibitem[OW93]{OstrovskyW93}
Rafail Ostrovsky and Avi Wigderson.
\newblock One-way functions are essential for non-trivial zero-knowledge.
\newblock In {\em Proceedings of the 2nd Israel Symposium on Theory of
  Computing Systems}, ISTCS~'93, pages 3--17, 1993.

\bibitem[RRR16]{ReingoldRR16}
Omer Reingold, Ron Rothblum, and Guy Rothblum.
\newblock Constant-round interactive proofs for delegating computation.
\newblock In {\em Proceedings of the 48th ACM Symposium on the Theory of
  Computing}, STOC~'16, pages 49--62, 2016.

\bibitem[RS96]{RubinfeldS96}
Ronitt Rubinfeld and Madhu Sudan.
\newblock Robust characterizations of polynomials with applications to program
  testing.
\newblock {\em SIAM Journal on Computing}, 25(2):252--271, 1996.

\bibitem[RS05]{RazS05}
Ran Raz and Amir Shpilka.
\newblock Deterministic polynomial identity testing in non-commutative models.
\newblock {\em Computational Complexity}, 14(1):1--19, 2005.
\newblock Preliminary version appeared in CCC~'04.

\bibitem[Sch80]{Schwartz80}
Jacob~T. Schwartz.
\newblock Fast probabilistic algorithms for verification of polynomial
  identities.
\newblock {\em Journal of the ACM}, 27(4):701--717, 1980.

\bibitem[Sha92]{Shamir92}
Adi Shamir.
\newblock {IP} = {PSPACE}.
\newblock {\em Journal of the ACM}, 39(4):869--877, 1992.

\bibitem[SY10]{ShpilkaY10}
Amir Shpilka and Amir Yehudayoff.
\newblock Arithmetic circuits: {A} survey of recent results and open questions.
\newblock {\em Foundations and Trends in Theoretical Computer Science},
  5(3-4):207--388, 2010.

\bibitem[Zip79]{Zippel79}
Richard Zippel.
\newblock Probabilistic algorithms for sparse polynomials.
\newblock In {\em Proceedings of the 1979 International Symposium on Symbolic
  and Algebraic Computation}, EUROSAM~'79, pages 216--226, 1979.

\end{thebibliography}
%%%%%%%%%%%%%%%%%%%%%%%%%%%%%%%%%%%%%%%%%%%%%%%%%%%%%%%%%%%%%%%%%%%%%%%%%%%%%%%%
%%%%%%%%%%%%%%%%%%%%%%%%%%%%%%%%%%%%%%%%%%%%%%%%%%%%%%%%%%%%%%%%%%%%%%%%%%%%%%%%
%%%%%%%%%%%%%%%%%%%%%%%%%%%%%%%%%%%%%%%%%%%%%%%%%%%%%%%%%%%%%%%%%%%%%%%%%%%%%%%%
\end{document}